\documentclass[11pt]{article}

\usepackage{fullpage}
\usepackage{color}
\usepackage{graphicx}
\usepackage{epsfig}
\usepackage{amsthm}
\usepackage{latexsym}
\usepackage{amssymb}
\usepackage{amsmath,mathrsfs}
\usepackage[hypertexnames=false]{hyperref}
\usepackage{euscript}
\usepackage{nicematrix}
\usepackage{braket}
\usepackage{todonotes}
\usepackage{zzt}

\newcommand{\holodd}[0]{\textup{\text{Holant}}^\text{odd}}
\newcommand{\eoe}[0]{\text{HW}^=}
\newcommand{\eog}[0]{\text{HW}^\geq}
\newcommand{\eol}[0]{\text{HW}^\leq}
\newcommand{\eosg}[0]{\text{HW}^>}
\newcommand{\eosl}[0]{\text{HW}^<}

\newcommand{\eo}[0]{\textup{\textsf{EO}}}
\newcommand{\ceo}[0]{\textup{\#\textsf{EO}}}

\newcommand{\POLMID}{\mathsf{Pol}(\oplus_3)}
\newcommand{\POLUP}{\mathsf{Pol}^\uparrow(\oplus_3)}
\newcommand{\POLDOWN}{\mathsf{Pol}^\downarrow(\oplus_3)}
\newcommand{\hF}{\widehat{\mathcal{F}}}
\newcommand{\F}{\mathcal{F}}

\title{\bf A Dichotomy for Boolean Complex Holant Problems with Conjugate-Closed Signature Sets}

\author{
 \begin{tabular}{c@{\hspace{0.5cm}}c@{\hspace{0.5cm}}c}
 \normalsize Jincheng Guan\thanks{School of Computer Science and Technology \& Hefei National Laboratory, University of Science and Technology of China. 
 } &
\normalsize  Shuai Shao\footnotemark[1] &
\normalsize  Zhuxiao Tang\thanks{College of Computing \& Artificial Intelligence (CAI),
University of Wisconsin-Madison.} \\
 \normalsize  \url{guanjincheng@mail.ustc.edu.cn} & \normalsize \url{shao10@ustc.edu.cn}  & \normalsize \url{zztang@wisc.edu}\\
 \end{tabular}
}

\date{}

\begin{document}
\maketitle

\begin{abstract}
 We study Boolean Holant problems with complex-valued signature sets closed under conjugation. Such sets arise naturally in tensor-network expressions for classical strong simulation of quantum circuits. We prove a complexity dichotomy for such problems with an explicit tractability criterion. This extends the dichotomy for real-valued Holant problems, with the same four tractability conditions.

Our proofs use Xia's projective binary group framework and quantum entanglement theory.
The conjugate closure assumption precisely makes $k$-uniformity, directly applicable
to the classification of Holant problems, by realizing reduced density matrices via Holant gadgets. 
We also use the classification of absolutely maximally entangled states to resolve a particular $6$-ary obstruction in our
inductive proof of the \#P-hardness. 

\end{abstract}

\section*{AI disclosure}

We used ChatGPT 5.6 sol to help develop the arguments in Sections~\ref{sec: second order orthoganality} to~\ref{sec: 2n >= 10} and in \autoref{sec: AME62 has a common local unitary normal}.
The manuscript was verified and written by the authors, with the aim of making the proofs accessible to readers.
The authors take full responsibility for the correctness and originality of all content in this manuscript.
\newpage
\tableofcontents
\newpage

\section{Introduction}

It is an important problem in quantum computation to understand which quantum circuits admit efficient classical simulation. One task is to compute specified output probabilities and their marginals exactly, known as \emph{exact strong simulation}. The complexity of this task depends on the allowed gates, input states, and measurements. Several classification results~\cite{jozsa-vandenNest2014,koh2017extensions,hebenstreit2020matchgates,bu-koh2022halfgauss} describe how these choices affect the complexity of classical simulation. For instance,  non-adaptive Clifford circuits with computational basis inputs and measurements admit efficient strong simulation, while allowing arbitrary product-state inputs makes the problem \#P-hard~\cite{jozsa-vandenNest2014}. Beyond strong simulation, complexity dichotomies for sampling have been established for circuits generated by two-qubit commuting Hamiltonians~\cite{bouland2016commuting} and for conjugated Clifford circuits~\cite{bouland2018conjugated}.

Classical strong simulation can be naturally studied through tensor networks and Holant problems. The connection between restricted quantum computation and graph-based counting already appears in Valiant's classical simulation of matchgate circuits~\cite{valiant2002quantum}. More directly, Cai, Guo, and Williams showed that Clifford gates on qubits are precisely the Boolean affine signatures with unitary matrix representations~\cite{cai-guo-williams-clifford}. This identifies a class of  quantum circuits admitting efficient classical simulations with a tractable class in the Holant framework. This tractability result was further generalized to affine signatures over arbitrary finite domains~\cite{bu-koh2022halfgauss}.

To represent quantum circuits using tensor networks or the Holant framework, an $n$-qubit gate $U$ can be viewed as a complex-valued Boolean function (or signature) of $2n$ variables, with entries
$g_U(x,y)=\bra{y}U\ket{x}$ for $
    x,y\in\{0,1\}^n.$
Connecting gates corresponds to contracting the associated tensors. To compute an output probability, one contracts the network for the amplitude together with its complex conjugate. Thus, gate tensors and their entrywise complex conjugates occur naturally in the same computation. Together with the tensors specifying inputs and measurements and their conjugates, they form a conjugate-closed set of signatures. These computations are therefore naturally represented by Holant problems with conjugate-closed sets of complex-valued signatures. Motivated by this connection, we study the complexity of such Holant problems and prove that every such problem is either polynomial-time computable or \#P-hard, with an explicit tractability criterion. Our dichotomy result covers arbitrary signature grids, which generalize the tensor networks arising from quantum circuits.

We briefly recall the definition of Holant problems. A constraint function, also called a \emph{signature}, is a function $f:\{0,1\}^n\to\mathbb{C}$ for some $n>0$, called its \emph{arity}. 
Let $\mathcal{F}$ be a fixed set of (algebraic) complex valued signatures. A \emph{signature grid} over $\mathcal{F}$ is a pair $\Omega=(G,\pi)$, where $G=(V,E)$ is a graph without isolated vertices. The labeling $\pi$ assigns to each vertex $v\in V$ a signature $f_v\in\mathcal{F}$ of arity $\deg(v)$, together with a correspondence between the incident edges $E(v)$ and the input variables of $f_v$.

\begin{definition}[Holant problems]
The input to $\Holant(\mathcal{F})$ is a signature grid $\Omega=(G,\pi)$ over $\mathcal{F}$. The output is the partition function
\[
    \Holant(\Omega)
    =\sum_{\sigma:E\to\{0,1\}}
      \prod_{v\in V}f_v\bigl(\sigma|_{E(v)}\bigr),
\]
where $\sigma|_{E(v)}$ denotes the restriction of $\sigma$ to the incident edges at $v$, in the order specified by $\pi$.
\end{definition}

For a signature $f$, its conjugate $\overline f$ is defined by $\overline f(x)=\overline{f(x)}$ for every input $x$. A signature set $\mathcal{F}$ is \emph{conjugate closed} if for every $f\in \mathcal{F}$ we have $\overline{f}\in \mathcal{F}$.
In particular, every set of real-valued signatures is conjugate closed. Our dichotomy therefore naturally extends the dichotomy for real-valued Holant problems~\cite{realholant}. Moreover, exactly the same four tractability conditions remain exhaustive: every signature is a tensor product of unary and binary signatures, or the entire signature set can be transformed holographically into the product-type, affine, or local-affine class. Their precise definitions are given in \cref{subsec: tractable signatures}. Although no new tractable cases arise, the hardness proof must handle complex-valued signatures and local unitary transformations not covered by the real-valued argument.

Conjugate closure is not only natural for the study of classical simulation of quantum circuits. It also makes quantum entanglement theory, particularly $k$-uniformity, directly applicable to the classification of Holant problems. Entanglement theory has previously been used to obtain Holant dichotomies~\cite{Backens-holant-plus,Backens-Holant-c,odd_holant_real}. Our result develops this connection further by using $k$-uniformity and the uniqueness, up to local unitaries, of absolutely maximally entangled states.
A nonzero signature $f$ of arity $n$ is interchangeable with a normalized $n$-qubit state
$\ket f=\frac{1}{\lVert f\rVert}
    \sum_{x\in\{0,1\}^n}f(x)\ket x,$ and its norm is $ \lVert f\rVert^2=\sum_x|f(x)|^2$.
A pure state is \emph{$k$-uniform}~\cite{scott2004} if the reduced density matrix on every set of $k$ qubits is maximally mixed. 
An \emph{absolutely maximally entangled} state, denoted by $\operatorname{AME}(n,2)$ is an $n$-qubit state that is $\lfloor n/2\rfloor$-uniform. Such states are maximally entangled across every bipartition and are closely related to quantum error-correcting codes and quantum secret sharing~\cite{scott2004,helwig2012}.

Conjugate closure makes these quantum notions accessible through Holant gadgets. For $S\subseteq[n]$, let $M_S(f)$ be the matrix whose rows are indexed by assignments to $S$ and whose columns are indexed by assignments to its complement. Connecting $f$ and $\overline f$ along the complementary variables realizes a signature with matrix $M_S(f)M_S(f)^\dagger$. Dividing this matrix by the norm $\lVert f\rVert^2$ gives exactly the reduced density matrix of $\ket f$ on $S$. Consequently, $\ket f$ is $k$-uniform precisely when
\[
    M_S(f)M_S(f)^\dagger
    =\frac{\lVert f\rVert^2}{2^k}I_{2^k}
    \qquad\text{for every }S\subseteq[n]\text{ with }|S|=k,
\]
where $I_{2^k}$ is the identity matrix of order $2^k$.
We call this condition \emph{$k$-th order orthogonality}, see~\cref{subsec: signatures and quantum states} for detailed definitions. Thus, properties of reduced density matrices translate directly into restrictions on  signatures realizable via the above \emph{mating gadgets}.

A key application of $k$-uniformity and AME states is to resolve an $6$-ary obstruction in our inductive proof. 
We first show that a $6$-ary signature satisfying the remaining closure conditions must satisfy third order orthogonality, and hence define an $\operatorname{AME}(6,2)$ state. 
We then invoke the uniqueness of the six-qubit AME state that every $\operatorname{AME}(6,2)$ state is locally unitarily equivalent, up to a permutation of the qubits, to the quantum hexacode state~\cite{rains1999quantum}. This reduces all remaining $6$-ary signatures to a single local-unitary normal form. The local unitaries in this form need not themselves be available as Holant gadgets. We overcome this difficulty through further gadget constructions and complexity reductions, converting the entanglement classification result into the required hardness proof for Holant complexity classification.

Beyond its connection with classical simulation of quantum circuits and quantum entanglement theory, our theorem advances the general complexity classification of Boolean Holant problems. 
Significant progress has been made in this program. A complete dichotomy for complex-valued symmetric signatures was proved in \cite{cai-guo-tyson-vanishing}. For asymmetric signatures, dichotomies were established under the availability of auxiliary unary signatures, including $\Holant^\ast$, $\Holant^+$, and $\Holant^c$~\cite{cai-lu-xia-holant-star,Backens-holant-plus,real_holantc,Backens-Holant-c}. Without such assumptions, a dichotomy for nonnegative real-valued signatures was proved in \cite{Lin_Wang_Decomposition_lemma}, and a dichotomy for real-valued signatures when a nonzero odd-arity signature is present was proved in \cite{odd_holant_real}. These results were followed by the full real-valued dichotomy~\cite{realholant}.

Recently, a particularly interesting framework known as counting weighted Eulerian orientations (\#EO) has been introduced and widely studied~\cite{cai-fu-shao-eo,ShaoTang25,MengWX25}. This framework has been recognized as an important step towards a full dichotomy for complex-valued Holant problems~\cite{Holant_Odd,GSS26}.
More recently, an $\mathrm{FP}^{\mathrm{NP}}$ versus \#P-hard classification for weighted \#EO problems was proved~\cite{MengWX25}. Based on this result, an $\mathrm{FP}^{\mathrm{NP}}$ versus \#P-hard classification was also proved for complex-valued Holant problems with a nonzero odd-arity signature~\cite{Holant_Odd}. Subsequently,  for all cases on the $\mathrm{FP}^{\mathrm{NP}}$ side,  polynomial-time algorithms were obtained~\cite{GSS26}. This result completes both $\mathrm{FP}^{\mathrm{NP}}$ versus \#P-hard  classification results as standard $\mathrm{FP}$ versus \#P-hard dichotomies. In addition, a projective-group framework~\cite{MingjiProgram} was proposed for the complete complexity classification of complex-valued Holant problems.
Building on these results, our theorem gives a full classification for complex-valued Holant problems under conjugate closure, including the even-arity cases where the remarkable entanglement structure arises.

\section{Preliminaries}

The symbols $\le_T$ and
$\equiv_T$ denote polynomial-time Turing reducibility and equivalence.

We use $\mathfrak i$ to denote $\sqrt{-1}$.
We use the matrices
\[
I_2=\begin{bmatrix}1&0\\0&1\end{bmatrix},\quad
N_2=\begin{bmatrix}0&1\\1&0\end{bmatrix},\quad
H=\frac1{\sqrt2}\begin{bmatrix}1&1\\1&-1\end{bmatrix},\quad
K=\frac1{\sqrt2}\begin{bmatrix}1&1\\
\mathfrak i&-\mathfrak i\end{bmatrix}.
\]
The notation $\mathbb{C}^\times$ means $\mathbb{C}\setminus\{0\}$.
For a complex valued matrix $A$, write $A^{\mathtt T}$ for transpose,
$\overline A\,$ for entrywise complex conjugation, and
$A^\dagger=\overline A^{\mathtt T}$ for conjugate transpose.  For row
vectors $\mathbf u,\mathbf v\in\mathbb C^m$, our convention is
\[
\langle\mathbf u,\mathbf v\rangle
=\mathbf u\,\overline{\mathbf v}^{\mathtt T},\qquad
\|\mathbf u\|=\sqrt{\langle\mathbf u,\mathbf u\rangle}.
\]
We also use $|\mathbf u|$ for this norm when no confusion with scalar
absolute value can arise.

\subsection{Definitions and Notations}

A \emph{signature} of arity $n$ is a function
$f:\{0,1\}^n\to\mathbb C$. 
For a bit string $\alpha=\alpha_1\cdots\alpha_n$, let $|\alpha|=n$,
\(
\operatorname{wt}(\alpha)=\sum_{j=1}^n\alpha_j,
\)
and denote its bitwise complement by $\overline{\alpha}$.  
The symbols
$0^n=\vec 0^{\,n}$ and $1^n=\vec 1^{\,n}$ denote the all-zero and all-one
strings; the superscript is omitted when its value is clear.  Addition
$\oplus$ of bit strings is coordinatewise over $\mathbb Z_2$.
The \emph{support} of $f$ is
\(
\mathscr S(f)=\{\alpha\in\{0,1\}^n:f(\alpha)\ne0\}.
\)
The signature is zero, written $f\equiv0$, if $\mathscr S(f)=\emptyset$.
A signature is \emph{symmetric} if its value depends only on Hamming weight.
For such a signature of arity $n$, we use the compressed notation
\(
f=[f_0,f_1,\ldots,f_n],
\)
where $f_j$ is the value on every input of weight $j$. 
Let
\(
\Delta_0=[1,0], \,\Delta_1=[0,1].
\)
For $n\ge1$, the equality signature $(=_n)$ is one on $0^n$ and $1^n$ and
zero elsewhere.  The binary disequality signature is
$(\neq_2)=[0,1,0]$. 
More generally, $\neq_{2n}$ denotes the indicator of
\(
x_1= x_2=\ldots= x_n\neq
x_{n+1}=x_{n+2}=\ldots= x_{2n}.
\)
We
write $\mathcal{EQ}=\{=_n\mid n\ge1\}$ and
$\mathcal{DEQ}=\{\neq_{2n}\mid n\ge1\}$.
Also write $\mathcal{EQ}_k=\{=_{nk}\mid n\ge1\}$.
Let
\(
\mathscr E_n=\{\alpha\in\{0,1\}^n:\operatorname{wt}(\alpha)
\text{ is even}\},\,
\mathscr O_n=\{\alpha\in\{0,1\}^n:\operatorname{wt}(\alpha)
\text{ is odd}\}.
\)
A signature has \emph{even parity} or \emph{odd parity} if its support is
contained in $\mathscr E_n$ or $\mathscr O_n$, respectively; it has
\emph{parity} in either case.

We use $=_2^-$ to denote the binary signature $(1, 0, 0, -1)$ and $\neq_2^-$ to denote the binary signature $(0, 1, -1, 0)$.
We may also write $=_2$ as $=_2^+$ and $\neq_2$ as $\neq_2^+$. 
Let $\mathcal{B}=\{=^+_2, =_2^-, \neq_2^+, \neq_2^-\}$.
We call them Bell signatures which correspond to Bell states $|\Phi^+\rangle=|00\rangle+|11\rangle$, $|\Phi^-\rangle=|00\rangle-|11\rangle$, 
$|\Psi^+\rangle=|01\rangle+|10\rangle$ and 
$|\Psi^-\rangle=|01\rangle-|10\rangle$ in quantum information science \cite{bell1964einstein}.

We use $f_{i_1i_2\ldots i_k}^{a_1a_2\ldots a_k}$ to denote the signature obtained from $f$ by pinning the variables $x_{i_1},x_{i_2},\ldots, x_{i_k}$ to $a_{1},a_2,\ldots, a_k$ respectively.
For example, $f_{i}^a=f(x_1,\ldots,x_{i-1},a,x_{i+1},\ldots,x_n).$
For $S\subseteq[n]$, let $M_S(f)$ be the
$2^{|S|}\times2^{n-|S|}$ matrix whose rows are indexed by assignments
to the variables in $S$ and whose columns are indexed by assignments to
the remaining variables, both in lexicographic order.  We also write
$M_{S,T}(f)$ when the ordered row and column variable lists $S,T$
partition the variables.  In particular,
\[
M_i(f)=
\begin{bmatrix}\mathbf f_i^0\\ \mathbf f_i^1\end{bmatrix},
\qquad
M_{ij}(f)=
\begin{bmatrix}
\mathbf f_{ij}^{00}\\ \mathbf f_{ij}^{01}\\
\mathbf f_{ij}^{10}\\ \mathbf f_{ij}^{11}
\end{bmatrix},
\]\
where $ {\bf {f}}^{a}_{i}$ denotes the row vector indexed by $x_i=a$ in $M_{i}(f)$, and ${\bf {f}}^{ab}_{ij}$ denotes the row vector indexed by $(x_i, x_j)=(a, b)$ in $M_{i j}(f)$.
For a binary signature $b=b(x_1,x_2)$, $M(b)=\left[\begin{smallmatrix}
    b(0,0) & b(0,1)\\
    b(1,0) & b(1,1)
\end{smallmatrix}\right]$ denotes its $2\times2$ signature matrix.

A \emph{signature grid} over $\mathcal F$ is a pair
$\Omega=(G,\pi)$, where $G=(V,E)$ is a graph and $\pi$ assigns to every
$v\in V$ a signature $f_v\in\mathcal F$ of arity $\deg(v)$, and labels the incident edges $E(v)$ at $v$ with input variables of $f_v$.
We consider all 0-1 edge assignments $\sigma$, and each gives an evaluation $\prod_{v\in V}f_v(\sigma|_{E(v)})$, where $\sigma|_{E(v)}$ denotes the restriction of $\sigma$ to $E(v)$.

\begin{definition}[Holant Problems]
The input to the problem $\Holant(\mathcal{F})$ is a signature grid $\Omega=(G,\pi)$ over $\mathcal{F}$.
The output is the partition function
\[
\operatorname{Holant}({\Omega})
=\sum_{\sigma:E(G)\to\{0,1\}}
 \prod_{v\in V}f_v(\sigma|_{E(v)}).
\]
Bipartite Holant problems $\holant{\mathcal{F}}{\mathcal{G}}$
are {Holant} problems
 over bipartite graphs $H = (U,V,E)$,
where each vertex in $U$ or $V$ is labeled by a signature in $\mathcal{F}$ or $\mathcal{G}$ respectively.
When $\{f\}$ is a singleton set, we write $\Holant(\{f\})$ as $\Holant(f)$ and  $\Holant(\{f\}\cup \mathcal{F})$ as $\Holant(f, \mathcal{F}).$ 
\end{definition}

We say a signature $f$ is \emph{realizable} from $\F$, if $\Holant(\F,f)\leq_T \Holant(\F)$.
The \emph{counting constraint satisfaction} problem \#CSP($\mathcal{F}$) is defined as $\Holant(\mathcal{EQ}\mid \mathcal{F}).$
The problem \#$\mathrm{CSP}_k(\mathcal{F})$ is defined as $\holant{\mathcal{EQ}_k}{\mathcal{F}}$.

For a signature $f$ of arity $n$, define its conjugation $\overline{f}$ by $\overline{f}(\alpha)=\overline{f(\alpha)},\,\forall\alpha\in \{0,1\}^n$.
We define the \emph{arrow reversal} of $f$ by
\(
f^{\textsc{ar}}(\alpha)
:=\overline{f(\overline\alpha)},\,\forall\alpha\in \{0,1\}^n.
\)
We say $f$ satisfies \emph{arrow reversal symmetry} (\textsc{ars}) if
$f=f^{\textsc{ar}}$.  
A signature set $\mathcal F$ is
\emph{conjugate closed}, written $\mathcal F=\overline{\mathcal F}$, if
$f\in\mathcal F$ implies $\overline f\in\mathcal F$.  
In this paper, we focus on the case where $\mathcal{F}$ is conjugate closed.
A set $\mathcal G$
is \emph{arrow-reversal closed}, written
$\mathcal G=\mathcal G^{\textsc{ar}}$, if $g\in\mathcal G$ implies
$g^{\textsc{ar}}\in\mathcal G$.

Let $\mathcal U$ be the set consisting of the binary zero signature and
all binary signatures $b$ for which
\(
M(b)M(b)^\dagger=\lambda I_2
\text{ for some }\lambda>0.
\)
Similarly, $\mathcal O$ consists of the binary zero signature and the
real-valued binary signatures $b$ satisfying
$M(b)M(b)^{\mathtt T}=\lambda I_2$ for some $\lambda>0$.  

We use projective binary matrices in the arity-six argument.  For nonzero
matrices $A,B$, write $A\sim B$ if $A=\lambda B$ for some
$\lambda\in\mathbb C^\times$, and let $[A]$ denote the projective class.
Set
\[
X=N_2,\qquad
Y=\begin{bmatrix}0&-\mathfrak i\\ \mathfrak i&0\end{bmatrix},\qquad
Z=\begin{bmatrix}1&0\\0&-1\end{bmatrix},\qquad
J=\mathfrak iY=\begin{bmatrix}0&1\\-1&0\end{bmatrix}.
\]
Notice that $I_2,X,Z,J$ are matrices of the signatures $=_2^{+},\neq_2^{+},=_2^{-},\neq_2^{-}$, respectively.
For $\varepsilon_1,\varepsilon_2,\varepsilon_3\in\{\pm1\}$, put
\[
C_{\varepsilon_1,\varepsilon_2,\varepsilon_3}
=\frac12\bigl(I_2+\varepsilon_1\mathfrak iX
                 +\varepsilon_2J+\varepsilon_3\mathfrak iZ\bigr).
\]
The projective tetrahedral group and its normal Klein-four subgroup are
\[
\Gamma
=\{[I_2],[X],[Z],[J]\}
 \cup
 \{[C_{\varepsilon_1,\varepsilon_2,\varepsilon_3}]:
        \varepsilon_1,\varepsilon_2,\varepsilon_3\in\{\pm1\}\},
\qquad
K_4=\{[I_2],[X],[Z],[J]\}.
\]
Thus $\Gamma\cong A_4$ and $K_4\cong V_4$, where $A_4$ is the alternating group of order 12 and $V_4$ is the Klein group.
We suppress the projective
brackets when a displayed unitary representative is intended. 
In this paragraph, and everywhere in the quantum arguments, $Z$ is the Pauli
matrix rather than a holographic transformation.

\subsection{Holographic Transformation}
To introduce the idea of holographic transformation,
it is convenient to consider bipartite graphs.
For a general graph,
we can always transform it into a bipartite graph while preserving the Holant value,
as follows.
For each edge in the graph,
we replace it by a path of length two.
(This operation is called the \emph{2-stretch} of the graph and yields the edge-vertex incidence graph.)
Each new vertex is assigned the binary \textsc{Equality} signature $=_2$. Thus, we have $\holant{=_2}{\mathcal{F}}\equiv_T \Holant(\mathcal{F})$.

For an invertible $2$-by-$2$ matrix $T \in {\bf GL}_2({\mathbb{C}})$
 and a signature $f$ of arity $n$, written as
a column vector (covariant tensor) $f \in \mathbb{C}^{2^n}$, we denote by
$Tf = T^{\otimes n} f$ the transformed signature.
  For a signature set $\mathcal{F}$,
define $T\mathcal{F} = \{T f \mid  f \in \mathcal{F}\}$ the set of
transformed signatures.
For signatures written as
 row vectors (contravariant tensors) we define
$f T^{-1}$ and  $\mathcal{F} T^{-1}$ similarly.
Whenever we write $T f$ or $T \mathcal{F}$,
we view the signatures as column vectors;
similarly for $f T^{-1}$ or $\mathcal{F} T^{-1}$ as row vectors.
We can also represent $Tf$ as the matrix $M_{S_k}(Tf)$ with the assignments of variables in $S_k$ as row index and the assignments of the other $n-k$ variables as column index. 
Then,  we  have $M_{S_k}(Tf)=T^{\otimes k}M_{S_k}(f)(T^{\tt T})^{\otimes n-k}$. Similarly, $M_{S_k}(fT^{-1})=({T^{-1}}^{\tt T})^{\otimes k}M_{S_k}(f)(T^{-1})^{\otimes n-k}$.

Let $T \in {\bf GL}_2({\mathbb{C}})$.
The holographic transformation defined by $T$ is the following operation:
given a signature grid $\Omega = (H, \pi)$ of $\holant{\mathcal{F}}{\mathcal{G}}$,
for the same bipartite graph $H$,
we get a new signature grid $\Omega' = (H, \pi')$ of $\holant{\mathcal{F} T^{-1}}{T \mathcal{G}}$ by replacing each signature in
$\mathcal{F}$ or $\mathcal{G}$ with the corresponding signature in $\mathcal{F} T^{-1}$ or $T \mathcal{G}$.

\begin{theorem}[\cite{Valiant08}]\label{thm: holographic transformation}
 For every $T \in {\bf GL}_2({\mathbb{C}})$,
  $\Holant(\mathcal{F} \mid \mathcal{G}) \equiv_T \Holant(\mathcal{F} T^{-1} \mid T \mathcal{G}).$
\end{theorem}

Therefore,
a holographic transformation does not change the complexity of the Holant problem in the bipartite setting. 
A particular useful holographic transformation is the $K$-transformation.
Direct calculation shows $(=_2)K=(\neq_2)$, so
\(
\holant{=_2}{\mathcal F}
\equiv_T
\holant{\neq_2}{K^{-1}\mathcal F}.
\)
Accordingly, throughout this paper
\[
\widehat f:=K^{-1}f,
\qquad
\widehat{\mathcal F}:=K^{-1}\mathcal F.
\]
We also call
$\holant{\neq_2}{\widehat{\mathcal F}}$ the
{$K$-Holant} formulation of $\Holant(\mathcal F)$.

\begin{lemma}\label{lm: conjugation and K-trans}
Let $f$ be a complex-valued signature of arity $n$.  For every
$\alpha\in\{0,1\}^n$,
\(
\widehat{\overline f}(\alpha)
=\overline{\widehat f(\overline\alpha)}.
\)
\end{lemma}
\begin{proof}
Using $K^{\mathtt T}=N_2K^{-1}$, we have
\(
\widehat{\overline f}
=(K^{-1})^{\otimes n}\overline f
=\overline{(K^{\mathtt T})^{\otimes n}f}
=N_2^{\otimes n}\overline{\widehat f},
\)
which is the stated entrywise identity.
\end{proof}

The preceding lemma immediately gives
the following equivalence.
\begin{lemma}\label{lm: conjugate closed and ars closed}
$\mathcal F$ is conjugate closed if and only if
$\widehat{\mathcal F}$ is arrow-reversal closed.
\end{lemma}

\subsection{Signature Factorization}

For signatures $f$ and $g$ on disjoint variable sets $\{x_1,\ldots,x_n\}$ and $\{y_1,\ldots,y_m\}$, $f\otimes g$ denotes their tensor product, with $(f\otimes g)(x_1,\ldots, x_n,y_1,\ldots,y_m)=f(x_1,\ldots,x_n)\cdot g(y_1,\ldots, y_m)$.
Recall that by our definition, every signature has arity at least one.
A nonzero signature $g$ \emph{divides} $f$ denoted by $g \mid f$, if there is a signature $h$ such that  $f=g \otimes h$
(with possibly a permutation of variables) or there is a constant $\lambda$ such that $f= \lambda \cdot g$.
In the latter case, if $\lambda \neq 0$, then we also have $f \mid g$ since $g= \frac{1}{\lambda} \cdot f$.
For nonzero signatures, if both $g\mid f$ and $f \mid g$, then they are nonzero constant multiples of
each other, and 
we say $g$ is an \emph{associate} of $f$, 
denoted by $g \sim f$.
In terms of  this division relation,
the notions of \emph{irreducible}  signatures  and \emph{prime} signatures have been defined.
They are proved equivalent and thus, the \emph{unique prime factorization} (UPF) of signatures is established \cite{cai-fu-shao-eo}.

A nonzero signature $f$ is irreducible  if there are no signatures $g$ and $h$ such that $f=g\otimes h$. 
A nonzero signature $f$ is a prime signature
if $f \mid g\otimes h$
implies that $f \mid g$ or $f \mid h$.  These notions
are equivalent.
We say a signature $f$ is reducible  if $f = g \otimes h$,
for some signatures $g$ and $h$. All zero signatures 
of arity greater than 1 are reducible.
A prime factorization of a signature $f$ is $f=g_1\otimes \ldots \otimes g_k$ up to a permutation of
variables, where each $g_i$ is irreducible. 

\begin{lemma}[Unique prime factorization~\cite{cai-fu-shao-eo}]\label{unique}
Every nonzero signature $f$ has a prime factorization.
If  $f$ has  prime factorizations
 $f=g_1\otimes \ldots \otimes g_k$ and $f=h_1\otimes \ldots \otimes h_\ell$,
both up to a permutation of variables,
then $k=\ell$ and after reordering the factors we have $g_i \sim h_i$  for all $i$. 
\end{lemma}

For $k\ge1$, let
\[
\mathcal F^{\otimes k}
=\left\{\lambda\bigotimes_{j=1}^k f_j:
  \lambda\in\mathbb C^\times,\ f_j\in\mathcal F\right\},\qquad
\mathcal F^\otimes=\bigcup_{k\ge1}\mathcal F^{\otimes k}.
\]
We also write $\langle\mathcal F\rangle$ for this normalized tensor
closure.

\begin{lemma}[\cite{Lin_Wang_Decomposition_lemma} Lemma 3.1]\label{lm: decomposition of tensor power}
    Let $\F$ be a set of signatures. Then
    $$
        \Holant(\F,f)\leq_T \Holant(\F,f^{\otimes d}),
    $$
    for all $d\geq 1$.
\end{lemma}

\subsection{Gadget Construction}

An \emph{$\mathcal F$-gate} is a finite signature grid over $\mathcal F$
with some dangling edges; summing over its internal edges realizes the
signature on the dangling edges. 
It is clear that if $f$ is the signature of an $\mathcal F$-gate, it is realizable from $\mathcal{F}.$
In the following we do not distinguish an $\mathcal{F}$-gate and its signature.

We record some frequently used gadgets in this paper. 
The following four basic constructions are illustrated in \cref{fig:gadget-constructions}.

\paragraph{Merging}
Given a binary signature $b$ and a signature $f$ with distinct variables $x_i,x_j$, the signature obtained \emph{merging} (or \emph{contracting}) variables $x_i,x_j$ of $f$ using $b$ is the following signature
\[
(\partial_{ij}^b f)(x_{[n]\setminus\{i,j\}})
=\sum_{a,c\in\{0,1\}}b(a,c)
  f(x_1,\ldots,x_{i-1},a,x_{i+1},\ldots,x_{j-1},c,x_{j+1},\ldots,x_n).
\]
When $b$ is $=_2$, we simply say merge $x_i$ and $x_j$ of $f$. 
We abbreviate
$\partial_{ij}f:=\partial_{ij}^{=_2}f=f_{ij}^{00}+f_{ij}^{11}$.
We also use the notation $\partial^{M(b)}_{ij}f$ for $\partial^{b}_{ij}f$, where $M(b)$ is the $2$ by 2 signature matrix of $b$.
We use $\partial_{(i_1j_1)\ldots(i_kj_k)}f$ to denote $\partial_{i_1j_1}\ldots\partial_{i_kj_k}f$ for short.
It is clear that the operators $\partial_{ij}$ and $\partial_{k\ell}$ commute if $\{i,j\}\cap \{k,\ell\}=\emptyset.$
For a signature class $\mathscr C$, the notation
$f\in\int\mathscr C$ means that
$\partial_{ij}f\in\mathscr C$ for every pair $i\ne j$.

\paragraph{Binary extension}
    Given a binary signature $b(x_1,x_2)$ and an arity-$n$ signature $f(y_1,\ldots, y_n)$, we may connect $x_i$ of $b$ to $y_j$ of $f$, and obtain a new signature.
    This is called a \emph{binary extension} or \emph{binary modification} to $f$.

\begin{lemma}\label{lem-binary-sim}
    Let ${b_1}(x_1, y_1), {b_2}(x_2, y_2)\in {\mathcal{U}}$. 
    If by connecting the variable $x_1$ of ${b_1}$ and the variable $x_2$ of ${b_2}$ using $=_2$, we get $\lambda\cdot (=_2)(y_1, y_2)$ for some $\lambda\in \mathbb{C}^\times$, then ${b_1}\sim \overline{b_2}$. 
    Moreover, by connecting the variable $y_1$ of ${b_1}$ and the variable $y_2$ of ${b_2}$, we will get $\lambda\cdot (=_2)(x_1, x_2)$.
\end{lemma}

\begin{proof}
Clearly both $b_1$ and $b_2$ are non-zero, otherwise connecting them we obtain a zero signature.
Since $b_1, b_2\in \mathcal{U}$ and they are non-zero, there exists $\lambda_1,\lambda_2\in \mathbb C^\times$ such that $M(b_1)M(b_1)^\dagger=\lambda_1I_2$ and  $M(b_2)M(b_2)^\dagger=\lambda_2I_2$.
By connecting $x_1$ of $b_1$ and $x_2$ of $b_2$, we obtain the binary signature $b_3(y_1,y_2)$ with matrix $M(b_3)=M(b_1)^{\tt T}M(b_2)=\lambda I_2$ for some $\lambda\in \mathbb{C}^\times.$
So $M(b_2)=\lambda (M(b_1)^{\tt T})^{-1}=\frac{\lambda}{\lambda_1} \overline{M(b_1)}$.
Hence $b_1\sim \overline{b_2}$.
Also, by connecting the variable $y_1$ of ${b_1}$ and the variable $y_2$ of ${b_2}$, we will get the signature $b_4(x_1,x_2)$ with matrix $M(b_4)=M(b_1)M(b_2)^{\tt T}=M(b_1)\cdot \frac{\lambda}{\lambda_1}M(b_1)^\dagger=\lambda\cdot(=_2)(x_1, x_2).$
\end{proof}

\paragraph{Mating and first order orthogonality}
Given a (complex-valued) signature $f$ of arity $n\geqslant 2$, we connect $f$ and $\overline{f}$ in the following manner:
Fix a set $S$ of $n-m$ variables among all $n$ variables of $f$. 
For each $x_k\in S$, connect $x_k$ of $f$ with the corresponding variable $x_k$ of $\overline{f}$ using $=_2$. 
The variables that are not in $S$ are called dangling variables.
For $m=1$, there is one dangling variable $x_i$. Then, 
the mating construction  realizes a binary signature, denoted by $\mathfrak m_{i}f$.
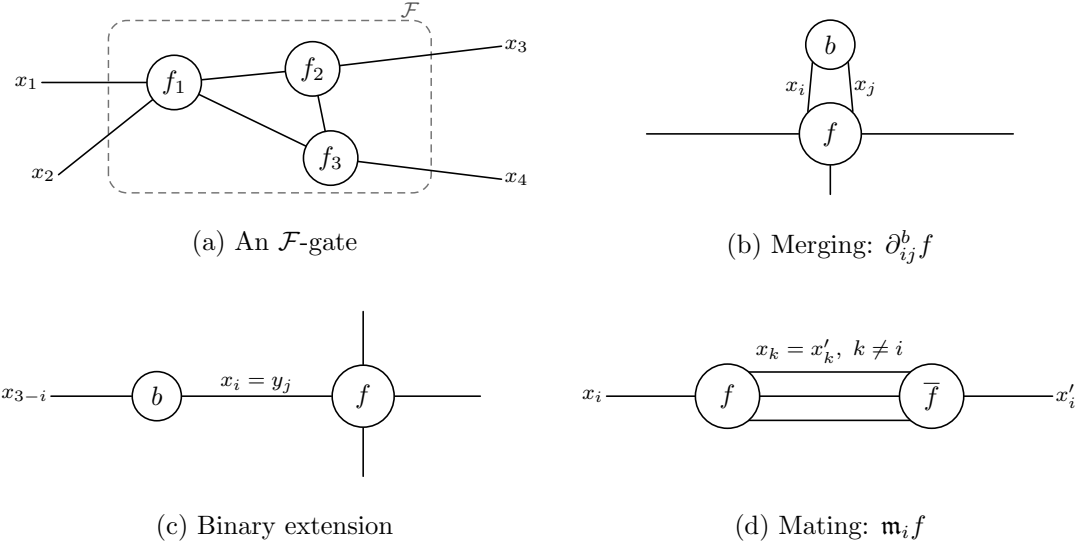
\begin{figure}[t]
\centering
\begin{tikzpicture}[
    wire/.style={draw=black, line width=0.6pt, line cap=round},
    boundary/.style={draw=black!55, densely dashed, line width=0.55pt,
        rounded corners=7pt},
    signature/.style={draw=black, circle, fill=white, minimum size=7.2mm,
        inner sep=0pt, line width=0.6pt, font=\small},
    small signature/.style={signature, minimum size=6.5mm},
    port label/.style={font=\scriptsize, inner sep=1pt},
    panel label/.style={font=\small, anchor=north}
]

\begin{scope}[shift={(0,3.75)}]
    \path[use as bounding box] (0,0) rectangle (6.75,3.35);

    \draw[boundary] (1.18,0.64) rectangle (5.45,2.92);
    \node[port label, anchor=south east, text=black!70]
        at (5.33,2.92) {$\mathcal F$};

    \draw[wire] (0.30,2.10) -- (2.05,2.10);
    \draw[wire] (0.52,0.88) -- (2.05,2.10);
    \draw[wire] (2.05,2.10) -- (3.88,2.28);
    \draw[wire] (2.05,2.10) -- (4.12,1.10);
    \draw[wire] (3.88,2.28) -- (4.12,1.10);
    \draw[wire] (3.88,2.28) -- (6.38,2.58);
    \draw[wire] (4.12,1.10) -- (6.38,0.82);

    \node[signature] at (2.05,2.10) {$f_1$};
    \node[signature] at (3.88,2.28) {$f_2$};
    \node[signature] at (4.12,1.10) {$f_3$};

    \node[port label, left] at (0.30,2.10) {$x_1$};
    \node[port label, left] at (0.52,0.88) {$x_2$};
    \node[port label, right] at (6.38,2.58) {$x_3$};
    \node[port label, right] at (6.38,0.82) {$x_4$};
    \node[panel label] at (3.38,0.28)
        {\textup{(a)} An $\mathcal F$-gate};
\end{scope}

\begin{scope}[shift={(7.35,3.75)}]
    \path[use as bounding box] (0,0) rectangle (6.75,3.35);

    \node[small signature] (merge-b) at (3.38,2.60) {$b$};
    \node[signature, minimum size=8.2mm] (merge-f) at (3.38,1.42) {$f$};

    \draw[wire] (merge-b.225) --
        node[port label, left, pos=0.53] {$x_i$} (merge-f.135);
    \draw[wire] (merge-b.315) --
        node[port label, right, pos=0.53] {$x_j$} (merge-f.45);
    \draw[wire] (0.95,1.42) -- (merge-f.180);
    \draw[wire] (merge-f.0) -- (5.80,1.42);
    \draw[wire] (3.38,0.62) -- (merge-f.270);

    \node[panel label] at (3.38,0.28)
        {\textup{(b)} Merging: $\partial_{ij}^{b}f$};
\end{scope}

\begin{scope}
    \path[use as bounding box] (0,0) rectangle (6.75,3.35);

    \node[small signature] (ext-b) at (1.82,1.70) {$b$};
    \node[signature, minimum size=8.2mm] (ext-f) at (4.55,1.70) {$f$};

    \draw[wire] (0.42,1.70) -- (ext-b.180);
    \draw[wire] (ext-b.0) --
        node[port label, above] {$x_i=y_j$} (ext-f.180);
    \draw[wire] (ext-f.0) -- (6.10,1.70);
    \draw[wire] (4.55,2.82) -- (ext-f.90);
    \draw[wire] (4.55,0.64) -- (ext-f.270);

    \node[port label, left] at (0.42,1.70) {$x_{3-i}$};
    \node[panel label] at (3.38,0.28)
        {\textup{(c)} Binary extension};
\end{scope}

\begin{scope}[shift={(7.35,0)}]
    \path[use as bounding box] (0,0) rectangle (6.75,3.35);

    \draw[wire] (0.42,1.70) -- (2.02,1.70);
    \draw[wire] (2.02,2.02) -- (4.72,2.02);
    \draw[wire] (2.02,1.70) -- (4.72,1.70);
    \draw[wire] (2.02,1.38) -- (4.72,1.38);
    \draw[wire] (4.72,1.70) -- (6.32,1.70);

    \node[signature, minimum size=8.5mm] at (2.02,1.70) {$f$};
    \node[signature, minimum size=8.5mm] at (4.72,1.70) {$\overline f$};

    \node[port label, left] at (0.42,1.70) {$x_i$};
    \node[port label, right] at (6.32,1.70) {$x_i'$};
    \node[port label, fill=white] at (3.37,2.30)
        {$x_k=x_k',\ k\ne i$};
    \node[panel label] at (3.38,0.28)
        {\textup{(d)} Mating: $\mathfrak m_i f$};
\end{scope}

\end{tikzpicture}
\caption{Basic gadget constructions.
\textup{(a)} An $\mathcal F$-gate, with every internal signature in
$\mathcal F$.
\textup{(b)} Merging the variables $x_i$ and $x_j$ of $f$ with the binary
signature $b$, producing $\partial_{ij}^{b}f$.
\textup{(c)} A binary extension formed by identifying $x_i$ of $b$ with
$y_j$ of $f$.
\textup{(d)} Mating $f$ with $\overline f$ on all variables except $x_i$,
producing $\mathfrak m_i f$; the three parallel wires represent these
$n-1$ contractions.
Joined half-edges are contracted, whereas unjoined half-edges are free
variables.}
\label{fig:gadget-constructions}
\end{figure}

It can be represented by matrix multiplication. 
We have 
\begin{equation}\label{m-form}
M(\mathfrak m_{i}f)=M_{i}(f)I_2^{\otimes (n-1)}M^{\tt T}_{i}(\overline{f})=M_i(f)M_i(f)^{\dagger}
=\begin{bmatrix}
{\bf {f}}^{0}_i\\
{\bf {f}}^{1}_i\\
\end{bmatrix}
\left[\begin{matrix}
{{\bf {\overline{f}}}^{0}_i}^{\tt T} &{{\bf {\overline{f}}}^{1}_i}^{\tt T}\\
\end{matrix}\right]
=\left[\begin{matrix}
|{\bf f}_i^0|^2 &  \langle {\bf f}_i^0, {\bf f}_i^1 \rangle\\
\langle {\bf f}_i^1, {\bf f}_i^0 \rangle & |{\bf f}_i^1|^2\\
\end{matrix}\right],
\end{equation}
where $\langle \cdot, \cdot\rangle$ denotes the 
inner product and $|\cdot|$ denotes the  norm defined by this inner product.
Note that $\langle{\bf f}^{0}_i, {\bf f}^{1}_i\rangle=\overline{\langle{\bf f}^{1}_i, {\bf f}^{0}_i\rangle}$ and $|\langle{\bf f}^{0}_i, {\bf f}^{1}_i\rangle|\leqslant|{\bf f}^{0}_i||{\bf f}^{1}_i|$ by the Cauchy-Schwarz inequality.

\begin{definition}[First order orthogonality \cite{odd_holant_real}]\label{def-first-order}
       Let $f$ be a complex-valued signature of arity $n \geqslant 2$. It satisfies the \emph{first order orthogonality}, denoted by $f\in$ {\sc 1st-Orth}, if there exists some $\mu\neq 0$ such that for all indices $i\in [n]$, the entries of $f$ satisfy the following equations
\begin{equation*}
    |{\bf f}_{i}^{0}|^2=|{\bf f}_{i}^{1}|^2=\mu, \text{ and } \langle{\bf f}_{i}^{0}, {\bf f}_{i}^{1}\rangle =0.
\end{equation*}
    \end{definition}
    
\begin{proposition}\label{prop: invariant 1st-Orth}
    {\sc 1st-Orth} is invariant up to conjugation and holographic transformation by $T\in \mathbf{U}_2(\mathbb{C})$.
\end{proposition}

\begin{proof}
    Suppose $f$ has arity $n$.
    Clearly $f\in$ {\sc 1st-Orth} $\iff $ $\overline{f}\in$ {\sc 1st-Orth}.
    Now take $T\in \mathbf{U}_2(\mathbb{C})$, we have
    \begin{equation*}
        \begin{aligned}
            M(\mathfrak m_i (Tf))
            =&M_i(Tf)M_i(Tf)^{\dagger}\\
            =&TM_i(f)(T^{\tt T})^{\otimes n-1}\overline{T}^{\otimes n-1}M_i(f)^\dagger T^\dagger\\
            =&TM_i(f)M_i(f)^\dagger T^\dagger\\
            =&TM(\mathfrak m_i f)T^\dagger.\\
        \end{aligned}
    \end{equation*}
    So, for every $\mu>0$, $M(\mathfrak m_i f)=\mu I_2 \iff M(\mathfrak m_i (Tf))=\mu I_2$.
    That is, $f\in$ {\sc 1st-Orth} $\iff $ $Tf\in$ {\sc 1st-Orth}.
\end{proof}

\begin{remark}
    Since $K\in \mathbf{U}_2(\mathbb{C})$, $f\in$ {\sc 1st-Orth} $\iff \widehat{f}\in $ {\sc 1st-Orth}.
\end{remark}

\begin{lemma}\label{lm: 1st-Orth uniform}
    Let $f$ be a signature of arity $n$.
    If for all indices $i\in [n]$, $M(\mathfrak{m}_{i}f)= \mu_i I_2$ for some real $\mu_i\neq 0$, then $f$ satisfies {\sc 1st-Orth} (i.e., all $\mu_i$ have the same value).
\end{lemma}
\begin{proof}
    For an arbitrary $i\in[n]$, connecting two variables of $\mathfrak m_if$ by $=_2$, we will get Holant value $2\mu_i\neq 0.$
    On the other hand, this value is equal to $\sum_{\alpha\in\{0,1\}^n}f(\alpha)\overline{f(\alpha)}$ which is independent of $i$.
    Therefore, all $\mu_i$ have the same value.
\end{proof}

\begin{lemma}\label{lm: not 1st-orth is hard}
    Let $\mathcal{F}$ be a conjugate closed set of signatures containing a nonzero signature that does not satisfy {\sc 1st-Orth}. 
    If $\mathcal{F}$ does not satisfy condition \rm{(\ref{cond: tractable classes})},
    then $\Holant(\mathcal{F})$ is \#P-hard.
\end{lemma}
 
\begin{proof}
    Let $f\in\F$ be a nonzero signature that does not satisfy {\sc 1st-Orth}, and consider $M_{i}(f)$ for an arbitrary index $i$.
    Clearly, $M(\mathfrak{m}_if)=M_{i}(f)M^\dagger_{i}(f)$ is a Hermitian positive semi-definite matrix, which is diagonalizable with two non-negative real eigenvalues $\lambda_i\geqslant\mu_i\geqslant 0$.
    These two eigenvalues are not both zero since $f \not\equiv 0$, and so $M(\mathfrak{m}_if)\neq 0$. Thus, $\lambda_i\neq 0$.
    Then, $|\frac{\mu_i}{\lambda_i}|=1$ iff $\lambda_i=\mu_i$. 

    Since $f$ does not satisfy {\sc 1st-Orth}, by \cref{lm: 1st-Orth uniform}, there is an index $i$ such that $M(\mathfrak{m}_if)\neq  \mu_i I_2$ for any real $\mu_i\neq 0$. 
    Thus, $M(\mathfrak{m}_if)$ has two eigenvalues with different norms. 
    By \cref{lm: 2 by 2 interpolation}, we can realize a nonzero binary signature $g$ such that $M(g)$ is degenerate. This implies that  $g$ can be factorized as a tensor product of two nonzero unary signatures. 
    By \cref{lm: decomposition}, we can realize a nonzero unary signature.
    By \cref{lm: odd holant with conjugation}, $\Holant(\mathcal{F})$ is \#P-hard.
\end{proof}

\subsection{Tractable Signatures}\label{subsec: tractable signatures}

We give some known signature sets that define polynomial time computable (tractable) counting problems.

\begin{definition}
    The class $\mathscr T$ consists of all signatures that, up to a
    permutation of variables, are tensor products of unary and binary signatures.
\end{definition}

\begin{definition}[Affine signatures]\label{def: affine}
A signature $f(x_1,\ldots,x_n)$ is \emph{affine} if it has the form
\[
f(x_1,\ldots,x_n)
=\lambda\,\chi_{A\mathbf x=0}\,
 \mathfrak i^{\,Q(x_1,\ldots,x_n)},
\]
where $\lambda\in\mathbb C$,
$\mathbf x=(x_1,\ldots,x_n,1)^{\mathsf T}$, $A$ is a matrix over
$\mathbb Z_2$, and $Q\in\mathbb Z_4[x_1,\ldots,x_n]$ is multi-linear polynomial of
degree at most two with even mixed coefficients.  
Equivalently,
\[
Q(x_1,\ldots,x_n)
=a_0+\sum_{k=1}^n a_kx_k
 +\sum_{1\le i<j\le n}2b_{ij}x_ix_j,
\qquad a_k,b_{ij}\in\mathbb Z_4.
\]
The equation $A\mathbf x=0$ is evaluated over $\mathbb Z_2$, and $\chi$ is a 0-1 indicator function such that $\chi_{A\mathbf{x}=0}$ is 1 iff $A\mathbf{x}=0.$
The class of affine signatures is denoted by $\mathscr A$. 
\end{definition}

We say that
$f$ has \emph{affine support} if $\mathscr S(f)$ is an affine subspace of
$\mathbb Z_2^n$.
Clearly, any affine signature has affine support. 

\begin{lemma}\label{lm: affine singular value}
    Let $f\in \mathscr{A}$ be a signature of arity n.
    For every subset of indices $S\subseteq[n]$, the nonzero rows of the $2^{|S|}$ by $2^{n-|S|}$ matrix $M_{S,\overline{S}}(f)$ are either orthogonal or proportional.
    Moreover, the $2^{|S|}$ by $2^{n-|S|}$ matrix $M_{S,\overline{S}}(f)$ has equal non-zero singular values, and $\mathrm{rank}(M_{S,\overline{S}}(f))$ is the power of 2.
\end{lemma}

\begin{proof}
    By the definition of affine signature, it has normal form $f(\mathbf{x})=\lambda \mathbf{1}_{A\mathbf{x}=b}\mathfrak{i}^{Q(\mathbf{x})}$, where the solution set of $A\mathbf x=b$ over $\mathbb Z_2$ is an affine subspace over $\mathbb Z_2^n$ and the $Q(\mathbf{x})$ is a multi-linear polynomial of degree at most two with even mixed coefficients.

    For any subset $S\subseteq [n]$, let $\mathbf{x}=(\mathbf{x}_S,\mathbf{x}_{\overline{S}})$ and  $\mathbf{x}_S,\mathbf{x}_{\overline{S}}$ be the vector of variables in $S$ and $\overline{S}$ respectively. Thus, the signature matrix $M_{S,\overline{S}}(f)$ has row index $\mathbf{x}_S$ and column index $\mathbf{x}_{\overline{S}}$. Since the support of $f$ is an affine subspace, we can assume that $\mathscr{S}(f)=L+(\mathbf{x}_S',\mathbf{x}_{\overline{S}}')$ where $L$ is the solution set of $A\mathbf{x}=0$ over $\mathbb Z_2$, hence $L$ is an linear subspace over $\mathbb{Z}_2^n$, and $(\mathbf{x}_S',\mathbf{x}_{\overline{S}}')$ is a special solution of $A\mathbf{x}=b$.

    Let $W:= \{\mathbf{u}_S\in \{0,1\}^{|S|}: \exists \mathbf{v}_{\overline{S}},(\mathbf{u}_S,\mathbf{v}_{\overline{S}})\in L\}$. This $W$ help us to give the set of all row indices of all nonzero rows, which is the set $\mathbf{x}_S+W$. Let $V:= \{\mathbf{v}_{\overline{S}}\in \{0,1\}^{|\overline{S}|}:(\mathbf{0},\mathbf{v}_{\overline{S}})\in L\}$. This $V$ help us to give the set of column indices that the entry of the matrix is nonzero when giving a nonzero row index. Giving an index $\mathbf{x}_S^{(1)}$ of a nonzero row, there must exist a index $\mathbf{x}_{\overline{S}}^{(1)}$ such that $(\mathbf{x}_S^{(1)},\mathbf{x}_{\overline{S}}^{(1)})\in L+(\mathbf{x}_S',\mathbf{x}_{\overline{S}}')$ since the row is not all zero. Then the index of column that the entry is nonzero belongs to $\mathbf{x}_{\overline{S}}^{(1)}+V$, since for any $\mathbf{v}_{\overline{S}}\in V$, we have $(\mathbf{x}_S^{(1)},\mathbf{x}_{\overline{S}}^{(1)}+\mathbf{v}_{\overline{S}})=(\mathbf{x}_S^{(1)},\mathbf{x}_{\overline{S}}^{(1)})+(\mathbf{0},\mathbf{v}_{\overline{S}})\in L+(\mathbf{x}_S',\mathbf{x}_{\overline{S}}').$ 
    
    Let $\mathscr{S}(f_{S}^{\mathbf{x}_S^{(1)}})=\{\mathbf{x}_{\overline{S}}^{(1)}\in\{0,1\}^{|S|}:(\mathbf{x}_S^{(1)},\mathbf{x}_{\overline{S}}^{(1)})\in L+(\mathbf{x}_S',\mathbf{x}_{\overline{S}}')\}$ be the partly support set of $f$ when fixed the variables in $S$ to $\mathbf{x}_S^{(1)}$. Thus, $\mathscr{S}(f_{S}^{\mathbf{x}_S^{(1)}})= \mathbf{x}_{\overline{S}}^{(1)}+V$ where $\mathbf{x}_{\overline{S}}^{(1)}$ is an arbitrary column index such that $(\mathbf{x}_S^{(1)},\mathbf{x}_{\overline{S}}^{(1)})\in L+(\mathbf{x}_S',\mathbf{x}_{\overline{S}}')$. Hence, for any nonzero row $\mathbf{x}_S^{(1)}$, the partly support $\mathscr{S}(f_{S}^{\mathbf{x}_S^{(1)}})$ is a coset of $V$. Suppose that $\mathscr{S}(f_{S}^{\mathbf{x}_S^{(1)}})=\mathbf{x}_{\overline{S}}^{(1)}+V$ and $\mathscr{S}(f_{S}^{\mathbf{x}_S^{(2)}})=\mathbf{x}_{\overline{S}}^{(2)}+V$ are intersecting, then exist $\mathbf{v}_{\overline{S}}^{(1)},\mathbf{v}_{\overline{S}}^{(2)}\in V$ such that $\mathbf{x}_{\overline{S}}^{(1)}+\mathbf{v}_{\overline{S}}^{(1)}=\mathbf{x}_{\overline{S}}^{(2)}+\mathbf{v}_{\overline{S}}^{(2)}$. We can get $\mathbf{x}_{\overline{S}}^{(1)}-\mathbf{x}_{\overline{S}}^{(2)}=\mathbf{v}_{\overline{S}}^{(2)}-\mathbf{v}_{\overline{S}}^{(1)}\in L$. Thus, for any $\mathbf{x}_{\overline{S}}^{(1)}+\mathbf{v}_{\overline{S}}^{'}\in \mathscr{S}(f_{S}^{\mathbf{x}_S^{(1)}})$, we have $\mathbf{x}_{\overline{S}}^{(1)}+\mathbf{v}_{\overline{S}}^{'}=\mathbf{x}_{\overline{S}}^{(2)}+(\mathbf{x}_{\overline{S}}^{(1)}-\mathbf{x}_{\overline{S}}^{(2)}+\mathbf{v}_{\overline{S}}^{'})\in \mathscr{S}(f_{S}^{\mathbf{x}_S^{(2)}})$. Then $\mathscr{S}(f_{S}^{\mathbf{x}_S^{(1)}})\subseteq \mathscr{S}(f_{S}^{\mathbf{x}_S^{(2)}})$, by symmetric, we can get $\mathscr{S}(f_{S}^{\mathbf{x}_S^{(2)}})\subseteq \mathscr{S}(f_{S}^{\mathbf{x}_S^{(1)}})$. Hence, we can get $\mathscr{S}(f_{S}^{\mathbf{x}_S^{(1)}})= \mathscr{S}(f_{S}^{\mathbf{x}_S^{(2)}})$. This indicates that for any two coset of $V$, either they are identical or they are disjoint. And we also have $\|f_{S}^{\mathbf{x}_S^{(i)}}\|^2=|\lambda|^2|V|$ have the same norm for any nonzero row.

    Until now, if for any the two nonzero rows $\mathbf{x}_S^{(i)},\mathbf{x}_S^{(j)}$ and these partly support $\mathscr{S}(f_{S}^{\mathbf{x}_S^{(i)}}), \mathscr{S}(f_{S}^{\mathbf{x}_S^{(j)}})$ are disjoint, then obviously $\langle f_{S}^{\mathbf{x}_S^{(i)}},f_{S}^{\mathbf{x}_S^{(j)}}\rangle=0$. Next, we consider the case that the two rows have the same partly support $\mathscr{S}(f_{S}^{\mathbf{x}_S^{(i)}})=\mathscr{S}(f_{S}^{\mathbf{x}_S^{(j)}})$, we will prove that $f_{S}^{\mathbf{x}_S^{(i)}},f_{S}^{\mathbf{x}_S^{(j)}}$ are proportional.

    Suppose that for two nonzero rows $\mathbf{x}_S^{(i)},\mathbf{x}_S^{(j)}$, we have that $\mathscr{S}(f_{S}^{\mathbf{x}_S^{(i)}})=\mathbf{x}_{\overline{S}}^{(i)}+V=\mathbf{x}_{\overline{S}}^{(j)}+V =\mathscr{S}(f_{S}^{\mathbf{x}_S^{(j)}})$ where $(\mathbf{x}_S^{(i)},\mathbf{x}_{\overline{S}}^{(i)}),(\mathbf{x}_S^{(j)},\mathbf{x}_{\overline{S}}^{(j)})\in L+(\mathbf{x}_S',\mathbf{x}_{\overline{S}}')$. For any $\mathbf{x}_{\overline{S}}^{(*)}\in \mathscr{S}(f_{S}^{\mathbf{x}_S^{(i)}})=\mathscr{S}(f_{S}^{\mathbf{x}_S^{(j)}})$, we have $(\mathbf{x}_S^{(i)}+\mathbf{x}_S^{(j)},\mathbf{0})=(\mathbf{x}_S^{(i)},\mathbf{x}_{\overline{S}}^{(*)})+(\mathbf{x}_S^{(j)},\mathbf{x}_{\overline{S}}^{(*)})\in L$. Then we can define $W_0=\{\mathbf{d}_S\in\{0,1\}^{|S|}:(\mathbf{d}_S,\mathbf{0})\in L\}$ is a linear subspace. This $W_0$ will help us analyze the class of rows where all rows in the class are pairwise proportional.  

    Let $\mathbf{d}_S=\mathbf{x}_S^{(i)}+\mathbf{x}_S^{(j)}$, we know that $(\mathbf{d}_S,\mathbf{0})\in L$. Suppose that $Q(\mathbf{x}_S,\mathbf{x}_{\overline{S}})=Q_S(\mathbf{x}_S)+Q_{\overline{S}}(\mathbf{x}_{\overline{S}})+2B(\mathbf{x}_S,\mathbf{x}_{\overline{S}}) \mod 4$. Then we analyze the distance between $Q(\mathbf{x}_S^{(i)},\mathbf{x}_{\overline{S}}^{(*)})$ and $Q(\mathbf{x}_S^{(j)},\mathbf{x}_{\overline{S}}^{(*)})$. We have $Q(\mathbf{x}_S^{(i)},\mathbf{x}_{\overline{S}}^{(*)})-Q(\mathbf{x}_S^{(j)},\mathbf{x}_{\overline{S}}^{(*)})=Q(\mathbf{x}_S^{(j)}+\mathbf{d}_S,\mathbf{x}_{\overline{S}}^{(*)})-Q(\mathbf{x}_S^{(j)},\mathbf{x}_{\overline{S}}^{(*)})=c_{\mathbf{x}_S^{(j)},\mathbf{d}_S}+2l_d(\mathbf{x}_{\overline{S}}^{(*)})\mod 4$, where $c_{\mathbf{x}_S^{(j)},\mathbf{d}_S}$ is independent to $\mathbf{x}_{\overline{S}}^{(*)}$ and $l_d(\mathbf{x}_{\overline{S}}^{(*)})$ is a linear polynomial of $\mathbf{x}_{\overline{S}}^{(*)}$. Thus, 
    \begin{equation*}
\begin{aligned}
&f\left(
  \mathbf{x}_S^{(i)},
  \mathbf{x}_{\overline S}^{(*)}
 \right)                                                     \\
&\quad =
 f\left(
  \mathbf{x}_S^{(j)}+\mathbf d_S,
  \mathbf{x}_{\overline S}^{(*)}
 \right)                                                     \\
&\quad =
 \lambda
 \mathfrak{i}^{
   Q\left(
     \mathbf{x}_S^{(j)},
     \mathbf{x}_{\overline S}^{(*)}
   \right)
 }
 \mathfrak{i}^{
   c_{\mathbf{x}_S^{(j)},\mathbf d_S}
   +\ell_{\mathbf d}\left(
      \mathbf{x}_{\overline S}^{(*)}
    \right) 
 }\\
 &\quad =
 \omega_{\mathbf{x}_S^{(j)},\mathbf d_S}(-1)^{\ell_{\mathbf d}}f\left(
     \mathbf{x}_S^{(j)},
     \mathbf{x}_{\overline S}^{(*)}
   \right),
\end{aligned}
\end{equation*}
where $\omega_{\mathbf{x}_S^{(j)},\mathbf d_S}=\mathfrak{i}^{c_{\mathbf{x}_S^{(j)},\mathbf d_S}}$ is independent to $\mathbf{x}_{\overline{S}}^{(*)}$ and thus $|\omega_{\mathbf{x}_S^{(j)},\mathbf d_S}|=1$. Hence, we have that the ratio of the two entries $f\left(
  \mathbf{x}_S^{(i)},
  \mathbf{x}_{\overline S}^{(*)}
 \right)  ,f\left(
     \mathbf{x}_S^{(j)},
     \mathbf{x}_{\overline S}^{(*)}
   \right)$ is $\omega_{\mathbf{x}_S^{(j)},\mathbf d_S}(-1)^{\ell_d(\mathbf{x}_{\overline{S}}^{(*)})}$.

Recall that $\mathscr{S}(f_{S}^{\mathbf{x}_S^{(i)}})=\mathbf{x}_{\overline{S}}^{(i)}+V=\mathbf{x}_{\overline{S}}^{(j)}+V =\mathscr{S}(f_{S}^{\mathbf{x}_S^{(j)}})$. If $\ell_d$ is constant zero, then we have $\ell_d(\mathbf{x}_{\overline{S}}^{(*)})=\ell_d(\mathbf{x}_{\overline{S}}^{(i)})+\ell_d(\mathbf{v_{\overline{S}}})=\ell_d(\mathbf{x}_{\overline{S}}^{(i)})$ where $\mathbf{v}_{\overline{S}}\in V$ and $\mathbf{x}_{\overline{S}}^{(*)}=\mathbf{x}_{\overline{S}}^{(i)}+\mathbf{v_{\overline{S}}}$. Thus, $f\left(
  \mathbf{x}_S^{(i)},
  \mathbf{x}_{\overline S}^{(*)}
 \right)  = \omega_{\mathbf{x}_S^{(j)},\mathbf d_S} f\left(
     \mathbf{x}_S^{(j)},
     \mathbf{x}_{\overline S}^{(*)}
   \right)$. If $\ell_d$ is not constant zero, choosing $\mathbf{v}_{\overline{S}}^{(0)}\in V$ such that $\ell_d(\mathbf{v}_{\overline{S}}^{(0)})=1$. Then we have $S=\sum_{\mathbf{v}_{\overline{S}}\in V}(-1)^{\ell_d(\mathbf{v}_{\overline{S}})}=\sum_{\mathbf{v}_{\overline{S}}\in V}(-1)^{\ell_d(\mathbf{v}_{\overline{S}})+1}=-S$ since $\mathbf{v}_{\overline{S}}+V=V$. Thus, $S=0$ and this indicates that $\langle f_{S}^{\mathbf{x}_S^{(i)}},f_{S}^{\mathbf{x}_S^{(j)}}\rangle=\sum_{\mathbf{w}_{\overline{S}}\in V+\mathbf{x}_{\overline{S}}^{(i)}} f\left(
  \mathbf{x}_S^{(i)},
  \mathbf{w}_{\overline{S}}
 \right)  \overline{f\left(
     \mathbf{x}_S^{(j)},
     \mathbf{w}_{\overline{S}}
   \right)}=\overline{\omega_{\mathbf{x}_S^{(j)},\mathbf d_S}}\sum_{\mathbf{w}_{\overline{S}}\in V+\mathbf{x}_{\overline{S}}^{(i)}}(-1)^{\ell_d(\mathbf{w}_{\overline{S}})}|f\left(
     \mathbf{x}_S^{(j)},
     \mathbf{w}_{\overline{S}}
   \right)|^2=\overline{\omega_{\mathbf{x}_S^{(j)},\mathbf d_S}}|\lambda|^2\sum_{\mathbf{v}_{\overline{S}}\in V}(-1)^{\ell_d(\mathbf{v}_{\overline{S}})+\ell_d(\mathbf{x}_{\overline{S}}^{(i)})}=\overline{\omega_{\mathbf{x}_S^{(j)},\mathbf d_S}}|\lambda|^2(-1)^{\ell_d(\mathbf{x}_{\overline{S}}^{(i)})}S=0$.

   Hence, for any two nonzero rows, either they are orthogonal or they are proportional, i.e., $\langle f_{S}^{\mathbf{x}_S^{(i)}},f_{S}^{\mathbf{x}_S^{(j)}}\rangle=0$ or $f_{S}^{\mathbf{x}_S^{(i)}}=\omega_{\mathbf{x}_S^{(i)},\mathbf{x}_S^{(j)}}f_{S}^{\mathbf{x}_S^{(j)}}$ where $|\omega_{\mathbf{x}_S^{(i)},\mathbf{x}_S^{(j)}}|=1$.

   Let $H=\{\mathbf{d}_S\in W_0: \ell_d(\mathbf{v}_{\overline{S}})=0\quad \text{for every }\mathbf{v}_{\overline{S}}\in V\}$ and $H$ is a linear subspace. Thus, when given any nonzero row vector $f_S^{\mathbf{x}_{S}}$, their are $|H|$ many rows that are proportional to $f_S^{\mathbf{x}_{S}}$(include itself). And since the number of nonzero rows is $|W|$, thus there are $|W|/|H|$ many different proportional classes such the elements in one class are pairwise proportional and the elements in the different classes are orthogonal. Since $W,H$ are linear subspace, then $|W|/|H|=2^{\dim W-\dim H}$.

   Finally, using the above proofs, we can represent the vector in one class are $\omega_1r,\dots,\omega_{|H|}r$ where $|\omega_i|=1$ for $i\in [|H|]$. Then the eigenvalue of $(\omega_1r,\dots,\omega_{|H|}r)^T(\overline{\omega_1r},\dots,\overline{\omega_{|H|}r})$ is $|\lambda|^2|V||H|$. Thus, the matrix $M_{S,\overline{S}}(f)$ has same singular value $|\lambda|\sqrt{|V||H|}$. And we can get that the $\operatorname{rank}(M_{S,\overline{S}})=|W|/|H|=2^{\dim W-\dim H}$ is the power of $2$.
    
\end{proof}

\begin{lemma}\label{lm: affine is 1st-orth}
    Let $f\in \mathscr{A}$ be an irreducible signature of arity $n\ge 2$.
    Then $f$ satisfies {\sc 1st-Orth}.
\end{lemma}
\begin{proof}
    By~\cref{lm: affine singular value}, for any $S={i}$ and $\overline{S}=[n]\setminus \{i\}$, then we can get that $M_{S,\overline{S}}(f)$ has two rows, $\|f_i^0\|^2=\|f_i^1\|^2$, and either $f_i^0,f_i^1$ are orthogonal or proportional. If $f_i^0,f_i^1$ are proportional, then $\operatorname{rank}(M_{S,\overline{S}}(f))=1$ and this indicates that $f$ is irreducible. Hence, the only case is that $f_i^0,f_i^1$ are orthogonal, namely $\langle f_i^0,f_i^1\rangle=0$. This gives the {\sc 1st-Orth}.
\end{proof}

\begin{definition}[Product-type signatures]\label{def: product-type}
A signature is of \emph{product type} if it is an ordinary pointwise
product of unary signatures, binary equality signatures, and binary
disequality signatures; different factors may involve overlapping
variables.  The class of all product-type signatures is denoted by
$\mathscr P$.
\end{definition}

Note that the product in \cref{def: product-type} are ordinary products of functions (not tensor products); in particular they may be applied on overlapping sets of variables.

\cref{def: product-type} is succinct.
But an alternative definition of $\mathscr{P}$ is
useful. This is given below in \cref{definition-product-1}.
\begin{definition}[Product-type signatures]\label{definition-product-1}
Let $\mathcal{E}$ be the set of all signatures $f$ such that
the support set
$\mathscr{S}(f)$ is contained in two antipodal points, i.e.,
if $f$ has arity $n$,
then $f$ is zero except on (possibly) two inputs $\alpha
= (a_1, a_2, \ldots, a_n)$
and $\overline{\alpha} =
(\overline{a}_1, \overline{a}_2, \ldots, \overline{a}_n)=(1- a_1, 1- a_2,
 \ldots, 1 - a_n)$.
Then $\mathscr{P}=\langle \mathcal{E}\rangle$.
\end{definition}

By Definition~\ref{definition-product-1}, if a signature $f\in\mathscr{P}$, then the support of $f$ is affine.
Thus the number of nonzero entries of $f$ is a power of 2.

Let $\alpha=e^{\frac{\pi \mathfrak i}{4}}$ such that $\alpha^2=\mathfrak i.$
For $s\in\mathbb Z$, let
\(
T_{\alpha^s}=\left[\begin{smallmatrix}1&0\\0&\alpha^s\end{smallmatrix}\right].
\)

\begin{definition}[Local-affine signatures]
An arity-$n$ signature $f$ is \emph{local-affine} if, for every
$\sigma=s_1\cdots s_n\in\mathscr S(f)$,
\(
\left(T_{\alpha^{s_1}}\otimes\cdots\otimes
      T_{\alpha^{s_n}}\right)f\in\mathscr A.
\)
The class of local-affine signatures is denoted by $\mathscr L$.
\end{definition}

We also use
\(
\alpha\mathscr A
=\left\{T_\alpha^{\otimes\operatorname{arity}(f)}f\mid 
  f\in\mathscr A\right\}.
\)

\begin{definition}[Transformable]
Let $\mathscr C$ be a class of signatures.  A signature set $\mathcal F$
is \emph{$\mathscr C$-transformable} if there exists
$T\in\mathrm{GL}_2(\mathbb C)$ such that
\[
(=_2)(T^{-1})^{\otimes2}\in\mathscr C
\qquad\text{and}\qquad
T\mathcal F\subseteq\mathscr C.
\]
\end{definition}

These definitions give the four alternatives in
\textup{(\ref{cond: tractable classes})}.  If
$\mathcal F\subseteq\mathscr T$, or if $\mathcal F$ is
$\mathscr P$-, $\mathscr A$-, or $\mathscr L$-transformable, then
$\Holant(\mathcal F)$ is computable in polynomial time
by~\cite{realholant}.

\begin{theorem}[\cite{realholant}]\label{thm: tractable classes}
    Let $\mathcal{F}$ be a set of complex valued signatures.
    Then $\Holant(\mathcal{F})$ is tractable if
    \begin{equation}\label{cond: tractable classes}
    \mathcal{F}\subseteq \mathscr{T}, ~~~
     \mathcal{F} \text{ is }\mathscr{P}\text{-transformable,}~~~
     \mathcal{F} \text{ is } \mathscr{A}\text{-transformable,~~~or~} 
     \mathcal{F} \text{ is } \mathscr{L}\text{-transformable.}\tag{$\mathrm{ \textcolor{red}{T}}$}
     \end{equation}
\end{theorem}

The following lemma can be easily checked from the above definitions.
\begin{lemma}\label{lm: APL-transformable}
    Let $f$ be an irreducible signature of arity $n\ge 2$.
    If $f$ is $\mathscr{A}$- or $\mathscr{P}$- or $\mathscr{L}$-transformable, then there exists $T_1,T_2,\ldots,T_n\in \mathrm{GL}_2(\C)$ such that $f=(T_1\otimes T_2\otimes \cdots \otimes T_n)a$, where $a\in \mathscr{A}.$
\end{lemma}
\begin{proof}
    Since $f$ is irreducible, $f$ is a non-zero signature.
    If $f$ is $\mathscr{A}$-transformable, there exists $T\in \mathrm
    GL_2(\C)$ such that $Tf=a$, where $a\in \mathscr{A}$.
    Let $T_1=T_2=\cdots=T_n=T^{-1}$, then $f=(T_1\otimes T_2 \otimes \cdots \otimes T_n)a.$
    If $f$ is $\mathscr{L}$-transformable, pick an arbitrary $s_1\cdots s_n\in \mathscr{S}(f)$, we have $T_{\alpha^{s_1}}\otimes\cdots \otimes T_{\alpha^{s_n}}f=a\in\mathscr{A}.$
    Let $T_i=T_{\alpha^{s_i}}^{-1}$ for $1\le i\le n$, the lemma is proved.

    Now we consider the case that $f$ is $\mathscr{P}$-transformable.
    There exists $T\in \mathrm{GL}_2(\C)$ such that $Tf=g\in \mathscr{P}.$
    Since holographic transformation doesn't change irreducibility, $g$ is also irreducible.
    By \cref{definition-product-1}, $g\in \mathcal{E}$.
    i.e. $\mathscr{S}(g)\subseteq \{\alpha,\overline{\alpha}\}$, where $\alpha\in \{0,1\}^n$.
    If $|\mathscr{S}(g)|=0$, $g$ is the zero signature, contradicting irreducibility.
    If $|\mathscr{S}(g)|=1$, $g$ is a tensor product of unary signatures, again contradicting irreducibility.
    So $\mathscr{S}(g)=\{\alpha,\overline{\alpha}\}.$
    Suppose $g(\alpha)=p$ and $g(\overline{\alpha})=q$, $pq\neq 0.$
    Without loss of generality, assume $\alpha_1=0.$
    Let $T_1=\left[\begin{smallmatrix}
        p & 0\\
        0 & q
    \end{smallmatrix}\right],T_2=\cdots=T_n=I_2.$
    Let $a(x)=1$ iff $x=\alpha$ or $x=\overline{\alpha}.$
    Clearly $a\in \mathscr{A}.$
    We claim $f=(T_1\otimes T_2\otimes \cdots T_n)a.$
    To see this,
    first notice $\mathscr{S}(f)=\mathscr{S}(a)$ since $T_1,\ldots, T_n$ are all diagonal matrices.
    Then $f(\alpha)=p\cdot a(\alpha)=p$, $f(\overline{\alpha})=q\cdot a(\overline{\alpha})=q.$
    This finishes the proof.
\end{proof}

We give the following lemma about the invariance of conjugate closure, condition~\eqref{cond: tractable classes} and irreducibility under real orthogonal transformation.
The proof is deferred to~\autoref{subapp: proof of subsec tractable signature}.

\begin{restatable}{lemma}
{InvarianceUnderRealOrth}\label{lem: invariance under real orthogonal transformation}
Let $\F$ be a set of complex-valued signatures and let $O\in \mathbf{O}_2$ be a real orthogonal $2\times 2$ matrix. Then the following statements hold:
\begin{enumerate}
    \item $\F$ is conjugate closed if and only if $O\F$ is conjugate closed.
    \item $\F$ satisfies condition~\eqref{cond: tractable classes} if and only if $O\F$ satisfies condition~\eqref{cond: tractable classes}.
    \item For every nonzero signature $f$, the signature $f$ is irreducible if and only if $Of$ is irreducible.
\end{enumerate}
Consequently, under a real orthogonal holographic transformation, conjugate closure and outside condition~\eqref{cond: tractable classes} are preserved preserved for signature set, while irreducibility is preserved for each signature.
\end{restatable}

\subsection{Known Dichotomies and Hardness Results}\label{subsec: hardness results}


\paragraph{Dichotomy for Counting Eulerian Orientations}
$\newline$
  
For $\bowtie \,\in\{=,\ge,\le,>,<\}$, define
\(
\mathrm{HW}^{\bowtie}
=\left\{\sigma\in\{0,1\}^*\mid 
  \operatorname{wt}(\sigma)\mathrel{\bowtie}\frac{|\sigma|}{2}\right\}.
\)
With a slightly abuse of notation,
we also write $f\in\mathrm{HW}^{\bowtie}$ when
$\mathscr S(f)\subseteq\mathrm{HW}^{\bowtie}$.
An arity-$2d$ signature
$f$ is an \emph{$\eo$
signature} if $f\in\eoe$.
For a set $\mathcal F$ of $\eo$
signatures, we use
\[
\ceo(\mathcal F):=\holant{\neq_2}{\mathcal F}.
\]
For completeness, we now give the $\#\eo$ support notation used
in~\cite{GSS26}.  
For an $\eo$ signature $f$, and for all
$\beta,\gamma,\delta\in\mathscr S(f)$, define
\begin{align*}
f\in\POLMID
&\iff \beta\oplus\gamma\oplus\delta\in\mathscr S(f),\\
f\in\POLUP
&\iff \beta\oplus\gamma\oplus\delta
      \in\mathscr S(f)\cup\mathrm{HW}^{>},\\
f\in\POLDOWN
&\iff \beta\oplus\gamma\oplus\delta
      \in\mathscr S(f)\cup\mathrm{HW}^{<}.
\end{align*}
The latter two classes are the up- and down-quasi-polymorphism classes.
A \emph{perfect pairing} $M$ of the variables of an arity-$2d$
signature is a partition into $d$ unordered pairs.  It determines
\[
\{0,1\}^{M}
=\{x\in\{0,1\}^{2d}\mid x_p\ne x_q
  \text{ for every }\{p,q\}\in M\}.
\]
For any signature class $\mathscr C$, define
\[
\eo^{\mathscr C}
=\left\{f\mid f\text{ is an $\eo$ signature and }
  f|_{\{0,1\}^{M}}\in\mathscr C
  \text{ for every perfect pairing }M \text{ of } f\right\}.
\]

\begin{theorem}[Dichotomy for \#EO~\cite{MengWX25},\cite{GSS26}]
\label{thm: dichotomy for EO}
Let $\mathcal F$ be a finite set of complex-valued $\eo$ signatures.  If
\[
\bigl(\mathcal F\subseteq\POLUP
      \ \text{or}\ \mathcal F\subseteq\POLDOWN\bigr)
\quad\text{and}\quad
\bigl(\mathcal F\subseteq\eo^{\mathscr A}
      \ \text{or}\ \mathcal F\subseteq\eo^{\mathscr P}\bigr),
\]
then $\ceo(\mathcal F)$ is computable in polynomial time.  Otherwise,
$\ceo(\mathcal F)$ is \#P-hard.
\end{theorem}
The two alternatives are uniform over $\mathcal F$: signatures from the
up and down classes cannot be mixed, nor can
$\eo^{\mathscr A}$ and $\eo^{\mathscr P}$ be mixed.

\begin{lemma}\label{lm: conjugation_polymorphism_affine support}
    Suppose $\mathcal{F}$ is a conjugate closed signature set.
    If $\widehat{\mathcal{F}}\subseteq \eog$ or $\widehat{\mathcal{F}}\subseteq \eol$, then every $\widehat{f}\in \widehat{\mathcal{F}}$ is $\eo$.
    Moreover, if $\widehat{\mathcal{F}}\subseteq \POLUP$ or $\widehat{\mathcal{F}}\subseteq \POLDOWN$, then every $\widehat{f}\in \widehat{\mathcal{F}}$ has affine support.
\end{lemma}

\begin{proof}
    We assume $\widehat{\mathcal{F}}\subseteq \eog$, the other case can be handled similarly.
    By \cref{lm: conjugate closed and ars closed}, $\widehat{\mathcal{F}}$ is arrow reversal closed.
    Hence $\widehat{\mathcal{F}}\subseteq \eog\cap \eol$, i.e., $\widehat{\mathcal{F}}$ is a set of $\eo$ signatures.
    We assume $\widehat{\mathcal{F}}\subseteq \POLUP$ without loss of generality.
    Pick arbitrary $\widehat{f}\in \widehat{\mathcal{F}}$ and $\alpha,\beta,\gamma\in \mathscr{S}(\widehat{f})$, we have $\alpha\oplus \beta\oplus \gamma\in \mathscr{S}(\widehat{f})\cup \eosg$.
    By \cref{lm: conjugate closed and ars closed}, $\widehat{f}^{\textsc ar}\in \widehat{\mathcal{F}}$.
    Notice that $\overline{\alpha},\overline{\beta},\overline{\gamma}\in \mathscr{S}(\widehat{f}^{\textsc ar})$.
    Then $\overline{\alpha\oplus\beta\oplus\gamma}\in \mathscr{S}(\widehat{f}^{\textsc ar})\cup \eosg$.
    It follows that $\alpha\oplus\beta\oplus\gamma\in \mathscr{S}(\widehat{f})\cup \eosl$.
    Thus, $\alpha\oplus\beta\oplus\gamma\in (\mathscr{S}(\widehat{f})\cup \eosl)\cap (\mathscr{S}(\widehat{f})\cup \eosg)=\mathscr{S}(\widehat{f})$.
    Therefore, $\widehat{f}$ has affine support.
\end{proof}

\paragraph{Dichotomy for Holant with an Odd Arity Signature}
$\newline$
A signature of arity $k$ is \emph{single-weighted} if its support is
contained in the strings of one Hamming weight $d$.  In that case define
the associated $\eo$ signature
\[
f_{\to\eo}=
\begin{cases}
f\otimes\Delta_0^{\otimes(2d-k)},&2d\ge k,\\
f\otimes\Delta_1^{\otimes(k-2d)},&2d<k.
\end{cases}
\]
Also set $f|_{\eo}:=f|_{\mathrm{HW}^{=}}$, and extend both operations
elementwise to signature sets:
\[
\mathcal F|_{\eo}=\{f|_{\eo}:f\in\mathcal F\},\qquad
\mathcal F_{\to\eo}=\{f_{\to\eo}:f\in\mathcal F\}.
\]
Finally,
\[
\mathcal M
=\{f:f(\sigma)=0\text{ whenever }\operatorname{wt}(\sigma)>1\}
\]
is the class of weighted-matching signatures.
Notice that 
\[
X\mathcal M
=\{f:f(\sigma)=0\text{ whenever }\operatorname{wt}(\sigma)<n-1\}
\]

\begin{theorem}[Dichotomy for $\holodd$ \cite{Holant_Odd}]\label{thm: holant odd}
    Let $\mathcal{F}$ be a set of signatures that contains a non-zero signature of odd arity, then $\Holant(\mathcal{F})$ is \#P-hard unless:
    \begin{enumerate}
        \item $\widehat{\mathcal{F}}$ only contains $\eog$ signatures (or $\eol$ signatures respectively), and $\widehat{\mathcal{F}}|_{\eo}\subseteq \POLUP$ or $\widehat{\mathcal{F}}|_{\eo}\subseteq \POLDOWN$, and $\widehat{\mathcal{F}}|_{\eo}\subseteq \eo^\mathscr{A}$ or $\widehat{\mathcal{F}}|_{\eo}\subseteq \eo^\mathscr{P}$;
        
        \item 
        All signatures in $\widehat{\mathcal{F}}$ are single-weighted. 
        There exists a signature in $\eosg$ and another in $\eosl$ belonging to $\widehat{\mathcal{F}}$.  
        In addition, $\widehat{\mathcal{F}}_{\to\eo}\subseteq \POLUP$ or $\widehat{\mathcal{F}}_{\to\eo}\subseteq \POLDOWN$, and $\widehat{\mathcal{F}}_{\to\eo}\subseteq \eo^\mathscr{A}$ or $\widehat{\mathcal{F}}_{\to\eo}\subseteq \eo^\mathscr{P}$;

        \item $\widehat{\mathcal{F}}\subseteq\langle \mathcal{M}\rangle$ or $\widehat{\mathcal{F}}\subseteq\langle X\mathcal{M}\rangle$;

        \item $\mathcal{F}\subseteq\mathscr{T}$;
        
        \item $\mathcal{F}$ is $\mathscr{A}$-transformable;

        \item $\mathcal{F}$ is $\mathscr{P}$-transformable;

        \item $\mathcal{F}$ is $\mathscr{L}$-transformable;
    \end{enumerate}
    We denote Case 1-7 as condition $(\mathcal{PC})$.
\end{theorem}

\begin{theorem}[\cite{Holant_Odd}]\label{thm: decomposition lemma}
    Let $\mathcal{F}$ be a set of signatures and $f$, $g$ be two signatures. Then

$$\Holant(f,g,\mathcal{F})\equiv_T\Holant(f\otimes g,\mathcal{F})$$ 
holds unless $\widehat{\mathcal{F}}'=\widehat{\mathcal{F}}\cup\{\widehat{f}\otimes\widehat{g}\}$ only contains $\eog$ signatures (or $\eol$ signatures respectively).
\end{theorem}

\begin{lemma}[Decomposition lemma]\label{lm: decomposition}
    Let $\mathcal{F}$ be a conjugate closed set of signatures.
    If a non-zero signature $f\in \mathcal{F}$ has a factorization $g\otimes h$, then $\Holant(g, h, \mathcal{F})\equiv_T\Holant( \mathcal{F})$.
\end{lemma}

\begin{proof}
    By \cref{thm: decomposition lemma}, we only need to consider the case that $\widehat{\mathcal{F}}$ only contains $\eog$ signatures (or $\eol$ signatures respectively).
    Since $\mathcal{F}$ is conjugate closed, $\widehat{\mathcal{F}}$ is $\eo$ by \cref{lm: conjugation_polymorphism_affine support}.
    By \cref{thm: dichotomy for EO}, either $\Holant(\mathcal{F})\equiv_T\holant{\neq_2}{\widehat{\mathcal{F}}}$ is \#P-hard, then trivially $\Holant(g,h,\mathcal{F})\equiv_T\Holant(\mathcal{F})$, or $\holant{\neq_2}{\hF}$ is tractable.
    In the tractable case, $\widehat{\mathcal{F}}\subseteq \POLUP$ or $\widehat{\mathcal{F}}\subseteq \POLDOWN$ and $\widehat{\mathcal{F}}\subseteq \eo^{\mathscr{A}}$ or $\widehat{\mathcal{F}}\subseteq \eo^{\mathscr{P}}$.
    By \cref{lm: conjugation_polymorphism_affine support}, every $\widehat{f}\in \mathcal{F}$ has affine support.
    Combined with $\widehat{\mathcal{F}}\subseteq \eo^{\mathscr{A}}$ or $\widehat{\mathcal{F}}\subseteq \eo^{\mathscr{P}}$, we have $\widehat{\mathcal{F}}\subseteq \mathscr{A}$ or $\widehat{\mathcal{F}}\subseteq\mathscr{P}$.
    Since $\widehat{f}=\widehat{g}\otimes\widehat{h}$, $\widehat{g}$ and $\widehat{h}$ are both affine or both product.
    Then trivially $\holant{\neq_2}{\widehat{g},\widehat{h},\widehat{\mathcal{F}}}\equiv_T \holant{\neq_2}{\widehat{\mathcal{F}}}$, which is equivalent to $\Holant(g, h, \mathcal{F})\equiv_T\Holant( \mathcal{F}).$ 
\end{proof}

\begin{lemma}\label{lm: odd holant with conjugation}
    Let $\mathcal{F}$ be a conjugate closed set of signatures containing a nonzero signature of odd arity.
    If $\mathcal{F}$ does not satisfy condition \rm{(\ref{cond: tractable classes})},
    then $\Holant(\mathcal{F})$ is \#P-hard.
\end{lemma}

\begin{proof}
    Notice that Case 4-7 in $(\mathcal{PC})$ is {\rm(\ref{cond: tractable classes})}.
    By \cref{thm: holant odd}, we only need to show that Case 1-3 in $(\mathcal{PC})$ cannot happen.
    \begin{enumerate}
        \item 
        Case 1 in $(\mathcal{PC})$ happens.
        By \cref{lm: conjugation_polymorphism_affine support}, $\hF$ is $\eo$ and has affine support.
        Combined with $\widehat{\mathcal{F}}|_{\eo}\subseteq \eo^\mathscr{A}$ or $\widehat{\mathcal{F}}|_{\eo}\subseteq \eo^\mathscr{P}$, we have $\hF\subseteq \mathscr{A}$ or $\hF\subseteq \mathscr{P}$.
        So $\F$ is $\mathscr{A}$-transformable or $\mathscr{P}$-transformable, contradicting that $\F$ does not satisfy \rm{(\ref{cond: tractable classes})}.

        \item 
        Case 2 in $(\mathcal{PC})$ happens.
        By \cref{lm: conjugation_polymorphism_affine support}, $\widehat{\mathcal{F}}_{\to\eo}$ has affine support.
        Combined with $\widehat{\mathcal{F}}_{\to\eo}\subseteq \eo^\mathscr{A}$ or $\widehat{\mathcal{F}}_{\to\eo}\subseteq \eo^\mathscr{P}$, we have
        $\widehat{\mathcal{F}}_{\to\eo}\subseteq \mathscr{A}$ or  $\widehat{\mathcal{F}}_{\to\eo}\subseteq \mathscr{P}$.
        So $\hF\subseteq \mathscr{A}$ or $\hF\subseteq \mathscr{P}$,
        contradicting \rm{(\ref{cond: tractable classes})}.

        \item 
        Case 3 in $(\mathcal{PC})$ happens.
        Pick an arbitrary $\widehat{f}\in \hF$.
        Write its unique prime factorization $\widehat{f}=\widehat{q_1}\otimes \widehat{q_2}\otimes\cdots\otimes \widehat{q_s}.$
        We first assume $\hF\in \braket{\widehat{\mathcal{M}}}.$
        Then every $\widehat{q_i}$ is supported on Hamming weight at most 1.
        By \cref{lm: conjugate closed and ars closed}, $\hF$ is arrow reversal closed.
        Taking {\sc ars} doesn't change factorization and irreducibility.
        So $\widehat{f}^{\textsc ar}=\widehat{q_1}^{\textsc ar}\otimes \widehat{q_2}^{\textsc ar}\otimes\cdots\otimes \widehat{q_3}^{\textsc ar}\in \hF$, where every $\widehat{q_i}^{\textsc ar}$ is also supported on Hamming weight at most 1.
        This is impossible if some $\widehat{q_i}$ has arity greater than 2.
        Thus, every $q_i$ has arity at most 2.
        So $\widehat{f}\in {\mathscr{T}}$, and $\hF\subseteq \mathscr{T}$.
        After a $K$-transformation,  $\F\subseteq \mathscr{T}$.
        Contradiction.
    \end{enumerate}
\end{proof}

\paragraph{Dichotomy for Quaternary Signatures}
$\newline$

\begin{theorem}[Dichotomy for six-vertex models]\label{thm: six vertex model}
Let $f$ be a 4-ary signature with the signature matrix
$M_{x_1x_2, x_4x_3}=\begin{bmatrix}
0 & 0 & 0 & a\\
0 & b & c & 0\\
0 & z & y & 0\\
x & 0 & 0 & 0\\
\end{bmatrix}$, then
$\Holant(\neq_2\mid f)$ is \#P-hard except for the following cases:
\begin{itemize}
\item $f\in\mathscr{P}$;
\item $f\in\mathscr{A}$;
\item there is a zero in each pair $(a,x), (b,y), (c,z)$;
\end{itemize}
in which cases Holant$(\neq_2\mid f)$ is computable in polynomial time.
For the last tractable case, $f\in \braket{\mathcal{M}}$ or $f\in \braket{X\mathcal{M}}$. 
\end{theorem}

For an arity 4 signature $h$, we say that $h$ has \emph{$H$-form}
if $h\in\mathrm{HW}^{\ge}$ or $h\in\mathrm{HW}^{\le}$, and define its
reduced form by $r(h):=h|_{\eo}$.  Equivalently, $M(h)$ has one of the
following two forms:
\[
\begin{bmatrix}
0&0&0&b\\
0&c&d&\ast\\
0&w&z&\ast\\
y&\ast&\ast&\ast
\end{bmatrix}
\quad\text{or}\quad
\begin{bmatrix}
\ast&\ast&\ast&b\\
\ast&c&d&0\\
\ast&w&z&0\\
y&0&0&0
\end{bmatrix},
\qquad
M(r(h))=
\begin{bmatrix}
0&0&0&b\\
0&c&d&0\\
0&w&z&0\\
y&0&0&0
\end{bmatrix}.
\]
If every right-hand vertex of an instance of
$\holant{\neq_2}{h}$ is assigned $h$, then the total Hamming weight on
the right is exactly half the number of incident edges.  Hence, when
$h$ has $H$-form, every local assignment occurring in a nonzero global
assignment has weight two.  Therefore $\holant{\neq_2}{h}\equiv_T\holant{\neq_2}{r(h)}.$

Thus the reduced problem is a six-vertex model, equivalently a quaternary
$\#\mathrm{EO}$ problem.  We place it together with the following
eight-vertex result since both serve as quaternary vertex-model base
cases used later.

Following~\cite{HuangFu4Regular}, let $\mathscr H$ denote the set of
arity-four signatures $f$ for which $\widehat f$ has $H$-form and
$\holant{\neq_2}{r(\widehat f)}$ is tractable.

\begin{theorem}[Dichotomy for a single quaternary signature
\cite{HuangFu4Regular}]\label{thm: single quaternary dichotomy}
    Let $f$ be an arity-four signature.  Then $\holant{=_2}{f}$ is
    \#P-hard unless one of the following holds:
    \begin{enumerate}
        \item $f\in\langle K\mathcal M\rangle$ or
        $f\in\langle KX\mathcal M\rangle$;
        \item $f$ is $\mathscr A$-transformable;
        \item $f$ is $\mathscr L$-transformable;
        \item $f$ is $\mathscr P$-transformable;
        \item $f\in\mathscr H$;
        \item $f\in\langle\mathscr T\rangle$.
    \end{enumerate}
    In these cases, the problem is computable in polynomial time.
\end{theorem}

\paragraph{Other Hardness Results}
The following is a standard lemma by polynomial interpolation. See for example, \cite{jcbook}.
\begin{restatable}{lemma}{InterpolationBinary}\label{lm: 2 by 2 interpolation}
    Let $f$ and $g$ be nonzero binary signatures with $M(f)=P^{-1}\left[\begin{smallmatrix}
    \lambda_1 & 0\\
    0 & \lambda_2\\
    \end{smallmatrix}\right]P$ and
    $M(g)=P^{-1}\left[\begin{smallmatrix}
    1 & 0\\
    0 & 0\\
    \end{smallmatrix}\right]P$ for some invertible matrix $P$.
    If $\lambda_1\neq 0$ and $\frac{\lambda_2}{\lambda_1}$ is not a root of unity,
    then $$\Holant(g, \mathcal{F})\leqslant_T\Holant(f, \mathcal{F})$$ for any signature set $\mathcal{F}$.
\end{restatable}

\begin{theorem}[\cite{real_holantc}]\label{thm: CSP2}
    Let $\mathcal{F}$ be any set of signatures. 
    Then $\#\operatorname{CSP}_2(\mathcal{F})$ is $\#\operatorname{P}$-hard unless $\mathcal{F}\subseteq\mathscr{A}$ or $\mathcal{F}\subseteq\alpha\mathscr{A}$ or $\mathcal{F}\subseteq\mathscr{P}$ or $\mathcal{F}\subseteq\mathscr{L}$, in which cases the problem is computable in polynomial time.
\end{theorem}

\begin{lemma}[\cite{jcbook}]\label{lm: CSP2 to holant =4}
    $\#\mathrm{CSP}_2(\mathcal{F})\leqslant_T\Holant(=_4, \mathcal F)$.
\end{lemma}

\begin{restatable}{lemma}{CSPwithEQ}\label{lm: CSP with =2}
    Suppose $\mathcal{F}$ doesn't satisfy \eqref{cond: tractable classes}.
    Then for every $T\in \mathrm{GL}_2(\C)$, $\#\mathrm{CSP}_2(T(\mathcal{F}\cup \{=_2\}))$ is \#P-hard.
\end{restatable}

\begin{restatable}{lemma}{DisFourgivesHard}\label{lm: disequality 4 gives dichotomy}
    Let $\F$ be a set of signatures. Then $\holant{\neq_2}{\neq_4,\hF}$ is either polynomial-time solvable or $\#\operatorname{P}$-hard.
\end{restatable}

\begin{restatable}{lemma}{ConjuDisFourHard}\label{lem: In conjugate closed setting disequality 4 gives hardness}
    Let $\F$ be a conjugate closed signature set. If $\F$ does not satisfy the condition \rm{(\ref{cond: tractable classes})}, then $\holant{\neq_2}{\neq_4,\hF}$ is $\#\operatorname{P}$-hard.
\end{restatable}

The proofs of~\cref{lm: CSP with =2}, ~\cref{lm: disequality 4 gives dichotomy} and~\cref{lem: In conjugate closed setting disequality 4 gives hardness} are given in \autoref{app: disequality reduction}.

\subsection{Signatures and Quantum States}\label{subsec: signatures and quantum states}
Similar to \cref{def-first-order}, we can define $k$-th order orthogonality for a signature of arity $n\ge 2k$:
\begin{definition}[$k$-th order orthogonality]\label{def: k-th order orthogonality}
    A signature $f$ of arity $n$ satisfies $k$-th order orthogonality, denoted by {\sc $k$th-Orth} ($n\ge 2k$), if there exists some $\mu\neq 0$ such that for all $S\subseteq [n]$ of size $k$, $ M_S(f)M_S(f)^\dagger=\mu I_{2^k}.$
\end{definition}

Similar to \cref{prop: invariant 1st-Orth}, we have the following invariance property:
\begin{proposition}\label{prop: invariant kth Orth}
    {\sc $k$th-Orth} is invariant up to conjugation and holographic transformation by $T\in \mathbf{U}_2(\mathbb{C})$.
\end{proposition}

The proof of~\cref{prop: invariant kth Orth} is in~\cref{subapp: proof of subsec signatures and quantum states}. In particular, every real orthogonal matrix
$O\in\mathbf{O}_2$ is unitary, and hence {\sc $k$th-Orth} is invariant up to orthogonal transformation by $O\in\mathbf{O}_2$.
Together with \cref{lem: invariance under real orthogonal transformation},
this shows that a real orthogonal transformation preserves conjugate
closure, failure of condition~\eqref{cond: tractable classes},
irreducibility, and {\sc $k$th-Orth}.

There is a natural interpretation of the $k$-th order orthogonality if we view $f$ as a quantum state.

Given two signatures $f$ and $g$ both of arity $n$, their inner product is defined as $\braket{f,g}=\sum_{x\in\{0,1\}^n}f(x)\overline{g(x)}.$
Define $||f||=\sqrt{\braket{f,f}}.$
A nonzero arity-$n$ signature $f$ determines the normalized $n$-qubit
state
\[
\ket f=\frac{1}{\|f\|}
 \sum_{x\in\{0,1\}^n}f(x)\ket x.
\]
For $S\subseteq[n]$, let $\overline S=[n]\setminus S$ and define the
reduced density matrix
\(
\rho_S=\operatorname{Tr}_{\overline S}(\ket f\!\bra f)
\)
(see \autoref{sec: AME62 has a common local unitary normal} for the exact definition and more details).
The state $\ket f$ is \emph{$k$-uniform} if
$\rho_S=2^{-k}I_{2^k}$ for every $S\subseteq[n]$ with $|S|=k$.
Equivalently, $f$ satisfies \textsc{$k$th-Orth}.  
An $\operatorname{AME}(n,2)$ state is a $\lfloor n/2\rfloor$-uniform $n$-qubit state.  
Equivalently, we say a non-zero arity $n$ signature $f$ is an AME($n,2$) state if $f$ satisfies {\sc $\lfloor n/2\rfloor$th-Orth}.
We use a signature and its
associated quantum state interchangeably when no confusion can arise.  

\paragraph{Pauli representation}
Define the Hilbert–Schmidt inner product on $\mathbb{C}^{n\times n}$ by $\braket{A,B}:=\mathrm{Tr}(A^\dagger B).$
The four Pauli matrices are 
$$
\sigma_0=I=\left[\begin{matrix}
    1 & 0\\
    0 & 1\\
\end{matrix}\right],
\sigma_1=X=\left[\begin{matrix}
    0 & 1\\
    1 & 0\\
\end{matrix}\right],
\sigma_2=Y=\left[\begin{matrix}
    0 & -\mathfrak i\\
    \mathfrak i & 0\\
\end{matrix}\right],
\sigma_3=Z=\left[\begin{matrix}
    1 & 0\\
    0 & -1\\
\end{matrix}\right]
$$
such that $\braket{\sigma_i,\sigma_j}=2\delta_{ij}$.
Thus they form an orthogonal basis of the vector space of $\mathbb{C}^{2\times 2}.$
More generally, the $4^m$ matrices $\sigma_{i_1}\otimes\cdots\otimes  \sigma_{i_m}$ form an orthogonal basis of $\mathbb{C}^{2^m\times 2^m}$.
Therefore every $m$-qubit density matrix has a unique expansion
$$
    \rho=\frac{1}{2^m}\sum_{i_1,\ldots,i_m}r_{i_1,\ldots,i_m}\sigma_{i_1}\otimes \cdots\otimes \sigma_{i_m},
$$
where $r_{i_1,\ldots,i_m}=\mathrm{Tr}((\sigma_{i_1}\otimes\cdots\otimes\sigma_{i_m})\rho)=\mathrm{Tr}(\rho(\sigma_{i_1}\otimes\cdots\otimes\sigma_{i_m}))$.

\subsection{Algebra Toolkits}
\paragraph{Homomorphism from $\mathbf{SU}(2)$ to $\mathbf{SO}(3)$}\label{para: SU2 and SO3}

We briefly recall the standard double-covering map from $\mathbf{SU}(2)$ to $\mathbf{SO}(3).$
See for example, \cite{Hall2015LieGroups}.
For $\mathbf{x}=(x_1,x_2,x_3)^{\tt T}\in\mathbb R^3$, write
\[
    \sigma(\mathbf{x})=\mathbf{x}\mathbin{\cdot}\boldsymbol{\sigma}
    :=x_1\sigma_1+x_2\sigma_2+x_3\sigma_3=\left[\begin{matrix}
        x_3 & x_1-\mathfrak{i}x_2\\
        x_1+ \mathfrak{i}x_2 & -x_3
    \end{matrix}\right].
\]
The map $\mathbf{x}\mapsto\sigma(\mathbf{x})$ identifies $\mathbb R^3$ with the real vector
space $V$ of traceless Hermitian $2\times 2$ matrices.  
For $U\in\mathbf{SU}(2)$, define a linear map $\Phi_U:V\to V$ by
$
    \Phi_U(X)=UXU^{\dagger}.
$

\begin{proposition}[\cite{Hall2015LieGroups}]
    The map $U\mapsto \Phi_U$ is a 2-1 and onto map of $\mathbf{SU}(2)$ to $\mathbf{SO}(3)$, with kernel equal to $\{I,-I\}$. i.e., $\mathbf{SU}(2)/\{\pm I\}\cong \mathbf{SO}(3).$
\end{proposition}

For $U\in \mathbf{SU}(2)$, let $R_U\in \mathbf{SO}(3)$ be the matrix of the linear map $\Phi_U.$ 
By definition, we have $U(x_1\sigma_1+x_2\sigma_2+x_3\sigma_3)U^\dagger=R_U(x_1,x_2,x_3)^{\tt T}\cdot \boldsymbol{\sigma},\,\forall (x_1,x_2,x_3)\in \mathbb{R}^3.$
So $U\sigma_pU^{\dagger}=\sum_{r=1}^3(R_U)_{rp}\sigma_r$, for $p=1,2,3.$

\section{Main Theorem and Proof Outline}

\begin{theorem}[Main theorem]\label{thm: main theorem}
    Let $\mathcal{F}$ be a complex signature set that is closed under conjugation. 
    Then $\Holant(\mathcal{F})$ is \#P-hard unless $\mathcal{F}$ satisfies {\rm (\ref{cond: tractable classes})}, in which case it is tractable.
\end{theorem}

\paragraph{Proof Overview}
The tractable cases of the main theorem follow from \cref{thm: tractable classes}.
It remains to prove that $\Holant(\mathcal{F})$ is \#P-hard if $\mathcal{F}$ does not satisfy condition~(\ref{cond: tractable classes}).

We first generalize the second order orthogonality ({\sc 2nd-Orth}) in \cite{realholant} to the conjugate closed setting.
This requires a genuinely new argument heavily depending on a recent dichotomy on quaternary signatures (\cref{thm: single quaternary dichotomy}).
Throughout the rest of the proof, {\sc 2nd-Orth} plays an important role.

Suppose $\mathcal{F}$ fails condition \eqref{cond: tractable classes}.
In particular, $\mathcal{F}\not\subseteq \mathscr{T}$.
Since $\mathcal{U}^{\otimes}\subseteq \mathscr{T}$, we also have $\mathcal{F}\not\subseteq \mathcal{U}^{\otimes}$.
Thus, there is a nonzero signature $f\in\mathcal{F}$ of arity $2n$ such that $f\notin\mathcal{U}^{\otimes}$.
We prove \#P-hardness by induction on $2n$.
The base case $2n=2$ is handled by \textsc{1st-Orth} (\cref{lm: not 1st-orth is hard}).
The case $2n=4$ is handled by \textsc{2nd-Orth} and the nonexistence of an $\operatorname{AME}(4,2)$ state (\cref{lm: base case 2n=4}).

For higher arities, we aim to realize a signature $f'$ from $f$ with lower arity while keeping the property $f'\notin\mathcal{U}^{\otimes}$.
Regular approaches for arity reduction are merging, mating, and factorization.
However, this approach does not always succeed due to the existence of the special signatures $f_6$ and $f_8$ discovered in~\cite{realholant}.
These signatures themselves are not in $\mathcal{U}^{\otimes}$, but the lower-arity signatures obtained from them by merging, mating, and factorization all lie in $\mathcal{U}^{\otimes}$.
We therefore need separate proofs for $2n=6$ and $2n=8$.

For $2n=6$, our idea is to first realize $f_6$, and then prove that $\Holant(f_6,\mathcal{F})$ is \#P-hard under the assumption that $\mathcal{F}\not\subseteq\mathscr{A}$.
In \cref{sec: Third Order Orth}, we first realize an $\operatorname{AME}(6,2)$ state, or equivalently, a signature satisfying \textsc{3rd-Orth}.
We use the framework of Xia~\cite{MingjiProgram}, which shows that, unless hardness already follows, the ``projective binary group'' generated by $f$ is a finite subgroup of $\mathbf{PU}(2)=\mathbf{SU}(2)/\{\pm I_2\}$.
We give an exposition of this proof in \autoref{appendix: mingji}.
Since $\mathbf{SU}(2)/\{\pm I_2\}\cong\mathbf{SO}(3)$, we use the classification of finite subgroups of $\mathbf{SO}(3)$ to consider all possible cases of the projective binary group in \cref{sec: Third Order Orth}.
For the polyhedral group case (\cref{subsec: polyhedral groups}), the large dihedral group case (\cref{subsec: dihedral group}), and the standard Klein group case (\cref{subsubsec: Standard Case Klein}), we prove that a signature satisfying \textsc{3rd-Orth} can be realized.
We also prove that the non-standard Klein group case (\cref{subsubsec: Non-standard Klein}) and the cyclic group cases (\cref{subsec: Trivial group,subsec: Order Two group,subsec: Cyclic group of Order at least three}) cannot occur.

Following the realization of an $\operatorname{AME}(6,2)$ state in~\cref{sec: Third Order Orth}, we show how to recover $f_6$ in the~\cref{sec:recover-f6}. By the uniqueness of $\operatorname{AME}(6,2)$ states up to local
unitary transformations and input permutations, which is proved in~\cref{sec: AME62 has a common local unitary normal}, we can realize a signature that is locally unitarily equivalent to $f_6$. However, the corresponding local unitary matrices may not be able to realized. This makes us to find the relations between the unitary matrices. Thus, in~\cref{lem:f6-common-orthogonal}, we show that either $\Holant(\mathcal{F})$ is \#P-hard, or a real orthogonal holographic transformation places these matrices projectively in the tetrahedral group. After the transformation, in~\cref{lem:f6-group-alternatives}, we then proved that the projective binary group of the $\Holant$ problem: it is either the tetrahedral group $\Gamma$ or a subgroup of $K_4$. We deal with the two cases in~\cref{lem:f6-tetrahedral-case,lem:f6-klein-case}, both cases give us the realization of $f_6$ and four Bell signatures. 

In \cref{sec: f6 is hard}, we prove hardness when $f_6$ is available.
The proof has two cases.
If every signature in $\mathcal{F}$ has parity, we apply a technique called ``realification" to prove every signature in $\mathcal{F}$ is \emph{real} up to a complex phase, and then apply the real Holant dichotomy (\cref{subsec: f6 parity}).
Otherwise, we reduce the arity while preserving the non-parity property until we obtain a binary signature without parity.
We then extend the ``non-$\mathcal{B}$ hard'' property of $f_6$ in~\cite{realholant} to the ``non-$\Gamma$ hard'' property and complete the proof using a minimal counterexample argument (\cref{subsec: f6 non-parity}).

For $2n=8$, we first realize $f_8$ in \cref{subsec: realize f8 Reed Muller}.
In \cref{subsec: f8 available hardness}, we prove that $\Holant(f_8,\mathcal{F})$ is \#P-hard under the assumption that $\mathcal{F}\not\subseteq\mathscr{A}$ .
The proof follows the framework for real Holant problems~\cite{realholant} and again uses a realification argument similar to that in \cref{subsec: f6 parity}.
Finally, for $2n\geq 10$, 
we can reduce the arity while keeping the resulting signature outside $\mathcal{U}^{\otimes}$, thus completing the induction (\cref{sec: 2n >= 10}).

\section{Second Order Orthogonality}\label{sec: second order orthoganality}
In this section, we will prove that if $\mathcal{F}$ contains an irreducible signature that doesn't satisfy {\sc 2nd-Orth}, then $\Holant(\mathcal{F})$ is \#P-hard (\cref{lm: not 2nd-orth is hard}).
First, \cref{def: k-th order orthogonality} specializes to the following when $k=2$.
\begin{definition}[Second order orthogonality \cite{realholant}]\label{def: second-order-orth}
    Let $f$ be a complex-valued signature of arity $n \geqslant 4$. It satisfies the \emph{second order orthogonality ({\sc 2nd-Orth})} if there exists some $\lambda\neq 0$ such that for all pairs of indices $\{i, j\}\subseteq [n]$, the entries of $f$ satisfy 
    \begin{equation*}
        |{\bf f}_{ij}^{00}|^2=|{\bf f}_{ij}^{01}|^2=|{\bf f}_{ij}^{10}|^2=|{\bf f}_{ij}^{11}|^2=\lambda, 
    \text{ ~~~~\rm and }~~~~ \langle{\bf f}_{ij}^{ab}, {\bf f}_{ij}^{cd}\rangle =0 ~~~~\text{\rm  for  all } (a, b)\neq (c, d).
    \end{equation*}
\end{definition}

Similar to \cref{lm: 1st-Orth uniform}, one can easily prove:
\begin{lemma}\label{lm: 2nd-Orth uniform}
    Let $f$ be a signature of arity $n$.
    If for all indices $i\neq j\in [n]$, $M(\mathfrak{m}_{ij}f)= \lambda_{ij} I_4$ for some real $\lambda_{ij}\neq 0$, then $f$ satisfies {\sc 2nd-Orth} (i.e., all $\lambda_{ij}$ have the same value).
\end{lemma}

We first record a consequence of {\sc 2nd-Orth} that will be frequently used in the subsequent sections.
\begin{lemma}\label{lm: 2nd-orth contraction nonzero}
    Suppose $f$ satisfies {\sc 2nd-Orth}, and $b$ is a non-zero binary signature.
    Then for every pair of indices $\{i,j\}$, $\partial_{ij}^bf$ is a non-zero signature.
    Moreover, $|\partial_{ij}^bf|^2=\lambda|b|^2$, where $\lambda=|f_{ij}^{00}|^2=|f_{ij}^{01}|^2=|f_{ij}^{10}|^2=|f_{ij}^{11}|^2>0.$
\end{lemma}
\begin{proof}
    $$|\partial_{ij}^{b}f|^2
    =
    \sum_{r,s,p,q\in\{0,1\}}
    b(r,s)\overline{b(p,q)}
    \langle f_{ij}^{rs},f_{ij}^{pq}\rangle=\sum_{r,s\in\{0,1\}}|b(r,s)|^2\braket{f_{ij}^{rs},f_{ij}^{rs}}=|b|^2\cdot\lambda$$
    by {\sc 2nd-Orth}, where $\lambda=|f_{ij}^{00}|^2=|f_{ij}^{01}|^2=|f_{ij}^{10}|^2=|f_{ij}^{11}|^2>0$.
\end{proof}

Recall $\mathfrak{m}_{ij}f$ is the 4-arity signature representing the mating of $f$ and $\overline{f}$ by their corresponding variables except for $x_i$ and $x_j$.
Note that $f$ satisfies {\sc 2nd-Orth} iff $M(\mathfrak{m}_{ij}f)=\lambda I_4$ for some $\lambda>0$.
In the setting of $\holant{\neq_2}{\hF}$, the mating operation corresponds to connecting variables except for $x_i$ and $x_j$ of $\widehat{f}$ and $\widehat{\overline{f}}$ using $\neq_2$.
We denote the resulting signature by $\widehat{\mathfrak m}_{ij} \widehat{f}$.
Clearly $\widehat{\mathfrak m}_{ij} \widehat{f}=\widehat{\mathfrak{m}_{ij}f}$, so $f$ satisfies {\sc 2nd-Orth} iff $M(\widehat{\mathfrak m}_{ij} \widehat{f})=\lambda N_4$ for some $\lambda>0$.

In the following, we assume that $\mathcal{F}$ is conjugate closed and doesn't satisfy condition \rm{(\ref{cond: tractable classes})}.
Recall that $\hF$ is arrow reversal closed by \cref{lm: conjugate closed and ars closed}.
We want to show that every irreducible signature of arity $2n\ge 4$ must satisfy {\sc 2nd-Orth}, otherwise $\Holant(\mathcal{F})$ is \#P-hard.
First we can realize the 4-ary signature $\mathfrak{m}_{ij}f$ for every pair of indices $\{i,j\}$.
We consider all possible cases that $\mathfrak{m}_{ij}f$ falls into the tractable classes of \cref{thm: single quaternary dichotomy}.

\subsection{Unary and Binary Tensor}
We first consider the case that $\mathfrak{m}_{ij}f\in \braket{\mathcal{T}}$ for some pair of indices $\{i,j\}.$

\begin{lemma}\label{lm: 2nd-orth mating reducible}
    Let $f\in \mathcal{F}$ be a signature of arity $2n\ge 4.$
    If there exists pair of indices $\{i,j\}$ such that $\mathfrak{m}_{ij}f\in \braket{\mathcal{T}}$, then one of the following holds:
    \begin{itemize}
        \item $\Holant(\mathcal{F})$ is \#P-hard.
        \item $f$ is reducible.
        \item $M(\mathfrak{m}_{ij}f)=\lambda_{ij}I_4$ for some $\lambda_{ij}>0.$
    \end{itemize}
\end{lemma}
\begin{proof}
    If $\mathfrak{m}_{ij}f$ fails {\sc 1st-Orth}, then $\Holant(\mathcal{F})$ is \#P-hard by \cref{lm: not 1st-orth is hard}.
    Now we assume that $\mathfrak{m}_{ij}f$ satisfies {\sc 1st-Orth}.
    If $\mathfrak{m}_{ij}f$ has a unary factor $u$, $u$ is realizable by \cref{lm: decomposition}.
    By \cref{lm: odd holant with conjugation}, $\Holant(\mathcal{F}
    )$ is \#P-hard.
    In the following, assume that $\mathfrak{m}_{ij}f=b_1\otimes b_2$ for two binary signatures $b_1$ and $b_2$.
    Let $R=M(\mathfrak{m}_{ij}f),\,U=M(b_1)$ and $V=M(b_2)$.
    Since $\mathfrak{m}_{ij}f=b_1\otimes b_2$, there are three possibilities.
    \begin{itemize}
        \item $R_{ab,cd}=U_{ab}V_{cd}$ for every $\{a,b,c,d\}\in\{1,2\}.$
        Then $\mathrm{rank}(R)=1.$
        Since $\mathrm{rank}(R)=\mathrm{rank}(M(\mathfrak{m}_{ij}f))=\mathrm{rank}(M_{ij}(f)M_{ij}(f)^\dagger)=\mathrm{rank}(M_{ij}f)$, we have $\mathrm{rank}(M_{ij}f)=1.$
        Therefore, $f$ factors into two binaries.
        Thus $f$ is reducible.
        
        \item $R_{ab,cd}=U_{ac}V_{bd}$ for every $\{a,b,c,d\}\in\{1,2\}.$
        Then $R=U\otimes V.$
        Recall that $R_{ab,cd}=\braket{\widehat{\mathbf{f}}_{ij}^{ab},\widehat{\mathbf{f}}_{ij}^{cd}}.$
        By {\sc 1st-Orth} at $x_i$, there exist some $\mu>0$ such that $\mu \delta_{ac}=\braket{\widehat{\mathbf{f}_i^a},\widehat{\mathbf{f}_i^c}}=R_{a1,c1}+R_{a2,c2}=U_{ac}(V_{11}+V_{22})$, for every $a,c\in\{1,2\}.$
        Set $s=V_{11}+V_{22}$.
        Taking $a=c=1$ gives $\mu=U_{11}s>0$, so $s\neq 0.$
        Therefore $U_{11}=U_{22}=\frac{\mu}{s}$ and $U_{12}=U_{21}=0$.
        Thus $U=uI_2$ for some $u\neq 0.$
        Similarly, {\sc 1st-Orth} at $x_j$ gives $\mu\delta_{bd}=(U_{11}+U_{22})V_{bd}=2uV_{bd}.$
        Hence $V=\frac{\mu}{2u}I_2.$
        So $R=U\otimes V=\frac{\mu}{2}I_4$, $\mu>0.$
        It follows that $M(\widehat{\mathfrak{m}}_{ij}\widehat{f})= \frac{\mu}{2} N_4$.
        Taking $\lambda_{ij}=\frac{\mu}{2}$ completes this case.
    
        \item $R_{ab,cd}=U_{ad}V_{bc}$ for every $\{a,b,c,d\}\in\{1,2\}.$
        {\sc 1st-Orth} at $x_i$ says there exists some $\mu>0$, \(\mu\delta_{ac}=R_{a1,c1}+R_{a2,c2}=U_{a1}V_{1c}+U_{a2}V_{2c}=(UV)_{ac}\), for every $\{a,c\}\in\{1,2\}.$
        Thus $UV=\mu I_2,\,\det(U)\det(V)=\mu^2.$
        Consider $R'=U\otimes V$, so $R'_{ab,cd}=U_{ac}V_{bd}.$
        Exchanging the two middle columns of $R'$ indexed by $12$ and $21$ switches the column index $c$ and $d$.
        So we obtain $U_{ad}V_{bc}=R_{ab,cd}.$
        Thus, $\det(R)=-\det(R')=-\det(U)^2\det(V)^2=-\mu^4<0.$
        However, $\det(R)=\det(M(\mathfrak{m}_{ij}f))=\det(M_{ij}(f)M_{ij}(f)^\dagger)\ge0$, contradiction.
        So this case cannot happen.
    \end{itemize}
    In all above cases, either $f$ is reducible, or $M(\mathfrak{m}_{ij}f)=\lambda_{ij}I_4$ for some $\lambda_{ij}>0.$
    The lemma is proved.
\end{proof}

In the following, we may assume that $\mathfrak{m}_{ij}f$ is irreducible for every pair of indices $\{i,j\}$.

\subsection{Affine, Product, Local-affine Transformable}

In this subsection, we consider the case that some irreducible $\mathfrak{m}_{ij}f$ is $\mathscr{A}$- or $\mathscr{P}$- or $\mathscr{L}$-transformable.
We prove either $\Holant(\mathcal{F})$ is \#P-hard or $M(\mathfrak{m}_{ij}f)=\lambda_{ij}I_4$ for some $\lambda_{ij}>0.$

\begin{lemma}\label{lm: APL-transformable singular value}
    Let $h$ be an irreducible signature of arity 4, and $h$ satisfies {\sc 1st-Orth}.
    If $h$ is $\mathscr{A}$- or $\mathscr{P}$- or $\mathscr{L}$-transformable, then for any $i,j\in[4]$, the 4 by 4 signature matrix $M_{ij}(h)$ has equal non-zero singular values and $\mathrm{rank}(M_{ij}(h))\in \{2,4\}.$
\end{lemma}
\begin{proof}
    By \cref{lm: APL-transformable}, there exists $T_1,T_2,T_3,T_4\in \mathrm{GL}_2(\C)$ such that $h=(T_1\otimes T_2 \otimes T_3 \otimes T_4)a$, where $a\in \mathscr{A}.$
    By \cref{lm: affine is 1st-orth}, $a$ satisfies {\sc 1st-Orth}.
    Consider the singular value decomposition for every $T_i, \,1\le i\le 4$.
    We have $T_i=U_i\mathrm{diag}(\sigma_{i,1},\sigma_{i,2})V_i^\dagger$, where $U_i,V_i$ are 2 by 2 unitary matrices, and $\sigma_{i,1}\ge\sigma_{i,2}>0$ are the singular values of $T_i.$
    Let $c_i=\sqrt{\sigma_{i,1}\cdot \sigma_{i,2}}$, $r_i=\sqrt{\sigma_{i,1}/\sigma_{i,2}}$, $D_i=\mathrm{diag}(r_i,r_i^{-1})$, we have $T_i=c_iU_i D_i V_i^\dagger$ for $1\le i\le 4.$
    Let $c=c_1c_2c_3c_4$. 
    Then
    $$
    h=c(U_1\otimes \cdots \otimes U_4)(D_1\otimes \cdots \otimes D_4)(V_1^\dagger \otimes \cdots \otimes V_4^\dagger)a.
    $$
    Let $f=(V_1^\dagger\otimes \cdots\otimes V_4^\dagger)a$ and $g=c^{-1}(U_1^\dagger\otimes \cdots \otimes U_4^\dagger)h$, we have
    \begin{equation}\label{eq: g = Df}
        g=(D_1\otimes \cdots \otimes D_4)f,\qquad D_i=\mathrm{diag}(r_i,r_i^{-1}).
    \end{equation}
    Since scalar and binary extension by unitary matrices do not change {\sc 1st-Orth}, both $f$ and $g$ satisfies {\sc 1st-Orth}.
    
    We next show $f=g$.
    Let $x=(x_1,x_2,x_3,x_4)\in \{0,1\}^{4}.$
    By \eqref{eq: g = Df}, $g(x)=d(x)f(x)$, where $d(x)=\prod_{i=1}^4r_i^{1-2x_i}>0$.
    Because $f$ satisfies {\sc 1st-Orth}, $||f_i^0||^2=||f_i^1||^2$ for every $1\le i\le 4$.
    So $\sum_{x\in\{0,1\}^4,\,x_i=0}|f(x)|^2=\sum_{x\in\{0,1\}^4,\,x_i=1}|f(x)|^2$, for every $1\le i\le 4.$
    Define $p(x)=|f(x)|^2/||f||^2$ for every $x\in\{0,1\}^4.$
    Then 
    $$
    \sum_{x\in\{0,1\}^4}(1-2x_i)p(x)=\sum_{x\in\{0,1\}^4,x_i=0}p(x)-\sum_{x\in\{0,1\}^4,x_i=1}p(x)=0,\quad \forall i, 1\le i\le 4.
    $$
    Thus
    \begin{equation}\label{eq: ||g|| over ||f||}
        \begin{aligned}
            \frac{||g||^2}{||f||^2}=&\frac{\sum_{x\in\{0,1\}^4}d(x)^2|f(x)|^2}{||f||^2}
            =\sum_{x\in\{0,1\}^4}p(x)d(x)^2\ge \prod_{x\in\{0,1\}^4}(d(x)^2)^{p(x)}\\
            =&\prod_{x\in\{0,1\}^4}\prod_{i=1}^4(r_i^{1-2x_i})^{2p(x)}
            =\prod_{i=1}^4r_i^{2\sum_{x\in\{0,1\}^4}(1-2x_i)p(x)}=1.
        \end{aligned}
    \end{equation}
    The inequality is by AM-GM.
    By \eqref{eq: ||g|| over ||f||}, we know $||g||^2\ge ||f||^2.$

    Conversely, we have 
    \begin{equation}\label{eq: f = Dg}
        f=(D_1^{-1}\otimes \cdots \otimes D_4^{-1})g,\qquad D_i^{-1}=\mathrm{diag}(r_i^{-1},r_i).
    \end{equation}
    Using the fact that $g$ satisfies {\sc 1st-Orth}, the same argument above gives $||f||^2\ge ||g||^2.$
    So $||f||=||g||,$ and the inequality in AM-GM becomes equality.
    So $d(x)^2=c$ for all $x$ with $p(x)>0$.
    Notice that $1=\prod_{x\in\{0,1\}^4}d(x)^{2p(x)}=\prod_{x\in\{0,1\}^4}c^{p(x)}=c^{\sum_{x\in\{0,1\}^4}p(x)}=c$.
    Also, $p(x)>0\iff x\in \mathscr{S}(f)$.
    So $d(x)=1$ for every $x\in \mathscr{S}(f).$
    Therefore, $f(x)=g(x)$.

    Recall $f=(V_1^\dagger\otimes \cdots\otimes V_4^\dagger)a$ and $g=c^{-1}(U_1^\dagger\otimes \cdots \otimes U_4^\dagger)h$.
    So $h$ is a scalar multiple of local unitary transformation of \(a\). 
    This doesn't change the rank and singular values of the signature matrix.
    By \cref{lm: affine singular value}, for any $\{i,j\}\in [4]$, $M_{ij}(a)$ has equal non-zero singular values, and $\mathrm{rank}(M_{ij}(a))\in\{1,2,4\}.$
    So $M_{ij}(h)$ has equal non-zero singular values, and $\mathrm{rank}(M_{ij}(h))\in\{1,2,4\}.$

    It remains to exclude that case where $\mathrm{rank}(M_{ij}(h))=1.$
    If so, $M_{ij}(h)=\mathbf{uu^{\tt T}}$ for some $\mathbf{u}\in \mathbb{C}^4.$
    Then $h$ is a tensor product of two binary signatures, both with truth table $\mathbf{u}.$
    This contradicts $h$ is irreducible.
    The lemma is proved.
\end{proof}

The following lemma proves $M(\mathfrak{m}_{ij}f)=\lambda_{ij}I_4$ when $\mathrm{rank}(M(\mathfrak{m}_{ij}f))=4.$

\begin{lemma}\label{lm: APL-trans rank 4}
    Let $f\in \mathcal{F}$ be a signature of arity $2n\ge 4$.
    If there exists a pair of indices $\{i,j\}$ such that (1) $\mathfrak{m}_{ij}f$ is irreducible and satisfies {\sc 1st-Orth}; (2) $\mathfrak{m}_{ij}f$ is $\mathscr{A}$- or $\mathscr{P}$- or $\mathscr{L}$-transformable; (3) $\mathrm{rank}(M(\mathfrak{m}_{ij}f))=4$, then $M(\mathfrak{m}_{ij}f)=\lambda_{ij}I_4$ for some $\lambda_{ij}>0.$
\end{lemma}
\begin{proof}
    Let $h=\mathfrak{m}_{ij}f$.
    Apply \cref{lm: APL-transformable singular value}, $M_{ij}(h)$ has 4 equal singular values $\lambda_{ij}$.
    Because $M_{ij}(h)=M(\mathfrak{m}_{ij}f)=M_{ij}(f)M_{ij}(f)^\dagger$ is Hermitian and positive semi-definite, its eigenvalues and singular values are the same, and are all equal to $\lambda_{ij}\ge 0$.
    Since $\mathrm{rank}(M_{ij}(h))=4$, $\lambda_{ij}>0.$
    By spectral decomposition of Hermitian matrix, there exists unitary matrix $U$ such that  $M_{ij}(h)=U\mathrm{diag}(\lambda_{ij},\lambda_{ij},\lambda_{ij},\lambda_{ij})U^{\dagger}=\lambda_{ij}UU^\dagger=\lambda_{ij}I_4$.
    The lemma is proved.
\end{proof}

In the following, we will show $\Holant(\mathcal{F})$ is \#P-hard if $\mathrm{rank}(M(\mathfrak{m}_{ij}f))=2$ for some pair of indices $\{i,j\}.$

\begin{definition}
Let $R\in\mathbb C^{4\times4}$ have rows and columns indexed
by pairs of bits. 
Its partial traces are the $2\times2$ matrices
defined by
\[
  (\operatorname{Tr}_1R)_{b,d}
  :=\sum_{a\in\{0,1\}}R_{ab,ad},
  \qquad
  (\operatorname{Tr}_2R)_{a,c}
  :=\sum_{b\in\{0,1\}}R_{ab,cb}.
\]
For an arity 4 signature $g$, we write
$\operatorname{Tr}_k(g):=\operatorname{Tr}_k(M(g))$ for $k=1,2$,
where $M(g)_{ab,cd}=g(a,b,c,d)$.
\end{definition}

\begin{lemma}\label{lem:mating-partial-traces}
Let $f$ be a signature of arity $2n\ge 4$ satisfying {\sc 1st-Orth}.
Then there exists some $\mu>0$, for every pair of distinct indices $\{i,j\}$,
\(
  \operatorname{Tr}_1(\mathfrak m_{ij}f)
  =\operatorname{Tr}_2(\mathfrak m_{ij}f)
  =\mu I_2.
\)
\end{lemma}

\begin{proof}
By permuting variables, assume that $(i,j)=(1,2)$.
Since $f$ satisfies {\sc 1st-Orth}, there exists $\mu>0$ such that
$M_1(f)M_1(f)^\dagger=M_2(f)M_2(f)^\dagger=\mu I_2$.

Put $R=M(\mathfrak m_{12}f)$. By the definition of mating,
$R_{ab,cd}=\sum_{x\in\{0,1\}^{2n-2}}
f(a,b,x)\overline{f(c,d,x)}$.
Consequently, $(\operatorname{Tr}_2R)_{a,c}=\sum_{b\in\{0,1\},x\in\{0,1\}^{2n-2}}
f(a,b,x)\overline{f(c,b,x)}$,
which is precisely the $(a,c)$ entry of $M_1(f)M_1(f)^\dagger$.
Thus $\operatorname{Tr}_2R=\mu I_2$.
Summing over the first coordinate instead gives
$\operatorname{Tr}_1R=M_2(f)M_2(f)^\dagger=\mu I_2$,
proving the claim.
\end{proof}

\begin{lemma}\label{lm:rank-two-matrix-form}
Let $h$ be an arity 4 signature.
Suppose that $M(h)$ is Hermitian, positive semi-definite of rank two with two nonzero equal eigenvalues, and
$\operatorname{Tr}_1R=\operatorname{Tr}_2R=\mu I_2$
for some $\mu>0$.
Then there exist traceless Hermitian matrices
$A,B\in\mathbb C^{2\times2}$ with $A^2=B^2=I_2$ such that
\(
  M(h)=\frac{\mu}{2}(I_4+A\otimes B).
\)
\end{lemma}

\begin{proof}
Let $R=M(h)$, and $\lambda>0$ be the common nonzero eigenvalue of $R$.
Since $\mathrm{rank}(R)=2$, we have $\operatorname{Tr}R=2\lambda$.
Also,
$\operatorname{Tr}R=\operatorname{Tr}(\operatorname{Tr}_1R)=2\mu$,
so $\lambda=\mu$.
By the spectral decomposition of Hermitian matrices, we have $R=U\mathrm{diag}(\lambda,\lambda,0,0)U^{\dagger}$ for some unitary matrix $U$.
Hence $R^2=\lambda R=\mu R$.
Let $W:=2R/\mu-I_4$.
It is Hermitian, satisfies $W^2=I_4$,
and has both partial traces zero.

Writing $W$ in $2\times2$ blocks, Hermitian symmetry and
$\operatorname{Tr}_1W=0$ give
\[
  W=\begin{pmatrix}
      D&C\\
      C^\dagger&-D
    \end{pmatrix},
  \qquad D=D^\dagger.
\]
The condition $\operatorname{Tr}_2W=0$ further gives
$\operatorname{Tr}D=\operatorname{Tr}C=0$.
Comparing blocks in $W^2=I_4$, we obtain
$D^2+CC^\dagger=D^2+C^\dagger C=I_2$ and $DC=CD$.
Therefore $C$ is normal and commutes with the Hermitian matrix $D$.

By simultaneous unitary diagonalization, there is a unitary matrix
$V$ such that
$V^\dagger DV=\operatorname{diag}(d,-d)$ and
$V^\dagger CV=\operatorname{diag}(z,-z)$,
where $d\in\mathbb R$ and $z\in\mathbb C$.
Here the opposite diagonal entries follow from the zero traces.
The identity $D^2+CC^\dagger=I_2$ gives $d^2+|z|^2=1$.

Define
\[
  A=\begin{pmatrix}d&z\\\overline z&-d\end{pmatrix},
  \qquad
  B=V\begin{pmatrix}1&0\\0&-1\end{pmatrix}V^\dagger.
\]
Both matrices are traceless and Hermitian.
Moreover, $B^2=I_2$ and
$A^2=(d^2+|z|^2)I_2=I_2$.
Since $D=dB$ and $C=zB$, the block expression for $W$
gives $W=A\otimes B$.
Substituting the definition of $W$ proves the claimed formula.
\end{proof}

\begin{lemma}\label{lm:realize-I-plus-AA}
Let $\mathcal F$ be a set of signatures. 
Suppose an arity 4 signature $h$ is realizable from $\mathcal F$ and
$M(h)=\lambda(I_4+A\otimes B)$, where $\lambda>0$
and $A,B\in\mathbb C^{2\times2}$ are traceless Hermitian matrices
satisfying $A^2=B^2=I_2$.
Then we can realize a signature $g$ with matrix
$M(g)=I_4+A\otimes A$.
\end{lemma}

\begin{proof}
Connect variables $x_2,x_4$ of the first copy of $h$ to variables
$x_4,x_2$ of the second, respectively.
The resulting signature is
$g'(a,b,c,d)=\sum_{x,y\in\{0,1\}}h(a,x,c,y)h(b,y,d,x)$.
Since $\operatorname{Tr}B=0$ and
$\operatorname{Tr}(B^2)=2$, direct expansion gives
\[
\begin{aligned}
g'(a,b,c,d)
&=\lambda^2\sum_{x,y\in\{0,1\}}
  (\delta_{ac}\delta_{xy}+A_{ac}B_{xy})
  (\delta_{bd}\delta_{yx}+A_{bd}B_{yx})\\
&=2\lambda^2
  (\delta_{ac}\delta_{bd}+A_{ac}A_{bd}).
\end{aligned}
\]
Thus $M(g')=2\lambda^2(I_4+A\otimes A)$.
Normalizing by $2\lambda^2$ proves the lemma.
\end{proof}

\begin{lemma}\label{lm: I+AA implies hard}
    Let $\mathcal{F}$ be a conjugate closed set of signatures.
    If from $\mathcal{F}$ we can realize a 4-ary signature $g$ with matrix $M(g)=I_4+A\otimes A$, with $A=A^{\dagger}, A^2=I_2$ and $\mathrm{Tr}(A)=0$, then $\Holant(\mathcal{F})$ is \#P-hard.
\end{lemma}
\begin{proof}
    By assumption, we may write the spectral decomposition of $A$ as $A=U\mathrm{diag}(1,-1)U^\dagger$, where $U=[\mathbf{u}_1, \mathbf{u}_2]$ is a 2 by 2 unitary matrix.
    Consider $U^{\tt T}U=\left[\begin{matrix}
        \mathbf{u}_1^{\tt T}\mathbf{u}_1 & \mathbf{u}_1^{\tt T}\mathbf{u}_2\\
        \mathbf{u}_2^{\tt T}\mathbf{u}_1 & \mathbf{u}_2^{\tt T}\mathbf{u}_2
    \end{matrix}\right]=\left[\begin{matrix}
        q_1 & z\\
        z & q_2
    \end{matrix}\right],$
    where $q_1=\mathbf{u}_1^{\tt T}\mathbf{u}_1,q_2=\mathbf{u}_2^{\tt T}\mathbf{u}_2$ and $z=\mathbf{u}_1^{\tt T}\mathbf{u}_2=\mathbf{u}_2^{\tt T}\mathbf{u}_1.$
    Since $U$ is unitary, $U^{\tt T}U$ is also unitary.
    So $|q_1|^2+|z|^2=|q_2|^2+|z|^2$, $|q_1|=|q_2|.$
    We may choose a complex phase of the eigenvectors $\mathbf{u}_1$ and $\mathbf{u}_2$, such that $q_1=\mathbf{u}_1^{\tt T}\mathbf{u}_1=\mathbf{u}_2^{\tt T}\mathbf{u}_2=q_2=r\ge 0$.
    \begin{itemize}
        \item If $r=0$, then $\mathbf{u}_1^{\tt T}\mathbf{u}_1=\mathbf{u}_2^{\tt T}\mathbf{u}_2=0$.
        So $\{\mathbf{u_1},\mathbf{u_2}\}=\{\frac{1}{\sqrt{2}}[1, \mathfrak{i}]^{\tt T},\frac{1}{\sqrt{2}}[1,-\mathfrak{i}]^{\tt T}\}$.
        So $A=U\mathrm{diag}(1,-1)U^\dagger=\pm Y$, where $Y=\left[\begin{matrix}
            0 & -\mathfrak{i}\\
            \mathfrak{i} & 0
        \end{matrix}\right]$.
        So $M(g)=I_4+Y\otimes Y.$
        Direct computation shows
        \[
        (K^{-1})^{\otimes 2}M(g)((K^{-1})^{\tt T})^{\otimes 2}=\begin{bmatrix}
            0& 0 & 0 & 2\\
            0& 0 & 0 & 0\\
            0& 0 & 0 & 0\\
            2& 0 & 0 & 0\\
        \end{bmatrix}.
        \]
        So $K^{-1}g=2(\neq_4).$
        By a $K$-holographic transformation, we have $\Holant(\neq_2\mid \widehat{\mathcal{F}},\neq_4)\le_T \Holant(\mathcal{F},g)\le_T\Holant(\mathcal{F}).$
        By \cref{lem: In conjugate closed setting disequality 4 gives hardness}, $\Holant(\neq_2\mid \widehat{\mathcal{F}},\neq_4)$ is \#P-hard.
        It follows that $\Holant(\mathcal{F})$ is \#P-hard.
        
        \item Suppose $r>0.$
        Consider merging variables $x_1$ and $x_2$ of $g$ and obtain the signature $b=\partial_{12}g.$
        Since $g(x_1,x_2,x_3,x_4)=\delta_{x_1,x_3}\delta_{x_2,x_4}+A_{x_1,x_3}A_{x_2,x_4}$, we have 
        $b(x_3,x_4)=\partial_{12}g(x_3,x_4)=\sum_{t\in\{0,1\}}g(t,t,x_3,x_4)=\sum_{t\in\{0,1\}}(\delta_{t,x_3}\delta_{t,x_4}+A_{t,x_3}A_{t,x_4})=\delta_{x_3,x_4}+(A^{\tt T}A)_{x_3,x_4}.$
        So $M(b)=I_2+A^{\tt T}A.$
        Let $B=M(b).$
        We have
        \begin{equation}
            \begin{aligned}
                U^{\tt T}BU
                =&U^{\tt T}U+U^{\tt T}A^{\tt T}AU\\
                =&U^{\tt T}U+U^{\tt T}\overline{U}\mathrm{diag}(1,-1)U^{\tt T}U\mathrm{diag}(1,-1)U^{\dagger}\\
                =& \left[\begin{matrix}
                r & z\\
                z & r
                \end{matrix}\right]
                +\mathrm{diag}(1,-1)\left[\begin{matrix}
                r & z\\
                z & r
                \end{matrix}\right]\mathrm{diag}(1,-1)\\
                =&\left[\begin{matrix}
                2r & 0\\
                0 & 2r
                \end{matrix}\right].
            \end{aligned}
        \end{equation}
        Thus, $B=2r\overline{U}U^\dagger.$
        Let $c=\frac{b}{2r}$ and $C=M(c)$.
        Since $\mathcal{F}$ is conjugate closed, $\overline{c}$ is also realizable, with matrix $\overline{C}=UU^{\tt T}.$
        Connecting $c$ and $\overline{c}$ using $=_2$, we can realize the signature with matrix $C\overline{C}=\overline{U} U^\dagger UU^{\tt T}=I_2$, i.e., we can realize $=_2$ on RHS.
        Now connect two copies of $\overline{c}$ with variables $x_3$ and $x_4$ of $g$ respectively.
        This will give the 4-ary signature $e=2\cdot U^{\otimes}(=_4)$.
        To see this, first notice that $M(g)=I_2\otimes I_2+A\otimes A=(\mathbf{u}_1\mathbf{u}_1^\dagger+\mathbf{u}_2\mathbf{u}_2^\dagger)^{\otimes 2}+(\mathbf{u}_1\mathbf{u}_1^\dagger-\mathbf{u}_2\mathbf{u}_2^\dagger)^{\otimes 2}=2\sum_{s=0}^1\mathbf{u}_s\otimes \mathbf{u}_s\otimes \overline{\mathbf{u}_s}\otimes \overline{\mathbf{u}_s}$.
        Then $e=(I_2\otimes I_2\otimes \overline{C}\otimes\overline{C})g=2(\mathbf{u}_0^{\otimes 4}+\mathbf{u}_1^{\otimes 4})=2\cdot U^{\otimes}(=_4).$
        On the other hand, $c$ is realizable on LHS by connecting two copies of $=_2$ on LHS with one copy of $c$ on RHS.
        Therefore,
        \[
        \Holant(c\mid \mathcal{F}\cup\{=_2,e\})\le_T\Holant(\mathcal{F}).
        \]
        Consider the holographic transformation by $U$.
        We have
        \[
        \Holant(=_2\mid U^\dagger (\mathcal{F}\cup\{=_2\}),=_4)
        \le_T
        \Holant(c\mid \mathcal{F}\cup\{=_2,e\})\le_T\Holant(\mathcal{F}).
        \]
        By \cref{lm: CSP2 to holant =4}, $\mathrm{CSP}_2({U^\dagger}(\mathcal{F}\cup\{=_2\}))\le_T\Holant(=_2\mid U^\dagger (\mathcal{F}\cup\{=_2\}),=_4).$
        By \cref{lm: CSP with =2}, $\mathrm{CSP}_2({U^\dagger}(\mathcal{F}\cup\{=_2\}))$ is \#P-hard.
        It follows that $\Holant(\mathcal{F})$ is \#P-hard.
    \end{itemize}
    
\end{proof}

\begin{lemma}\label{lm: APL-trans rank 2}
    Let $f\in \mathcal{F}$ be a signature of arity $2n\ge 4$.
    If there exists a pair of indices $\{i,j\}$ such that (1) $\mathfrak{m}_{ij}f$ is $\mathscr{A}$- or $\mathscr{P}$- or $\mathscr{L}$-transformable; (2) $\mathfrak{m}_{ij}f$ is irreducible and satisfies {\sc 1st-Orth}; (3) $\mathrm{rank}(M(\mathfrak{m}_{ij}f))=2$, then $\Holant(\mathcal{F})$ is \#P-hard.
\end{lemma}
\begin{proof}
    If $f$ doesn't satisfy {\sc 1st-Orth}, then $\Holant(\mathcal{F})$ is \#P-hard by \cref{lm: not 1st-orth is hard}.
    Now assume that $f$ satisfies {\sc 1st-Orth}.
    By \cref{lem:mating-partial-traces}, there exists $\mu>0$ such that $\mathrm{Tr}_1(\mathfrak{m}_{ij}f)=\mathrm{Tr}_2(\mathfrak{m}_{ij}f)=\mu I_2.$
    
    We realize $h:=\mathfrak{m}_{ij}f$ by mating.
    Since $h$ is irreducible, satisfies {\sc 1st-Orth}, and is $\mathscr{A}$- or $\mathscr{P}$- or $\mathscr{L}$-transformable, 
    \cref{lm: APL-transformable singular value} applies.
    Also, $\mathrm{rank}(M(h))=2$ by assumption, so $M(h)$ has two equal non-zero singular values.
    Notice $M(h)=M_{ij}(f)M_{ij}(f)^\dagger$ is Hermitian and positive semi-definite, so it has two equal non-zero eigenvalues.
    Then \cref{lm:rank-two-matrix-form} applies.
    We have $M(h)=\frac{\mu}{2}(I_4+A\otimes B)$, where $A,B\in \C^{2\times 2}$ are traceless Hermitian matrices with $A^2=B^2=I_2.$
    By \cref{lm:realize-I-plus-AA}, we can realize an arity 4 signature $g$ with matrix $M(g)=I_4+A\otimes A.$
    By \cref{lm: I+AA implies hard}, $\Holant(\mathcal{F})$ is \#P-hard.
    The lemma is proved.
\end{proof}

\begin{lemma}\label{lm: mating is APL-trans implies 2nd-orth}
    Let $f\in \mathcal{\mathcal{F}}$ be a signature of arity $2n\ge 4$.
    If there exists a pair of indices $\{i,j\}$ such that $\mathfrak{m}_{ij}f$ is $\mathscr{A}$- or $\mathscr{P}$- or $\mathscr{L}$-transformable, then either $f$ is reducible, or $\Holant(\mathcal{F})$ is \#P-hard, or $M(\mathfrak{m}_{ij}f)=\lambda_{ij}I_4$ for some $\lambda_{ij}>0.$
\end{lemma}
\begin{proof}
    If $\mathfrak{m}_{ij}f$ has a unary factor, $\Holant(\mathcal{F})$ is \#P-hard by \cref{lm: odd holant with conjugation}.
    If $\mathfrak{m}_{ij}f$ is a tensor product or two binary signatures, the lemma is proved by \cref{lm: 2nd-orth mating reducible}.
    Next we assume that $\mathfrak{m}_{ij}f$ is irreducible.
    
    If $\mathfrak{m}_{ij}f$ doesn't satisfy {\sc 1st-Orth}, then $\Holant(\mathcal{F})$ is \#P-hard by \cref{lm: not 1st-orth is hard}.
    Now we assume that $\mathfrak{m}_{ij}f$ satisfies {\sc 1st-Orth}.
    By \cref{lm: APL-transformable singular value}, $\mathrm{rank}(M(\mathfrak{m}_{ij}f))=2$ or $4$.
    If $\mathrm{rank}(M(\mathfrak{m}_{ij}f))=2$, then $\Holant(\mathcal{F})$ is \#P-hard by \cref{lm: APL-trans rank 2}.
    If $\mathrm{rank}(M(\mathfrak{m}_{ij}f))=4$, then $M(\mathfrak{m}_{ij}f)=\lambda_{ij}I_4$ for some $\lambda_{ij}>0$ by \cref{lm: APL-trans rank 4}.
\end{proof}

\subsection{Weighted Matching and $\mathscr{H}$ Type}

We next show that the case $\mathfrak{m}_{ij}f\in \braket{K\mathcal{M}}$ or $\mathfrak{m}_{ij}f\in \braket{KX\mathcal{M}}$ cannot happen.

\begin{lemma}\label{lm: mating is self inversive}
    Let $h=\mathfrak{m}_{ij}f$.
    Then $\widehat{h}(a,b,c,d)=\overline{\widehat{h}(\overline{c},\overline{d},\overline{a},\overline{b})}$, for every $a,b,c,d\in \{0,1\}.$
\end{lemma}
\begin{proof}
    
    Recall in the mating gadget $\widehat{\mathfrak{m}}_{ij}\widehat{f}$, we connect corresponding variables of $\widehat{f}$ and $\widehat{\overline{f}}$ using $\neq_2.$
    Assign $a,b$ to the two dangling edges on LHS, and $c,d$ to the two dangling edges on RHS.
    Now reflect the gadget by connecting $\widehat{\overline{f}}$ and $\widehat{f}$ using $\neq_2$, and flip the inputs on all dangling edges.
    The inputs change from $(a,b,c,d)$ to $(\overline{c},\overline{d},\overline{a},\overline{b}).$
    Since $\widehat{\overline{f}}(x)=\overline{\widehat{f}(\overline{x})}$ by \cref{lm: conjugation and K-trans},
    in the reversed gadget, the evaluation of $\widehat{\overline{f}}(x)$ on LHS equals to the conjugation of $\widehat{f}(\overline{x})$, while the evaluation of $\widehat{f}(x)$ on RHS equals to the conjugation of $\widehat{\overline{f}}(x)$.
    That is, $\widehat{h}(a,b,c,d)=\overline{\widehat{h}(\overline{c},\overline{d},\overline{a},\overline{b})}$, for every $a,b,c,d\in \{0,1\}.$
\end{proof}

\begin{lemma}\label{lm: 2nd-Orth M and XM cannot happen}
    Let $f\in \mathcal{F}$ be a signature of arity $2n\ge 4$.
    For an arbitrary pair of indices $\{i,j\}$, if $\mathfrak{m}_{ij}f$ is irreducible, 
    then $\mathfrak{m}_{ij}f\notin \braket{K\mathcal{M}}$ and $\mathfrak{m}_{ij}f\notin \braket{KX\mathcal{M}}$.
\end{lemma}

\begin{proof}
    Suppose for a contradiction that $\mathfrak{m}_{ij}f\in \braket{K\mathcal{M}}$ or $\mathfrak{m}_{ij}f\in \braket{KX\mathcal{M}}$.
    Let $h=\mathfrak{m}_{ij}f$.
    Then $\widehat{h}\in \braket{\mathcal{M}}$ or $\widehat{h}\in \braket{X\mathcal{M}}.$
    Since $h$ is irreducible, $\widehat{h}$ is also irreducible.
    So $\widehat{h}\in \mathcal{M}$ or $\widehat{h}\in X\mathcal{M}.$
    Thus $\widehat{h}$ is supported on Hamming weight at most 1 or at least 3 by the definition of $\mathcal{M}$ and $X\mathcal{M}$.

    If $\mathscr{S}(\widehat{h})$ contains an element of Hamming weight 1 (resp. 3), then by \cref{lm: mating is self inversive}, $\mathscr{S}(\widehat{h})$ also contains an element of Hamming weight 3 (resp. 1).
    This contradicts $\widehat{h}\in {\mathcal{M}}$ or $\widehat{h}\in {X\mathcal{M}}.$
    Therefore, $\widehat{h}=0$.
    This contradicts $\widehat{h}$ is irreducible.
\end{proof}

Finally, we deal with the case that some $\mathfrak{m}_{ij}f\in \mathscr{H}.$

\begin{lemma}\label{lm: 2nd-orth H type}
    Let $f\in \mathcal{F}$ be a signature of arity $2n\ge 4$.
    For an arbitrary pair of indices $\{i,j\}$, if $\mathfrak{m}_{ij}f\in \mathscr{H}$, then either $f$ is reducible, or $\Holant(\mathcal{F})$ is \#P-hard, or $M(\mathfrak{m}_{ij}f)=\lambda_{ij}f$ for some $\lambda_{ij}>0.$
\end{lemma}
\begin{proof}
    Let $h=\mathfrak{m}_{ij}f$ and $\widehat{h}=\widehat{\mathfrak{m}}_{ij}\widehat{f}.$
    By the definition of $\mathscr{H}$, $\widehat{h}$ is supported on $\mathrm{HW}^{\ge}$ or $\mathrm{HW}^{\le}$.
    By \cref{lm: mating is self inversive}, $\widehat{h}(a,b,c,d)=\overline{\widehat{h}(\overline{c},\overline{d},\overline{a},\overline{b})}$, for every $a,b,c,d\in \{0,1\}.$
    So $\widehat{h}$ is supported on Hamming weight $2.$
    Therefore, $\Holant(\neq_2\mid \widehat{h})$ is a six-vertex model.
    By \cref{thm: six vertex model}, either
    $\Holant(\neq_2\mid \widehat{h})\le_T\Holant(\mathcal{F})$ is \#P-hard, or $\widehat{h}\in \mathscr{A}$, or $\widehat{h}\in \mathscr{P}$, or $\widehat{h}\in \braket{\mathcal{M}}$, or $\widehat{h}\in \braket{X\mathcal{M}}$.
    If $\widehat{h}\in \mathscr{A}$ or $\widehat{h}\in \mathscr{P}$,
    then $h$ is $\mathscr{A}$- or $\mathscr{P}$-transformable, and the lemma follows from \cref{lm: mating is APL-trans implies 2nd-orth}.
    Also, the case $\widehat{h}\in \braket{\mathcal{M}}$ or $\widehat{h}\in \braket{X\mathcal{M}}$ cannot happen by \cref{lm: 2nd-Orth M and XM cannot happen}.
    The lemma is proved.
\end{proof}

\subsection{Proof of Second Order Orthogonality}

\begin{lemma}\label{lm: not 2nd-orth is hard}
    Let $f\in {\mathcal{F}}$ be a irreducible signature of arity $2n \geqslant 4$.
    If $f$ does not satisfy {\sc 2nd-Orth}, then $\Holant({\mathcal{F}})$ is \#P-hard.
\end{lemma}

\begin{proof}
    Fix an arbitrary pair of indices $\{i,j\}$, we can realize the arity 4 signature $\mathfrak{m}_{ij}f$ by mating.
    By \cref{thm: single quaternary dichotomy}, either $\Holant(\mathfrak{m}_{ij}f)\le_T\Holant(\mathcal{F})$ is \#P-hard, or $\mathfrak{m}_{ij}f$ falls in the six tractable classes of \cref{thm: single quaternary dichotomy}.
    
    The first two classes $\braket{K\mathcal{M}}$ and $\braket{KX\mathcal{M}}$ cannot occur by \cref{lm: 2nd-Orth M and XM cannot happen}.
    
    If $\mathfrak{m}_{ij}f$ is $\mathscr{A}$- or $\mathscr{P}$- or $\mathscr{L}$-transformable, then by \cref{lm: mating is APL-trans implies 2nd-orth}, either $f$ is reducible, contradiction; or $\Holant(\mathcal{F})$ is \#P-hard; or $M(\mathfrak{m}_{ij}f)=\lambda_{ij}I_4$ for some $\lambda_{ij}>0.$
    By \cref{lm: 2nd-Orth uniform}, all $\lambda_{ij}$ have the same value $\lambda$. 
    So $f$ satisfies {\sc 2nd-Orth}, unless $\Holant(\mathcal{F})$ is \#P-hard.

    If $f\in \mathscr{H}$, the lemma follows from \cref{lm: 2nd-orth H type} and \cref{lm: 2nd-Orth uniform}.
    If $f\in \braket{\mathscr{T}}$, the lemma follows from from \cref{lm: 2nd-orth mating reducible} and \cref{lm: 2nd-Orth uniform}.
\end{proof}

\section{The Induction Framework and Base Case}\label{sec-induction}
    In this section, we introduce the induction framework and handle the base cases (\cref{lm: base case 2n=2,lm: base case 2n=4}).
    Recall that we assume that ${\mathcal{F}}$ does not satisfy condition \rm{(\ref{cond: tractable classes})}, and we want to prove $\Holant(\mathcal{F})$ is \#P-hard.
    By assumption, ${\mathcal{F}}\not\subseteq \mathscr{T}$.
    Also, since  ${\mathcal{U}}^{\otimes}\subseteq \mathscr{T}$, ${\mathcal{F}}\not\subseteq {\mathcal{U}}^{\otimes}$.
    Thus, there is a nonzero signature ${f}\in {\mathcal{F}}$ of arity $2n$ such that ${f}\notin {\mathcal{U}}^{\otimes}$.
    We want to achieve a proof of  \#P-hardness by induction on $2n$. 
    
\subsection{Base Case}\label{subsec: Base Case}
    We first consider the base that $2n=2$.
    Notice that a nonzero binary signature  ${f}$ satisfies {\sc 1st-Orth} iff its matrix form (as a 2-by-2 matrix) is projective unitary.
    Thus, ${f}\notin{\mathcal{U}}$ implies that it does not satisfy {\sc 1st-Orth}.
    Then, we have the following result.

\begin{lemma}\label{lm: base case 2n=2}
    Let $\mathcal F$ contain a binary signature $f\notin \mathcal{U}^{\otimes}$.
    Then, $\Holant(\mathcal{F})$ is \#P-hard.
\end{lemma}

\begin{proof}
    We prove this lemma in the setting of $\Holant(\mathcal{F})$.
    Since ${\mathcal{U}}^{\otimes}$ contains the binary zero signature, $f\notin {\mathcal{U}}^{\otimes}$ implies that $f\not\equiv 0$.
    If $f$ is reducible, then it is a tensor product of two nonzero unary signatures. 
    By \cref{lm: decomposition}, either $\Holant(\mathcal{F})$ is \#P-hard, or we can realize a nonzero unary signature by factorization, and we are done by \cref{lm: odd holant with conjugation}.
    Otherwise, $f$ is irreducible. 
    Since $f\notin \mathcal{U}^{\otimes}$, $f$ does not satisfy {\sc 1st-Orth}. 
    By \cref{lm: not 1st-orth is hard},  $\Holant(\mathcal{F})$ is \#P-hard.
\end{proof}

Next, we consider the base case $2n=4$.
\begin{lemma}\label{lm: base case 2n=4}
    Let ${\mathcal{F}}$ contain a $4$-ary signature ${f}\notin {\mathcal{U}}^{\otimes}$.
    Then, $\Holant({\mathcal F})$ is {\rm \#}P-hard.
\end{lemma}

\begin{proof}
    Since ${f}\notin{{\mathcal{U}}}^{\otimes}$, $f\not\equiv 0$.
    First, we may assume that $f$ is irreducible. 
    Otherwise, we can realize a nonzero unary signature or a binary signature that is not in ${\mathcal{U}}$.
    Then, by \cref{lm: odd holant with conjugation,lm: base case 2n=2}, we have \#P-hardness.
    Since $f$ is irreducible, we may further assume that ${f}$ satisfies {\sc 2nd-Orth}.
    Otherwise, by \cref{lm: not 2nd-orth is hard}, we get \#P-hardness. 
    However, this is impossible due to the non-existence of an \rm{AME}$(4,2)$ state {\cite{AME42}}.
    A contradiction!
\end{proof}

\subsection{Induction Framework}\label{subsec: Induction framework}
The general induction framework is  that we start with a signature $f$ of arity $2n>4$ that is not in ${\mathcal{U}}^{\otimes}$, and realize a signature $g$ of arity $2k \leqslant 2n-2$ that is also not in ${\mathcal{U}}^{\otimes}$, or we can directly show $\Holant({\mathcal F})$ is {\rm \#}P-hard.
If we can reduce the arity down to $2$ (by a sequence of reductions of length independent of the problem instance size), then we have a binary signature ${b}\notin{\mathcal{U}}$. 
By \cref{lm: not 1st-orth is hard}, we are done.

\vspace{0.5em}
\noindent\textbf{Induction Hypothesis:} If we can realize a $2k$-arity $(1\le k<n)$ signature $f$ from $\mathcal{F}$ such that $f\notin \mathcal{U}^{\otimes}$, then $\Holant(\mathcal{F})$ is \#P-hard.
\vspace{0.5em}

For the inductive step, we first consider the case that ${f}$ is reducible.
Suppose that ${f}={f_1} \otimes {f_2}$.
If ${f_1}$ or ${f_2}$ have odd arity, then we can realize a signature of odd arity by factorization and we are done.
Otherwise, ${f_1}$ and ${f_2}$ both have even arity. 
Since ${f}\notin {\mathcal{U}}^{\otimes}$, we know ${f_1}$ and ${f_2}$ cannot both be in ${\mathcal{U}}^{\otimes}$. 
Then, we can realize a signature of lower arity that is not in ${\mathcal{U}}^{\otimes}$ by factorization. 
We are done.

In the following we may assume that $f$ is irreducible.
Then, we may further assume that $f$ satisfies {\sc 2nd-Orth}. 
Otherwise, we get \#P-hardness by \cref{lm: not 2nd-orth is hard}.
We use merging with $=_2$ to realize signatures of arity $2n-2$ from $f$.
Consider ${\partial}_{ij} f$ for all pairs of indices $\{i, j\}$. 
If there exists a pair $\{i, j\}$ such that ${\partial}_{ij} f \notin {\mathcal{U}}^{\otimes}$, 
then we can realize $g={\partial}_{ij}f$ which has arity $2n-2$, and we are done.
Thus, we may assume ${\partial}_{ij} f \in {\mathcal{U}}^{\otimes}$ for all $\{i, j\}$. 
We denote this property by $ f\in{\int}{\mathcal{U}}^{\otimes}$.

Now, can assume that $f$ is irreducible, satisfies {\sc 2nd-Orth} and $f\in {\int}{\mathcal{U}}^{\otimes}$. 
By factorization, we can realize all the binary factors of $\partial_{ij}f$ for any pair ${i,j}$.
Then we can merge $f$ by any of the binary factor $b$ and realize $\partial^b_{ij}f$.
We may further assume that $\partial^b_{ij}f\in \mathcal{U}^{\otimes}$, otherwise we obtain \#P-hardness by the Induction Hypothesis.
Again we can realize all binary factors of $\partial^b_{ij}f\in \mathcal{U}^{\otimes}$, and so on.
In each step, we can also concatenate two binary factors using $=_2$.
Finally, we can realize a set of binary signatures from $f$ by consecutive merging, factorization and concatenation.
Clearly, if $f$ satisfies all the assumptions above, so does $\overline{f}$.
If we can realize $b$ from $f$, then we can also realize $\overline{b}$ from $\overline{f}$.
We summarize the above generating process as follows:

\begin{definition}[Projective binary group] \label{def: generating process}
    Given a signature $f$ of arity $2n$, we recursively define a sequence of sets:
    \begin{enumerate}
        \item $B_0(f)=\{=_2\}$.
        
        \item Given $B_{i-1}(f)$, $A_i(f)$ consists of all binary factors of $\partial^{b}_{ij}f\in \mathcal{U}^{\otimes}$ and $\partial^{b}_{ij}\overline{f}\in \mathcal{U}^{\otimes}$ for arbitrary distinct pair $\{i,j\}$ and $b$ arbitrarily chosen from $B_{i-1}(f)$.

        \item Given $A_i(f)$, $B_i(f)$ consists of all binary signatures realized by connecting $k$ signatures from $A_i(f)$ as a path for arbitrary $k\ge 1$.
    \end{enumerate}
    Let $B(f)=\mathop{\bigcup}\limits_{i=1}^{\infty}B_i(f)$, and call it the projective binary group of $f$.
\end{definition}

We may assume that $f$ has the following closure property, otherwise \#P-hardness follows.
\begin{proposition}[$B(f)$ closure]\label{prop: binary closure}
    For all pairs of indices $\{i,j\}$ and $b\in B(f)$, $\partial^b_{ij}f\in B(f)^{\otimes}$.
\end{proposition}

Next, we explain the notion ``projective binary group".
Every binary signature $b(x,y)\in \mathcal{U}$ can be represented as a $2\times 2$ matrix $M(b)=\left[\begin{smallmatrix}
    b(00) & b(01)\\
    b(10) & b(11)
\end{smallmatrix}\right]\in \mathbf{PU}(2)$.
Since normalizing by a nonzero constant doesn't change the complexity, we can realize every binary signature in $[b]:=\{c\cdot b\mid c\in \mathbb{C}\setminus\{0\}\}.$
We identify every binary signature $b$ in $B(f)$ with its projective class $[b]$, and its representative can be chosen in $\mathbf{SU}(2).$ 
Note that there are two possible choices of representatives, according to normalizing by $\pm\det(M(b))$.

We can define the multiplication $b_3:=b_1\cdot b_2$ by the resulting signature connecting $b_1$ and $b_2$ using $=_2$.
Then $b_3$ has the matrix $M(b_3)=M(b_1)M(b_2)$.
In the following, we always assume $M(b)\in \mathbf{SU}(2)$ unless otherwise stated.
An important observation by Xia in \cite{MingjiProgram} is that:

\begin{restatable}{theorem}{MingjiTheorem}\label{thm: binary set is finite group}
    $B(f)$ is 
    a finite subgroup of $\mathbf{PU}(2)=\mathbf{SU}(2)/\{\pm I_2\}$, otherwise $\Holant(\mathcal{F})$ is \#P-hard.
\end{restatable}
We give an exposition of Xia's proof in \autoref{appendix: mingji}.
Since $\mathbf{SU}(2)/\{\pm I_2\}\cong\mathbf{SO}(3)$, $B(f)$ is isomorphic to a finite subgroup of $\mathbf{SO}(3).$
All finite subgroups of $\mathbf{SO}(3)$ are: cyclic groups, dihedral groups and polyhedral groups.
Moreover, in our conjugate closed setting, once we can realize $b$, we can realize $\overline{b}$.
Also, we can rename the two variables of $b$.
Therefore, we may assume that $B(f)$ is closed under taking conjugation and transpose.

Similarly, we can do a $K$-transformation and define the projective binary group $\widehat{B}(\widehat{f})$ is the $\holant{\neq_2}{\hF}$ setting.
We have $\widehat{B}(\widehat{f})=K^{-1}B(f).$
The multiplication $\widehat{b_3}=\widehat{b_1}\cdot \widehat{b_2}$ has the matrix $M(\widehat{b_3})=M(\widehat{b_1})N_2M(\widehat{b_2})$, and $\widehat{B}(\widehat{f})$ has the identity elemetn $\widehat{e}$ with matrix $M(\widehat{e})=N_2.$

\section{Third Order Orthogonality}\label{sec: Third Order Orth}
Recall that if $\mathcal{F}$ contains an odd arity signature, then the main theorem is true by \cref{lm: odd holant with conjugation}.
Also, \cref{subsec: Base Case} proves the main theorem if $\mathcal{F}$ only consists of signatures of arity no larger than $4$.
In this section, we assume $\mathcal{F}$ contains a signature of arity $6$.
The induction framework \cref{subsec: Induction framework} shows we may assume the following, otherwise $\Holant(\mathcal{F})$ is already \#P-hard:

\begin{equation}\label{eq: arity 6 assumption}
    f\in \mathcal{F} \text{ is arity } 6, \text{irreducible, satisfies {\sc 2nd-Orth} and } B(f) \text{ closure \cref{prop: binary closure}.}
\end{equation}

In this section, we are going to prove the following theorem:

\begin{restatable}{theorem}{ThirdOrderOrth}\label{thm: third order orthogonality}
    Under assumption \eqref{eq: arity 6 assumption}, $f$ satisfies {\sc 3rd-Orth}.
    Equivalently, $\frac{\ket{f}}{||f||}$ is an AME(6,2) state.
\end{restatable}

The proof enumerates all possible cases of the finite group $B(f)\le \mathbf{SO}(3)$.
We fix arbitrary distinct indices $\{i,j,k\}\subseteq [6]$.
For the polyhedral group case (\cref{subsec: polyhedral groups}), large dihedral group case (\cref{subsec: dihedral group}) and standard case of Klein group (\cref{subsubsec: Standard Case Klein}), we prove that $f$ satisfies {\sc 3rd-Orth}, by showing the coefficients of Pauli weight $3$ in the Pauli representation of the marginal density matrix $\rho_{ijk}$ are all 0. 
We also prove the non-standard case of Klein group (\cref{subsubsec: Non-standard Klein}) and the cyclic group case (\cref{subsec: Trivial group,subsec: Order Two group,subsec: Cyclic group of Order at least three}) can not happen.

\subsection{Polyhedral Groups}\label{subsec: polyhedral groups}

In this subsection, we consider the case that $B(f)\le \mathbf{SO}(3)$ is a polyhedral group.
It can be the tetrahedral group $A_4$, the octahedral group $S_4$, or the icosahedral group $A_5$.

Let $\rho=\ket{f}\bra{f}$ be the density matrix of $\ket{f}$.
By merging variables $x_i,x_j,x_k$ of $f$ and variables $x_i,x_j,x_k$ of $\overline{f}$ using $=_2$, we can realize the state with marginal density matrix in Pauli representation:
$$\rho_{ijk}=\frac{1}{8}\sum_{p,q,r=0}^3 t_{pqr} \sigma_p\otimes\sigma_q\otimes\sigma_r.$$
By {\sc 2nd-Orth}, we have 
\begin{lemma}\label{lm: Pauli rep by 2nd-Orth}
    $$\rho_{ijk}=\frac{1}{8}\left(I_8+\sum_{p,q,r=1}^3 t_{pqr} \sigma_p\otimes\sigma_q\otimes\sigma_r\right).$$
\end{lemma}
\begin{proof}
    First, $t_{000}=\mathrm{Tr}(\rho_{ijk}(\sigma_0\otimes\sigma_0\otimes\sigma_0))=\mathrm{Tr}(\rho_{ijk})=1.$
    We only need to prove $t_{pqr}=0$ if there are $1$ or $2$ zeros among $\{p,q,r\}.$
    For example,
    $t_{120}=\mathrm{Tr}(\rho_{ijk}(\sigma_1\otimes \sigma_2\otimes I))=\mathrm{Tr}(\rho_{ij}(\sigma_1\otimes \sigma_2)).$
    By {\sc 2nd-Orth}, $\rho_{ij}=\frac{1}{4}I$.
    So $t_{120}=\frac{1}{4}\mathrm{Tr}(\sigma_1\otimes \sigma_2)=\frac{1}{4}\mathrm{Tr}(\sigma_1)\mathrm{Tr}(\sigma_2)=0.$
    Similarly, $t_{pqr}=0$ if there are $1$ or $2$ zeros among $\{p,q,r\}.$
\end{proof}

So $f$ satisfies $3$rd-Orth iff $t_{pqr}=0$ for all $p,q,r\in\{1,2,3\}.$

\begin{definition}
    For every $b\in B(f)$ with matrix $M(b)\in \mathbf{SU}(2)$, 
    define $\operatorname{vec}(b)=(b_{00},b_{01},b_{10},b_{11})^{\tt T}$, $R_b\in \mathbf{SO}(3)$ be the adjoint rotation of $b$, and the matrix $C(b)\in \mathbb{R}^{3\times 3}$ by $C(b)_{pq}=\operatorname{vec}(b)^\dagger (\sigma_p\otimes\sigma_q) \operatorname{vec}(b)$ for $p,q\in\{1,2,3\}$.
\end{definition}

\begin{lemma}\label{lm: Cb and Rb}
    $C(b)=2R_{b^{-1}}^{\tt T}S$, where $S=\mathrm{diag}(1,-1,1).$
\end{lemma}
\begin{proof}
    Recall the correspondence between $M(b)\in\mathbf{SU}(2)$ and $R_b\in \mathbf{SO}(3)$ in \cref{para: SU2 and SO3} gives
    $$M(b) \sigma_pM(b)^\dagger=\sum_{r=1}^3(R_b)_{rp}\sigma_r\implies M(b)^\dagger\sigma_pM(b)=\sum_{r=1}^3(R_{b^{-1}})_{rp}\sigma_r$$ 
    for $p=1,2,3.$
    Let $s_1=1, s_2=-1, s_3=1$, notice that $s_q\sigma_q=\sigma_q^{\tt T}$ for $q=1,2,3.$
    Then
    \begin{equation*}
        \begin{aligned}
            C(b)_{pq}=&\operatorname{vec}(b)^\dagger(\sigma_p\otimes\sigma_q) \operatorname{vec}(b)=\mathrm{Tr}(M(b)^\dagger \sigma_p M(b)\sigma_q^{\tt T})\\
            =&\mathrm{Tr}\left(\sum_{r=1}^3(R_{b^{-1}})_{rp}\sigma_rs_q\sigma_q\right)
            =s_q \sum_{r=1}^3 (R_{b^{-1}})_{rp}\mathrm{Tr}(\sigma_r\sigma_q)\\
            =&2s_q \sum_{r=1}^3 (R_{b^{-1}})_{rp}\delta_{rq}
            =2s_q(R_{b^{-1}})_{qp}
            =(2R_{b^{-1}}^{\tt T}S)_{pq}.
        \end{aligned}
    \end{equation*}
\end{proof}

\begin{lemma}[Pauli slice]\label{lm: Pauli slice}
    For every $b\in B(f)$ and $r=1,2,3$, we have $\sum_{p,q=1}^3C(b)_{pq}t_{p,q,r}=0.$
\end{lemma}

\begin{proof}
    Suppose after merging $f$ and $\overline{f}$, the dangling edges are $x_i,x_j,x_k$ and $y_i,y_j,y_k$.
    Since $b\in B(f)$, we have $\overline{b}\in B(f)$.
    Connect $x_i$ and $x_j$ using $\overline{b}$, $y_i$ and $y_j$ using $b$, we get a binary signature $g(x_k,y_k)\in \mathcal{U}$ by binary closure \cref{prop: binary closure}.
    We have
    \begin{equation*}
        \begin{aligned}
            \mu I_2
            =&M(g)M(g)^{\dagger}=(\operatorname{vec}(b)^\dagger \otimes I_2) \rho_{ijk}(\operatorname{vec}(b) \otimes I_2)\\
            =&\frac{1}{8}(\operatorname{vec}(b)^\dagger \otimes I_2)\left(I_2^{\otimes3}+\sum_{p,q,r=1}^3 t_{pqr} \sigma_p\otimes\sigma_q\otimes\sigma_r\right)(\operatorname{vec}(b) \otimes I_2)\\
            =& \frac{1}{8}\left(\operatorname{vec}(b)^{\dagger}\operatorname{vec}(b)\otimes I_2+\sum_{p,q,r=1}^3t_{pqr}(\operatorname{vec}(b)^\dagger (\sigma_p\otimes\sigma_q) \operatorname{vec}(b))\sigma_r\right)\\
            =&\frac{|\operatorname{vec}(b)|^2}{8}I_2+\frac{1}{8}\sum_{r=1}^3\left(\sum_{p,q=1}^3 C(b)_{pq}t_{pqr}\right)\sigma_r.&
        \end{aligned}
    \end{equation*}
    Since $\{I_2,\sigma_1,\sigma_2,\sigma_3\}$ is a basis of $\mathbb{C}^{2\times 2}$, we have $\sum_{p,q=1}^3 C(b)_{pq}t_{pqr}=0$ for $r=1,2,3.$
\end{proof}
\begin{remark}
    For $r=1,2,3$, we define $T^{(r)}$ be the $3\times 3$ matrix with entries $T^{(r)}_{pq}=t_{pqr}$.
    This is a slice of the the Pauli coefficients $t_{pqr}.$
    \cref{lm: Pauli slice} says $T^{(r)}$ is orthogonal to $C(b)$ under the Frobenius inner product in $M_3(\mathbb{R}).$
    Similarly, we can define $T^{(p)}$ and $T^{(q)}$, and they are also orthogonal to $C(b)$.
\end{remark}

The following lemma utilizes the assumption that $B(f)=A_4,S_4$ or $A_5$.
\begin{lemma}\label{lm: rep lemma}
    Let $W=\mathrm{Span}_{\mathbb{R}}\{C(b)\mid b\in B(f)\}$, then $W=M_3(\mathbb{R})$.
\end{lemma}
\begin{proof}
    By Section 4.8 and Theorem 4.5.1 in \cite{reptheory}, for $B(f)=A_4,S_4,A_5$, the natural three-dimensional rotation representation $\rho:b\mapsto R_b$ is irreducible.
    Apply the density theorem (Theorem 3.2.2 in \cite{reptheory}), we have $\mathrm{Span}_{\mathbb{C}}\{R_b\mid b\in B(f)\}=M_3(\mathbb{C})$.
    Since all matrices $R_b$ are real, $\mathrm{Span}_{\mathbb{R}}\{R_b\mid b\in B(f)\}=M_3(\mathbb{R})$.
    By \cref{lm: Cb and Rb}, $W=M_3(\mathbb{R}).$
\end{proof}

\begin{lemma}\label{lm: arity 6 polyhedral group 3rd orth}
    If $B(f)$ is isomorphic to a polyhedral group, then $f$ satisfies {\sc 3rd-Orth}.
\end{lemma}
\begin{proof}
    By \cref{lm: Pauli slice}, the matrix $T^{(r)}$ with entries $(T^{(r)})_{pq}=t_{pqr}$ is in $W^{\perp}$. 
    Here we use the Frobenius inner product in $M_3(\mathbb{R})$: $\braket{A,B}=\mathrm{Tr}(A^TB).$
    By \cref{lm: rep lemma}, $W^\perp=\{O\}.$ 
    Thus, $T^{(r)}=O$ for every $r=1,2,3$, or $t_{pqr}=0$ for every $p,q,r\in\{1,2,3\}.$
\end{proof}

\subsection{Dihedral Groups of Order at Least Six}\label{subsec: dihedral group}
In this subsection, we assume that $B(f)\cong D_n(n\ge 3)$ which is the dihedral group of order $2n.$
We still use the Pauli representation of $\rho_{ijk}$ in \cref{subsec: polyhedral groups}.

Suppose in $\mathbb{R}^3$, the regular $n$-polygon $P_n$ lies on the $x$-$y$ plane, and $P_n$ is symmetric relative to the $x$-axis.
Then $D_{n}$ is generated by a rotation $r$ through angle $\theta=\frac{2\pi}{n}$ around the $z$-axis, and a rotation $h$ through angle $\pi$ around the $x$-axis.
Specifically,
$$r=\left[\begin{matrix}
    \cos \theta & -\sin \theta &0\\
    \sin \theta & \cos \theta & 0\\
    0 & 0 & 1
\end{matrix}\right],\quad
h=\left[\begin{matrix}
    1 & 0 & 0\\
    0 & -1 & 0\\
    0 & 0 & -1
\end{matrix}\right].
$$
These matrices satisfy $r^n=I,h^2=I,hrh=r^{-1}.$
Therefore $D_n=\braket{r,h}=\{r^k,hr^k\mid k\in[n]\}.$
Let $L$ be the space spanned by the $z$-axis in $\mathbb{R}^3,$ and $P$ be the $x$-$y$ plane.
Notice that every element in $D_n$ is block diagonal relative to $P\oplus L.$
So $\mathrm{Span}_{\mathbb{R}}\{R_b\mid b\in B(f)\}\subseteq\mathrm{End}(P)\oplus \mathrm{End}(L)$.
By inspection, we can see $D_n$ generates every element in $\mathrm{End}(P)\oplus \mathrm{End}(L)$.
Thus, $\mathrm{Span}_{\mathbb{R}}\{R_b\mid b\in B(f)\}=\mathrm{End}(P)\oplus \mathrm{End}(L)$.
By \cref{lm: Cb and Rb}, $W=\mathrm{Span}_{\mathbb{R}}\{C(b)\mid b\in B(f)\}=\mathrm{End}(P)\oplus \mathrm{End}(L).$
Next we compute $W^\perp.$
Suppose $X\in W^\perp$, then $\mathrm{Tr}(C^{\tt T}X)=0$ for every $C\in \mathrm{End}(P)\oplus \mathrm{End}(L).$
$C$ has the form $\left[\begin{smallmatrix}
    A & 0\\
    0 & \alpha
\end{smallmatrix}\right]$ where $A\in M_2(\mathbb{R})$ and $\alpha\in \mathbb{R}$.
We write $X=\left[\begin{smallmatrix}
    Y & u\\
    v & \beta
\end{smallmatrix}\right]$ where $Y\in M_{2}(\mathbb{R})$, $u\in \mathbb{R}^{2\times 1}$, $v\in \mathbb{R}^{1\times 2}$, $\beta\in \mathbb{R}$.
Then 
$$
\mathrm{Tr}\left(
\left[\begin{matrix}
    A^{\tt T} & 0\\
    0 & \alpha
\end{matrix}\right]
\left[\begin{matrix}
    Y & u\\
    v & \beta
\end{matrix}\right]
\right)
=\mathrm{Tr}(A^{\tt T}Y)+\alpha\beta=0
$$
for every $A\in M_{2}(\mathbb{R})$ and $\alpha\in \mathbb{R}.$
Thus, $Y=O$ and $\beta=0.$
We have 
$$
W^{\perp}=\{\left[\begin{matrix}
    O & u\\
    v & 0
\end{matrix}\right]\mid u\in \mathbb{R}^{2\times 1}, v\in \mathbb{R}^{1\times 2}\}.
$$
By Pauli slice \cref{lm: Pauli slice}, $T^{(p)},T^{(q)}$ and $T^{(r)}\in W^\perp$ for every $p,q,r\in\{1,2,3\}.$
By inspection, the only possibility is that $t_{pqr}=0$ for every $p,q,r\in\{1,2,3\}.$
Thus, we prove the following lemma:
\begin{lemma}\label{lm: arity 6 dihedral group 3rd orth}
    If $B(f)\cong D_{n}(n\ge 3)$, then $f$ satisfies {\sc 3rd-Orth}.
\end{lemma}

\subsection{Klein Group}\label{subsec: Klein group}
Let $B(f)=\{b_0,b_1,b_2,b_3\}\cong K_4$.
It is proved in \cite{MingjiProgram} that we only need to consider the following two cases:
\begin{enumerate}
    \item
    (Standard Case).
    $$
    M(b_0)=\left[\begin{matrix}
    1 & 0\\
    0 & 1
    \end{matrix}\right],
    M(b_1)=\left[\begin{matrix}
    0 & \mathfrak{i}\\
    \mathfrak{i} & 0
    \end{matrix}\right],
    M(b_2)=\left[\begin{matrix}
    \mathfrak{i} & 0\\
    0 & -\mathfrak{i}
    \end{matrix}\right],
    M(b_3)=\left[\begin{matrix}
    0 & 1\\
    -1 & 0
    \end{matrix}\right].$$
        
    \item
    (Non-standard Case).
    $$M(b_0)=\left[\begin{matrix}
    1 & 0\\
    0 & 1
    \end{matrix}\right],
    M(b_1)=\left[\begin{matrix}
    \mathfrak{i} & 0\\
    0 & -\mathfrak{i}
    \end{matrix}\right],
    M(b_2)=\frac{1}{\sqrt{2}}\left[\begin{matrix}
    0 & 1+\mathfrak{i}\\
    -1+\mathfrak{i} & 0
    \end{matrix}\right],
    M(b_3)=\frac{1}{\sqrt{2}}\left[\begin{matrix}
    0 & -1+\mathfrak{i}\\
    1+\mathfrak{i} & 0
    \end{matrix}\right].$$
\end{enumerate}

\subsubsection{Standard Case}\label{subsubsec: Standard Case Klein}
By \cref{para: SU2 and SO3}, we can compute the corresponding rotation matrices of $b_0,b_1,b_2,b_3$ in $\mathbf{SO}(3)$:
$$
R_{b_0}=
\left[\begin{matrix}
    1 & 0 & 0\\
    0 & 1 & 0\\
    0 & 0 & 1\\
\end{matrix}\right],
R_{b_1}=
\left[\begin{matrix}
    1 & 0 & 0\\
    0 & -1 & 0\\
    0 & 0 & -1\\
\end{matrix}\right],
R_{b_2}=
\left[\begin{matrix}
    -1 & 0 & 0\\
    0 & -1 & 0\\
    0 & 0 & 1\\
\end{matrix}\right],
R_{b_3}=
\left[\begin{matrix}
    -1 & 0 & 0\\
    0 & 1 & 0\\
    0 & 0 & -1\\
\end{matrix}\right].
$$
So $\mathrm{Span}_{\mathbb R}\{R_{b}\mid b\in B(f)\}=\{\mathrm{diag}(x,y,z)\mid x,y,z\in \mathbb R\}.$
By \cref{lm: Cb and Rb}, $W:=\mathrm{Span}_{\mathbb R}\{C(b)\mid b\in B(f)\}=\{\mathrm{diag}(x,y,z)\mid x,y,z\in \mathbb R\}.$
So $W^\perp=\{A\in M_3(\mathbb{R})\mid A_{11}=A_{22}=A_{33}=0\}.$
By Pauli slice \cref{lm: Pauli slice}, $t_{pqr}\neq 0\implies p,q,r\text{ are mutually distinct}.$

In the standard case, by the realification (or realnumberizing) method (see \cite{MingjiProgram} v3, Section 5), $f$ is a real signature up to a normalization.
Since normalization doesn't change the complexity, we assume $f$ is real.
Then $\rho_{ijk}=M(f)M(f)^{\dagger}=M(f)M(f)^{\tt T}=\rho_{ijk}^{\tt T}.$
It follows that $t_{pqr}=\mathrm{Tr}(\rho_{ijk}(\sigma_p\otimes \sigma_q\otimes \sigma_r))=\mathrm{Tr}((\sigma_p\otimes \sigma_q\otimes \sigma_r)^{\tt T}\rho_{ijk}^{\tt T})=\mathrm{Tr}(\rho_{ijk}(\sigma_p\otimes \sigma_q\otimes \sigma_r)^{\tt T})$.
Assume that $p,q,r$ are mutually distinct. 
Then $(p,q,r)$ contains exactly one $2$.
Notice that $\sigma_1^{\tt T}=\sigma_1,\sigma_2^{\tt T}=-\sigma_2,\sigma_3^{\tt T}=\sigma_3$, we have  $(\sigma_p\otimes \sigma_q\otimes \sigma_r^{\tt T})=-\sigma_p\otimes \sigma_q\otimes \sigma_r.$
So $t_{pqr}=\mathrm{Tr}(\rho_{ijk}(\sigma_p\otimes \sigma_q\otimes \sigma_r)^{\tt T})=-\mathrm{Tr}(\rho_{ijk}(\sigma_p\otimes \sigma_q\otimes \sigma_r))=-t_{pqr}$.
Thus, $t_{qpr}=0$ if $p,q,r$ are mutually distinct.
We have proved:
\begin{lemma}\label{lm: K4 standard case}
    In the standard case of $B(f)\cong K_4$, $f$ satisfies {\sc 3rd-Orth}.
\end{lemma}

\subsubsection{Non-standard Case}\label{subsubsec: Non-standard Klein}
In this case the Pauli slice constraints alone do not force all the
weight-three Pauli coefficients to vanish.  We instead use the full
factorization assertion in \cref{prop: binary closure}.

\begin{lemma}[Common matching]\label{lm: common matching}
    Up to a permutation of variables,
    \begin{equation}\label{eq: common matching}           
        f=\sum_{s=0}^3w_sb_s(x_1,x_2)b_{p(s)}(x_3,x_4)b_{q(s)}(x_5,x_6),
    \end{equation}
    where $w_s\neq 0\,(\forall s, 0\le s\le 3)$ and $p,q$ are both permutations of $\{0,1,2,3\}.$
\end{lemma}
\begin{proof}
    Suppose $f$ is normalized with norm 1.
    Define $u_i=\frac{1}{\sqrt{2}}({b_i}(00),b_i(01),b_i(10),b_i(11))^{\tt T}$ be the normalized vector.
    Notice that 
    $$\braket{u_i,u_j}=\frac{1}{2}\mathrm{Tr}(M(b_i)^\dagger M(b_j))=\delta_{ij},\quad  \forall i,j=0,1,2,3.$$ 
    So $\{u_0,u_1,u_2,u_3\}$ is an orthonormal basis for the vector space $ (\mathbb C^{4}).$
    Let $F=M_{x_1x_2,x_3x_4x_5x_6}(f)\in \mathbb{C}^{4\times 16}$, and $F_y$ denote the column in $F$ indexed by $y=(x_3,x_4,x_5,x_6).$ Then $F_y$ can be expanded as $F_y=\sum_{s=0}^3 c_s(y)u_s$, where $c_s(y)\in \mathbb{C}.$  
    Define $c_s=(c_s(0000),\ldots,c_s(1111))^{\tt T}\in \mathbb{C}^{16}.$
    Then $F=\sum_{s=0}^3 u_s c_s^{\tt T}$.
    The coefficient is $c_s(y)=u_s^\dagger F_y=\frac{1}{\sqrt{2}}\partial_{12}^{\overline{b_s}}f(y).$
    Since $f$ satisfies {\sc 2nd-Orth}, we have $$\frac{1}{4}\sum_{s=0}^3u_su_s^\dagger=\frac{1}{4}I_4=\rho_{12}=FF^\dagger=\sum_{s,t=0}^3u_sc_s^{\tt T}\overline{c_t}u_t^\dagger=\sum_{s,t=0}^3\braket{c_t,c_s}u_su_t^\dagger.$$
    Then $\braket{c_t,c_s}=\frac{1}{4}\delta_{st}$.
    Let $g_s=2c_s$, then $\braket{g_s,g_t}=\delta_{st}$, and $F=\frac{1}{2}\sum_{s=0}^3 u_sg_s^{\tt T}$.
    We can view $g_s$ as a signature on $y=(x_3,x_4,x_5,x_6).$
    Then $g_s=\sqrt{2}\partial_{12}^{\overline{b_s}}f$ and
    \begin{equation*}\label{eq: Schmidt decomposition of f}
        f(x_1,x_2,x_3,x_4,x_5,x_6)=\frac{1}{2}\sum_{s=0}^3 u_s(x_1,x_2)\otimes g_s(x_3,x_4,x_5,x_6).
    \end{equation*}
    Tracing out $x_1,x_2$, we have 
    \begin{equation}\label{eq: tracing out x1 x2}
        \rho_{3456}=\frac{1}{4}\sum_{s=0}^3g_sg_s^{\dagger}.
    \end{equation}
    Here \(g_sg_s^\dagger\) simply means the matrix whose \((x,y)\)-entry is \(g_s(x)\overline{g_s(y)}\).
    Since we can realize $g_s=\sqrt{2}\partial_{12}^{\overline{b_s}}f$ for $s=0,1,2,3$,
    by assumption \eqref{eq: arity 6 assumption}, $g_s$ factors into two binaries $b_{\ell_s}$ and $b_{k_s}.$
    The two binaries match up the four variables $\{x_3,x_4,x_5,x_6\}$ into two pairs.
    Consider the pair $\{3,4\}$.
    Let $k_{34}$ be the number of the four signatures \(g_0,g_1,g_2,g_3\) whose matching contains the edge \(34\).
    Tracing out $x_5,x_6$ in \eqref{eq: tracing out x1 x2},
    \begin{itemize}
        \item If $g_s=u_{\ell_s}(x_3,x_4)\otimes u_{k_s}(x_5,x_6)$, then the $s$-th term in \eqref{eq: tracing out x1 x2} contributes $\frac{1}{4}{u_{\ell_s}}{u_{\ell_s}}^\dagger.$
        
        \item If $g_s=u_{\ell_s}(x_3,x_5)\otimes u_{k_s}(x_4,x_6)$, then the $s$-th term in \eqref{eq: tracing out x1 x2} contributes $\frac{1}{4}(\frac{1}{2}I_2\otimes \frac{1}{2}I_2)$.
    \end{itemize}
    By {\sc 2nd-Orth}, $$\rho_{34}=\frac{1}{4}\left(\sum_{s: 34 \text{ is matched in } g_s}u_{\ell_s}u_{\ell_s}^\dagger+(4-k_{34})\frac{1}{4}I_4 \right)=\frac{1}{4}I_4.$$
    Multiplying by \(4\) and rearranging gives
    \begin{equation*}
        \sum_{s:\,34\text{ is matched in }g_s}u_{\ell_s}u_{\ell_s}^\dagger = \frac{k_{34}}4I_4.
    \end{equation*}
    Let \(m_j\) be the number of times the label \(u_j\) occurs on the edge \(34\).
    So
    $$\sum_{j=0}^3 m_j u_{j}u_j^\dagger=\frac{k_{34}}{4}I_4=\frac{k_{34}}{4} \sum_{j=0}^3 u_ju_j^\dagger.$$
    Multiply $u_i$ on both sides, we have $m_iu_i=\frac{k_{34}}{4}u_i,\,\forall i=0,1,2,3.$
    So $m_0=m_1=m_2=m_3=\frac{k_{34}}{4}.$
    Since $m_i\in \mathbb{N},$ we have $k_{34}=0$ or $4.$
    Moreover, if $k_{34}=4$, then $m_0=m_1=m_2=m_3=1.$
    Similarly, we can define $k_e$ be number of the four signatures $g_0,g_1,g_2,g_3$ whose matching contains the edge $e=ij$, and the same argument shows $k_e=0$ or $4.$

    Now we count the edge occurrences and show that $k_{34}=k_{56}=4$ after possibly renaming the variables.
    The four signatures \(g_0,g_1,g_2,g_3\) each have a perfect matching consisting of two edges. 
    Therefore the total number of edge occurrences is
    $4\cdot2=8.$
    For every residual edge \(e\), since $k_e=0$ or $4$, its number of occurrences is either \(0\) or \(4\). The only way these counts can sum to \(8\) is that exactly two edges occur four times each.
    Therefore all four contractions use the same matching.
    i.e., after renaming variables possibly, we may assume $$g_s(x_3,x_4,x_5,x_6)=u_{p(s)}(x_3,x_4)\otimes u_{q(s)}(x_5,x_6),\quad \text{for all } s=0,1,2,3.$$
    On either edge of that common matching, the factors \(u_0,u_1,u_2,u_3\) occur exactly once.
    That is, $p(s),q(s)$ are both permutations of $\{0,1,2,3\}.$
    The lemma is proven.
\end{proof}

\begin{lemma}\label{lm: K4 non-standard}
Under assumption \eqref{eq: arity 6 assumption}, the non-standard case of
$B(f)\cong K_4$ cannot occur.
\end{lemma}

\begin{proof}
    We note that $$M(b_r)M(b_s)=M(b_{r\oplus s}),$$
    where $\oplus$ is the commutative addition of the indices of the Klein group $B(f)$ defined by
    $$
    s\oplus s=0, \quad 1\oplus2=3, \quad 1\oplus 3=2,\quad 2\oplus 3=1.
    $$
    Also note that $B(f)$ has the transpose closure property:
    $$
    M(b_0)^{\tt T}=M(b_0), \quad
    M(b_1)^{\tt T}=M(b_1), \quad
    M(b_2)^{\tt T}=M(b_3), \quad 
    M(b_3)^{\tt T}=M(b_2).
    $$
    Let \(\tau(s)\) denote the label obtained by transpose. 
    Thus
    \[
    \tau(0)=0,\quad \tau(1)=1,\quad
    \tau(2)=3,\quad \tau(3)=2.
    \]
    Now we start from the common matching expression \eqref{eq: common matching} in \cref{lm: common matching}:
    \begin{equation*}    
        f=\sum_{s=0}^3w_sb_s(x_1,x_2)b_{p(s)}(x_3,x_4)b_{q(s)}(x_5,x_6),
    \end{equation*}
    Merge $x_2$ and $x_3$ using $=_2$, we get 
    \begin{equation}\label{eq: first expansion of g}
        \begin{aligned}
            \partial_{23}f(x_1,x_4,x_5,x_6)
            =&\sum_{s=0}^3\sum_{a=0}^1 w_sb_s(x_1,a)b_{p(s)}(a,x_4)b_{q(s)}(x_5,x_6)\\
            =&\sum_{s=0}^3w_sb_{s\oplus p(s)}(x_1,x_4)b_{q(s)}(x_5,x_6).
        \end{aligned}
    \end{equation}
    
    Next we show that $s\mapsto s\oplus p(s)$ is also a permutation of $\{0,1,2,3\}.$
    Since $\{b_0,b_1,b_2,b_3\}$ is a basis for the set of binary signatures, $\{b_r\otimes b_t\mid 0\le r,t\le 3\}$ is a basis for the set of 4-ary signatures.
    We can expand $\partial_{23}f$ under this basis:
    \begin{equation}\label{eq: second expansion of g}
        \partial_{23}f(x_1,x_4,x_5,x_6)=\sum_{r,t=0}^3 A_{rt}b_r(x_1,x_4)b_t(x_5,x_6).
    \end{equation}
    Compare \eqref{eq: first expansion of g} and \eqref{eq: second expansion of g}, we have $A_{s\oplus p(s),q(s)}=w_s.$
    Since $q(s)$ is a permutation, every column of $A$ receices a non-zero entry. 
    So $\mathrm{rank}(A)=\#\{s\oplus p(s)\mid s=0,1,2,3\}.$

    On the other hand, $\partial_{23}f$ factors into two binaries by assumption \eqref{eq: arity 6 assumption}.
    There are three cases:
    \begin{enumerate}
        \item $\partial_{23}f(x_1,x_4,x_5,x_6)=\lambda b_{i}(x_1,x_4)b_j(x_5,x_6)$.
        Then the matrix $A$ only has one non-zero entry $A_{ij}=\lambda.$
        So $\mathrm{rank}(A)=1.$
        \item $\partial_{23}f(x_1,x_4,x_5,x_6)=\lambda b_{i}(x_1,x_5)b_j(x_4,x_6)$.
        Then $M_{x_1x_4,x_5x_6}(\partial_{23}f)=\lambda M_{x_1,x_5}(b_i)\otimes M_{x_4,x_6}(b_j)$ and thus $\mathrm{rank}(M_{x_1x_4,x_5x_6}(\partial_{23}f))=2\times 2=4.$
        Basis change preserves the rank.
        So $\mathrm{rank}(A)=4.$
        \item $\partial_{23}f(x_1,x_4,x_5,x_6)=\lambda b_{i}(x_1,x_6)b_j(x_4,x_5)$.
        Similar to Case 2, $\mathrm{rank}(A)=4.$
    \end{enumerate}
    So far we proved $\mathrm{rank}(A)=1$ or $4$.
    If $\mathrm{rank}(A)=1,$ then $s\oplus p(s)=r$ for $s=0,1,2,3.$
    By \eqref{eq: first expansion of g}, $$\partial_{23}f=\sum_{s=0}^3w_sb_r(x_1,x_4)b_{q(s)}(x_5,x_6)=b_r(x_1,x_4)\otimes \sum_{s=0}^3w_sb_{q(s)}(x_5,x_6).$$
    But $q(s)$ is a permutation, \(b_0,b_1,b_2,b_3\) are linearly independent, and all the coefficients $w_s\neq 0$.
    Therefore $\sum_{s=0}^3 w_s b_{q(s)}$ cannot be proportional to a single member of \(B(f)\).
    This contradicts \cref{prop: binary closure}.
    Thus, $\mathrm{rank}(A)=\#\{s\oplus p(s)\mid s=0,1,2,3\}=4.$
    We proved that $s\mapsto s\oplus p(s)$ is a permutation of $\{0,1,2,3\}.$

    We now consider merging $x_1$ and $x_3$ using $=_2.$
    \begin{equation}\label{eq: merge 13}
        \begin{aligned}
            \partial_{13}f(x_2,x_4,x_5,x_6)
            =&\sum_{s=0}^3\sum_{a=0}^1 w_sb_s(a,x_2)b_{p(s)}(a,x_4)b_{q(s)}(x_5,x_6)\\
            =&\sum_{s=0}^3w_sb_{\tau(s)\oplus p(s)}(x_2,x_4)b_{q(s)}(x_5,x_6).
        \end{aligned}
    \end{equation}
    By the same rank argument above, $s\mapsto \tau(s)\oplus p(s)$ is a permutation of $\{0,1,2,3\}.$
    We are ready to arrive at a contradiction by the following claim:
    \begin{claim}\label{claim: impossible permutations}
        It is impossible that $p(s),s\oplus p(s),\tau(s)\oplus p(s)$ are all permutations of $\{0,1,2,3\}.$
    \end{claim}
    \begin{claimproof}{\cref{claim: impossible permutations}}
        Let
        $
        d=p(0), r(s)=p(s)\oplus d.
        $
        Since $p$ is a permutation, $r$ is also a
        permutation, and $r(0)=0$.
        By assumption, $s\mapsto s\oplus p(s)$ is a permutation. 
        If $r(s)=s$ for some $s\neq0$, then
        $
        s\oplus p(s)=s\oplus r(s)\oplus d=d=0\oplus p(0),
        $
        a contradiction. 
        Hence $r$ has no fixed point on $\{1,2,3\}$, so
        \[
        r|_{\{1,2,3\}}\in\{(1\,2\,3),(1\,3\,2)\}.
        \]
        In the first case, $r(2)=3=\tau(2)$; in the second case,
        $r(3)=2=\tau(3)$. Thus, in either case, there exists
        $s_*\in\{2,3\}$ such that $r(s_*)=\tau(s_*)$. Consequently,
        \[
        \tau(s_*)\oplus p(s_*)
         =\tau(s_*)\oplus r(s_*)\oplus d
         =d
         =\tau(0)\oplus p(0).
        \]
        Therefore $s\mapsto\tau(s)\oplus p(s)$ is not injective, hence is not
        a permutation. It is thus impossible for
        $
        p(s), s\oplus p(s), \tau(s)\oplus p(s)
        $
        to all be permutations of $\{0,1,2,3\}$.
    \end{claimproof}
    The lemma is proved.
\end{proof}

\subsection{Trivial Group}\label{subsec: Trivial group}
In this subsection, we assume $B(f)\cong C_1$ is the trivial group, i.e., the only binary signature generated from $f$ is $=_2$.
For a pair of distinct indices $\{i,j\}$, we can view it as an edge in $K_6$.
We write $\partial_{ij}f=\partial_ef$.
\cref{prop: binary closure} gives
\begin{equation}
    \partial_e f
    =
    \lambda_e \bigotimes_{d\in M_e} (=_2)_d ,
    \qquad \lambda_e\neq 0,
    \label{eq:C1-contraction}
\end{equation}
where \(M_e\) is a perfect matching of the four vertices outside \(e\);
in particular, \(M_e\) consists of two disjoint edges.
We first show that all \(\lambda_e\) have the same modulus.  
By {\sc 2nd-Orth}, there is a common \(\mu>0\) such that
$
    |f_e^{ab}|^2=\mu,
    \text{for every \(e\) and \(a,b\in\{0,1\}\),}
$
and the four slices \(f_e^{ab}\) are pairwise orthogonal.  Hence
$
    |\partial_e f|^2
    =
    |f_e^{00}+f_e^{11}|^2
    =
    2\mu.
$
Since \(|(=_2)|^2=2\), the right-hand side of
\eqref{eq:C1-contraction} has squared norm \(4|\lambda_e|^2\).
Consequently,
$
    |\lambda_e|^2=\frac{\mu}{2},
$
which is independent of \(e\).

Now let \(e,d\) be disjoint edges.  The contractions commute:
$
    \partial_d\partial_e f
    =
    \partial_e\partial_d f.
$
If \(d\in M_e\), then the second contraction closes an equality loop
and contributes a factor \(2\).  If \(d\notin M_e\), then \(d\) crosses
the two matching edges and contributes no additional scalar.  In both
cases the remaining tensor is equality on the two uncontracted
vertices.  Therefore
\begin{equation}
    \lambda_e\,2^{[d\in M_e]}
    =
    \lambda_d\,2^{[e\in M_d]},
    \label{eq:C1-double-contraction}
\end{equation}
where \([P]\in\{0,1\}\) is the indicator of the statement \(P\).
Taking absolute values in \eqref{eq:C1-double-contraction} and using
\(|\lambda_e|=|\lambda_d|\) gives the reciprocity relation
\begin{equation}
    d\in M_e
    \quad\Longleftrightarrow\quad
    e\in M_d.
    \label{eq:C1-reciprocity}
\end{equation}
The powers of \(2\) in \eqref{eq:C1-double-contraction} then agree, so
\(\lambda_e=\lambda_d\) whenever \(e\cap d=\varnothing\).  
Hence there is
one common nonzero scalar
$
    \lambda_e=\lambda\neq 0
    \text{ for every edge \(e\).}
$
Define
$
    F_e:=\{e\}\cup M_e.
$
Then \(F_e\) is a perfect matching of \(K_6\).  
If \(d\in M_e\), write
\(M_e=\{d,c\}\).  By \eqref{eq:C1-reciprocity}, \(e\in M_d\), and the
remaining edge is necessarily \(c\).  Thus
$
    M_d=\{e,c\},
    \text{ and therefore }
    F_d=F_e.
$
It follows that the distinct sets \(F_e\) partition the \(15\) edges of
\(K_6\) into five perfect matchings.  
Up to relabelling, they are
\begin{equation}
\begin{aligned}
    &12\mid 34\mid 56, \qquad
     13\mid 25\mid 46, \qquad
     14\mid 26\mid 35,\\
    &15\mid 24\mid 36, \qquad
     16\mid 23\mid 45.
\end{aligned}
\label{eq:C1-factorization}
\end{equation}
We now derive a contradiction from four of these factors.  From
\[
    M_{12}=\{34,56\},\qquad
    M_{13}=\{25,46\},\qquad
    M_{25}=\{13,46\},\qquad
    M_{14}=\{26,35\},
\]
equation \eqref{eq:C1-contraction}, evaluated at suitable residual
assignments, gives
\begin{align}
    f_{001010}+f_{111010}&=0,
    &
    f_{000110}+f_{110110}&=0,
    \label{eq:C1-eight-a}\\
    f_{010010}+f_{111010}&=\lambda,
    &
    f_{000110}+f_{101110}&=0,
    \label{eq:C1-eight-b}\\
    f_{000000}+f_{010010}&=\lambda,
    &
    f_{100100}+f_{110110}&=0,
    \label{eq:C1-eight-c}\\
    f_{000000}+f_{100100}&=\lambda,
    &
    f_{001010}+f_{101110}&=\lambda.
    \label{eq:C1-eight-d}
\end{align}
The first seven equations imply
\[
\begin{aligned}
    f_{001010}&=-f_{111010},&
    f_{101110}&=f_{110110},\\
    f_{000000}&=f_{111010},&
    f_{100100}&=-f_{110110}.
\end{aligned}
\]
Using the first equation in \eqref{eq:C1-eight-d}, we therefore obtain
\[
\begin{aligned}
    f_{001010}+f_{101110}
    &=-f_{111010}+f_{110110}\\
    &=-\bigl(f_{000000}+f_{100100}\bigr)\\
    &=-\lambda.
\end{aligned}
\]
However, the second equation in \eqref{eq:C1-eight-d} says that the same
quantity is \(\lambda\).  Hence \(\lambda=-\lambda\), so
\(\lambda=0\), contradicting \(\lambda\neq 0\).
Therefore, we proved:
\begin{lemma}\label{lm: trivial group}
    Under assumption \eqref{eq: arity 6 assumption}, the case $B(f)\cong C_1$ cannot occur.
\end{lemma}

\subsection{Order Two Group}\label{subsec: Order Two group}
In this subsection, we assume that $B(f)\cong C_2$, the group of order $2$. 
We will show this case cannot occur under assumption \eqref{eq: arity 6 assumption}.
Suppose $B(f)=\{e,b\}$, where $M(e)=I_2$ and $b^2=e.$
\begin{lemma}\label{lm: C2 standard form}
    Up to a normalization and a real orthogonal holographic transformation, there are two cases: 1. (Symmetric Case): $M(b)=\left[\begin{matrix}
        \mathfrak{i} & 0\\
        0 & -\mathfrak{i}
    \end{matrix}\right]$
    or 2. (Asymmetric Case):
    $M(b)=\left[\begin{matrix}
        0 & 1\\
        -1 & 0
    \end{matrix}\right].$
\end{lemma}
\begin{proof}
    Notice that $B(f)$ is closed under transpose. 
    So $M(b)^{\tt T}=\pm M(b).$
    \begin{enumerate}
        \item $M(b)^{\tt T}=M(b)$.
        Assume $M(b)=\left[\begin{matrix}
            x & y\\
            y & z
        \end{matrix}\right]\in \mathbf{SU}(2).$
        Since $b$ has order two,
        we have $M(b)^2=\left[\begin{matrix}
            x^2+y^2 & xy+yz\\
            xy+yz & y^2+z^2
        \end{matrix}\right]=\pm I.$
        So $y(x+z)=0$ and $x^2=z^2.$
        \begin{itemize}
            \item 
            If $y=0$, $M(b)=\left[\begin{matrix}
            x & 0\\
            0 & z
            \end{matrix}\right].$
            By $\det(M(b))=xz=1$ and $x^2=z^2$, we have $x^4=1$.
            If $x=\pm 1$, then $z=\pm 1$ by $xz=1$, contradicting $b$ has order 2.
            So $x=\pm \mathfrak{i}.$
            It follows that $z=\mp \mathfrak{i}.$
            Up to a normalization, $M(b)=\left[\begin{matrix}
                \mathfrak{i} & 0\\
                0 & -\mathfrak{i}
            \end{matrix}\right].$

            \item If $x+z=0$, $M(b)=\left[\begin{matrix}
                x & y\\
                y & -x
            \end{matrix}\right].$
            By $M(b)\in \mathbf{SU}(2)$, 
            we have $x^2+y^2=-1$ and 
            $M(b)M(b)^\dagger=\left[\begin{matrix}
                |x|^2+|y|^2 & x\overline{y}-\overline{x}y\\
                \overline{x}y-x\overline{y} & |x|^2+|y|^2
            \end{matrix}\right]=\pm I_2$. 
            So $x\overline{y}=\overline{x}y$.
            If $xy\neq 0$, then $\arg(x)=\arg(y)$.
            After a normalization, we may assume $x,y\in \mathbb{R}$.
            If $xy=0$, we may still normalize and assume $x,y\in \mathbb{R}.$
            Since $x^2+y^2=-1$, we may assume $x=\mathfrak{i}\cos(2\theta)$ and $y=\mathfrak{i}\sin(2\theta).$
            Let $T=\left[\begin{matrix}
                \cos \theta & -\sin \theta\\
                \sin \theta & \cos \theta
            \end{matrix}\right]\in \mathbf{SO}(2).$
            A direct calculation shows that $T^{\tt T}M(b)T=\left[\begin{matrix}
                \mathfrak{i} & 0\\
                0 & -\mathfrak{i}
            \end{matrix}\right].$
            So the lemma is true by a $T$-transformation.
        \end{itemize}
        
        \item 
        $M(b)^{\tt T}=-M(b).$
        Assume that 
        $M(b)=\left[\begin{matrix}
                0 & x\\
                -x & 0
            \end{matrix}\right].$
        Since $M(b)\in \mathbf{SU}(2)$, we have $x^2=1$, $x=\pm 1$.
        So up to a normalization, $M(b)=\left[\begin{matrix}
                0 & 1\\
                -1 & 0
            \end{matrix}\right].$
    \end{enumerate}
\end{proof}

For simplicity, in the standard case, we no longer require $M(b)\in \mathbf{SU}(2)$ 
and assume that $M(b)=\left[\begin{matrix}
    1 & 0\\
    0 & -1
\end{matrix}\right].$
We first record two lemmas that will be used in both cases:
\begin{lemma}\label{lem:c2-two-projection}
    Suppose that $g=\lambda\phi(x_\ell,x_m)\psi(x_u,x_v)$ for some $\lambda\neq 0$, $\phi,\psi\in\{e,b\}$ and $\{1,2\}\subseteq\{\ell,m,u,v\}.$
    If $\{1,2\}=\{\ell,m\}$ or $\{1,2\}=\{u,v\}$, then exactly one of
    $\partial_{12}^{e}g$ and $\partial_{12}^{b}g$ is nonzero.  
    Otherwise, both $\partial_{12}^{e}g$ and $\partial_{12}^{b}g$ are nonzero and are not proportional.
\end{lemma}

\begin{proof}
    Define $\phi^p=e$ if $p \equiv 0 \text{ (mod }{2})$, and $\phi^p=b$ if $p \equiv 1 \text{ (mod }{2})$.
    First suppose that $g=\lambda\phi^p(x_1,x_2)\phi^q(x_u,x_v)$ for some $p,q\in \N.$
    Notice that
    \[
     \sum_{\alpha,\beta\in\{0,1\}}
     \phi^p(\alpha,\beta)\phi^t(\alpha,\beta)
     =
     \begin{cases}
     2,&p \equiv t \text{ (mod }{2})\\
     0,&p \not\equiv t \text{ (mod }{2}).
     \end{cases}
    \]
    Consequently, merging $x_1$ and $x_2$ by $e=\phi^t \,(t \text{ is even})$ or $b=\phi^t\,(t \text{ is odd})$ gives exactly one nonzero result.

    Now suppose that $g=\lambda\phi^p(x_1,x_u)\phi^q(x_2,x_v)$.
    In the symmetric case, merging $x_1$ and $x_2$ by $\phi^t$ produces $\partial_{12}^{\phi^t} g=\lambda\phi^{p+q+t}(x_u,x_v)$.
    The two signatures $\partial_{12}^eg$ and $\partial_{12}^bg$ therefore produce the two distinct factors $e$ and $b$, up to nonzero scalars.
    In the asymmetric case, $M_{x_u,x_v}(\partial_{12}^{\phi^t} g)=\lambda M(\phi^p)^{\tt T}M(\phi^t)M(\phi^q).$
    For $t=0$ and $t=1$, these matrices are projectively $I_2$ and $M(b)$. 
    Hence they are again nonzero and not proportional.
\end{proof}

\begin{lemma}\label{lm: C2 contraction orth}
    For any pair of indices $\{i,j\},$ $|\partial_{ij}^{e}f|=|\partial_{ij}^bf|> 0$ and $\braket{\partial_{ij}^ef,\partial_{ij}^bf}=0.$
\end{lemma}

\begin{proof}
    In the symmetric case, $\partial_{ij}^ef=f_{ij}^{00}+f_{ij}^{11}$, and $\partial_{ij}^bf=f_{ij}^{00}-f_{ij}^{11}$.
    By {\sc 2nd-Orth}, $|f_{ij}^{00}|=|f_{ij}^{11}|>0$ and $\braket{f_{ij}^{00},f_{ij}^{11}}=0.$
    So
    \begin{equation*}
        \begin{aligned}
            |\partial_{ij}^{e}f|^2=\braket{f_{ij}^{00}+f_{ij}^{11},f_{ij}^{00}+f_{ij}^{11}}=|f_{ij}^{00}|^2+|f_{ij}^{11}|^2>0.
        \end{aligned}
    \end{equation*}
    Similarly, $|\partial_{ij}^{b}f|^2=|f_{ij}^{00}|^2+|f_{ij}^{11}|^2=|\partial_{ij}^{e}f|^2>0.$
    Also,
    \begin{equation*}
        \begin{aligned}
            \braket{\partial_{ij}^{e}f,\partial_{ij}^{b}f}=\braket{f_{ij}^{00}+f_{ij}^{11},f_{ij}^{00}-f_{ij}^{11}}=|f_{ij}^{00}|^2-|f_{ij}^{11}|^2-\braket{f_{ij}^{00},f_{ij}^{11}}+\braket{f_{ij}^{11},f_{ij}^{00}}
            =0.
        \end{aligned}
    \end{equation*}
    In the asymmetric case, $\partial_{ij}^ef=f_{ij}^{01}+f_{ij}^{10}$, and $\partial_{ij}^bf=f_{ij}^{01}-f_{ij}^{10}$.
    {\sc 2nd-Orth} gives $|f_{ij}^{01}|=|f_{ij}^{10}|>0$ and $\braket{f_{ij}^{01},f_{ij}^{10}}=0.$
    So $|\partial_{ij}^{e}f|^2=|\partial_{ij}^bf|^2=|f_{ij}^{01}|^2+|f_{ij}^{10}|^2> 0$ and $\braket{\partial_{ij}^ef,\partial_{ij}^bf}=0.$
\end{proof}

\subsubsection{Symmetric Case}

In this case, $B(f)=\{e,b\}$ where $M(e)=I_2$ and $M(b)=\left[\begin{matrix}
    1 & 0\\
    0 & -1
\end{matrix}\right]$.

\begin{lemma}\label{lm: C2 symmetric impossible}
    Under assumption \eqref{eq: arity 6 assumption}, the symmetric case of $B(f)\cong C_2$ cannot occur.
\end{lemma}
\begin{proof}
    
    By permuting the residual variables $(x_3,x_4,x_5,x_6)$ and applying local binary modification by $b$, we may normalize
    \begin{equation}
         \partial_{12}^{e}f=x\cdot e(x_3,x_4)e(x_5,x_6),
         \qquad x\in\mathbb C^\times.
         \label{eq:c2-symmetric-normalization}
    \end{equation}
    Here $x\neq 0$ by \cref{lm: C2 contraction orth}.
    With a little abuse of notation, we still use $f$ to denote the 6-ary signature after the binary modification.
    Notice that a binary modification by $b$ doesn't change assumption \eqref{eq: arity 6 assumption} and the conclusion in \cref{lem:c2-two-projection}.
    \begin{claim}\label{claim: C2 symmetry}
        Up to the symmetries preserving \eqref{eq:c2-symmetric-normalization}, there are three possibilities for $\partial_{12}^{b}f$:
        \[
         y\cdot e(x_3,x_4)b(x_5,x_6),\quad
         y\cdot b(x_3,x_4)b(x_5,x_6),\quad
         y\cdot e(x_3,x_5)b(x_4,x_6),
         \qquad y\in\mathbb C^\times.
        \]
    \end{claim}
    \begin{claimproof}{\cref{claim: C2 symmetry}}
        By \cref{lm: C2 contraction orth}, $\braket{\partial_{12}^{e}f,\partial_{12}^{b}f}=0.$
        If $\partial_{12}^bf=y\cdot \phi^p(x_3,x_4)\phi^q(x_5,x_6)$ where $y\neq 0$ by \cref{lm: C2 contraction orth}, orthogonality to \(\partial_{12}^{e}f=x\cdot e(x_3,x_4)e(x_5,x_6)\) requires at least one \(b\)-factor in $\phi^p$ and $\phi^q$. 
        Up to exchanging $(x_3,x_4)$ and $(x_5,x_6)$, we have $\partial_{12}^bf=y\cdot e(x_3,x_4)b(x_5,x_6)$ or $\partial_{12}^bf= y\cdot b(x_3,x_4)b(x_5,x_6).$
        Otherwise, $x_3$ and $x_4$ are in distinct binary factors of $\partial_{12}^bf$.
        Up to exchanging $x_5$ and $x_6$, we may assume
        $\partial_{12}^bf=y\cdot \phi^p(x_3,x_5)\phi^q(x_4,x_6), \,y\ne 0.$
        Then $0=\braket{\partial_{12}^ef,\partial_{12}^bf}=xy\cdot\mathrm{Tr}(I_2\cdot M(\phi^q)\cdot I_2\cdot M(\phi^p)^{\tt T})=(-1)^{p+q}.$
        So exactly one binary factor of $\partial_{12}^bf$ is $b$.
        Up to exchanging $(x_3,x_5)$ and $(x_4,x_6)$, $\partial_{12}^bf=y\cdot e(x_3,x_5)b(x_4,x_6).$
    \end{claimproof}

    Suppose first that
    $\partial_{12}^bf=y\cdot e(x_3,x_4)b(x_5,x_6),\,y\in\C^\times$, and put $g=\partial_{34}^{b}f$.  Commutativity of contractions gives
    \[
     \partial_{12}^{e}g
     =\partial_{34}^{b}\partial_{12}^{e}f=0,
     \qquad
     \partial_{12}^{b}g
     =\partial_{34}^{b}\partial_{12}^{b}f=0.
    \]
    Since $g\neq 0$ by \cref{lm: C2 contraction orth}, the above equations contradicts \cref{lem:c2-two-projection}.
    Next suppose that $\partial_{12}^{b}f=y\cdot b(x_3,x_4)b(x_5,x_6)$, and put $g=\partial_{35}^{e}f$.  Then
    \[
     \partial_{12}^{e}g
     =\partial_{35}^{e}\partial_{12}^{e}f=x\cdot e(x_4,x_6),
     \qquad
     \partial_{12}^{b}g
     =\partial_{35}^{e}\partial_{12}^{b}f=y\cdot e(x_4,x_6).
    \]
    Both projections are nonzero and proportional, again contradicting
    \cref{lem:c2-two-projection}.

    It remains to consider the case
    \begin{equation}
        \partial_{12}^{e}f=x\cdot e(x_3,x_4)e(x_5,x_6),
         \qquad
         \partial_{12}^{b}f=y\cdot e(x_3,x_5)b(x_4,x_6),\qquad x,y\in \C^\times.
         \label{eq:c2-symmetric-crossed}
    \end{equation}
    Put $g=\partial_{34}^{b}f$.  Commuting the two contractions gives
    \begin{equation}\label{eq: C2 sym}
        \partial_{12}^{e}g=\partial_{34}^b\partial_{12}^ef=0,
     \qquad
     \partial_{12}^{b}g=\partial_{34}^b\partial_{12}^bf=y\cdot e(x_5,x_6)\ne0,
    \end{equation}
    where the second equation follows from $I_2^{\tt T}ZZ=I_2$.

    By \cref{lem:c2-two-projection}, $e(x_1,x_2)\mid g$ or $b(x_1,x_2)\mid g$.
    If $g=z\cdot e(x_1,x_2)t(x_5,x_6)$ for some $t\in \{e,b\}$ and $z\in \C^\times$, then $\partial_{12}^e g=2z\cdot t(x_5,x_6)\neq 0$, contradiction.
    Thus
    $g=z\cdot b(x_1,x_2)t(x_5,x_6)$ for some $t\in\{e,b\}$.  
    Merge $x_1$ and $x_2$ by $b$ gives
    \(
     \partial_{12}^{b}g=2z\cdot t(x_5,x_6).
    \)
    Comparison with \eqref{eq: C2 sym} gives $|y|=2|z|$ and $t=e$.
    On the other hand, for any pair of indices $\{i,j\}$, we have $|\partial_{ij}^bf|^2=\lambda|b|^2=2\lambda$ by \cref{lm: 2nd-orth contraction nonzero}, where $\lambda=|f_{ij}^{00}|^2=|f_{ij}^{01}|^2=|f_{ij}^{10}|^2=|f_{ij}^{11}|^2>0$.
    Consequently $|\partial_{12}^{b}f|=|\partial_{34}^{b}f|=|g|>0$.
    So $|y|=|z|>0$, contradicting $|y|=2|z|$.
    Therefore, the case \eqref{eq:c2-symmetric-crossed} cannot occur.
\end{proof}

\subsubsection{Asymmetric Case}
In the asymmetric case, we have $B(f)=\{e,b\}$, where $M(e)=I_2$ and $M(b)=\left[\begin{matrix}
    0 & 1\\
    -1 & 0
\end{matrix}\right].$
\begin{lemma}\label{lm: C2 asymmetric impossible}
    Under assumption \eqref{eq: arity 6 assumption}, the asymmetric case of $B(f)\cong C_2$ cannot occur.
\end{lemma}
\begin{proof}
    Similar to the proof of \cref{lm: C2 symmetric impossible}, we may permute the variables and apply binary modification and assume
    \[
     \partial_{12}^{e}f=x\cdot e(x_3,x_4)e(x_5,x_6),
     \qquad x\in\mathbb C^\times.
    \] 
    By \cref{lm: C2 contraction orth}, $\partial_{ij}^{e}f\neq 0,\,\partial_{ij}^bf\neq 0,$ and $\braket{\partial_{ij}^ef,\partial_{ij}^bf}=0$ for any pair of indices ${i,j}.$ 
    Again orthogonality of $\partial_{12}^{e}f$ and
    $\partial_{12}^{b}f$ leaves three possibilities for $\partial_{12}^bf:$
    \[
     y\cdot e(x_3,x_4)b(x_5,x_6),\quad
     y\cdot b(x_3,x_4)b(x_5,x_6),\quad
     y\cdot e(x_3,x_5)b(x_4,x_6),
    \]
    for some $y\in \mathbb{C}^\times.$
    Suppose first that $\partial_{12}^{b}f=y\cdot e(x_3,x_4)b(x_5,x_6)$, and put $g=\partial_{34}^{b}f$.  Commutativity gives
    \[
     \partial_{12}^{e}g
     =\partial_{34}^{b}\partial_{12}^{e}f=0,
     \qquad
     \partial_{12}^{b}g
     =\partial_{34}^{b}\partial_{12}^{b}f=0.
    \]
    This contradicts \cref{lem:c2-two-projection}.
    Next suppose that $\partial_{12}^bf=y\cdot b(x_3,x_4)b(x_5,x_6)$, and put $g=\partial_{35}^{e}f$.  Since $M(b)^{\tt T}I_2M(b)=I_2$, we have
    \[
     \partial_{12}^{e}g=\partial_{35}^e\partial_{12}^ef=x\cdot e(x_4,x_6),
     \qquad
     \partial_{12}^{b}g=\partial_{35}^e\partial_{12}^bf=y\cdot e(x_4,x_6).
    \]
    They are nonzero and proportional, again contradicting \cref{lem:c2-two-projection}.

    It remains to consider the case
    \begin{equation}
        \partial_{12}^{e}f=x\cdot e(x_3,x_4)e(x_5,x_6),
        \qquad     \partial_{12}^{b}f=y\cdot e(x_3,x_5)b(x_4,x_6).
    \end{equation}
    Put $g=\partial_{34}^{b}f$.  Commuting the two contractions gives
    \begin{equation}\label{eq: C2 asym}
        \partial_{12}^{e}g=\partial_{34}^b\partial_{12}^ef=0,
     \qquad
     \partial_{12}^{b}g=\partial_{34}^b\partial_{12}^bf=-y\cdot e(x_5,x_6)\ne0,
    \end{equation}
    where the second identity follows from $I_2^{\tt T}JJ=-I_2$.

    By \cref{lem:c2-two-projection}, $e(x_1,x_2)\mid g$ or $b(x_1,x_2)\mid g$.
    If $g=z\cdot e(x_1,x_2)t(x_5,x_6)$ for some $t\in \{e,b\}$, then $\partial_{12}^e g=2z\cdot t(x_5,x_6)\neq 0$, contradiction.
    Thus
    $g=z\cdot b(x_1,x_2)t(x_5,x_6)$ for some $t\in\{e,b\}$.  
    Merge $x_1$ and $x_2$ by $b$ gives
    \(
     \partial_{12}^{b}g=2z\cdot t(x_5,x_6).
    \)
    Comparison with \eqref{eq: C2 asym} gives $|y|=2|z|$ and $t=e$.
    On the other hand, for any pair of indices $\{i,j\}$, we have $|\partial_{ij}^bf|^2=\lambda|b|^2=2\lambda$ by \cref{lm: 2nd-orth contraction nonzero}, where $\lambda=|f_{ij}^{00}|^2=|f_{ij}^{01}|^2=|f_{ij}^{10}|^2=|f_{ij}^{11}|^2>0$.
    Consequently $|\partial_{12}^{b}f|=|\partial_{34}^{b}f|=|g|>0$.
    So $|y|=|z|>0$, contradicting $|y|=2|z|$.
\end{proof}

\subsection{Cyclic Group of Order at Least Three}\label{subsec: Cyclic group of Order at least three}
In this subsection, we assume that $B(f)\cong C_n(n\ge 3)$ which is the cyclic group of order $n$.
We are going to prove that under assumption \eqref{eq: arity 6 assumption}, $B(f)\cong C_n$ cannot occur.
In \cite{MingjiProgram} v1, it is proved that we only need to consider two possible cases:
\begin{enumerate}
    \item 
    (Standard Case).
    In the setting of $\Holant(\mathcal{F}),$ up to a real orthogonal holographic transformation, 
    the projective binary group of $f$ is $B(f)=\{b^0, b^1,\ldots,b^{n-1}\}$, where $b$ is the binary signature $\left[\begin{matrix}
        1 & 0\\
        0 & \zeta
    \end{matrix}\right]$ and  $\zeta=e^{\frac{2\pi \mathfrak{i}}{n}}.$
    
    \item 
    (Non-standard Case).
    In the setting of $\holant{\neq_2}{\hF}$, the projective binary group of $\widehat{f}$ is $\widehat{B}(\widehat{f})=\{b^0,b^1,\ldots,b^n\}$,
    where $b$ is the binary signature $\left[\begin{matrix}
        0 & 1\\
        \zeta & 0
    \end{matrix}\right]$ and $\zeta=e^{\frac{2\pi \mathfrak{i}}{n}}.$
\end{enumerate}

\subsubsection{Standard Case}
By \cref{prop: binary closure}, after merging $x_i$ and $x_j$ using some binary signature $b^k$ in $B(f)$, $\partial_{ij}^{b^k}f$ factors into two binary signatures in $B(f)$. 
This two binaries matches up the remaining variables by pairs.
The following lemma shows that for any $\{i,j\}$ there is a common matching which is independent of $k\in [n]$.
\begin{lemma}[Common matching]\label{lm: common matching Cn}
    Fix an arbitrary pair of indices $\{i,j\}$.
    For any $b^k\in B(f)$, $\partial_{ij}^{b^k}f(x_\ell,x_m,x_p,x_q)$ factors into binaries $b^{k_1}(x_\ell,x_n)$ and $b^{k_2}(x_p,x_q)$ for some $k_1,k_2\in [n]$ after permuting the variables possibly, where $\{i,j,\ell,m,p,q\}=\{1,2,3,4,5,6\}$.
\end{lemma}
\begin{proof}
    Let $f^{(k)}=\partial_{ij}^{b^k}f,\, g=f_{ij}^{00},\,h=f_{ij}^{11}$ and $y=(x_\ell,x_m,x_p,x_q).$
    We have 
    \begin{equation}\label{eq: cyclic group merging}
        f^{(k)}(y)=\partial_{ij}^{b^{k}}f(y)=f_{ij}^{00}(y)+\zeta^k f_{ij}^{11}(y)=g(y)+\zeta^k h(y), \quad \text{for every } y\in \{0,1\}^4.
    \end{equation}
    By assumption \eqref{eq: arity 6 assumption}, every merging signature $f^{(k)}$ factors into two binary signatures in $B(f).$
    In the standard case, the two binary signatures match up $x_\ell,x_m,x_p,x_q$ into two equal pairs.
    So the support of $f^{(k)}$ is
    \begin{equation*}
        \begin{aligned}
            &S_1=\{0000,0011,1100,1111\}, \text{ or }\\
            &S_2=\{0000,0101,1010,1111\}, \text{ or }\\
    &S_3=\{0000,0110,1001,1111\},
        \end{aligned}
    \end{equation*}
    according to the matching $\ell m\mid pq, \ell p\mid mq$ or $\ell q\mid mp.$
    \begin{claim}\label{claim: cyclic group common matching}
        If there exists two distinct $k,k'\in [n]$, such that $\mathscr{S}(f^{(k)})=\mathscr{S}(f^{(k')})=S$, then $\mathscr{S}(f^{(t)})=S$ for every $t\in[n]$. 
    \end{claim}
    \begin{claimproof}{\cref{claim: cyclic group common matching}}
    By \eqref{eq: cyclic group merging}, we have
    \[\left[
    \begin{matrix}
        f^{(k)}(y)\\
        f^{(k')}(y)
    \end{matrix}
    \right]
    =\left[\begin{matrix}
        1 & \zeta^k\\
        1 & \zeta^{k'}
    \end{matrix}
    \right]
    \left[\begin{matrix}
        g(y)\\
        h(y)
    \end{matrix}
    \right],\quad \forall y\in \{0,1\}^4.
    \]
    Thus, $g(y)=h(y)=0 \iff f^{(k)}(y)=f^{(k')}(y)=0.$
    Since $\mathscr{S}(f^{(k)})=\mathscr{S}(f^{(k')})=S$, $g$ and $h$ vanishes outside $S$.
    By \eqref{eq: cyclic group merging}, $f^{(t)}$ vanishes outside $S$ for every $t\in [n]$.
    So $\mathscr{S}(f^{(t)})\subseteq S.$
    Since $\mathscr{S}(f^{(t)})=S_1,S_2$ or $S_3$, we have $\mathscr{S}(f)^{(t)}=S,$ for every $t\in [n].$
    \end{claimproof}
    
    If $n\ge 4$, then by pigeon hole principle, there exists distinct $k,k'\in [n]$ such that $\mathscr{S}(f^{(k)})=\mathscr{S}(f^{(k')})=S$.
    By \cref{claim: cyclic group common matching}, the lemma is proved.
    In the following, we may assume that $n=3$, and $\{\mathscr{S}(f^{(0)}),\mathscr{S}(f^{(1)}),\mathscr{S}(f^{(2)})\}=\{S_1,S_2,S_3\}.$
    Now $\zeta=\omega=e^{\frac{2\pi \mathfrak{i}}{3}}.$
    By \eqref{eq: cyclic group merging}, 
    \begin{equation}\label{eq: cyclic group n=3}
        f^{(0)}+\omega f^{(1)}+\omega^2 f^{(2)}=(1+\omega+\omega^2)g+(1+\omega^2+\omega^4)h=0.
    \end{equation}
    Notice that $y=0011$ is in $S_1$ but not in $S_2\cup S_3.$
    Evaluate \eqref{eq: cyclic group n=3} at $y=0011,$ LHS is nonzero, contradiction.
\end{proof}

\begin{lemma}\label{lm: cyclic group support size 2}
    Fix an arbitrary pair of indices $\{i,j\}$, $|\mathscr{S}(f_{ij}^{00})|=|\mathscr{S}(f_{ij}^{11})|=2.$ 
\end{lemma}
\begin{proof}
    Follow the notations in the proof of \cref{lm: common matching Cn},
    let $f^{(k)}=\partial_{ij}^{b^k}f,$ $g=f_{ij}^{00}$ and $h=f_{ij}^{11}.$
    Without loss of generality, we assume $\{i,j\}=\{5,6\}$.
    By \cref{lm: common matching Cn}, we may assume that for every $k\in[n]$, $f^{(k)}(x_1,x_2,x_3,x_4)=\lambda_k b^{^{k_1}}(x_1,x_2)b^{k_2}(x_3,x_4)$ for some $\lambda_k\neq 0$ and $k_1,k_2\in [n]$.
    Define $\widetilde{f^{(k)}}(x,y)=f^{(k)}(x,x,y,y),\,\widetilde{g}(x,y)=g(x,x,y,y),\,\widetilde{h}(x,y)=h(x,x,y,y)$ for every $x,y\in \{0,1\}.$
    By \eqref{eq: cyclic group merging}, we have $\widetilde{f^{(k)}}=\widetilde{g}+\zeta^k \widetilde{h}.$
    Notice that $$M(\widetilde{f^{(k)}})=\lambda_k\left[\begin{matrix}
        b^{(k_1)}(0,0)\\
        b^{(k_1)}(1,1)
    \end{matrix}\right]
    \left[\begin{matrix}
        b^{(k_2)}(0,0)&
        b^{(k_2)}(1,1)
    \end{matrix}\right].$$
    So $\mathrm{rank}(M(\widetilde{g}+\zeta^k\widetilde{h}))=\mathrm{rank}(M(\widetilde{f^{(k)}}))=1.$
    Hence $\det(M(\widetilde{g}+\zeta^k\widetilde{h}))=0.$
    Notice that $\det(M(\widetilde{g}+t\widetilde{h}))$ is a polynomial in $t$ of degree at most $2$, and it vanishes at $\zeta^k$, for all $k\in [n]$.
    Since $n\ge 3$, we have $\det(M(\widetilde{g}+t\widetilde{h}))\equiv 0.$
    \begin{claim}\label{claim: cyclic group rank 1}
        $\mathrm{rank}(M(\widetilde{g}))=\mathrm{rank}(M(\widetilde{h}))=1.$
    \end{claim}
    \begin{claimproof}{\cref{claim: cyclic group rank 1}}
        Let $t=0$ in $\det(M(\widetilde{g}+t\widetilde{h}))= 0$, we have $\det(M(\widetilde{g}))=0.$
        Hence $\mathrm{rank}(M(\widetilde{g}))\le 1.$
        Let $t\to +\infty$ in $\det(M(\widetilde{g}+t\widetilde{h}))= 0$, we have $\det(M(\widetilde{h}))=0$ and $\mathrm{rank}(M(\widetilde{h}))\le 1.$
        Next we show $\widetilde{g}\neq 0$ and $\widetilde{h}\ne 0.$
        By the assumption $f^{(k)}=\lambda_k b^{^{k_1}}(x_1,x_2)b^{k_2}(x_3,x_4)$, we have $\mathscr{S}(f^{(k)})=S$ for every $k\in [n]$, where $S=\{0000,0011,1100,1111\}.$
        By \eqref{eq: cyclic group merging}, we have 
        $$g=\frac{\zeta^{k_2}f^{(k_1)}-\zeta^{k_1}f^{(k_2)}}{\zeta^{k_2}-\zeta^{k_1}},\quad \text{for every distinct }k_1,k_2\in[n].$$
        So $\mathscr{S}(g)\subseteq S.$
        By assumption \eqref{eq: arity 6 assumption}, $f$ satisfies {\sc 2nd-Orth}.
        So $g=f_{ij}^{00}\neq 0.$
        It follows that $\widetilde{g}\neq 0$ by the definition of $\widetilde{g}$ and $\mathscr{S}(g)\subseteq S.$
        Thus, $\mathrm{rank}(M(\widetilde{g}))=1.$
        Similarly, we can show $\widetilde{h}\neq 0$ and $\mathrm{rank}(M(\widetilde{h}))=1.$
    \end{claimproof}
    By \cref{claim: cyclic group rank 1}, we may write $M(\widetilde{g})=uv^{\tt T}$ and $M(\widetilde{h})=wz^{\tt T}$, where $u,v,w,z\in \mathbb{C}^2.$
    If $u$ and $w$, $v$ and $z$ are both linearly independent, then $$M(\widetilde{g}+\widetilde{h})=\left[\begin{matrix}
        u & w
    \end{matrix}\right]
    \left[\begin{matrix}
        v^{\tt T} \\ z^{\tt T}
    \end{matrix}\right]$$
    has rank two, contradicting $\det(M(\widetilde{g}+t\widetilde{h}))\equiv 0.$
    Thus we may assume $u=w$ or $v=z$.
    After renaming variables, we may always assume $u=w.$
    Since $f$ satisfies {\sc 2nd-Orth}, we have $|g|=|h|$ and $\braket{g,h}=0.$
    So
    $$
    \braket{g,h}=(|u_0|^2+|u_1|^2)(v_0\overline{z_0}+v_1\overline{z_1})=0.
    $$
    Clearly $u\neq 0$, otherwise $\widetilde{g}=0.$
    So $v_0\overline{z_0}+v_1\overline{z_1}=0.$
    Also, $|g|=|h|$ gives $|v_0|^2+|v_1|^2=|z_0|^2+|z_1|^2.$ 
    \begin{claim}\label{claim: cyclic group equal norm}
        $|v_0|=|z_1|, |v_1|=|z_0|.$
    \end{claim}
    \begin{claimproof}{\cref{claim: cyclic group equal norm}}
        Let $r=(-\overline{v_1},\overline{v_0})^{\tt T}$, then $\braket{v,r}=0.$
        Since $\braket{v,z}=0$ and $v\neq 0$, there exists some $\lambda\in \mathbb{C}$ such that $z=\lambda r.$
        Substitute into $|v_0|^2+|v_1|^2=|z_0|^2+|z_1|^2$, we obtain $|v_0|^2+|v_1|^2=|\lambda|^2(|r_0|^2+|r_1|^2)=|\lambda|^2(|v_1|^2+|v_0|^2)$.
        So $|\lambda|=1.$
        Substituting back, we have $|v_0|=|z_1|$ and $|v_1|=|z_0|$.
    \end{claimproof}
    Now $\widetilde{f^{(k)}}=\widetilde{g}+\zeta^k\widetilde{h}=u(v+\zeta^k z)^{\tt T}$.
    So $f^{(k)}(x_1,x_2,x_3,x_4)=u(x_1,x_2)u'(x_3,x_4)$, where $u(0,0)=u_0$, $u(1,1)=u_1$, $u(0,1)=u(1,0)=0$; $u'(0,0)=v_0+\zeta^k z_0$, $u'(1,1)=v_1+\zeta^k z_1$ and $u'(0,1)=u'(1,0)=0.$
    By assumption \eqref{eq: arity 6 assumption}, both $u$ and $u'$ are proportional to some binary signatures in $B(f)$.
    Hence $|u_0|=|u_1|\neq 0$, and $|v_0+\zeta^k z_0|=|v_1+\zeta^k z_1|,\,\forall k\in [n].$
    Square both sides, we have
    \begin{equation*}
        |v_0|^2+2\Re(\zeta^k\overline{v_0}z_0)+|z_0|^2=|v_1|^2+2\Re(\zeta^k\overline{v_1}z_1)+|z_1|^2,\quad \forall k\in[n].
    \end{equation*}
    Using \cref{claim: cyclic group equal norm}, we have 
    \begin{equation*}
        \Re(\zeta^k(\overline{v_0}z_0-\overline{v_1}z_1))=0,\quad \forall k\in [n].
    \end{equation*}
    So $\overline{v_0}z_0-\overline{v_1}z_1=0.$
    Combining $v_0\overline{z_0}+v_1\overline{z_1}=0$, we have $\overline{v_0}z_0=\overline{v_1}z_1=0.$
    Since both $v$ and $z$ are non-zero, after possibly exchanging the coordinates, $v=(v_0,0)$ and $z=(0,z_1)$ with $v_0z_1\neq 0.$
    Consequently, $\widetilde{g}=uv^{\tt T}$ has two non-zero entries.
    By the definition of $\widetilde{g}$, we have $\mathscr{S}(g)=\mathscr{S}(\widetilde{g})=2$.
    Similarly, $\mathscr{S}(h)=2.$
    The lemma is proved.
\end{proof}

\begin{lemma}\label{lm: cyclic group even support}
    For any $x\in\mathscr{S}(f)$, $\mathrm{wt}(x)$ is even.
\end{lemma}
\begin{proof}
    Suppose $a=(a_1,a_2,a_3,a_4,a_5,a_6)\in \mathscr{S}(f)$ but $\mathrm{wt}(a)$ is odd.
    After renaming the variables, we may assume $a_1=a_2,a_3=a_4,a_5\neq a_6.$
    Consider merging $x_1$ and $x_2$ using $b^s\in B(f)$, and merging $x_3$ and $x_4$ using $b^t\in B(f)$.
    The resulting signature $\partial_{12}^{b^s}\partial^{b^t}_{34}f$ is also in $B(f)$ by assumption \eqref{eq: arity 6 assumption}.
    So $\partial_{12}^{b^s}\partial^{b^t}_{34}f(a_5,a_6)=0.$
    Thus,
    \begin{equation}\label{eq: cyclic group even support}
        \sum_{i,j\in\{0,1\}}(\zeta^s)^i(\zeta^t)^jf(i,i,j,j,a_5,a_6)=\partial_{12}^{b^s}\partial^{b^t}_{34}f(a_5,a_6)=0.
    \end{equation}
    Pick $s_1,s_2,t_1,t_2\in [n]$ such that $s_1\neq s_2$ and $t_1\neq t_2.$
    Then \eqref{eq: cyclic group even support} gives that 
    \begin{equation*}
        \left[\begin{matrix}
            1 & \zeta^{t_1} & \zeta^{s_1} & \zeta^{s_1t_1}\\
            1 & \zeta^{t_2} & \zeta^{s_1} & \zeta^{s_1t_2}\\
            1 & \zeta^{t_1} & \zeta^{s_2} & \zeta^{s_2t_1}\\
            1 & \zeta^{t_2} & \zeta^{s_2} & \zeta^{s_2t_2}\\
        \end{matrix}\right]
        \left[\begin{matrix}
            f(0,0,0,0,a_5,a_6)\\
            f(0,0,1,1,a_5,a_6)\\
            f(1,1,0,0,a_5,a_6)\\
            f(1,1,1,1,a_5,a_6)\\
        \end{matrix}\right]
        =\left[\begin{matrix}
            0\\
            0\\
            0\\
            0\\
        \end{matrix}\right].
    \end{equation*}
    The coefficient matrix is 
    $\left[\begin{smallmatrix}
        1 & \zeta^{s_1}\\
        1 & \zeta^{s_2}
    \end{smallmatrix}\right]\otimes\left[\begin{smallmatrix}
        1 & \zeta^{t_1}\\
        1 & \zeta^{t_2}
    \end{smallmatrix}\right]$, which has rank 4.
    Thus $f(a)=0$, contradiction.
\end{proof}

\begin{lemma}\label{lm: cyclic group impossible}
    Under assumption \eqref{eq: arity 6 assumption}, the standard case of $B(f)\cong C_n (n\ge 3)$ cannot occur.
\end{lemma}
\begin{proof}
    We count the number of pairs \((x,\{i,j\})\) such that $f(x)\ne0\text{ and }x_i=x_j.$
    There are fifteen choices of \(\{i,j\}\).
    For each pair, exactly four supported words satisfy $x_i=x_j$ by \cref{lm: cyclic group support size 2}.
    Thus the total number of incidences is
    $15\times4=60.$

    On the other hand,
    let \(e=\#\{x\in \mathscr{S}(f)\mid \mathrm{wt}(x)=0  \text{ or } 6\}\),
    \(m_2=\#\{x\in \mathscr{S}(f)\mid \mathrm{wt}(x)=2\}\) and \(m_4=\#\{x\in \mathscr{S}(f)\mid \mathrm{wt}(x)=2\}\). 
    Clearly
    $0\le e\le2.$
    A word of weight \(0\) or \(6\) has $\binom{6}{2}=15$ equal pairs $\{x_i,x_j\}$.
    A word of weight \(2\) has $\binom22+\binom42=1+6=7$ equal pairs. 
    A word of weight \(4\) has the same number.
    By \cref{lm: cyclic group even support}, the total number of pairs $(x,\{i,j\})$ such that $f(x)\neq 0$ and $x_i=x_j$ is
    $15e+7(m_2+m_4).$
    So $60=15e+7(m_2+m_4)$,
    $4\equiv e\pmod7.$
    Contradicting $0\le e\le 2.$
\end{proof}

\subsubsection{Non-standard Case}
The non-standard case can be handled similarly to the standard case by the following lemma.
We only sketch the proof and focus on differences from the standard case.
\begin{lemma} 
    Under assumption \eqref{eq: arity 6 assumption}, the non-standard case of $B(f)\cong C_n (n\ge 3)$ cannot occur.
\end{lemma}
\begin{proof}
    Fix a pair of indices $\{x_i,x_j\}.$
    Let $f^{(k)}=\partial_{ij}^{b^k}f,\, g=f_{ij}^{01},\, h=f_{ij}^{10}.$
    Without loss of generality, we assume $\{i,j\}=\{5,6\}.$
    We have 
    $
        f^{(k)}=g+\zeta^k h.
    $
    For every $k\in[n]$, $\mathscr{S}(f^{(k)})$ is one of the following:
    \begin{equation*}
        \begin{aligned}
        S_1&=\{0101,0110,1001,1010\}, \text{ or}\\
        S_2&=\{0011,0110,1001,1100\},\text{ or}\\
        S_3&=\{0011,0101,1010,1100\}.
        \end{aligned}
    \end{equation*}
    Again, if there exists two distinct $k,k'\in[n]$ such that $\mathscr{S}(f^{(k)})=\mathscr{S}(f^{(k')})=S$, then $\mathscr{S}(f^{(t)})=S$ for any $t\in[n]$.
    For $n\ge 4$, this is true by the pigeon hole principle.
    For $n=3$,
    unlike in \cref{lm: common matching Cn}, we require an exceptional argument here.
    Suppose $\{\mathscr{S}(f^{(0)}),\mathscr{S}(f^{(1)}),\mathscr{S}(f^{(2)})\}=\{S_1,S_2,S_3\}.$
    Then 
    \[\{f^{(0)},f^{(1)},f^{(2)}\}=\{\lambda_0b^{r_1}(x_1,x_2)b^{r_2}(x_3,x_4),\lambda_1b^{s_1}(x_1,x_3)b^{s_2}(x_2,x_4),\lambda_2b^{t_1}(x_1,x_4)b^{t_2}(x_2,x_3)\}\]
    for some $\lambda_0,\lambda_1,\lambda_2\neq 0$ and $r_1,r_2,s_1,s_2,t_1,t_2\in\{0,1,2\}.$
    By
    $
        f^{(k)}=g+\omega^k h
    $
    where $\omega=e^{\frac{2\pi \mathfrak{i}}{3}}$, we have
    $f^{(0)}+\omega f^{(1)}+\omega^2 f^{(2)}=0.$
    Absorbing the coefficients $\lambda_k$, we have
    $$
    \alpha b^{r_1}(x_1,x_2)b^{r_2}(x_3,x_4)+\beta b^{s_1}(x_1,x_3)b^{s_2}(x_2,x_4)+\gamma b^{t_1}(x_1,x_4)b^{t_2}(x_2,x_3)=0
    $$
    for some $\alpha,\beta,\gamma\neq0.$
    Evaluate at $(x_1,x_2,x_3,x_4)=0110,0011$ and $0101$ respectively, we have
    \begin{equation*}
        \begin{aligned}
            &\alpha \omega^{r_2}+\beta \omega^{s_2}=0,\\
            & \beta+\gamma=0,\\
            &\alpha+\gamma \omega^{t_2}=0.
        \end{aligned}
    \end{equation*}
    The second equation gives $\beta=-\gamma$.
    Plug into the first, $\alpha=\gamma\omega^{s_2-r_2}.$
    Plug into the third, $\omega^{t_2}=-\omega^{s_2-r_2}.$
    Cube both sides, we have $-1=1$, contradiction.
    Thus, the same statement (after replacing $B(f)$ by $\widehat{B}(\widehat{f})$) in \cref{lm: common matching Cn} can be proved in the non-standard case.
    Similar to \cref{lm: cyclic group support size 2}, we can prove $|\mathscr{S}(f_{ij}^{01})|=|\mathscr{S}(f_{ij}^{10})|=2$ for any pair of indices $\{i,j\}.$

    Instead of proving $f$ is supported on even Hamming weights like \cref{lm: cyclic group even support}, we will show that for any $x\in \mathscr{S}(f)$, $\mathrm{wt}(x)=0,6$ or $3$.
    Pick a non-constant $x\in \mathscr{S}(f)$.
    Choose two positions \(i,j\) on which \(x_i\ne x_j\), and let \(y\) denote the remaining four bits.
    Since $f(x)\neq 0$, the degree one polynomial
    \(
    P_y(t)=f_{ij}^{01}(y)+tf_{ij}^{10}(y)
    \)
    is not identically zero.
    Therefore, there exists some $k\in[n]$ such that $f^{(k)}(y)=P_y(\zeta^k)\neq 0.$
    Since $\mathscr{S}(f^{(k)})=S_1,S_2$ or $S_3$, we have $\mathrm{wt}(y)=2.$
    Hence $\mathrm{wt}(x)=3.$
    We have proved $\mathscr{S}(f)\subseteq \{000000,111111\}\cup \{x\mid x\in \{0,1\}^6,\mathrm{wt}(x)=3\}$.

    Now we count the number of pairs $(x, \{i,j\})$ such that $f(x)\neq 0$ and $x_i\ne x_j$.
    Again there are fifteen choices of \(\{i,j\}\), and
    for each pair, exactly four supported words satisfy $x_i=x_j$.
    Thus the total number of incidences is
    $15\times4=60.$
    On the other hand, the two constant words contribute zero unequal pairs, and
    a weight-three word has $3\times 3=9$ unequal pairs.
    So $60=9m$, where $m\in \mathbb{Z}$ is the number of $x\in \mathscr{S}(f)$ such that $\mathrm{wt}(x)=3.$
    This is a contradiction.
\end{proof}

\subsection{Synthesis of All Cases}

We are ready to prove the main theorem of this section.
\ThirdOrderOrth*

\begin{proof}
    By \cref{thm: binary set is finite group}, $B(f)$ is isomorphic to a finite subgroup of $\mathbf{SO}(3)$.
    It can either be one of the polyhedral groups, or a dihedral group, or a cyclic group.
    If $B(f)$ is a polyhedral group, or a dihedral group of order at least six, or the Klein group in the standard case, $f$ satisfies {\sc 3rd-Orth} by \cref{lm: arity 6 polyhedral group 3rd orth,lm: arity 6 dihedral group 3rd orth,lm: K4 standard case}, respectively.
    The non-standard case and the cyclic group case cannot happen by \cref{lm: K4 non-standard,lm: trivial group,lm: C2 symmetric impossible,lm: C2 asymmetric impossible,lm: cyclic group impossible}.
    The above cases are exhaustive.
\end{proof}

\section{Recovering \texorpdfstring{$f_6$}{f6} from Six-Qubit $\operatorname{AME}(6,2)$ Signatures}
\label{sec:recover-f6}

By \cref{thm: third order orthogonality}, we have realized a 6-ary signature $f$ and its normalized state as an $\mathrm{AME}(6,2)$ state.
In this section, we will prove either $\Holant(\F)$ is $\#\operatorname{P}$-hard, or we can realize an special 6-ary signature $f_6$ and four Bell signatures $\mathcal{B}=\{I_2,X,Z,J\}$ after a real orthogonal holographic transformation by $O$, i.e., $\Holant(O\F,f_6,\mathcal{B})\leq_T \Holant(\F)$. 
With this $f_6$ and the four Bell signatures $\mathcal{B}$ in hand, in the next~\cref{sec: f6 is hard}, we will prove that $\Holant(O\mathcal{F},f_6,\mathcal{B})$ is $\#\operatorname{P}$-hard. 

We retain the standing assumptions that $\mathcal F$ is closed under entrywise conjugation, and is
outside condition \eqref{cond: tractable classes}.

The sepcial signature $f_6$ is first found in the real Holant problems~\cite{realholant}, and it is a obstacle for the proof of hardness of real Holant~\cite{realholant}. In this paper, we illustrate that why such a special signature can be realized in real Holant and our conjugate-closed Holant, it is corresponding to the $\mathrm{AME}(6,2)$ state in the quantum perspective. The condtion outside~\eqref{cond: tractable classes} and the induction hypothesis make the case $2n=6$ that either yields $\#\operatorname{P}$-hardness or an $\mathrm{AME}(6,2)$ state. And combined with the proof of uniqueness of the $\mathrm{AME}(6,2)$ state in~\cref{sec: AME62 has a common local unitary normal}, we can finally realize a concrete $\mathrm{AME}(6,2)$ state $f_6$ in our conjugate-closed Holant. Since the this $\mathrm{AME}(6,2)$ state $f_6$ has a simple expression, we can easily generalize the proof of $f_6$ giving $\#\operatorname{P}$-hardness in~\cite{realholant} to our conjugate-closed Holant problems.

We now give the definition of 6-ary signature $f_6$:
\[
 f_6(x)=\mathbf1_{x_1+\cdots+x_6=0}
 (-1)^{x_1+x_2+x_3+q(x)},\qquad
 q(x)=x_1x_2+x_1x_3+x_2x_3+x_1x_4+x_2x_5+x_3x_6,
\]
where Boolean arithmetic is over $\mathbb{F}_2$. With rows and columns indexed by $x_1x_2x_3$ and $x_4x_5x_6$, we have the signature matrix of $f_6$ as following:
\begingroup
\setlength{\arraycolsep}{4pt}
\renewcommand{\arraystretch}{1.1}
\[
M_{x_1x_2x_3,x_4x_5x_6}(f_6)=
\begin{pmatrix}
 1 &  0 &  0 &  1 &  0 &  1 &  1 &  0 \\
 0 &  1 & -1 &  0 & -1 &  0 &  0 &  1 \\
 0 & -1 &  1 &  0 & -1 &  0 &  0 &  1 \\
-1 &  0 &  0 & -1 &  0 &  1 &  1 &  0 \\
 0 & -1 & -1 &  0 &  1 &  0 &  0 &  1 \\
-1 &  0 &  0 &  1 &  0 & -1 &  1 &  0 \\
-1 &  0 &  0 &  1 &  0 &  1 & -1 &  0 \\
 0 & -1 & -1 &  0 & -1 &  0 &  0 & -1
\end{pmatrix}.
\]
\endgroup

Its normalized state is an $\mathrm{AME}(6,2)$ state. Now we have a 6-ary signature $f$ and its normalized state is $\mathrm{AME}(6,2)$, since $\mathrm{AME}(6,2)$ has one local-unitary normal form (see~\cref{sec: AME62 has a common local unitary normal} for the detailed proofs), after relabeling its inputs, we can get that 
\begin{equation}\label{eq:f6-local-unitary-form}
    f=\lambda(U_1\otimes\cdots \otimes U_6)f_6, \quad \lambda\neq 0, \quad U_iU^\dagger=I.
\end{equation}
One can refer to~\cref{fig:f6-local-unitary-gadget}. for a more intuitive understanding.
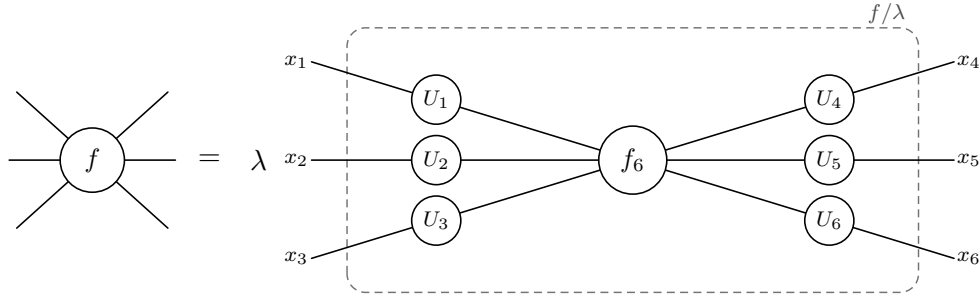
\begin{figure}[H]
\centering
\begin{tikzpicture}[
    wire/.style={draw=black, line width=0.6pt, line cap=round},
    boundary/.style={draw=black!55, densely dashed, line width=0.55pt,
        rounded corners=7pt},
    signature/.style={draw=black, circle, fill=white, minimum size=7.2mm,
        inner sep=0pt, line width=0.6pt, font=\small},
    local unitary/.style={signature, minimum size=6.5mm, font=\scriptsize},
    port label/.style={font=\scriptsize, inner sep=1pt}
]
    \path[use as bounding box] (0,0) rectangle (13.2,4.05);

    \coordinate (f) at (1.35,2.10);
    \draw[wire] (0.35,3.00) -- (f);
    \draw[wire] (0.25,2.10) -- (f);
    \draw[wire] (0.35,1.20) -- (f);
    \draw[wire] (f) -- (2.35,3.00);
    \draw[wire] (f) -- (2.45,2.10);
    \draw[wire] (f) -- (2.35,1.20);
    \node[signature, minimum size=8.5mm] at (f) {$f$};

    \node[font=\normalsize] at (2.90,2.10) {$=$};
    \node[font=\normalsize] at (3.55,2.10) {$\lambda$};

    \draw[boundary] (4.72,0.35) rectangle (12.28,3.85);
    \node[port label, anchor=south east, text=black!70]
        at (12.15,3.85) {$f/\lambda$};

    \coordinate (f6) at (8.50,2.10);
    \draw[wire] (4.25,3.40) -- (f6);
    \draw[wire] (4.25,2.10) -- (f6);
    \draw[wire] (4.25,0.80) -- (f6);
    \draw[wire] (f6) -- (12.75,3.40);
    \draw[wire] (f6) -- (12.75,2.10);
    \draw[wire] (f6) -- (12.75,0.80);

    \node[local unitary] at (5.90,2.90) {$U_1$};
    \node[local unitary] at (5.90,2.10) {$U_2$};
    \node[local unitary] at (5.90,1.30) {$U_3$};
    \node[local unitary] at (11.10,2.90) {$U_4$};
    \node[local unitary] at (11.10,2.10) {$U_5$};
    \node[local unitary] at (11.10,1.30) {$U_6$};
    \node[signature, minimum size=9mm] at (f6) {$f_6$};

    \node[port label, left]  at (4.25,3.40) {$x_1$};
    \node[port label, left]  at (4.25,2.10) {$x_2$};
    \node[port label, left]  at (4.25,0.80) {$x_3$};
    \node[port label, right] at (12.75,3.40) {$x_4$};
    \node[port label, right] at (12.75,2.10) {$x_5$};
    \node[port label, right] at (12.75,0.80) {$x_6$};
\end{tikzpicture}
\caption{Gadget form of~\eqref{eq:f6-local-unitary-form}.
The dashed enclosure realizes
$(U_1\otimes\cdots\otimes U_6)f_6=f/\lambda$, where $U_i$ is an binary extension on variable $x_i$ of $f_6$.
This need not be an $\mathcal F$-gate, since the local unitaries $U_i$ maybe can not be realized in $\mathcal F$.}
\label{fig:f6-local-unitary-gadget}
\end{figure}

However, the matrices $U_1,\dots,U_6$ maybe not available in $\Holant(\F)$, so instead of proving that we can directly realize $f_6$ in $\Holant(\F)$,
we will prove that we can convert $U_1,\dots,U_6$ to $V_1,\dots,V_6$ where $[V_i]\in \Gamma$ for $1\leq i\leq 6$ after a orthogonal holographic transformation. Let $h=O^{\otimes6}f=\lambda(V_1\otimes\cdots\otimes V_6)f_6$.
For each pair of variables, we contract $h$ using $=_2$ and extract
the two resulting binary signatures. Let $G$ be the group generated by
the projective classes of these factors, their conjugates, and their
transposes. We then treat the cases $G=\Gamma$ and $G\subseteq K_4$
to remove the remaining local matrices and obtain $f_6$ together
with the four Bell signatures.

The main theorem of the section is the following theorem.

\begin{theorem}[Realizing $f_6$]
\label{thm:f6-equality-realization}
Suppose that $\mathcal F$ doesn't satisfy condition~\eqref{cond: tractable classes} and contains a nonzero six-ary signature whose
normalized state is $\operatorname{AME}(6,2)$. Then either
$\operatorname{Holant}(\mathcal F)$ is $\#\mathrm P$-hard, or there
is a real orthogonal matrix $O$ such that
\(
 \operatorname{Holant}(O\mathcal F\cup\{f_6\}\cup\mathcal B)
 \leq_T \operatorname{Holant}(\mathcal F).
\)
\end{theorem}

\subsection{Binary Contractions and Hardness}
\label{subsec:f6-binary-contractions}

In this subsection, we first analyze binary contractions of $f_6$,
then use the local-unitary expression of $f$ to constrain the matrices
$U_i$. By the preceding reductions and \cref{lm: base case 2n=4},
realizing a nonzero four-ary signature outside $\mathcal U^\otimes$
gives hardness. Thus, unless hardness already follows, every nonzero
contraction of $f$ through a realizable binary signature belongs to
$\mathcal U^\otimes$.

Recall that the matrices $I=\left(\begin{matrix}
    1 & 0\\
    0 & 1\\
\end{matrix}\right),X=\left(\begin{matrix}0&1\\1&0\end{matrix}\right)$,
$Z=\left(\begin{matrix}1&0\\0&-1\end{matrix}\right)$, and
$J=\left(\begin{matrix}0&1\\-1&0\end{matrix}\right)$. For a nonzero matrix $A$, $[A]$ denotes its class up to a nonzero scalar. Also recall the following notation of Kelin group and tetrahedral,
\[
 C_{\epsilon_1,\epsilon_2,\epsilon_3}
 =\tfrac12(I+\epsilon_1\mathfrak iX+\epsilon_2J+\epsilon_3\mathfrak iZ),
 \qquad
 \Gamma=K_4\cup
 \{[C_{\epsilon_1,\epsilon_2,\epsilon_3}]:
          \epsilon_1,\epsilon_2,\epsilon_3\in\{\pm1\}\},
\]
where $K_4=\{[I],[X],[Z],[J]\}$. $K_4$ is the projective Kelin group and $\Gamma$ is the projective tetrahedral group. Both are closed under multiplication, inverse, transpose, and conjugation.
The matrices in $\mathcal B=\{I,X,Z,J\}$ are the four Bell signatures.

Also recall the gadget construction of $\Holant(\F)$, for a 6-ary signature $h$ and a binary matrix $Q$, write $\partial_{ij}^Qh=\sum_{x_i,x_j}Q(x_i,x_j)h(x)$, which is called a contraction of $h$ through $Q$ on variables $x_i,x_j$. In particular,
$\partial_{ij}^Ig$ joins ports $i,j$ directly by $=_2$.
We first determine when such a contraction of $f_6$ splits into two
binary factors.

\begin{lemma}[Binary contractions of $f_6$]
\label{lem:f6-contraction-test}
For every nonzero binary matrix $Q$ and every pair $1\le i<j\le6$,
the signature $\partial_{ij}^Qf_6$ is nonzero. $\partial_{ij}^Qf_6$ is a tensor product
of two binary signatures if and only if $[Q]\in\Gamma$.
In that case, both factors are projectively unitary and their
projective classes belong to $\Gamma$.
\end{lemma}
\begin{proof}[Proof of Lemma~\ref{lem:f6-contraction-test}]
Without loss of generality, we consider the first two variables $x_1,x_2$.
First take $i=1,j=2$ and write
$Q=\left(\begin{smallmatrix}a&b\\c&d\end{smallmatrix}\right)$.
By the definition of $f_6$, with row and column indices ordered
as $00,01,10,11$, we have
\[
M_{x_3x_4,x_5x_6}(\partial_{12}^Qf_6)=
\begin{pmatrix}
a-d&-b-c&b-c&a+d\\
c-b&a+d&a-d&b+c\\
-b-c&a-d&-a-d&c-b\\
-a-d&b-c&b+c&a-d
\end{pmatrix}.
\]
This matrix is nonzero when $Q\ne0$: its entries determine
$a-d,a+d,b-c,b+c$ and hence $a,b,c,d$.
The equality $\partial_{12}^Qf_6=A(x_3,x_4)B(x_5,x_6)$ holds
for nonzero binary signatures $A,B$ if and only if this matrix
has rank one. The other two possible factorizations are described
by $M_{x_3x_5,x_4x_6}(\partial_{12}^Qf_6)$ and
$M_{x_3x_6,x_4x_5}(\partial_{12}^Qf_6)$ in the same way.
A nonzero matrix has rank one if and only if all its
$2$-by-$2$ minors vanish. Applying this criterion gives the following
necessary conditions for each matrix:
\[
\begin{array}{c|l}
M_{x_3x_4,x_5x_6}&a^2+d^2=ab-cd=b^2+d^2=ac+bd=ad+bc=c^2-d^2=0,\\
M_{x_3x_5,x_4x_6}&a^2+d^2=ab+cd=b^2-d^2=ac-bd=ad+bc=c^2+d^2=0,\\
M_{x_3x_6,x_4x_5}&a^2-d^2=ab=b^2-c^2=ac=bd=cd=0.
\end{array}
\]
Here all three matrices are evaluated at $\partial_{12}^Qf_6$.
For example, in the displayed matrix the minors on rows $1,4$
and columns $1,4$, and on rows $1,2$ and columns $2,3$, are
$2(a^2+d^2)$ and $-2(ab-cd)$, respectively.
The last system gives either $b=c=0$, $a=\pm d\ne0$, or
$a=d=0$, $b=\pm c\ne0$, hence exactly the classes in $K_4$.
In either of the first two systems, $d=0$ forces $Q=0$.
Scaling to $d=1$, their solutions are respectively
$(a,b,c)=(\epsilon\mathfrak i,\eta\mathfrak i,-\epsilon\eta)$ and
$(\epsilon\mathfrak i,\eta,-\epsilon\eta\mathfrak i)$, where
$\epsilon,\eta\in\{\pm1\}$.
Their classes are $[C_{\eta,\epsilon\eta,\epsilon}]$ and
$[C_{-\epsilon\eta,\eta,\epsilon}]$, giving the remaining eight
classes of $\Gamma$.

Substituting these solutions gives the following factorizations,
up to a nonzero scalar. In a row with variable pairs $(x_a,x_b)$
and $(x_c,x_d)$, the last column lists the classes of
$A,B$ in $\partial_{12}^Qf_6=A(x_a,x_b)B(x_c,x_d)$.
\[
\begin{array}{c|c|c}
[Q]&\text{variable pairs}&\text{factor classes}\\ \hline
{[I]}&(x_3,x_6),(x_4,x_5)&[J],\ [X]\\
{[X]}&(x_3,x_6),(x_4,x_5)&[X],\ [Z]\\
{[Z]}&(x_3,x_6),(x_4,x_5)&[I],\ [I]\\
{[J]}&(x_3,x_6),(x_4,x_5)&[Z],\ [J]\\
{[C_{\eta,\epsilon\eta,\epsilon}]}&(x_3,x_4),(x_5,x_6)&
[ZC_{\eta,-\epsilon\eta,\epsilon}],\ [C_{-\eta,-\epsilon\eta,\epsilon}]\\
{[C_{-\epsilon\eta,\eta,\epsilon}]}&(x_3,x_5),(x_4,x_6)&
[ZC_{-\epsilon\eta,\eta,\epsilon}],\ [C_{\epsilon\eta,\eta,\epsilon}]
\end{array}
\]
Every listed class lies in $\Gamma$ and has a unitary representative.
This proves the converse as well. For each solution only one of the
three systems holds, and the two factors for that partition are
unique up to reciprocal scalars. Thus the conclusion holds for
any choice of factors.

\end{proof}

Contracting a binary signature $B$ with inputs $i,j$ of $f$ corresponds
to contracting $U_i^TBU_j$ with the same inputs of $f_6$.
The following theorem gives hardness when this relative matrix lies
outside $\Gamma$. We first use its equality case, $B=I$, to obtain
the conditions needed for a common change of basis.

\begin{theorem}[The relation between $U_i$ and $U_j$]
\label{thm:f6-nongamma-hard}
Let $\F$ be a conjugate closed signature set, $\F$ outside condition~\eqref{cond: tractable classes} and contains an $\operatorname{AME}(6,2)$ state $f=\lambda(U_1\otimes\cdots \otimes U_6)f_6$ where $\lambda\neq 0$ and $U_iU^\dagger=I$.
Suppose that a nonzero binary signature $B$ is available in $\Holant(\F)$.
If $[U_i^TBU_j]\notin\Gamma$ for some $i<j$, then
$\Holant(\mathcal F)$ is $\#\mathrm P$-hard.
In particular, hardness follows if $[U_i^TU_j]\notin\Gamma$ for some
$i<j$.
\end{theorem}

\begin{proof}[Proof of Theorem~\ref{thm:f6-nongamma-hard}]
We use the standing assumptions that $\mathcal F$ is fixed, finite,
algebraic, conjugate closed, and outside \textup{(T)}.
If $B$ is realized by a gadget, conjugating its vertex labels
realizes $\bar B$. More generally, the same proof applies when
$B,\bar B$ can be jointly adjoined by a polynomial-time Turing
reduction: work over that enlarged conjugate-closed family and
compose the reductions at the end.
Contract $B$ with inputs $i,j$ of $f$. The resulting signature is
\[
 g=\partial_{ij}^Bf
 =\lambda\left(\bigotimes_{k\ne i,j}U_k\right)
       \partial_{ij}^{\,U_i^TBU_j}f_6.
\]
The matrix on the right is $U_i^TBU_j$ because its $(a,b)$ entry is
$\sum_{s,t}U_i(s,a)B(s,t)U_j(t,b)$.
If its class lies outside $\Gamma$, \cref{lem:f6-contraction-test}
shows that $g$ is nonzero and is not a product of two binary signatures;
the remaining local unitaries preserve these properties.
In particular, $g\notin\mathcal U^\otimes$, since every nonzero
four-ary signature in $\mathcal U^\otimes$ is a product of two
binary signatures. By \cref{lm: base case 2n=4}, this gives hardness.
If $B,\bar B$ were adjoined, the enlarged family still contains
$\mathcal F$ and remains outside \textup{(T)}; composing the
reductions gives hardness of $\Holant(\mathcal F)$.
Taking $B=I$ proves the last assertion, since this contraction
simply joins the two inputs by an equality edge.
\end{proof}

Since $=_2$ is available, it remains to consider the case
$[U_i^TU_j]\in\Gamma$ for every $i<j$.

\subsection{A Common Real Orthogonal Transformation}
\label{subsec:f6-real-orthogonal}

Unless $\Holant(\F)$ is $\#\operatorname{P}$-hard, \cref{thm:f6-nongamma-hard} gives
$[U_i^TU_j]\in\Gamma$ for every $i<j$. The next lemma gives one real
orthogonal holographic transformation that converts all six local matrices $U_i$ into
$V_i$ which belongs to $\Gamma$ up to a scalar.

\begin{lemma}[A common orthogonal transformation]
\label{lem:f6-common-orthogonal}
Let $\F$ be a conjugate closed signature set, $\F$ outside condition~\eqref{cond: tractable classes} and contains an $\operatorname{AME}(6,2)$ state $f=\lambda(U_1\otimes\cdots \otimes U_6)f_6$ where $\lambda\neq 0$ and $U_iU^\dagger=I$.
If
$[U_i^TU_j]\in\Gamma$ for every $i<j$,
then there are a real orthogonal matrix $O$ and unitary matrices
$V_1,\ldots,V_6$ such that $OU_i=V_i$ and $[V_i]\in\Gamma$ for
every $i$. 
\end{lemma}

\begin{proof}[Proof of Lemma~\ref{lem:f6-common-orthogonal}]
Put $A_i=U_1^TU_i$ for $1\le i\le6$. These matrices are unitary,
and $[A_i]\in\Gamma$ for $i>1$. Since
$U_2^TU_3=A_2^TA_1^{-1}A_3$, closure of $\Gamma$ gives
$[A_1^{-1}]=[A_2^{-T}(U_2^TU_3)A_3^{-1}]\in\Gamma$ as well.

Let $R=\overline{U_1}$. Then $U_i=RA_i$ and
$S=R^TR=A_1^{-1}$ is symmetric and unitary.
The only symmetric classes in $\Gamma$ are $[I],[X],[Z]$:
$J$ is antisymmetric, and the off-diagonal entries of
$C_{\epsilon_1,\epsilon_2,\epsilon_3}$ differ by $\epsilon_2$.
Multiplying $R$ by a phase and all $A_i$ by its inverse, we may
assume $S\in\{I,X,Z\}$.

The following choices are unitary, have classes in $\Gamma$, and
satisfy $A_0^TSA_0=I$:
\[
\begin{array}{c|ccc}
S&I&X&Z\\ \hline
A_0&I&e^{\mathrm i\pi/4}C_{+,+,+}^{-1}
       &e^{-\mathrm i\pi/4}C_{+,+,-}^{-1}.
\end{array}
\]
Indeed, $C_{+,+,+}^TC_{+,+,+}=\mathrm iX$ and
$C_{+,+,-}^TC_{+,+,-}=-\mathrm iZ$.
Hence $RA_0$ is both unitary and orthogonal, so it is real.
Set $O=(RA_0)^T$ and $V_i=A_0^{-1}A_i$.
Then $O$ is real orthogonal, $OU_i=V_i$, and each $V_i$ is unitary
with $[V_i]\in\Gamma$.
\end{proof}

We now can realize
\begin{equation}
 h:=O^{\otimes6}f=\lambda(V_1\otimes\cdots\otimes V_6)f_6,
 \label{eq:f6-transformed-signature}
\end{equation}
on the $\Holant(O\mathcal{F})$, and 
$\Holant(\mathcal F)\equiv_T\Holant(\mathcal F)$.

In this basis, $[V_i^TBV_j]\in\Gamma$ if and only if $[B]\in\Gamma$,
since $[V_i],[V_j]\in\Gamma$ and $\Gamma$ is closed under transpose,
multiplication, and inversion. We can therefore state the hardness
condition directly for a binary signature realized over $\Holant(O\mathcal F)$.

\begin{lemma}[Non-$\Gamma$ hard]
\label{lem:f6-nongamma-binary}
Let $\F$ be a conjugate closed signature set, $\F$ outside condition~\eqref{cond: tractable classes} and contains an $\operatorname{AME}(6,2)$ state $f=\lambda(U_1\otimes\cdots \otimes U_6)f_6$ where $\lambda\neq 0$ and $U_iU^\dagger=I$.
Let $O$ be a real orthogonal matrix and 
$V_1,\ldots,V_6$ are unitary matrices such that $OU_i=V_i$ and $[V_i]\in\Gamma$ for
every $i$. If $\Holant(O\F)$ can realize a nonzero binary signature $B$
with $[B]\notin\Gamma$, then $\Holant(O\mathcal F)$, and hence
$\Holant(\mathcal F)$, is $\#\mathrm P$-hard.
\end{lemma}

\begin{proof}[Proof of Lemma~\ref{lem:f6-nongamma-binary}]
Since $[V_1],[V_2]\in\Gamma$ and $\Gamma$ is closed under transpose,
multiplication, and inversion, $[B]\notin\Gamma$ implies
$[V_1^TBV_2]\notin\Gamma$.
Apply \cref{thm:f6-nongamma-hard},
this gives that $\#\operatorname{P}$-hardness of $\Holant(O\mathcal{F})$.
\end{proof}

\subsection{Recovering \texorpdfstring{$f_6$}{f6} from Binary Factors}
\label{subsec:f6-binary-factors}
After the holographic transformation of~\cref{subsec:f6-real-orthogonal}, we can get a signature $h:=O^{\otimes6}f=\lambda(V_1\otimes\cdots\otimes V_6)f_6$. In this subsection, we will prove that we can realize $f_6$ and the four Bell signatures $\mathcal B$ in $\Holant(O\F,h)$. 

We define the projective binary group of $\Holant(O\mathcal{F})$, denoted by
\begin{equation}
 G:=\{[B]: B\ne 0 \text{ is a binary signature that can be realized in }\Holant(O\mathcal{F}).\}.
\end{equation}
where the $=_2$ is included, so $[I]\in G$.
By \cref{lem:f6-nongamma-binary}, either
$\Holant(O\mathcal F)$ is $\#\mathrm P$-hard or $G\subseteq\Gamma$.
We henceforth consider the latter case.

Joining binary gadgets gives matrix multiplication, exchanging their
inputs gives transposition, and conjugating their vertex labels gives
entrywise conjugation. Hence $G$ is closed under these operations.
Since $\Gamma$ is finite, inverses are powers, so $G$ is a subgroup
of $\Gamma$. In particular, every binary factor of an equality
contraction of $h$ has its class in $G$.

\begin{lemma}[Two possibilities for the binary group]
\label{lem:f6-group-alternatives}
For the projective binary group $G$ of $\Holant(O\F)$, either
$G=\Gamma$ or $G\subseteq K_4$.
\end{lemma}

\begin{proof}
Suppose $G\nsubseteq K_4$, and choose
$[C_{a,b,c}]\in G\setminus K_4$.
Since $G$ is closed under transpose and multiplication, it contains
the class of
\[
C_{a,b,c}^TC_{a,b,c}
=\begin{cases}
a\mathrm iX,&c=ab,\\
c\mathrm iZ,&c=-ab.
\end{cases}
\]
Thus $G$ contains a nonidentity element of $K_4$.
Conjugation by $[C_{a,b,c}]$ cyclically permutes
$[X],[Z],[J]$, so $K_4\subseteq G$.
The image of $[C_{a,b,c}]$ generates $\Gamma/K_4$, which has
order three. Therefore $G=\Gamma$.
\end{proof}

We handle these two cases separately.
If $G=\Gamma$, every local inverse $[V_i^{-1}]$ is realizable
over $\mathcal E'$, so the local matrices can be removed directly.
If $G\subseteq K_4$, all binary factors of the equality contractions
of $h$ lie in $K_4$. We use this restriction to analyze $h$;
it does not by itself imply that each $[V_i]$ lies in $K_4$.

\begin{lemma}[The tetrahedral case]
\label{lem:f6-tetrahedral-case}
For the projective binary group $G$ of $\Holant(O\F)$,
if $G=\Gamma$, then
\[
\Holant(O\mathcal F\cup\{f_6\}\cup\mathcal B)
\leq_T\Holant(O\mathcal F).
\]
\end{lemma}
\begin{proof}[Proof of Lemma~\ref{lem:f6-tetrahedral-case}]
Recall that
$h=\lambda(V_1\otimes\cdots\otimes V_6)f_6$, where
$\lambda\ne0$ and $[V_i]\in\Gamma$.
Since $G=\Gamma$, each $[V_i^{-1}]$ is available.
Attaching these binary signatures to the corresponding inputs of $h$
gives, up to a nonzero scalar,
$
(V_1^{-1}\otimes\cdots\otimes V_6^{-1})h=\lambda f_6.
$
The four Bell classes also belong to $G$, so all signatures in
$\mathcal B$ are available.
Thus,
$
\Holant(O\mathcal F\cup\{f_6\}\cup\mathcal B)
\leq_T\Holant(O\mathcal F).
$
\end{proof}

\begin{lemma}[The Klein case]
\label{lem:f6-klein-case}
For the projective binary group $G$ of $\Holant(O\F)$,
if $G\subseteq K_4$, then $G=K_4$ and 
\[
\Holant(O\mathcal F\cup\{f_6\}\cup\mathcal B)
\leq_T\Holant(O\mathcal F).
\]
\end{lemma}

\begin{proof}[Proof of Lemma~\ref{lem:f6-klein-case}]
Every binary factor of an equality contraction of $h$ has its
projective class in $G\subseteq K_4$.
We first use this restriction to show that, after relabeling its
inputs, $h$ differs from $f_6$ only by Bell matrices.
We then obtain the Bell signatures and remove these matrices.

For the calculation, use the fixed representative
$g=(I\otimes I\otimes(-J)\otimes X\otimes I\otimes I)f_6$.
The definition of $f_6$ gives
\[
\begin{aligned}
g(x)={}&I(x_1,x_2)I(x_3,x_6)I(x_4,x_5)\\
&+(\mathrm iX)(x_1,x_2)(\mathrm iZ)(x_3,x_6)J(x_4,x_5)\\
&+J(x_1,x_2)(\mathrm iX)(x_3,x_6)(\mathrm iZ)(x_4,x_5)\\
&+(\mathrm iZ)(x_1,x_2)J(x_3,x_6)(\mathrm iX)(x_4,x_5).
\end{aligned}
\]
Thus $h=\nu(B_1\otimes\cdots\otimes B_6)g$, where $\nu\ne0$
and $[B_i]\in\Gamma$.
Put $r=C_{+,+,+}$, and let $c_i\in\{0,1,2\}$ specify the
coset $[B_i]\in[r]^{c_i}K_4$.
All calculations with these indices are modulo $3$.
Multiplication adds coset indices, while transposition negates them:
indeed, $C_{a,b,c}^T=C_{a,-b,c}$, and
$[C_{a,b,c}]\in[r]K_4$ exactly when $abc=1$.
Consequently $B_i^TB_j$ has coset index $c_j-c_i$.

The four-term expansion of $g$ gives the following contractions.
Each row expresses $\partial_{ij}^Qg$, up to a nonzero scalar,
as $A(x_a,x_b)B(x_c,x_d)$.
The last column gives the coset indices of $[A],[B]$.
\[
\begin{array}{cc|c|c|c}
ij&Q&\text{variable pairs}&(A,B)&\text{coset indices}\\ \hline
12&I&(36)(45)&(I,I)&(0,0)\\
12&r&(34)(56)&
\left(
\begin{pmatrix}1&\mathrm i\\-1&\mathrm i\end{pmatrix},
\begin{pmatrix}1&-1\\-\mathrm i&-\mathrm i\end{pmatrix}
\right)&(2,1)\\
12&r^2&(35)(46)&
\left(
\begin{pmatrix}1&\mathrm i\\1&-\mathrm i\end{pmatrix},
\begin{pmatrix}-1&-1\\\mathrm i&-\mathrm i\end{pmatrix}
\right)&(2,1)\\
13&I&(25)(46)&(Z,Z)&(0,0)\\
13&r&(24)(56)&
\left(
\begin{pmatrix}1&\mathrm i\\-1&\mathrm i\end{pmatrix},
\begin{pmatrix}1&\mathrm i\\1&-\mathrm i\end{pmatrix}
\right)&(2,2)\\
13&r^2&(26)(45)&
\left(
\begin{pmatrix}1&-\mathrm i\\1&\mathrm i\end{pmatrix},
\begin{pmatrix}-1&1\\\mathrm i&\mathrm i\end{pmatrix}
\right)&(2,1).
\end{array}
\]
For example, contracting inputs $1,2$ replaces the first binary
matrix in each term of the expansion by its entrywise contraction
with $Q$. Substituting $Q=I,r,r^2$ and factoring gives the first
three rows. Summing over inputs $1,3$ gives the remaining rows.

The variable pairs and coset indices in each row depend only on
the coset of $[Q]$.
To see this, the same expansion shows that
\[
X_1X_4Z_5Z_6,\qquad
X_2X_4Z_3Z_5,\qquad
X_3X_4Z_2Z_6,\qquad
X_4X_5Z_1Z_2
\]
fix $g$, where a subscript indicates the input on which a matrix
acts and all other inputs carry $I$.
Their restrictions to inputs $1,2$ are
$X_1,X_2,Z_2,Z_1Z_2$, and their restrictions to inputs $1,3$ are
$X_1,Z_3,X_3,Z_1$.
In both cases these restrictions generate all two-input Bell
classes. Hence Bell matrices on the contracted inputs can be
transferred to the remaining inputs. They do not change the
variable pairs or the coset indices of the resulting factors.

Now consider an equality contraction of $h$ at inputs $i,j$.
It corresponds to contracting $B_i^TB_j$ into $g$.
A factor $A$ on remaining inputs $a,b$ becomes $B_aAB_b^T$.
If $[A]\in[r]^uK_4$, this factor has coset index $c_a+u-c_b$,
which must be zero because its class belongs to $G\subseteq K_4$.

These constraints depend only on differences of the $c_i$.
Subtracting $c_1$ from all six indices, the rows for $12$ and $13$
give the following possibilities:
\[
\begin{array}{c|ccc}
 &c_3=0&c_3=1&c_3=2\\ \hline
c_2=0&000000&001221&\text{none}\\
c_2=1&010212&\text{none}&012120\\
c_2=2&\text{none}&021102&022011.
\end{array}
\]
Each entry lists $(c_1,\ldots,c_6)$.
For example, $(c_2,c_3)=(0,1)$ gives
$c_3=c_6$, $c_4=c_5$, $c_2+2=c_4$, and $c_5+2=c_6$,
hence $001221$.
The other eight choices are obtained from the corresponding
rows of the contraction table in the same way.
Restoring the common shift gives at most eighteen tuples.

We next remove these coset indices using identities of $g$.
For $(P_\pi g)(x)=g(x_{\pi(1)},\ldots,x_{\pi(6)})$, the following
rows satisfy
$(S_1\otimes\cdots\otimes S_6)g=\mu P_\pi g$
for some $\mu\ne0$:
\[
\begin{array}{c|c}
(S_1,\ldots,S_6)&(\pi(1),\ldots,\pi(6))\\ \hline
(I,I,I,\mathrm iZ,\mathrm iZ,I)&(5,2,6,1,3,4)\\
(r,r,r,r,r,C_{-,+,-})&(6,5,3,1,4,2)\\
(I,I,r,C_{+,-,+},C_{+,-,+},C_{+,-,-})&(6,4,5,3,2,1).
\end{array}
\]
These identities follow by applying each row to the four-term
expansion of $g$: a binary factor $M$ on inputs $a,b$ becomes
$S_aMS_b^T$, and substitution gives the indicated permutation
of the same tensor, up to scalar.

Composition of these identities permutes and adds their coset
tuples. More explicitly, if
$L_Sg=\mu P_\pi g$ and $L_Tg=\eta P_\sigma g$, where
$L_S=\bigotimes_iS_i$ and $L_T=\bigotimes_iT_i$, then
\[
(P_\pi L_TP_\pi^{-1})L_Sg
=\mu\eta P_\pi P_\sigma g.
\]
The third row has tuple $001221$.
Composing with the first row permutes a tuple to
$(c_4,c_2,c_5,c_6,c_1,c_3)$.
After subtracting its first coordinate, this gives the cycle
\[
001221\longmapsto010212\longmapsto022011
\longmapsto021102\longmapsto012120\longmapsto001221.
\]
The identity gives $000000$, and the second row adds a common
shift of one. Thus these identities supply all eighteen tuples.

Choose an identity with $[S_i]K_4=[B_i]K_4$ for every $i$.
Then $D_i=B_iS_i^{-1}$ has class in $K_4$, and
\[
h=\nu\mu(D_1\otimes\cdots\otimes D_6)P_\pi g.
\]
This is a new expression for the same available signature $h$;
it does not require the matrices $S_i$ to be available.
Relabeling the inputs therefore leaves an available signature
equal, up to scalar, to $g$ with Bell matrices on its inputs.

It remains to obtain the Bell signatures.
The expansion of $g$ also gives
\[
g(x)=\mathbf1_{\sum_i x_i=0}(-1)^{q_0(x)+x_6},
\qquad
q_0(x)=x_2x_3+x_3x_6+x_1x_6+x_1x_5+x_2x_5.
\]
A Bell matrix flips an input bit and adds a linear phase,
up to scalar. Thus the relabeled $h$ has the form
$\alpha\mathbf1_{\sum_i x_i=s}
(-1)^{q_0(x)+\sum_i\ell_i x_i}$,
where $\alpha\ne0$ and $s,\ell_i\in\mathbb F_2$.

The following are six binary factors of its equality contractions.
The last two columns identify each factor's class as $[X^uZ^v]$.
\[
\begin{array}{c|c|cc}
\text{contracted inputs}&\text{factor inputs}&u&v\\ \hline
1,2&(x_4,x_5)&s+\ell_1+\ell_2&\ell_4+\ell_5\\
1,3&(x_2,x_5)&\ell_1+\ell_3&
1+\ell_1+\ell_2+\ell_3+\ell_5\\
2,3&(x_1,x_4)&s+1+\ell_2+\ell_3&
1+\ell_1+\ell_2+\ell_3+\ell_4\\ \hline
1,2&(x_3,x_6)&\ell_1+\ell_2&
1+\ell_1+\ell_2+\ell_3+\ell_6\\
1,4&(x_2,x_3)&s+\ell_1+\ell_4&
1+s+\ell_2+\ell_3\\
2,4&(x_1,x_6)&s+\ell_2+\ell_4&
1+s+\ell_1+\ell_6.
\end{array}
\]
To obtain these entries, set $x_i=x_j=t$ and sum over $t$.
For $ij=12,13,23,14,24$, the coefficients of $t$ in the exponent
are respectively
\[
\begin{gathered}
x_3+x_6+\ell_1+\ell_2,\qquad
x_2+x_5+\ell_1+\ell_3,\qquad
1+x_5+x_6+\ell_2+\ell_3,\\
x_5+x_6+\ell_1+\ell_4,\qquad
x_3+x_5+\ell_2+\ell_4.
\end{gathered}
\]
The sum forces this coefficient to vanish.
Together with the parity equation, this gives the two binary
supports; substitution gives their phases.
For example, $ij=12$ gives
$x_6=x_3+\ell_1+\ell_2$ and
$x_5=x_4+s+\ell_1+\ell_2$.
The remaining exponent is, up to a constant,
$(1+\ell_1+\ell_2+\ell_3+\ell_6)x_3
+(\ell_4+\ell_5)x_4$,
which gives both $12$ rows.

All six factor classes belong to $G$.
Projective multiplication of $X^uZ^v$ adds the pairs $(u,v)$
over $\mathbb F_2$.
The first three rows sum to $(1,0)$, and the last three sum to
$(0,1)$. Joining the corresponding factors therefore gives
$[X],[Z]\in G$.
Hence $K_4=\langle[X],[Z]\rangle\subseteq G$, and the assumption
$G\subseteq K_4$ yields $G=K_4$.

All Bell signatures are now available.
We can remove the six Bell matrices from the relabeled $h$ to
recover $g$, and then use the fixed Bell transformation defining
$g$ to recover $f_6$.
Only finitely many fixed constructions are used.
Removing their nonzero scalar factors gives
\[
\Holant(O\mathcal F\cup\{f_6\}\cup\mathcal B)
\leq_T\Holant(O\mathcal F).
\]
\end{proof}

\begin{proof}[Proof of \cref{thm:f6-equality-realization}]
Suppose that $\Holant(\mathcal F)$
is not $\#\mathrm P$-hard.
By the AME normal form, after relabeling inputs, write
$f=\lambda(U_1\otimes\cdots\otimes U_6)f_6$ with unitary $U_i$.
Applying \cref{thm:f6-nongamma-hard} with $B=I$ gives
$[U_i^TU_j]\in\Gamma$ for every $i<j$.
By \cref{lem:f6-common-orthogonal}, there is a real orthogonal
matrix $O$ such that
$h=O^{\otimes6}f=\lambda(V_1\otimes\cdots\otimes V_6)f_6$
with $[V_i]\in\Gamma$.

Let $G$ be the projective binary group of $\Holant(O\mathcal F)$.
By \cref{lem:f6-nongamma-binary}, $G\leq\Gamma$.
Lemma~\ref{lem:f6-group-alternatives} gives either $G=\Gamma$
or $G\subseteq K_4$.
Applying \cref{lem:f6-tetrahedral-case} or
\cref{lem:f6-klein-case}, respectively, and using orthogonal
holographic equivalence, we obtain
\[
\Holant(O\mathcal F\cup\{f_6\}\cup\mathcal B)
\leq_T\Holant(O\mathcal F)
\equiv_T\Holant(\mathcal F).
\]
We are done.
\end{proof}

\section{Hardness When $f_6$ is Available}\label{sec: f6 is hard}

In this section, we still assume that $\mathcal{F}$ only contains even arity signatures and doesn't satisfy condition (\ref{cond: tractable classes}).
By \cref{thm:f6-equality-realization}, we may assume that $\Holant(f_6,\mathcal{B},O\mathcal{F})\le_T \Holant(\mathcal{F})$.  Since the invariant of outside condtion~\eqref{cond: tractable classes}, which is prove in~\cref{lem: invariance under real orthogonal transformation}, we can use $\F$ instead of $O\F$ for simplicity in this section.
In this section, we prove that $\Holant(f_6,\mathcal{B},\mathcal{F})$ is \#P-hard (\cref{lm: f6 is hard}).
The proof consists of two cases.
In \cref{subsec: f6 parity}, we consider the case that every signature in $\mathcal{F}$ has parity.
We do a realification argument to prove that every signature in $\mathcal{F}$ is \emph{real} up to a complex phase, otherwise $\Holant(f_6,\mathcal{B},\mathcal{F})$ is already \#P-hard (\cref{lm: realification when parity}). 
Then the result in \cite{realholant} applies (\cref{lm: parity f6 is hard}).
In \cref{subsec: f6 non-parity}, $\mathcal{F}$ contains a signature with no parity.
Then we can do arity reduction to realize a binary signature $b$ with no parity.
By \cref{lem:f6-nongamma-binary}, $b\in \Gamma$, which is the tetrahedral group, otherwise $\Holant(\mathcal{F})$ is \#P-hard.
The binary signature $b$ and the four Bell signatures in $\mathcal{B}$ generate the whole tetrahedral group $\Gamma$.
We then prove \#P-hardness when all signatures in $\Gamma$ is available (\cref{lm: non-parity f6 is hard}).


\subsection{Parity Condition}\label{subsec: f6 parity}
In this subsection, we assume every signature in $\mathcal{F}$ has parity.
Recall
$$
\sigma_0=I=\left[\begin{matrix}
    1 & 0\\
    0 & 1\\
\end{matrix}\right],
\sigma_1=X=\left[\begin{matrix}
    0 & 1\\
    1 & 0\\
\end{matrix}\right],
\sigma_2=Y=\left[\begin{matrix}
    0 & -\mathfrak i\\
    \mathfrak i & 0\\
\end{matrix}\right],
\sigma_3=Z=\left[\begin{matrix}
    1 & 0\\
    0 & -1\\
\end{matrix}\right],
J=\left[\begin{matrix}
    0 & 1\\
    -1 & 0\\
\end{matrix}\right].
$$

\begin{definition}[Projectively real]
    A signature $f:\{0,1\}^d\to \mathbb{C}$ is projectively real, if there exists $\lambda\in \mathbb{C}^{\times}$ and a signature $g:\{0,1\}^d\to \mathbb{R}$ such that $f=\lambda g.$
\end{definition}

\begin{restatable}{definition}{Fgate}
    Let $\mathcal{G}=\{g\mid g \text{  is an $\mathcal{F}$-gate or $g$ is a factor of some $\mathcal{F}$-gate}\}.$
\end{restatable}

Clearly $\Holant(\mathcal{F})\equiv_T \Holant(\mathcal{G})$.
The reason for passing through the signature set from $\mathcal{F}$ to $\mathcal{G}$ is that, in the induction proof of \cref{lm: realification when parity}, we will realize some signature $g\in \mathcal{G}$ with parity from a signature $f\in \mathcal{F}$ with parity, and the induction hypothesis will be applied to $g.$
Also, in the non-parity case, proof of \cref{lm: all 12 binary available is hard}, we will do a minimal counterexample argument for a non-parity signature in $\mathcal{G}.$

\begin{lemma}\label{lm: non-zero Bell contraction}
    Let $f,g\in \mathcal{G}$ be two non-zero signatures both of arity $d\ge 4$.
    If $f$ and $g$ have opposite parity, 
    then there exists $\{i,j\}\subseteq [d]$ and $A\in \{I,X,Z,J\}$ such that $\partial_{ij}^Af\neq 0$ and $\partial_{ij}^Ag\neq 0$, otherwise $\Holant(\mathcal{F})$ is \#P-hard.
\end{lemma}
\begin{proof}
    Since $f$ and $g$ are non-zero and have opposite parity, we may pick $\alpha\in \mathscr{S}(f)$ and $\beta\in \mathscr{S}(g)$ such that $\gamma:=\alpha\oplus \beta$ has odd Hamming weight.
    Combined with $d\ge 4$, there exists three distinct positions $i,j,k$ where $\gamma_i=\gamma_j=\gamma_k.$
    Without loss of generality, we assume $(i,j,k)=(1,2,3).$
    Suppose for a contradiction that, for every $\{u,v\}\subseteq \{1,2,3\}$ and every $A\in\{I,X,Z,J\}$, $\partial_{uv}^Af=0$ or $\partial_{uv}^Ag=0$.
    
    We let $e_i$ be the vector of length $d$ where the $i$-th position is $1$ and other positions are $0$.
    Define 
    \begin{equation*}
        \begin{aligned}
            &a=f(\alpha)\neq 0,&
            &a'=g(\beta)\neq 0;\\
            &b=f(\alpha\oplus e_1\oplus e_2),
            &&b'=g(\beta\oplus e_1\oplus e_2);\\
            &c=f(\alpha\oplus e_1\oplus e_3),
            &&c'=g(\beta\oplus e_1\oplus e_3).
        \end{aligned}
    \end{equation*}
    If $\alpha_1=\alpha_2$, then $\beta_1=\beta_2$ by $\gamma_1=\gamma_2$.
    Consider merging variables $x_1$ and $x_2$ of $f$ and $g$ respectively by $I$,
    we have $\partial_{12}^If(\alpha_3\cdots\alpha_d)=a+b$ and $\partial_{12}^Ig(\beta_3\cdots\beta_d)=a'+b'$.
    Since $\partial_{uv}^If=0$ or $\partial_{uv}^Ig=0$, we have $(a+b)(a'+b')=0$.
    Consider merging variables $x_1$ and $x_2$ of $f$ and $g$ respectively by $Z$, we have $(a-b)(a'-b')=0$.
    If $\alpha_1\neq \alpha_2$, then $\beta_1\neq \beta_2$ by $\gamma_1=\gamma_2.$
    Consider merging by $X$ and $J$, we still have
    $$
    \begin{cases}
        (a+b)(a'+b')=0,\\
        (a-b)(a'-b')=0.
    \end{cases}
    $$
    Plus the above two equations, we have $aa'+bb'=0.$
    Now consider merging variables $x_1$ and $x_3$ of $f$ and $g$ respectively, we will obtain $aa'+cc'=0.$
    Again consider merging variables $x_2$ and $x_3$ of $f$ and $g$ respectively, we have
    $bb'+cc'=0.$
    Thus $aa'=bb'=cc'=0$, contradicting $aa'\neq 0.$
\end{proof}

\begin{lemma}\label{lm: realification when parity}
    Suppose $f\in \mathcal{G}$ has parity. 
    Then $f$ is projectively real, otherwise $\Holant(f_6,\mathcal{B},\mathcal{F})$ is \#P-hard.
\end{lemma}

\begin{proof}
    We prove this lemma by induction on the arity of $f$.
    If $f$ has arity $2$, then $f\in \Gamma$, otherwise $\Holant(\mathcal{F})$ is \#P-hard by \cref{lem:f6-nongamma-binary}.
    Since $f$ has parity, $f\in \{I,X,Z,J\}.$
    Then $f$ is real.
    If $f$ has arity $4$, then by \cref{lm: base case 2n=4}, $\Holant(\mathcal{F})$ is \#P-hard unless $f\in \mathcal{U}^{\otimes }$.
    Suppose $f=f_1\otimes f_2$, where $f_1,f_2\in \mathcal{U}.$
    Then we can realize $f_1$ and $f_2$ by \cref{lm: decomposition}.
    It is clear that the factor of a parity signature must satisfy parity as well.
    So both $f_1$ and $f_2$ have parity.
    Hence $f_1,f_2\in \{I,X,Z,J\}$, otherwise $\Holant(\mathcal{F})$ is \#P-hard, by the same argument as in the arity 2 case.
    Thus $f=f_1\otimes f_2$ is real.

    Next suppose $f$ has arity $2d\ge 6$.
    If $f$ is reducible, then we can realize its arbitrary factor $f_1$ in $\mathcal{G}$.
    Clearly $f_1$ has parity.
    By induction hypothesis, $f_1$ is projectively real, otherwise $\Holant(\mathcal{F})$ is \#P-hard.
    So $f$ is projectively real.
    In the following, we assume that $f$ is irreducible.
    By \cref{lm: not 2nd-orth is hard},
    we may assume that $f$ satisfies {\sc 2nd-Orth}. 
    Fix a pair of indices $\{i,j\}$ for some $i,j$, $1\le i<j \le 2d$.
    Consider $\partial_{ij}f=f_{ij}^{00}+f_{ij}^{11}$.
    By {\sc 2nd-Orth}, $|\partial_{ij}f|^2=|f_{ij}^{00}|^2+|f_{ij}^{11}|^2>0.$
    So $\partial_{ij}f\neq 0.$
    Similarly, $\partial^{A}_{ij}f\neq 0$ for every $A\in\{I,X,Z,J\}.$

    It is clear that $\partial_{ij}^Af$ has the same parity as $f$ for any $A\in \{I,Z\},$ and $\partial_{ij}^Af$ has the opposite parity to $f$ for any $A\in \{X,J\}.$
    So $\partial_{ij}^Af$ has arity $2(d-1)$ and it has parity.
    By induction hypothesis,
    we may assume that $\partial_{ij}^Af=\lambda_A\cdot  g_A$ for some $\lambda_{A}\in \mathbb{C}^{\times}$ and some non-zero real signature $g_A$, for every $A\in \{I,X,Z,J\}.$
    
    Now fix arbitrary $B\in \{I,Z\}$ and $C\in \{X,J\}.$
    Notice that $\partial_{ij}^Bf$ and $\partial_{ij}^Cf$ are non-zero signatures in $\mathcal{G}$ of arity at least 4, and they have opposite parity.
    By \cref{lm: non-zero Bell contraction}, there exists $A\in\{I,X,Z,J\}$ and a pair of indices $\{u,v\}$ with $1\le u<v\le 2d$ and $\{u,v\}\cap \{i,j\}=\emptyset,$ such that $\partial_{uv}^A\partial_{ij}^Bf\neq 0$ and $\partial_{uv}^A\partial_{ij}^Cf\neq 0.$
    Consider $\partial_{uv}^Af\in \mathcal{G}$, it is non-zero by {\sc 2nd-Orth} and has parity. 
    So we may write $\partial_{uv}^Af=\mu\cdot h$ for some $\mu\in \mathbb{C}^\times$ and non-zero real signature $h.$
    Commutativity gives that
    \begin{equation*}
        \begin{aligned}
            &\lambda_B\cdot \partial_{uv}^Ag_B=\partial_{uv}^A\partial_{ij}^Bf=\partial_{ij}^B\partial_{uv}^Af=\mu\cdot \partial_{ij}^Bh,\\
            &
            \lambda_C\cdot \partial_{uv}^Ag_C=\partial_{uv}^A\partial_{ij}^Cf=\partial_{ij}^C\partial_{uv}^Af=\mu\cdot \partial_{ij}^Ch.\\
        \end{aligned}
    \end{equation*}
    Notice that $\partial_{uv}^Ag_B,\,\partial_{ij}^Bh,\,\partial_{uv}^Ag_C,\,\partial_{ij}^Ch$ are all non-zero and real, we have $\frac{\lambda_B}{\mu}\in \mathbb{R}^\times$ and $\frac{\lambda_C}{\mu}\in \mathbb{R}^\times$.
    Thus $\frac{\lambda_B}{\lambda_C}\in \mathbb{R}^\times.$
    Since $B\in \{I,Z\}$ and $C\in\{X,J\}$ are arbitrarily picked, $\lambda_I,\lambda_X,\lambda_Z,\lambda_J$ are non-zero complex numbers with the same phase.
    Therefore, we may write $\partial_{ij}^Af=e^{\mathfrak i \theta} \cdot r_A$ for some $\theta\in [0,2\pi)$ and some non-zero real signature $r_A$, for every $A\in\{I,X,Z,J\}$.
    We may recover $f$ by 
    \begin{equation*}
        \begin{aligned}
            &f_{ij}^{00}=\frac{1}{2}(\partial_{ij}^If+\partial_{ij}^Zf)=\frac{e^{\mathfrak i \theta}}{2}(r_I+r_Z),\quad 
            f_{ij}^{11}=\frac{1}{2}(\partial_{ij}^If-\partial_{ij}^Zf)=\frac{e^{\mathfrak i \theta}}{2}(r_I-r_Z),\\
            &f_{ij}^{01}=\frac{1}{2}(\partial_{ij}^Xf+\partial_{ij}^Jf)=\frac{e^{\mathfrak i \theta}}{2}(r_X+r_J),\quad 
            f_{ij}^{10}=\frac{1}{2}(\partial_{ij}^Xf-\partial_{ij}^Jf)=\frac{e^{\mathfrak i \theta}}{2}(r_X-r_J).
        \end{aligned}
    \end{equation*}
    Therefore, $f$ is projectively real.
\end{proof}

\begin{lemma}\label{lm: parity f6 is hard}
    If every signature  $f\in\mathcal{F}$ has parity, then $\Holant(f_6,\mathcal{B},\mathcal{F})$ is \#P-hard.
\end{lemma}
\begin{proof}
    By \cref{lm: realification when parity}, every signature $f\in \mathcal{F}$ is projectively real, otherwise $\Holant(\mathcal{F})$ is \#P-hard and then we are done.
    Since a normalization doesn't change the complexity, we may assume that $\mathcal{F}$ is a set of real signatures.
    Then by Lemma 7.20 in \cite{realholant}, $\Holant(f_6,\mathcal{B},\mathcal{F})$ is \#P-hard.
\end{proof}

\subsection{Non-parity Condition}\label{subsec: f6 non-parity}

\begin{lemma}\label{lm: non-pairty reduction}
    If $f$ is a signature of arity $2d\ge 4$ with no parity, then we can realize a signature $g$ of arity $2(d-1)$ with no parity.
\end{lemma}
\begin{proof}
    The proof is the same as Lemma 7.1 in \cite{realholant}.
\end{proof}

By \cref{lm: non-pairty reduction}, we can realize a binary signature $b$ with no parity.
By \cref{lem:f6-nongamma-binary}, $b\in\Gamma,$ otherwise $\Holant(f_6,\mathcal{B},\mathcal{F})$ is \#P-hard.
So $b\in \Gamma\setminus \{I,X,Z,J\}\cong A_4 \setminus V_4$, where $V_4$ is the Klein group.
\textcolor{black}{Since an arbitrary element in $A_4 \setminus V_4$ and $V_4$ generate $A_4$,} we can realize all elements in $\Gamma.$
In the following, we will prove $\Holant(\Gamma,f_6,\mathcal{F})$ is \#P-hard.
Recall
\Fgate*
Now we may assume $\Gamma\cup \{f_6\}\subseteq \mathcal{G}.$

\begin{definition}[Pauli tensor]
    For $i_1,i_2,\ldots,i_k\in \{0,1,2,3\}$, we define the Pauli tensor matrix 
    $$Q_{i_1i_2\ldots i_k}=\sigma_{i_1}\otimes \sigma_{i_2}\otimes\cdots\otimes \sigma_{i_k}.$$
    The weight of $Q_{i_1i_2\ldots i_k}$, denoted by $\mathrm{wt}(Q_{i_1i_2\ldots i_k})$, is defined as $\#\{j\mid 1\le j \le k\mid i_j\neq 0\}$, the number of matrices in $\sigma_{i_1},\sigma_{i_2},\ldots,\sigma_{i_k}$ that are not equal to $\sigma_0=I.$
\end{definition}

\begin{lemma}\label{lm: conjugate action by Gamma}
    For any $A,B\in\{X,Y,Z\}$, there exists $\gamma\in \Gamma$ such that $\gamma A\gamma^\dagger =B.$
\end{lemma}
\begin{proof}
    Let $C=\frac{1}{2}(I+\mathfrak iX+J+\mathfrak iZ)=\frac{1}{2}\left[\begin{matrix}
        1+\mathfrak i & 1+\mathfrak i\\
        -1+\mathfrak i & 1-\mathfrak i
    \end{matrix}\right]\in \Gamma.$
    Direct cauculation shows that
    \begin{equation*}
        CXC^\dagger=Z,\quad 
        CZC^\dagger=Y,\quad
        CYC^\dagger=X.
    \end{equation*}
    Since the power of $C$ is also in $\Gamma,$ the lemma is proved.
\end{proof}

\begin{lemma}\label{lm: decrease Pauli weight by 2}
    Suppose $Q_{i_1i_2\ldots i_k}$ is a Pauli tensor of weight at least $3.$
    Then we can realize a signature $g_{i_1i_2\ldots i_k}\in \mathscr{A}\cap \mathcal{G}$ of arity $2k$ with matrix $U_{i_1i_2\ldots i_k}=M_{x_1\ldots x_k,\, x_{k+1}\ldots x_{2k}}(g_{i_1i_2\ldots i_k})$ which is unitary, such that $U_{i_1i_2\ldots i_k}Q_{i_1i_2\ldots i_k}U_{i_1i_2\ldots i_k}^\dagger$ is a Pauli tensor of weight decreased by 2 compared with the Pauli tensor $Q_{i_1i_2\ldots i_k}$.
\end{lemma}

\begin{proof}
    Since $Q_{i_1i_2\ldots i_k}$ has weight at least 3, we may assume $i_1i_2i_3\neq 0$ without loss of generality.
    First apply \cref{lm: conjugate action by Gamma}, we realize a 6-ary signature $g_6$ with matrix $M(g_6)=\gamma_1\otimes \gamma_2\otimes\gamma_3$ with $\gamma_1,\gamma_2,\gamma_3\in \Gamma$, such that $M(g_6)(\sigma_{i_1}\otimes\sigma_{i_2}\otimes \sigma_{i_3})M(g_6)^\dagger=X\otimes Z\otimes Z.$
    Now connect variables $x_4,x_5,x_6$ of $f_6$ and variables $x_1,x_2,x_3$ of $g_6$, we realize the signature $h_6$ with matrix $M(h_6)$, such that $M(h_6)(\sigma_{i_1}\otimes\sigma_{i_2}\otimes \sigma_{i_3})M(h_6)^\dagger=M(f_6)(X\otimes Z\otimes Z)M(f_6)^\dagger=4X\otimes I\otimes I.$
    Let $h_{i_1i_2\ldots i_k}=h_6 \otimes (=_2)^{\otimes (k-3)}$.
    Then $M(h_{i_1i_2\ldots i_k})Q_{i_1i_2\ldots i_k}M(h_{i_1i_2\ldots i_k})^\dagger$ has Pauli weight decreased by 2.
    Clearly $M(h_{i_1i_2\ldots i_k})$ is projectively unitary, because $f_6$ and $g_6$ are projectively unitary.
    Also $h_{i_1i_2\ldots,i_k}\in \mathscr{A}$, since all signatures in $\Gamma$ and $f_6$ are affine.
    Letting $g_{i_1i_2\ldots i_k}$ be the normalized signature of $h_{i_1i_2\ldots i_k}$ finishes the proof.
\end{proof}

Repeatedly apply \cref{lm: decrease Pauli weight by 2}, we have the following Corollary:

\begin{corollary}\label{cor: decrease Pauli weight}
    Suppose $Q_{i_1i_2\ldots i_k}$ is a Pauli tensor of weight at least $3.$
    Then we can realize a signature $f_{i_1i_2\ldots i_k}\in \mathscr{A}\cap\mathcal{G}$ of arity $2k$ with matrix $V_{i_1i_2\ldots i_k}=M_{x_1\ldots x_k,\, x_{k+1}\ldots x_{2k}}(f_{i_1i_2\ldots i_k})$ which is unitary, such that $V_{i_1i_2\ldots i_k}Q_{i_1i_2\ldots i_k}V_{i_1i_2\ldots i_k}^\dagger$ is a Pauli tensor of weight 1 or 2.
\end{corollary}

\begin{lemma}\label{lm: all 12 binary available is hard}
    Suppose $\mathcal{F}$ doesn't satisfy condition \eqref{cond: tractable classes}. 
    Then $\Holant(\Gamma, f_6, \mathcal{F})$ is \#P-hard.
\end{lemma}
\begin{proof}
    
    Since $\mathcal{F}$ doesn't satisfy condition (\ref{cond: tractable classes}), then there exists some $f\in \mathcal{G}$ such that $f\notin \mathscr{A}.$
    We choose such an $f$ with minimal arity $d$.
    If $d=2$, then $f\notin \Gamma$ since $\Gamma\subseteq \mathscr{A}$.
    By \cref{lem:f6-nongamma-binary}, $\Holant(\Gamma,f_6,\mathcal{F})$ is \#P-hard.
    In the following we assume $d\ge 4.$
    Let $\mathbf{f}$ denote the truth table of $f$, listed lexicographically as a column vector in $\mathbb{C}^{2^d}$.

    Consider an arbitrary Pauli tensor $Q_{i_1i_2\ldots i_d}.$
    By \cref{cor: decrease Pauli weight}, we can realize a signature $f_{i_1i_2\ldots i_d}\in \mathscr{A}\cap \mathcal{G}$ of arity $2d$ with matrix $V_{i_1i_2\ldots i_d}=M_{x_1\ldots x_d,\,x_{d+1}\ldots x_{2d}}(f_{i_1i_2\ldots i_d})$ such that $V_{i_1i_2\ldots i_d}Q_{i_1i_2\ldots i_d}V_{i_1i_2\ldots i_d}^\dagger$ is a Pauli tensor of weight 1 or 2.
    Connect variables $x_{d+1},\ldots, x_{2d}$ of $f_{i_1i_2\ldots i_d}$ and variables $x_1,\ldots, x_d$ of $f$, we realize a signature $g_{i_1i_2\ldots i_d}$ of arity $d$ with truth table $\mathbf{g}_{i_1i_2\ldots i_k}=V_{i_1i_2\ldots i_d}\mathbf{f}.$

    \begin{claim}\label{claim: g is irreducible}
        For every $i_1,i_2,\ldots, i_d\in [4]$, $g_{i_1i_2\ldots i_d}$ is irreducible. 
    \end{claim}
    \begin{claimproof}{\cref{claim: g is irreducible}}
        Suppose for a contradiction that $g_{i_1i_2\ldots i_d}$ is reducible for some $i_1,i_2,\ldots, i_d\in [4]$.
        Every factor $g'$ must be affine.
        Otherwise we can realize $g'\in \mathcal{G}$ of arity less than $d$, contradicting minimality of $f$.
        If every factor $g'\in \mathscr{A}$, then $g_{i_1i_2\ldots i_d}\in \mathscr{A}$.
        Since $f_{i_1i_2\ldots i_d}\in \mathscr{A}\cap \mathcal{G}$, we have $\overline{f_{i_1i_2\ldots i_d}}\in \mathscr{A}\cap \mathcal{G}$ by conjugate closure.
        Connect variables $x_{1},\ldots, x_{d}$ of $\overline{{f}_{i_1i_2\ldots i_d}}$ and variables $x_1,\ldots, x_d$ of $g_{i_1i_2\ldots i_d}$, we realize a signature with truth table $V_{i_1i_2\ldots i_d}^\dagger V_{i_1i_2\ldots i_d}\mathbf{f}=\mathbf{f}.$
        Because both $\overline{{f}_{i_1i_2\ldots i_d}}$ and $g_{i_1i_2\ldots i_d}$ are affine, $f$ is affine, contradiction.
    \end{claimproof}
    \begin{claim}\label{claim: reduce Pauli weight}
        There exists $i_1,i_2,\ldots, i_d\in [4]$ such that $g_{i_1i_2\ldots i_d}$ doesn't satisfy {\sc 2nd-Orth.}
    \end{claim}
    \begin{claimproof}{\cref{claim: reduce Pauli weight}}
        
        Suppose for a contradiction that, for all $i_1,i_2,\ldots, i_d\in [4]$, $g_{i_1i_2\ldots i_d}$ satisfies {\sc 2nd-Orth}.
        Consider the Pauli representation of $\mathbf{ff}^\dagger:$
        \begin{equation*}
            \mathbf{ff^\dagger}=\frac{1}{2^d}\sum_{i_1,i_2,\ldots,i_d\in[4]}\mathrm{Tr}(\mathbf{ff}^\dagger Q_{i_1i_2\ldots i_d})Q_{i_1i_2\ldots i_d}.
        \end{equation*}
        The Pauli coefficient is
        \begin{equation*}
            \mathrm{Tr}(\mathbf{ff}^\dagger Q_{i_1i_2\ldots i_d})=\mathbf{f}^\dagger Q_{i_1i_2\ldots i_d} \mathbf{f}=\mathbf{g}_{i_1i_2\ldots i_d}^\dagger V_{i_1i_2\ldots i_d}Q_{i_1i_2\ldots i_d}V_{i_1i_2\ldots i_d}^\dagger \mathbf{g}
            =\mathbf{g}_{i_1i_2\ldots i_d}^\dagger Q'_{i_1i_2\ldots i_d} \mathbf{g},
        \end{equation*}
        where $Q'_{i_1i_2\ldots i_d}$ is a Pauli tensor of weight 1 or 2 if $i_1i_2\cdots i_d\neq 0$.
        So $\mathrm{Tr}(\mathbf{ff}^\dagger Q_{i_1i_2\ldots i_d})=0$ if $i_1i_2\cdots i_d\neq 0$ by {\sc 2nd-Orth},
        and $\mathrm{Tr}(\mathbf{ff}^\dagger Q_{00\ldots 0})=|\mathbf{f}|^2>0.$
        Thus $\mathbf{ff}^\dagger={|\mathbf{f}|^2}/{2^d}\cdot I_{2^d}$.
        LHS has rank 1 while RHS has full rank. 
        Contradiction!
    \end{claimproof}
    By \cref{claim: g is irreducible}, \cref{claim: reduce Pauli weight}, there exists some $g_{i_1i_2\ldots i_d}\in \mathcal{G}$ of arity at least 4 that is irreducible and doesn't satify {\sc 2nd-Orth}.
    So by \cref{lm: not 2nd-orth is hard}, $\Holant(\Gamma, f_6, \mathcal{F})\equiv_T \Holant(\mathcal{G})$ is \#P-hard.
\end{proof}

\begin{lemma}\label{lm: non-parity f6 is hard}
    If $\mathcal{F}$ contains a signature with no parity, then $\Holant(f_6,\mathcal{B},\mathcal{F})$ is \#P-hard.
\end{lemma}

\begin{proof}
    By \cref{lm: non-pairty reduction}, we can realize a binary signature $b$ with no parity.
    By \cref{lem:f6-nongamma-binary}, $b\in\Gamma,$ otherwise $\Holant(\mathcal{F})$ is \#P-hard.
    So $b\in \Gamma\setminus \{I,X,Z,J\}\cong A_4 \setminus K_4.$
    Since an arbitrary element in $A_4 \setminus K_4$ and $K_4$ generates $A_4$, we can realize all elements in $\Gamma.$
    By \cref{lm: all 12 binary available is hard}, $\Holant(f_6,\mathcal{B},\mathcal{F})\equiv_T\Holant(f_6,\Gamma,\mathcal{F})$ is \#P-hard.
\end{proof}

\subsection{Hardness When $f_6$ is Available}

\begin{lemma}\label{lm: f6 is hard}
    If $\mathcal{F}$ doesn't satisfy condition \eqref{cond: tractable classes}, then $\Holant(f_6,\mathcal{B},\mathcal{F})$ is \#P-hard.
\end{lemma}

\begin{proof}
    Combine \cref{lm: parity f6 is hard,lm: non-parity f6 is hard}.
\end{proof}

The following lemma completes the induction proof for $2n=6$.
\begin{lemma}\label{lm: induction 2n=6}
    If $\mathcal{F}$ doesn't satisfy condition \eqref{cond: tractable classes} and it contains a signature ${f}$ of arity $6$ such that $f\notin \mathcal{U}^\otimes$, then $\Holant(\mathcal{F})$ is \#P-hard.
\end{lemma}

\begin{proof}
    We may assume \eqref{eq: arity 6 assumption} in \cref{sec: Third Order Orth}, otherwise $\Holant(\mathcal{F})$ is \#P-hard.
    By \cref{thm: third order orthogonality}, $\ket{f}$ is an AME(6,2) state after normalization.
    By \cref{thm:f6-equality-realization}, $\Holant(f_6,\mathcal{B},\mathcal{F})\equiv_T \Holant(\mathcal{F})$.
    Notice here we ignore the real orthogonal holographic transformation in \cref{thm:f6-equality-realization}, since condition \eqref{cond: tractable classes} is invariant under holographic transformation.
    Then by \cref{lm: f6 is hard}, $\Holant(\mathcal{F})$ is \#P-hard.
\end{proof}

\section{Final Obstacle: \texorpdfstring{$f_8$}{f8}}
\label{sec: 2n=8}
\label{sec: f8}

In this section, we will complete the induction proof for $2n=8$  (\cref{thm: base case 2n=8}).
We assume that $\mathcal F$ does not satisfy condition \rm{(\ref{cond: tractable classes})}, and $\mathcal{F}$ contains a signature $f\notin \mathcal{U}^\otimes.$

Shao and Cai discovered the signature $f_8$ in \cite{realholant}.
It is $f_8=\chi_C$ where
\begin{equation}\label{eq: f8}
 \begin{aligned}
  C=&\mathscr{S}(f_8)=\{(x_1, x_2, \ldots, x_8)\in \mathbb{Z}_2^{8} \mid~
x_1+x_2+x_3+x_4= 0, ~ x_5+x_6+x_7+x_8= 0,\\
  &\hspace{34ex}x_1+x_2+x_5+x_6= 0, ~ x_1+x_3+x_5+x_7= 0\}.\\
=&\{00000000, 00001111, 00110011, 00111100, 01010101, 01011010, 01100110, 01101001,\\
& ~10010110,  10011001, 10100101, 10101010, 11000011, 11001100, 11110000, 11111111\}.
\end{aligned}
\end{equation}

$f_8$ have some extraordinary properties.
For example, it satisfies {\sc 2nd-Orth} and $f_8\in \int \mathcal{U}^{\otimes}$.
Thus, simple merging and factorization cannot realize a signature outside $\mathcal{U}^{\otimes}$.
In \cref{subsec: realize f8 Reed Muller}, we will realize $f_8$.
In \cref{subsec: f8 available hardness}, we prove \#P-hardness when $f_8$ is available.

\subsection{Realizing \texorpdfstring{$f_8$}{f8} Through the Lens of Reed-Muller Code}
\label{subsec: realize f8 Reed Muller}

\begin{definition}[Strong equality property]\label{def: Strong equality property}
    Let $\mathcal{E}^\otimes=\{\lambda\cdot (=_2)^{\otimes n}\mid \lambda\in \mathbb{C}^\times,n\ge1\}$ be the set of tensor products of binary $=_2$ up to a non-zero scalar.
    We say a signature $h$ satisfies the \emph{strong equality property}, if $h\in \int \mathcal{E}^{\otimes}$, i.e., for any pair of indices $\{i,j\}$, $\partial_{ij}h\in \mathcal{E}^\otimes.$
\end{definition}

It can be checked that $f_8$ satisfies the strong equality property.
The following lemma is essentially the same as
\cite[Lemma~8.1]{realholant}.
The only difference is that when doing binary extension, we use the conjugate transpose of the binary instead of the binary itself.

\begin{lemma}
\label{lm: unitary wire normalization}
Let $f\notin \mathcal{U}^\otimes$ be a signature of arity $2n\ge 8$ in $\mathcal{F}$.
Then one of the following holds:
\begin{itemize}
    \item $\Holant(\mathcal{F})$ is \#P-hard;
    \item There is a signature $g\notin \mathcal{U}^\otimes$ of arity $2k\le 2n-2$ that is realizable from $f$;
    \item There is an irreducible signature $f^\ast$ of arity $2n$ that is realizable from $f$ with the strong equality property.
\end{itemize}
\end{lemma}

\begin{proof}
Since ${f}\notin {\mathcal{U}}^{\otimes}$, $ f\not\equiv 0$.
Again, we may assume that ${f}$ is irreducible.
Otherwise, by factorization, we can realize a nonzero signature of odd arity and we get \#P-hardness by \cref{lm: odd holant with conjugation}, or we can realize a signature of lower even arity that is not in $\mathcal{U}^{\otimes}$ and we are done.
Under the assumption that $f$ is irreducible, we may further assume that $f$ satisfies {\sc 2nd-Orth} by \cref{lm: not 2nd-orth is hard}.
Consider signatures ${\partial}_{ij}{f}$ for all pairs of indices $\{i, j\}$. 
If there exists a pair $\{i, j\}$ such that ${\partial}_{ij}{f}\notin {\mathcal{U}}^{\otimes}$, then let ${g}={\partial}_{ij}{f}$, and we are done.
Thus, we may also assume that ${f}\in {\int}{\mathcal{U}}^{\otimes}$.  

Since $f$ satisfies {\sc 2nd-Orth}, $\partial_{ij}f$ is non-zero by \cref{lm: 2nd-orth contraction nonzero}.
Since $ {\partial}_{12} {f} \in  {\mathcal{U}}^{\otimes}$, without loss of generality, we may assume that in the UPF of $ {\partial}_{12} {f}$, variables $x_3$ and $x_4$ appear in one  binary signature $b_1(x_3, x_4)$, $x_5$ and $x_6$ appear in one binary signature $b_2(x_5, x_6)$ and so on. 
Thus, we have 
\begin{equation*}
 {\partial}_{12} {f}=   {b_1}(x_3, x_4)\otimes  {b_2}(x_5, x_6)\otimes {b_3}(x_7, x_8)\otimes \ldots \otimes  {b_{n-1}}(x_{2n-1}, x_{2n}).
\end{equation*}

All these binary factors $b_1,b_2,\ldots,b_{n-2}$ are realizable by
\cref{lm: decomposition}.  
Put $B_i=M(b_i)$.  
Since
$b_i\in\mathcal U$ and $b_i\neq 0$, $B_i$ is projectively unitary.  Reversing the two arguments gives $B_i^{\mathtt T}$, and conjugate closure gives
$\overline B_i$.  
Hence the binary signature $b_i^{-1}$ with matrix $B_i^{-1}\propto \overline{B_i}^{\tt T}$ is also available.
We consider the following gadget construction on $f$.
For $1\leqslant i\leqslant n-1$, by a slight abuse of names of variables, we connect the variable $x_{2i+1}$ of $f$   
with the variable $x_{2i+2}$ of $  {b^{-1}_{i}}(x_{2i+1}, x_{2i+2})$. 
We get a signature $f'$ of arity $2n$. 
We denote this gadget construction by $G_1$ and we write $f'$ as $G_1 \circ f$.
$G_1$ is constructed by extending variables of $f$ using binary signatures realized from ${\partial}_{12}{f}$. 
It does not change the irreducibility of $f$.  
Thus, $f'$ is irreducible since $f$ is irreducible.
Similarly, we may assume that $f'\in {\int}  {\mathcal{U}}^{\otimes}$. Otherwise, we are done.

Consider the signature ${\partial}_{12}{f'}$. 
Since the above gadget construction $G_1$  does not touch variables $x_1$ and $x_2$ of $f$, $G_1$
commutes with the merging gadget ${\partial}_{12}$. 
Thus, ${\partial}_{12}{f'}$ can be realized by performing the gadget construction $G_1$ on ${\partial}_{12}{f}$, which connects each binary signature ${b_{i}}(x_{2i+1}, x_{2i+2})$ in the UPF of ${\partial}_{12}{f}$ with  ${b^{-1}_{i}}(x_{2i+1}, x_{2i+2})$. 
Thus, each binary signature ${b_{i}}$ in ${\partial}_{12}{f}$ is canceled up to a nonzero scalar after this gadget construction $G_1$. 
After normalization and renaming variables, we have 
\begin{equation}\label{eqn:lm5-partial12-f}
{\partial}_{12}{f'}= (=_2)(x_3, x_4)\otimes (=_2)(x_5, x_6)\otimes \ldots \otimes (=_2)(x_{2n-1}, x_{2n}).
\end{equation}
Thus, ${\partial}_{12}{f'}\in {\mathcal{E}}^{\otimes}$.
Moreover, for all pairs of indices $\{i, j\}$ disjoint with $\{1, 2\}$,
we have 
\begin{equation}\label{equ-ij-12-in-D}
    {\partial}_{(ij)(12)}{f'}\in {\mathcal{E}}^{\otimes}, \text{ ~and hence~ } {\partial}_{(ij)(12)}{f'}\not\equiv 0.
    \end{equation}
A fortiori, for all pairs of indices $\{i, j\}$ disjoint with $\{1, 2\}$, ${\partial}_{ij}{f'}\not\equiv 0$.

Next, we show that we can realize an irreducible signature $f^\ast$ of arity $2n$ from ${f'}$ such that ${f^\ast}$ has the strong equality property.
We first prove the following claim.

\begin{claim}\label{claim: strong equality property}
    Let ${h}\in {\int}{{\mathcal{U}}}^{\otimes}$ be a signature of arity $2n\geqslant 8$. 
    If $\partial_{ij}{h}$ is a tensor product of $=_2$ up to a non-zero scalar for all $\{i, j\}$ disjoint with $\{1, 2\}$, then ${h}$ has strong equality property.
\end{claim}

\begin{claimproof}{\cref{claim: strong equality property}}
    Clearly, we only need to show that ${\partial}_{1k}{h}\in {\mathcal{E}}^{\otimes}$ for all $2 \leqslant k \leqslant 2n$. Then, by symmetry we also have ${\partial}_{2k}{h}\in {\mathcal{E}}^{\otimes}$ for $k=1$ and all $3 \leqslant k \leqslant 2n$. This will prove ${h}\in {\int}{\mathcal{E}}^{\otimes}.$
    Consider ${\partial}_{1k}{h}$ for an arbitrary $2\leqslant k \leqslant 2n$.
   Since for all $\{i, j\}$ disjoint with $\{1, 2\}$, we have
    ${\partial}_{ij}{h}\in {\mathcal{E}}^{\otimes}$, a fortiori
    for all  $\{i, j\}$ disjoint with $\{1, 2\} \cup \{k\}$,
    \begin{equation}\label{eqn:partial-1kj-in-sec8}
    \partial_{(1k)(ij)}{h}\in {\mathcal{E}}^{\otimes}.
   \end{equation}
    Since ${h}$ has arity $2n\geqslant 8$, we can pick a pair of indices $\{i, j\}$ disjoint with $\{1, 2\} \cup \{k\}$.
    Since $\partial_{(1k)(ij)}{h}\in {\mathcal{E}}^{\otimes}$, which is nonzero,  a  fortiori  we have
    ${\partial}_{1k}{h} \not \equiv 0$.
    So we may consider  the UPF of ${\partial}_{1k}{h}$,
    which is known to be in ${\mathcal{U}}^{\otimes}$.
    For a contradiction,
    suppose that  there is a binary signature  ${b_1}$ (as a factor of ${\partial}_{1k}{h}$) such that ${b_1}$ is not an associate of $=_2$. 
    Among the two variables in the scope of ${b_1}$,  at least one is not $x_2$. 
    We pick such a variable $x_s$ where $x_s \neq x_2$. 
    Then, we consider another binary signature ${b_2}$ in the UPF of ${\partial}_{1k}{h}$. 
    \begin{itemize}
        \item 
    If ${b_2} =\lambda \cdot (=_2)$, for some  nonzero scalar
    $\lambda$, then we pick a variable $x_t$ in the scope of ${b_2}$ that is not $x_2$. 
    Consider ${\partial}_{(st)(1k)}{h}$. 
    When merging variables $x_s$ and $x_t$ of ${\partial}_{1k}{h}$, 
    we connect the variable $x_s$ of  ${b_1}$ with the variable $x_t$ of $\lambda  \cdot (=_2)$, and the resulting binary signature is just $\lambda \cdot{b_1}$, which is not an associate of $=_2$. 
    Thus, we have ${\partial}_{(st)(1k)}{h}\notin {\mathcal{E}}^{\otimes}$.
    
    \item Otherwise,  ${b_2}$ is not an associate of $=_2$. Since ${\partial}_{1k}{h}$ has arity $2n-2\geqslant 6$, we can find another binary signature ${b_3}$ in the UPF of ${\partial}_{1k}{h}$. 
    We pick a variable $x_t$ in the scope of ${b_3}$ that is not $x_2$. 
    Consider ${\partial}_{(st)(1k)}{h}$. 
    Now, when merging variables $x_s$ and $x_t$ of ${\partial}_{1k}{h}$, 
    the binary signature ${b_2}$ is untouched. 
    Thus, we have ${b_2}\mid {\partial}_{(st)(1k)}{h}$, which implies that ${\partial}_{(st)(1k)}{h}\notin {\mathcal{E}}^{\otimes}$.
   \end{itemize}
   
   Note that in both cases, $\{s,t\} \cap (\{1,2\} \cup \{k\}) = \emptyset$. Therefore
   the two cases above both contradict
   (\ref{eqn:partial-1kj-in-sec8}) by picking $\{i, j\}=\{s, t\}$.
      Thus, ${\partial}_{1k}{h}\in {\mathcal{E}}^{\otimes}$ for all $2 \leqslant k \leqslant 2n$. Then similarly, we can show that ${\partial}_{2k}{h}\in {\mathcal{E}}^{\otimes}$ for all $3 \leqslant k \leqslant 2n$. 
\end{claimproof}

Remember that ${\partial}_{ij}{f'}\not\equiv 0$ for all $\{i, j\}$ disjoint with $\{1, 2\}$.
We consider the UPF of ${\partial}_{ij}{f'}$. Since ${f'}\in {\int}{\mathcal{U}}^{\otimes}$, there are two cases depending on whether variables $x_1$ and $x_2$ appear in one  binary signature or two distinct binary signatures.

\vspace{1ex}
\noindent{\bf Case 1.}
    For every $\{i, j\}$ disjoint with $\{1, 2\}$, in the UPF of ${\partial}_{ij}{f'}$, $x_1$ and $x_2$ appear in one nonzero binary signature
    ${b_{ij}}(x_1, x_2) \in {\mathcal{U}}$.  
    In other words, for every $\{i, j\}$ disjoint with $\{1, 2\}$, $${\partial}_{ij}{f'}={b_{ij}}(x_1, x_2)\otimes {g_{ij}}, ~~\text{ for some } {g_{ij}}\not\equiv 0.$$
    (These factors ${b_{ij}}$ and $ {g_{ij}}$ are nonzero since ${\partial}_{ij}{f'}\not\equiv 0$.)
  Then, ${g_{ij}} \sim {\partial}_{(12)(ij)}{f'}$, and by
    (\ref{equ-ij-12-in-D}), we have ${g_{ij}}\in{\mathcal{E}}^{\otimes}$. 
    Also for $\{k, \ell\}$ disjoint with both $\{i, j\}$
    and $\{1, 2\}$, ${\partial}_{(k\ell)(ij)}{f'}
    \not \equiv 0$ since
    ${\partial}_{(12)(k \ell)(ij)} {f'} = {\partial}_{(ij)(k \ell)(12)} {f'} 
    \not \equiv 0$.
    
    We first show that for any two pairs $\{i, j\}\neq \{k, \ell\}$ that are both disjoint with $\{1, 2\}$,
    ${b_{ij}}(x_1, x_2)\sim {b_{k\ell}}(x_1, x_2)$. 
    If $\{i, j\}$ is disjoint with $\{k, \ell\}$, then  ${b_{ij}}(x_1, x_2)\mid {\partial}_{(k\ell)(ij)}{f'}$ and ${b_{k\ell}}(x_1, x_2)\mid {\partial}_{(ij)(k\ell)}{f'}$. 
    Since ${\partial}_{(k\ell)(ij)}{f'}={\partial}_{(ij)(k\ell)}{f'}\not\equiv 0$, by \cref{unique}, 
    we have ${b_{ij}}(x_1, x_2)\sim {b_{k\ell}}(x_1, x_2)$. 
    Otherwise, $\{i, j\}$ and $\{k, \ell\}$ are not disjoint. Since ${f'}$ has arity $\geqslant 8$, we can find another pair of indices $\{s, t\}$ such that it is disjoint with $\{1, 2\}\cup \{i, j\} \cup \{k, \ell\}$. 
    Then, by the above argument, we have ${b_{ij}}(x_1, x_2)\sim {b_{st}}(x_1, x_2),$ and  ${b_{st}}(x_1, x_2)\sim {b_{k\ell}}(x_1, x_2).$
    Thus, ${b_{ij}}(x_1, x_2)\sim {b_{k\ell}}(x_1, x_2).$
    We can use a binary signature ${b}(x_1, x_2)$  to denote these binary signature ${b_{ij}}(x_1, x_2)$ for all $\{i, j\}$ disjoint with $\{1, 2\}$. 
    Then, ${b}(x_1, x_2) \mid {\partial}_{ij}{f'}$ for all $\{i, j\}$ disjoint with $\{1, 2\}$. 
    Also,  ${b}(x_1, x_2)$ is realizable  from $ f'$ by merging and factorization.
    Then $b^{-1}(x_1,x_2)$ is realizable by taking conjugation and renaming the two variables in $b.$
    
    Then, we consider the following gadget construction $G_2$ on ${f'}$. 
    By a slight abuse of variable names, 
    we connect the variable $x_1$ of ${f'}$ with the variable $x_2$ of $b^{-1}(x_1, x_2)$ and we get a signature ${f^\ast}$. 
    Clearly, $G_2$ is constructed by extending variables of ${f'}$. It does not change the irreducibility of ${f'}$.
    Thus, ${f^\ast}$ is irreducible.
    Again, we may assume that ${f^\ast}\in {\int}{\mathcal{U}}^{\otimes}$.
    Consider ${\partial}_{ij}{f^\ast}$ for all $\{i, j\}$ disjoint with $\{1, 2\}$. 
    Since the above gadget construction $G_2$ only touches the variable $x_1$ of $f'$, it commutes with the merging operation ${\partial}_{ij}$.
    Thus, ${\partial}_{ij}{f^\ast}$ can be realized by performing the gadget construction $G_2$ on ${\partial}_{ij}{f'}$, i.e.,
    connecting the binary signature $ b(x_1, x_2)$ in the UPF of ${\partial}_{ij}{f'}$ with itself, which changes $b(x_1, x_2)$ to $=_2$ up to some nonzero scalar $\lambda_{ij}$.
    Thus,  for all $\{i, j\}$ disjoint with $\{1, 2\}$, after renaming variables, we have
    $${\partial}_{ij}{f^\ast}=\lambda_{ij}\cdot(=_2)(x_1, x_2) \otimes {g_{ij}}\in {\mathcal{E}}^{\otimes}.$$
        Thus, ${\partial}_{ij}{f^\ast}\in {\mathcal{E}}^{\otimes}$ for all $\{i, j\}$ disjoint with $\{1, 2\}$. By our \cref{claim: strong equality property}, ${f^\ast}\in {\int}{\mathcal{E}}^{\otimes}$.
        We are done with {Case 1.}

      \vspace{1ex}
    \noindent{\bf Case 2.}  There is a pair of indices $\{i, j\}$ disjoint with $\{1, 2\}$ such that $x_1$ and $x_2$ appear in two distinct nonzero binary signatures ${b_1'}(x_1, x_u)$ and ${b_2'}(x_2, x_v)$ in the UPF of ${\partial}_{ij}{f'}$. 
 In other words, there exits $\{i, j\}$ such that  \begin{equation}
    {\partial}_{ij}{f'}={b'_1}(x_1, x_u)\otimes {b'_2}(x_2, x_v) \otimes {h_{ij}}, \text{ for some } {h_{ij}}\not\equiv 0.
    \end{equation}
    Since ${h_{ij}} \mid {\partial}_{(12)(ij)}{f'}$ and ${\partial}_{(12)(ij)}{f'}\in 
    {\mathcal{E}}^{\otimes}$, we have ${h_{ij}}\in {\mathcal{E}}^{\otimes}$.
    Also, after merging variables $x_1$ and $x_2$ in ${\partial}_{ij}{f'}$, variables $x_u$ and $x_v$ form a binary $=_2$ up to a nonzero scalar, since we already know ${\partial}_{(12)(ij)}{f'}\in 
    {\mathcal{E}}^{\otimes}$. 
    \textcolor{black}{By \cref{lem-binary-sim}, we have ${b'_1} \sim \overline{b'_2}$.
    Also, connecting the variable $x_u$ of ${b'_1}$ and the variable  $x_v$ of ${b'_2}$ will give the binary signature $\lambda \cdot=_{2}(x_1, x_2)$ as well.}
    
    We consider the following gadget construction $G_3$ on ${f'}$.
    By a slight abuse of variable names, 
    we connect variables $x_1$ and $x_2$ of ${f'}$ with the variable $x_1$ of $\overline{b'_1}(x_1,x_u)$ and $x_2$ of $\overline{b'_2}(x_2,x_v)$ respectively. 
    We get a signature ${f^{\ast}}$.
    Again, ${f^{\ast}}$ is irreducible since the gadget construction $G_3$ does not change the irreducibility of ${f'}.$
    Also, we may assume that ${f^{\ast}}\in{\int}{\mathcal{U}}^{\otimes}$.
    Otherwise, we are done.
    Consider ${\partial}_{ij}{f^{\ast}}$.
    Similarly, by the commutitivity of the  gadget construction $G_3$ and the merging gadget ${\partial}_{ij}$, 
    ${\partial}_{ij}{f^{\ast}}$ can be realized by connecting variables $x_1$ and $x_2$ of ${\partial}_{ij}{f'}$ with the variable $x_1$ of $\overline{b'_1}$ and the variable $x_2$ of $\overline{b'_2}$ respectively.
    After renaming variables, we have
    \begin{equation}\label{form_ij}
        {\partial}_{ij}{f^{\ast}}=\lambda_{ij} \cdot (=_2)(x_1, x_u)\otimes (=_2) \left(x_2, x_v\right) \otimes  h_{ij}\in \mathcal{E}^{\otimes}.
    \end{equation}

We now show that  ${\partial}_{12}{f^{\ast}}\in {\mathcal{E}}^{\otimes}$. Note that it is realized in the following way; 
we first connect variables $x_1$ and $x_2$ of ${f'}$ with the variable $x_1$ of $\overline{b'_1}(x_1, x_u)$ and the variable $x_2$ of $\overline{b'_2}(x_2, x_v)$ respectively to get ${f^{\ast}}$, and then after renaming variables $x_u$ and $x_v$ to $x_1$ and $x_2$ respectively, we merge them using $=_2$ (see Figure~\ref{fig:f^ast-f'}(a)). 
By 
associativity of gadget constructions, we can change the order; 
we first connect the variable $x_u$ of ${b'_1}(x_1, x_u)$ with the variable $x_v$ of ${b'_2}(x_2, x_v)$ (using $=_2$), and then we use the resulting binary signature to connect variables $x_1$ and $x_2$ of ${f'}$ (edges are connected using $=_2$). 
Note that connecting $x_u$ of ${b'_1}(x_1, x_u)$ with $x_v$ of ${b'_2}(x_2, x_v)$ gives $\lambda\cdot (=_2)$ up to a nonzero scalar $\lambda$, and $\lambda\cdot (=_2)$ is unchanged by extending both of its two variables with $=_2$ (see Figure~\ref{fig:f^ast-f'}(b)).
Thus, ${\partial}_{12}{f^{\ast}}$ is actually realized by merging $x_1$ and $x_2$ of ${f'}$ up to a nonzero scalar. 
Thus, we have ${\partial}_{12}{f^{\ast}}\sim {\partial}_{12}{f'},$ and hence  ${\partial}_{12}{f^{\ast}}\in {\mathcal{E}}^{\otimes}$,
by the form (\ref{eqn:lm5-partial12-f}) of ${\partial}_{12}{f'}$.

\begin{figure}[!h]
    \centering
    \includegraphics[height=4.8cm]{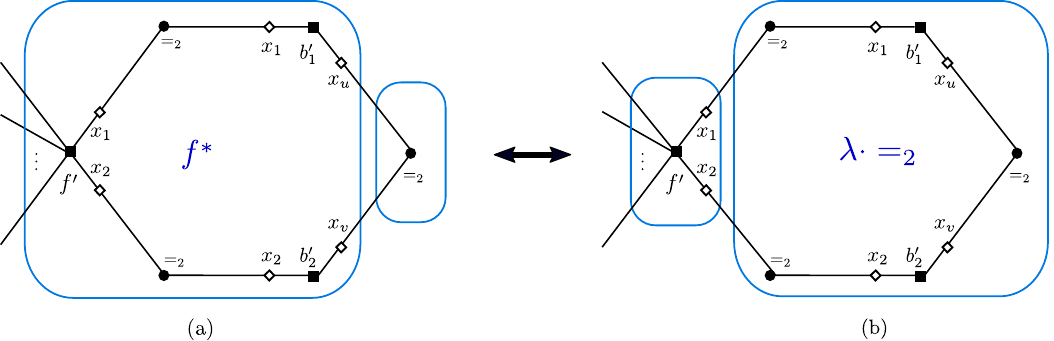}
    \caption{\cite{realholant} Gadget constructions of ${\partial}_{12}{{f^\ast}}$  and ${\partial}_{12}{{f'}}$}
    \label{fig:f^ast-f'}
\end{figure}

Then, we show that  ${\partial}_{st}{f^{\ast}}\in {\mathcal{E}}^{\otimes}$ for all pairs of indices $\{s, t\}$ disjoint with $\{1, 2, i, j\}$ and $\{s, t\} \neq \{u, v\}$ where $u$ and $v$ are named in (\ref{form_ij}). 
Clearly, ${\partial}_{st}{f^{\ast}}\not\equiv 0$ since ${\partial}_{(st)(12)}{f^{\ast}}\in{\mathcal{E}}^{\otimes}$.
 We first show that in the UPF of ${\partial}_{st}{f^{\ast}}$, $x_1$ and $x_2$ appear in two distinct nonzero binary signatures. 
 Otherwise, for a contradiction, suppose that there is a nonzero binary signature ${b^\ast}(x_1, x_2)$ such that ${b^\ast}(x_1, x_2)\mid {\partial}_{st}{f^{\ast}}$. 
 Then, ${b^\ast}(x_1, x_2)\mid{\partial}_{(ij)(st)}{f^{\ast}}={\partial}_{(st)(ij)}{f^{\ast}}\not\equiv 0$. 
 By the form (\ref{form_ij}) of ${\partial}_{ij}{f^{\ast}}$,  
 the only way that $x_1$ and $x_2$ can form a nonzero binary signature in ${\partial}_{(st)(ij)}{f^{\ast}}$ is that the merging gadget is actually merging $x_u$ and $x_v$. Thus, $\{s, t\}=\{u, v\}$. Contradiction. 
Therefore, for some $i'$ and $j'$, we have 
\begin{equation}\label{form_st}
{\partial}_{st}{f^{\ast}}={b^\ast_{st1}}(x_1, x_{i'})\otimes {b^\ast_{st2}}(x_2, x_{j'})\otimes{h_{st}},
\end{equation}
for some ${b^\ast_{st1}}(x_1, x_{i'}),  {b^\ast_{st2}}(x_2, x_{j'}), {h_{st}} \not \equiv 0$ since  ${\partial}_{st}{f^{\ast}}\not\equiv 0$.
  Since ${h_{st}} \mid {\partial}_{(12)(st)}{f^\ast}$ and ${\partial}_{(12)(st)}{f^\ast}\in 
    {\mathcal{E}}^{\otimes}$, 
we have ${h_{st}} \in {\mathcal{E}}^{\otimes}$.
Also, by \cref{lem-binary-sim}, ${b^\ast_{st1}}  \sim \overline{b^\ast_{st2}}$.
For a contradiction,
suppose that ${\partial}_{st}{f^{\ast}}\notin {\mathcal{E}}^{\otimes}$, then  ${b^\ast_{st1}}(x_1, x_{i'}) 
\not\sim (=_2)$,
and ${b^\ast_{st2}}(x_2, x_{j'})  
\not\sim (=_2)$.
Consider the signature ${\partial}_{(st)(ij)}{f^{\ast}}$. 
Since $\{s, t\}\neq \{u, v\}$,  by the form (\ref{form_ij}) of ${\partial}_{ij}{f^{\ast}}$, 
$x_1$ and $x_2$ appear in two binary signatures in the UPF of ${\partial}_{(st)(ij)}{f^{\ast}}.$
Remember that ${\partial}_{(st)(ij)}{f^{\ast}}={\partial}_{(ij)(st)}{f^{\ast}}$.
By the form (\ref{form_st}) of ${\partial}_{st}{f^{\ast}}$, 
if
$\{i', j'\}= \{i, j\}$, 
then, after merging $x_i$ and $x_j$ of ${\partial}_{st}{f^{\ast}}$, $x_1$ and $x_2$ will form a new binary  signature in ${\partial}_{(ij)(st)}{f^{\ast}}$. Contradiction.
Thus, $\{i', j'\}\neq \{i, j\}$.
Then, when merging $x_i$ and $x_j$ of ${\partial}_{st}{f^{\ast}}$,
among ${b^\ast_{st1}}(x_1, x_{i'})$ and ${b^\ast_{st2}}(x_2, x_{j'})$, at least one binary signature is untouched. 
Thus, ${\partial}_{(ij)(st)}{f^{\ast}}$ has a factor that is not an associate of $=_2$. 
Contradicting ${\partial}_{(ij)(st)}{f^{\ast}}\in {\mathcal{E}}^{\otimes}$,
which is a consequence of (\ref{form_ij}).
Thus, ${\partial}_{st}{f^{\ast}}\in {\mathcal{E}}^{\otimes}$.

Then, we show that ${\partial}_{uv}{f^{\ast}} \in {\mathcal{E}}^{\otimes}$. 
Recall the form (\ref{form_ij}) of ${\partial}_{ij}{f^{\ast}}$. Clearly, $\{u, v\}$ is disjoint with $\{1, 2, i, j\}$.
Also, ${\partial}_{uv}{f^{\ast}}\not\equiv 0$ since ${\partial}_{(ij)(uv)}{f^{\ast}} \in {\mathcal{E}}^{\otimes}.$
Consider the UPF of ${\partial}_{uv}{f^{\ast}}$.
\begin{itemize}
    \item 
If $x_1$ and $x_2$ appear in one nonzero binary signature ${b^\ast_{uv}}(x_1, x_2)$, then $${\partial}_{uv}{f^{\ast}}={b^\ast_{uv}}(x_1, x_2)\otimes {g_{uv}} ~~~~\text{ for some } {g_{uv}} \not\equiv 0.$$
Then, we have
 ${g_{uv}} \sim
 {\partial}_{(12)(uv)}{f^{\ast}} \in 
{\mathcal{E}}^{\otimes}$ since ${\partial}_{12}{f^{\ast}}\in {\mathcal{E}}^{\otimes}$.
Also, since ${b^\ast_{uv}}(x_1, x_2)\mid {\partial}_{(ij)(uv)}{f^{\ast}}\in {\mathcal{E}}^{\otimes}$, we have ${b^\ast_{uv}}(x_1, x_2)\in {\mathcal{E}}^{\otimes}$. Hence, ${\partial}_{uv}{f^{\ast}}\in {\mathcal{E}}^{\otimes}$.
\item If $x_1$ and $x_2$ appear in two distinct nonzero binary signatures ${b^\ast_{uv1}}(x_1, x_{i'})$ and ${b^\ast_{uv2}}(x_2, x_{j'})$, then
 $${\partial}_{uv}{f^{\ast}}={b^\ast_{uv1}}(x_1, x_{i'})\otimes {b^\ast_{uv2}}(x_2, x_{j'})\otimes{h_{uv}} ~~~~\text{ for some } {h_{uv}}\not\equiv 0.$$
 
Then, we have $ {h_{uv}} \in {\mathcal{E}}^{\otimes}$ since ${\partial}_{(12)(uv)}{f^{\ast}} \in 
{\mathcal{E}}^{\otimes}$. 
By the form (\ref{form_ij}) of ${\partial}_{ij}{f^{\ast}}$, 
after merging variables $x_u$ and $x_v$ of  ${\partial}_{ij}{f^{\ast}}$,
variables $x_1$ and $x_2$ form a binary $=_2$ in ${\partial}_{(uv)(ij)}{f^{\ast}}={\partial}_{(ij)(uv)}{f^{\ast}}$. 
On the other hand,
by the form of ${\partial}_{uv}{f^{\ast}}$, the only way that $x_1$ and $x_2$  form a binary after merging two variables in ${\partial}_{uv}{f^{\ast}}$ is to merge  $x_{i'}$ and $x_{j'}$. 
Thus, we have $\{i', j'\}=\{i, j\}$. Since ${f^{\ast}}$ has arity $2n \geqslant 8$, we can find another pair of indices $\{s, t\}$ disjoint with $\{1, 2, i, j, u, v\}$. 
When merging variables $x_s$ and $x_t$ in ${\partial}_{uv}{f^{\ast}}$, 
binary signatures ${b^\ast_{uv1}}(x_1, x_{i'})$ and ${b^\ast_{uv2}}(x_2, x_{j'})$ are untouched. 
Thus, we have ${b^\ast_{uv1}}(x_1, x_{i'})\otimes {b^\ast_{uv2}}(x_2, x_{j'})\mid {\partial}_{(st)(uv)}{f^{\ast}}$.
As showed above, we have ${\partial}_{st}{f^{\ast}}\in {\mathcal{E}}^{\otimes}$ and then  ${\partial}_{(st)(uv)}{f^{\ast}}\in {\mathcal{E}}^{\otimes}$.
Thus, ${b^\ast_{uv1}}(x_1, x_{i'})\otimes {b^\ast_{uv2}}(x_2, x_{j'})\in {\mathcal{E}}^{\otimes}$ and then ${\partial}_{uv}{f^{\ast}}\in {\mathcal{E}}^{\otimes}$.
\end{itemize}

So far, we have shown that ${\partial}_{12}{f^{\ast}}\in {\mathcal{E}}^{\otimes}$, ${\partial}_{ij}{f^{\ast}}\in {\mathcal{E}}^{\otimes}$ and ${\partial}_{st}{f^{\ast}}\in {\mathcal{E}}^{\otimes}$ for all $\{s, t\}$ disjoint with $\{1, 2, i, j\}$.
If we can further show that ${\partial}_{ik}{f^{\ast}}\in {\mathcal{E}}^{\otimes}$ for all $k\neq 1, 2, i, j$, and then symmetrically ${\partial}_{jk}{f^{\ast}}\in {\mathcal{E}}^{\otimes}$ for all $k\neq 1, 2, i, j$, then 
${\partial}_{st}{f^{\ast}}\in {\mathcal{E}}^{\otimes}$ for all $\{s, t\}$  disjoint with $\{1, 2\}$.
Thus, by our \cref{claim: strong equality property},  ${f^\ast}\in {\int}\mathcal{E}^{\otimes}$. This will finish the proof of Case 2.


Now we prove ${\partial}_{ik}{f^{\ast}}\in {\mathcal{E}}^{\otimes}$ for all $k\neq 1, 2, i, j$. 
Since ${\partial}_{(ik)(12)}{f^{\ast}}\in \mathcal{E}^{\otimes}$,
we have  ${\partial}_{ik}{f^{\ast}}\not\equiv 0$. 
So we can consider the UPF
of ${\partial}_{ik}{f^{\ast}}$. 
\begin{itemize}
    \item 
If $x_1$ and $x_2$ appear in one nonzero binary signature, then  $${\partial}_{ik}{f^{\ast}}={b^\ast_{ik}}(x_1, x_2)\otimes {g_{ik}} ~~~~\text{ for some } {g_{ik}}\in \mathcal{E}^{\otimes}.$$ 
Here, ${g_{ik}}\in \mathcal{E}^{\otimes}$
since ${\partial}_{(ik)(12)}{f^{\ast}}
\in \mathcal{E}^{\otimes}$.
Since ${f^\ast}$ has arity $2n \geqslant 8$,
we can pick a pair of indices $\{s, t\}$ disjoint with $\{1, 2, i, j, k\}$, and merge variables $x_s$ and $x_t$ of ${\partial}_{ik}{f^{\ast}}$.
Then, ${b^\ast_{ik}}(x_1, x_2)\mid {\partial}_{(st)(ik)}{f^{\ast}}.$
Since ${\partial}_{st}{f^{\ast}}\in \mathcal{E}^{\otimes}$,
${\partial}_{(st)(ik)}{f^{\ast}}={\partial}_{(ik)(st)}{f^{\ast}}\in {\mathcal{E}}^{\otimes}$. Thus, ${b^\ast_{ik}}(x_1, x_2) \in {\mathcal{E}}^{\otimes}$ and then ${\partial}_{ik}{f^{\ast}}\in {\mathcal{E}}^{\otimes}$.
\item    If $x_1$ and $x_2$ appear in two nonzero distinct binary signatures,
then $${\partial}_{ik}{f^{\ast}}={b^\ast_{ik1}}(x_1, x_p)\otimes {b^\ast_{ik2}}(x_2, x_q)\otimes{h_{ik}} ~~~~\text{ for some } {h_{ik}} \in \mathcal{E}^{\otimes}.$$
Again, here ${h_{ik}} \in \mathcal{E}^{\otimes}$
since ${\partial}_{(ik)(12)}{f^{\ast}}
\in \mathcal{E}^{\otimes}$.
By connecting variables $x_1$ and $x_2$ of ${\partial}_{ik}{f^{\ast}}$, $x_p$ and $x_q$ will form a binary $=_2$ up to a nonzero scalar
(this binary signature is $=_2$
because we know that
 ${\partial}_{(ik)(12)}{f^{\ast}}\in \mathcal{E}^{\otimes}$).
By \cref{lem-binary-sim}, 
as the type of binary signatures,
${{b^\ast_{ik1}}} \sim \overline{b^\ast_{ik2}}.$
Between $x_p$ and $x_q$, at least one of them is not $x_j$; suppose that it is $x_p$. 
We pick a variable $x_r$ in the scope of ${h_{ik}}$ that is also not $x_j$ (such a  variable $x_r$  exists as $2n \geqslant 8$). Then, by merging $x_p$ and $x_r$ of ${\partial}_{ik}{f^{\ast}}$, 
the binary signature ${b^\ast_{ik2}}(x_2, x_q)$ is untouched. 
Since $\{p, r\}$ is disjoint with $\{1, 2, i, j\}$, we have ${b^\ast_{ik2}}(x_2, x_q)\mid {\partial}_{(ik)(pr)}{f^{\ast}}\in {\mathcal{E}}^{\otimes}$. Thus, we have ${b^\ast_{ik2}}(x_2, x_q)
\in {\mathcal{E}}^{\otimes}$ and so does
${{b^\ast_{ik1}}}(x_1, x_p)$. 
Thus, ${\partial}_{ik}{f^{\ast}}\in {\mathcal{E}}^{\otimes}$.
    \end{itemize}

As remarked earlier,  by symmetry, we also have ${\partial}_{jk}{f^{\ast}}\in {\mathcal{E}}^{\otimes}$ for all $k\neq 1, 2, i, j$. Thus, we are done with Case 2.

Thus, an irreducible signature ${f^\ast}\in \int\mathcal{E}^{\otimes}$ of arity $2n$ is realized from $ f$.
\end{proof}

We next determine the exact form of an arity-$8$ signature with the strong equality property. 
By \cref{lm: unitary wire normalization}, we may assume that $\mathcal{F}$ contains an arity $8$ signature $h$ satisfying the strong equality property.
Otherwise, the induction is completed by the first two items in \cref{lm: unitary wire normalization}.

\begin{definition}\label{def: f8 matching}
    Let $h\in \mathcal{F}$ be an arity 8 signature satisfying the strong equality property.
    By \cref{def: Strong equality property}, for any pair of distinct indices $e=\{i,j\}$, $$\partial_{e}h=\partial_{ij}h=\lambda_{ij}\cdot (=_2)(x_{i_1},x_{j_1})\cdot(=_2)(x_{i_2},x_{j_2})\cdot (=_2)(x_{i_3},x_{j_3}), \quad \lambda_{ij}\in \C^\times,$$
    and $\{i,j\}\cup \{i_1,j_1\}\cup\{i_2,j_2\}\cup\{i_3,j_3\}=\{1,2,\ldots, 8\}.$
    We define $M_e=\{\{i_1,j_1\},\{i_2,j_2\},\{i_3,j_3\}\}$ and $F_e=\{i,j\}\cup M_e.$
    For two pairs of distinct indices $e_1=\{i,j\}$ and $e_2=\{k,\ell\}$, define $e_1\sim e_2$ if $e_1\in F_{e_2}$.
\end{definition}

\begin{lemma}\label{lm: f8 matching same norm}
    The $\lambda_{ij}$ in \cref{def: f8 matching} have the same norm for every pair of distinct indices $\{i,j\}$, otherwise $\Holant(\mathcal{F})$ is \#P-hard.
\end{lemma}
\begin{proof}
    We may assume $h$ is irreducible and satisfies {\sc 2nd-Orth}, otherwise we obtain $\Holant(\mathcal{F})$ is \#P-hard.
    By \cref{lm: 2nd-orth contraction nonzero}, $|\partial_{ij}f|^2=\lambda\cdot |(=_2)|^2=2\lambda$ for every pair of distinct indices $\{i,j\}$, where $\lambda=|f_{ij}^{00}|^2=|f_{ij}^{01}|^2=|f_{ij}^{10}|^2=|f_{ij}^{11}|^2>0$.
    On the other hand, $|\partial_{ij}f|^2=|\lambda_{ij}|^2\cdot \prod_{k=1}^3|(=_2)(x_{i_k},x_{j_k})|^2=8|\lambda_{ij}|^2.$
    Therefore, $|\lambda_{ij}|$ are all the same.
\end{proof}

\begin{lemma}\label{lm: f8 matching equivalence}
    The relation ``$\sim$'' in \cref{def: f8 matching} is an equivalence relation, otherwise $\Holant(\mathcal{F})$ is \#P-hard.
\end{lemma}
\begin{proof}
    First $M_e$ and $F_e$ are well defined for any pair of indices $e=\{i,j\}$, by the UPF of $\partial_eh.$
    Clearly $e\in F_e$, so $\sim$ is self-inversive.
    
    Now suppose $d=\{k,\ell\}$ and $e=\{i,j\}$ are two pairs of distinct indices such that $d\neq e$ and $d\sim e$. 
    Then $d\in M_e$ and $d\cap e=\emptyset.$
    Suppose $M_e=\{\{k,\ell\},\{s,t\},\{u,v\}\}.$
    Then
    \begin{equation}\label{eq: contract h using e}
        \partial_{ij}h=\lambda_{ij}\cdot (=_2)(x_{k},x_{\ell})\cdot (=_2)(x_{s},x_{t})\cdot (=_2)(x_{u},x_{v}).
    \end{equation}
    Suppose $M_d=\{\{i_1,j_1\},\{i_2,j_2\},\{i_3,j_3\}\}$, then 
    \begin{equation}\label{eq: contract h using d}
        \partial_{k\ell}h=\lambda_{k\ell}\cdot (=_2)(x_{i_1},x_{j_1})\cdot (=_2)(x_{i_2},x_{j_2})\cdot (=_2)(x_{i_3},x_{j_3}).
    \end{equation}
    Consider $\partial_{(k\ell)(ij)}h=\partial_{(ij)(k\ell)}h$.
    By \eqref{eq: contract h using e}, we have 
    \begin{equation}\label{eq: contracting h by e then d}
        \partial_{(k\ell)(ij)}h=2\lambda_{ij}\cdot (=_2)(x_{s},x_{t})\cdot (=_2)(x_{u},x_{v}).
    \end{equation} 
    Suppose for a contradiction that $e\notin M_d$.
    Then $x_i$ and $x_j$ appear in two distinct binary factors in \eqref{eq: contract h using d}. 
    Without loss of generality, $i=i_1$ and $j=j_2$.
    Merging $x_i$ and $x_j$ results in 
    \begin{equation}\label{eq: contracting h using d then e}
        \partial_{(ij)(k\ell)}h=\lambda_{k\ell}\cdot (=_2)(x_{j_1},x_{i_2})\cdot (=_2)(x_{i_3},x_{j_3}).
    \end{equation}
    Compare \eqref{eq: contracting h by e then d} and \eqref{eq: contracting h using d then e}, we have $2\lambda_{ij}=\lambda_{k\ell}.$
    However, by \cref{lm: f8 matching same norm}, $|\lambda_{ij}|=|\lambda_{k\ell}|>0$, contradiction.
    Therefore, $e\in M_d$.
    Without loss of generality, $i=i_1$ and $j=j_1$.
    Then 
    \begin{equation}\label{eq: contracting h using d then e 2}
        \partial_{(ij)(k\ell)}h=2\lambda_{k\ell}\cdot (=_2)(x_{i_2},x_{j_2})\cdot (=_2)(x_{i_3},x_{j_3}).
    \end{equation}
    Compare \eqref{eq: contracting h by e then d} and \eqref{eq: contracting h using d then e 2}, we have $\lambda_{ij}=\lambda_{k\ell}.$
    Thus, $\sim$ is reflective.
    Moreover, if $d\in M_e$, then $\lambda_d=\lambda_e.$

    Finally we show transitivity.
    Suppose $c=\{s,t\},d=\{k,\ell\},e=\{i,j\}$, and $c\sim d,\,d\sim e$.
    If any two of them are equal, then trivially $c\sim e.$
    Now assume $c,d,e$ are mutually distinct.
    We have $c\in M_d$ and $d\in M_e$.
    So $\{s,t\}\cap\{k,\ell\}=\emptyset$ and $\{k,\ell\} \cap \{i,j\}=\emptyset.$
    Suppose 
    \begin{equation}\label{eq: transitivity partial e h}
        \partial_{ij}h=\lambda_{ij}\cdot (=_2)(x_k,x_{\ell})\cdot (=_2)\cdot (=_2).
    \end{equation}
    By reflectivity, $e\in M_d$.
    Also $c\in M_d$ and $c\neq e$.
    So $\{k,\ell\}\cap\{i,j\}=\emptyset$ and
    \begin{equation}\label{eq: transitivity partial d h}
        \partial_{k\ell}h=\lambda_{k\ell}\cdot (=_2)(x_i,x_{j})\cdot (=_2)(x_s,x_t)\cdot (=_2).
    \end{equation}
    Now consider $\partial_{(st)(ij)(k\ell)}h=\partial_{(st)(k\ell)(ij)}h$.
    By \eqref{eq: transitivity partial d h}, $\partial_{(st)(ij)(k\ell)}h=4\lambda_{k\ell}\cdot(=_2)$.
    Suppose for a contradiction that $c\notin M_e$.
    Then $s$ and $t$ appear in the second and third binary factors respectively in \eqref{eq: transitivity partial e h}.
    So $\partial_{(st)(k\ell)(ij)}h=2\lambda_{ij}\cdot(=_2)$.
    We obtain $4\lambda_{k\ell}=2\lambda_{ij}>0$, contradiction.
    So $c\in M_e$ and thus $c\sim e$.
    
    We have proved $\sim$ is self-inversive, reflexive and transitive. 
    So it is an equivalence relation. 
\end{proof}

\begin{lemma}\label{lm: f8 same coe}
    The $\lambda_{ij}$ in \cref{def: f8 matching} are all the same, otherwise $\Holant(\mathcal{F})$ is \#P-hard.
\end{lemma}
\begin{proof}
    In the proof of \cref{lm: f8 matching equivalence}, we have proved $\lambda_{ij}$ are all the same for $\{i,j\}$ in the same equivalent class of $\sim.$
    We only need to show $\lambda_{ij}=\lambda_{k\ell}$ for $\{i,j\}\not \sim \{k,\ell\}.$
    Let $e=\{i,j\},\,d=\{k,\ell\}.$
    We have $\partial_{ij}h=\lambda_{ij}\cdot (=_2)^{\otimes 3}$
    and $\partial_{k\ell}h=\lambda_{k\ell}\cdot (=_2)^{\otimes 3}$.
    Since $\{i,j\}\not \sim\{k,\ell\}$, we know $d\notin M_e$ and $e\notin M_d.$
    So $x_k$ and $x_\ell$ appear in distinct binary factors of $\partial_{ij}h$.
    Then $\partial_{(k\ell)(ij)}h=\lambda_{ij}\cdot (=_2)^{\otimes 2}.$
    Similarly, $x_i$ and $x_j$ appear in distinct binary factors of $\partial_{k\ell}h$, and $\partial_{(ij)(k\ell)}h=\lambda_{k\ell}\cdot (=_2)^{\otimes 2}.$
    By $\partial_{(k\ell)(ij)}h=\partial_{(ij)(k\ell)}h$, we have $\lambda_{ij}=\lambda_{k\ell}$.
\end{proof}

An equivalent class of $\sim$ in \cref{def: f8 matching} can be viewed as a perfect matching of the complete graph $K_8$ with vertices labeled by $\{1,2,\ldots, 8\}.$
Since $K_8$ has 28 edges, and every equivalent class has size $4$, there are 7 equivalent classes in total.
They are 7 perfect matchings, or a 1-factorization of $K_8$.

\begin{lemma}
\label{lm: f8 one factorization}
After identifying the vertices of $K_8$ with the points of
$\mathbb F_2^3$, the seven perfect matchings determined by the equivalent classes of $\sim$ are
\begin{equation}
    F_e=\bigl\{\{z,z+v_e\}\mid z\in\mathbb F_2^3\bigr\}, \qquad v_e\in\mathbb F_2^3\setminus\{0\}.
 \label{eq: translation matchings}
\end{equation}
\end{lemma}

\begin{proof}
Let $\mathcal{PM}$ be the 7 perfect matchings determined by the equivalent classes of $\sim$.
The union of two perfect matchings in $\mathcal{PM}$ is either one $8$-cycle or two $4$-cycles.  
The first case is impossible.  Indeed, label an alternating $8$-cycle so that the two perfect matchings are
\[
 F_1=\{\{1,2\},\{3,4\},\{5,6\},\{7,8\}\},\qquad
 F_2=\{\{2,3\},\{4,5\},\{6,7\},\{8,1\}\}.
\]
Merging first $x_1,x_2$ and then $x_4,x_5$ leaves the residual matching
$\{\{3,6\},\{7,8\}\}$.
By commutativity, it is equivalent to first merge $x_4,x_5$ and then $x_1,x_2$, leaving $\{\{3,8\}, \{6,7\}\}$.  
This contradicts the UPF of $\partial_{(12)(34)}h$.
Thus every two perfect matchings form two $4$-cycles.

For every perfect matching $F\in\mathcal{PM}$, we identify $F$ with a map $m_F:V(K_8)\to V(K_8)$, $u\mapsto v$ if $\{u,v\}\in F.$
So $m_F$ is an involution map, i.e., $m_F^{(2)}(u)=m_F(m_F(u))=u,\,\forall u\in [8].$
Let $G=\{m_F\mid F\in \mathcal{PM}\}\cup \{\mathrm{id}\}$.
Consider two perfect matchings $F_1,F_2\in \mathcal{PM}$.
Since $F_1\cup F_2$ form two 4-cycles of $K_8$, we have $m_{F_1}m_{F_2}=m_{F_2}m_{F_1}$.
Now fix one vertex $z_0\in V(K_8)$.
Given two maps $m_1,m_2\in G$, let $m_3$ be the unique member in $G$ satisfying $m_3(z_0)=m_1m_2(z_0)$.
This $m_3$ exists, because if $m_1=m_2$ then $m_3=\mathrm{id}$, otherwise it is the map determined by the matching containing the edge $\{z_0,m_1m_2(z_0)\}.$
For every $m\in G$, commutativity gives $m_3(m(z_0))=m(m_3(z_0))=m(m_1m_2(z_0))=m_1m_2(m(z_0))$.  
The eight points $\{m(z_0)\mid m\in G\}$ exhaust the vertex set, so $m_3=m_1m_2$.  
Thus, $G$ is an abelian group of
order eight, and every non-zero element in $G$ has order $2$.
Hence $G\cong \mathbb F_2^3$. 
Now we identify the vertex $m(z_0)\in V(K_8)$ with $m\in G\cong \mathbb F_2^3,\,\forall m\in G.$
Then the seven non-zero elements in $G$ are exactly \eqref{eq: translation matchings}.
\end{proof}

In the following, we give an interpretation of the signature $f_8$ through the lens of the Reed-Muller code $RM(1,3).$
Let
\begin{equation}
 C=RM(1,3)
 =\left\{(a_0+a\cdot z)_{z\in\mathbb F_2^3}\mid 
 a_0\in\mathbb F_2,\ a\in\mathbb F_2^3\right\}.
 \label{eq: RM code definition}
\end{equation}
The coordinates are ordered as
$000,001,010,011,100,101,110,111$.  Thus $C$ is the self-dual
$[8,4,4]$ first order Reed-Muller code, with codewords
\begin{align*}
&00000000,00001111,00110011,00111100,01010101,01011010,\notag\\
&01100110,01101001,10010110,10011001,10100101,10101010,\notag\\
&11000011,11001100,11110000,11111111.
\end{align*}
Therefore, $f_8=\chi_C.$

\begin{lemma}\label{lm: f8 indistinguishble}
    Let $h\in \mathcal{F}$ be an arity $8$ signature satisfying the strong equality property.
    Then there exists $\lambda\in \mathbb C^\times$ such that $\partial_{ij}h= \lambda\cdot\partial_{ij}f_8$ for every pair of distinct indices $\{i,j\},$ otherwise $\Holant(\mathcal{F})$ is \#P-hard.
\end{lemma}
\begin{proof}
    By \cref{lm: f8 same coe}, unless $\Holant(\mathcal{F})$ is \#P-hard, there exists $\lambda\in \mathbb C^\times$ such that 
    \begin{equation*}
        \partial_{e}h=\lambda\cdot\bigotimes_{d\in M_e}(=_2)_d
    \end{equation*}
    for every pair of indices $e=\{i,j\}$, where $(=_2)_d$ is the binary $=_2$ on the two variables corresponding to the two indices in $d.$
    By \cref{lm: f8 one factorization}, we may label $\{1,2,\ldots, 8\}$ by elements in $\mathbb F_2^3$, such that $F_e=\{\{z,z+v_e\}\mid z\in \mathbb{F}_2^3\}$ for some $v_e\in \mathbb F_2^3\setminus\{0\}.$
    Since $e\in F_e$, there exists some $z_0\in \mathbb F_2^3$ such that $e=\{z_0,z_0+v_e\}.$
    Suppose $M_e=\{\{z_1,z_1+v_e\},\{z_2,z_2+v_e\},\{z_3,z_3+v_e\}\}$ for some $z_1,z_2,z_3\in \mathbb F_2^3.$
    Then 
    \begin{equation}\label{eq: f8 indistinguishble partial e h}
        \partial_{e}h=\lambda\cdot\bigotimes_{i=1}^3(=_2)(x_{z_i},x_{z_i+v_e}).
    \end{equation}
    Now consider $\partial_{e}f_8$.
    Since $e=\{z_0,z_0+v_e\}$, merging $x_{z_0}$ and $x_{z_0+v_e}$ requires $x_{z_0}=x_{z_0+v_e}.$
    By the definition of $C=RM(1,3)$ and $f_8=\chi_C$, the $y$-th coordinate of the codeword indexed by $(a_0,a)$ is $x_y=a_0+a\cdot y.$
    So $x_{z_0}=x_{z_0+v_e}\iff a_0+a\cdot z_0=a_0+a\cdot (z_0+v_e)\iff a\cdot v_e=0.$
    Once $a\cdot v_e=0$, we have $x_{y+v_e}=a_0+a\cdot (y+v_e)=a_0+a\cdot y=x_y$, for every $y\in \mathbb F_2^3.$
    Therefore, $x_{z_i}=x_{z_i+v_e}$ in $\partial_{e}f_8$ for $i=1,2,3.$
    Direct computation shows that merging any two variables of $f_8$, it becomes a tensor product of three binary equalities \cite{realholant}.
    So
    \begin{equation}\label{eq: f8 indistinguishble partial e f8}
        \partial_{e}f_8=\bigotimes_{i=1}^3(=_2)(x_{z_i},x_{z_i+v_e}).
    \end{equation}
    Comparing \eqref{eq: f8 indistinguishble partial e h} and \eqref{eq: f8 indistinguishble partial e f8} gives $ \partial_{ij}h=\lambda\cdot\partial_{ij}f_8$, for every pair of indices $\{i,j\}$.
\end{proof}

\begin{lemma}
\label{lm: f8 normalized equations}
Let $h\in \mathcal{F}$ be an arity 8 signature satisfying the strong equality property. 
After normalizing $h$,
unless $\Holant(\mathcal{F})$ is \#P-hard,
there exists
$\alpha,\beta\in\mathbb C$ such that
\begin{equation}
 h= f_8+\alpha u_+^{\otimes8}+\beta u_-^{\otimes8}.
 \label{eq: f8 kernel family}
\end{equation}
where $u_+=[1, \mathfrak i], u_-=[1, -\mathfrak i]$ are unary signatures.
Moreover, $h$ satisfies {\sc 2nd-Orth} if and only if
\begin{equation}
 \operatorname{Re}(\alpha)+8|\alpha|^2=0,
 \qquad
 \operatorname{Re}(\beta)+8|\beta|^2=0.
 \label{eq: f8 second orth circles}
\end{equation}
\end{lemma}

\begin{proof}
By \cref{lm: f8 indistinguishble}, there exists $\lambda\in \mathbb C^\times$ such that $\partial_{ij}h=\lambda\cdot \partial_{ij}f_8$ for every pair of distinct indices $\{i,j\}.$
We normalize $h$ by $\lambda$, and with a slight abuse of notation, still use $h$ to denote the normalized signature.
Let $h'=h-f_8$.
Then $\partial_{ij}h'=0$ for every pair of distinct indices $\{i,j\}.$
Notice that $\{u_+,u_{-}\}$ form a basis for $\C^2$.
So $\{u_{\epsilon_1}\otimes u_{\epsilon_2}\otimes\ldots \otimes u_{\epsilon_8}\mid \epsilon_1,\epsilon_2,\ldots,\epsilon_8\in\{+,-\}\}$ form a basis for $(\C^2)^8$.
The arity $8$ signature $h'$ can be viewed as a tensor in $(\C^2)^8.$
We expand $h'$ interms of this basis, and obtain
\begin{equation}\label{eq: expansion of h'}
    h'=\sum_{\epsilon_1,\ldots, \epsilon_8\in\{+,-\}}c'_{\epsilon_1,\ldots,\epsilon_8}u_{\epsilon_1}\otimes\ldots\otimes u_{\epsilon_8}=\sum_{S\subseteq[8]}c_S \bigotimes_{j\in S}u_+ \bigotimes_{j\notin S}u_{-},
\end{equation}
where $c_S=c'_{\epsilon_1,\ldots,\epsilon_8}$ if $S=\{i\in[8]\mid \epsilon_i=+\}.$
We note that
\begin{equation}
 u_+u_+^{\mathtt T}=u_-u_-^{\mathtt T}=0,\qquad
 u_+u_-^{\mathtt T}=u_-u_+^{\mathtt T}=2.
 \label{eq: isotropic unary basis}
\end{equation}
Now consider $\partial_{ij}h'.$
By \eqref{eq: isotropic unary basis}, every term in \eqref{eq: expansion of h'} with $\epsilon_i=\epsilon_j$ will vanish in $\partial_{ij}h'$, because merging $x_i$ and $x_j$ results in a factor $u_{\epsilon_i}u_{\epsilon_j}^{\tt T}.$
\begin{claim}\label{claim: Johnson graph}
    For every subset $T\subseteq[8]\setminus\{i,j\}$, the coefficients in \eqref{eq: expansion of h'} satisfy $c_{T\cup\{i\}}+c_{T\cup\{j\}}=0.$
\end{claim}
\begin{claimproof}{\cref{claim: Johnson graph}}
    Let $v_T:=\bigotimes_{j\in T} u_+\bigotimes_{j\in [8]\setminus\{i,j\}\setminus T}u_-.$
    After merging $x_i$ and $x_j$, $v_T$ is a term in the expansion of $\partial_{ij}h'$.
    We consider the coefficient of $v_T$ in $\partial_{ij}h'.$
    If a term in \eqref{eq: expansion of h'} survives in $\partial_{ij}h'$, then $\epsilon_i\neq\epsilon_j.$
    There are two possibilities, $(\epsilon_i,\epsilon_j)=(+,-)$ or $(\epsilon_i,\epsilon_j)=(-,+).$
    In the first case, the contribution for the coefficient of $v_T$ in $\partial_{ij}h'$ is $2c_{T\cup \{i\}}.$
    In the latter case, the contribution is $2c_{T\cup \{j\}}.$
    Hence the coefficient of $v_T$ in $\partial_{ij}h'$ is $2c_{T\cup \{i\}}+2c_{T\cup \{j\}}.$
    Since $\partial_{ij}h'=0$, and the terms $v_T$ are linearly independent, we have $c_{T\cup\{i\}}+c_{T\cup \{j\}}=0$, for every $T\subseteq[8]\setminus \{i,j\}.$
\end{claimproof}

\begin{claim}\label{claim: coe is 0 in expansion of h}
    In \eqref{eq: expansion of h'}, the coefficients $c_S=0$ if $1\le |S|\le 7$.
\end{claim}
\begin{claimproof}{\cref{claim: coe is 0 in expansion of h}}
    Fix \(1\le t\le 7\), and look only at coefficients \(c_S\) in \eqref{eq: expansion of h'} with $|S|=t$.
    Consider the Johnson graph $J(8,t)$, where the vertices are $t$-element subsets of [8], and two vertices are adjacent when the intersection of the two vertices (subsets) contains $t-1$ elements.
    By \cref{claim: Johnson graph}, whenever two subsets $S_1$ and $S_2$ are adjacent in $J(8,t)$, their coefficients $c_{S_1}$ and $c_{S_2}$ are negatives.
    We claim every vertex in $J(8,t)$ is contained in a triangle.
    To see this, pick an arbitrary vertex (subset) $S$.
    If $1\le t\le 6$, there are distinct elements $p,q,r\in [8]$ such that $p\in S$ and $q,r\notin S$.
    Let $S_1=S\setminus\{p\}\cup \{q\}$ and $S_1=S\setminus\{p\}\cup \{r\}$. 
    Then $S,S_1,S_2$ form a triangle in $J(8,t)$.
    If $t=7$, then $J(8,7)$ is the complete graph $K_8$, so every vertex is in a triangle.
    Along the triangle, we have $c_S+c_{S_1}=0\,,c_{S_1}+c_{S_2}=0\,,c_{S_2}+c_{S}=0\,$.
    So $c_S=0$.
\end{claimproof}
By \cref{claim: coe is 0 in expansion of h}, the expansion \eqref{eq: expansion of h'} simplifies to $h'=c_{\emptyset}u_+^{\otimes 8}+c_{[8]}u_-^{\otimes 8}.$
Let $\alpha=c_{\emptyset}$ and $\beta=c_{[8]}$, we have $h= f_8+\alpha u_+^{\otimes8}+\beta u_-^{\otimes8}.$

It remains to impose {\sc 2nd-Orth}.  The affine group of
$\mathbb F_2^3$ is transitive on pairs of coordinates, so it suffices to
compute the Gram matrix for one pair.  Put
\[
 E_-=2\operatorname{Re}(\alpha-\beta)
       +16(|\alpha|^2-|\beta|^2),\qquad
 E_+=2\operatorname{Re}(\alpha+\beta)
       +16(|\alpha|^2+|\beta|^2).
\]
For the four slices ordered as $00,01,10,11$, direct substitution into
$M_{12}(h)M_{12}(h)^\dagger$ gives
\begin{equation}
 D I_4+4
 \begin{bmatrix}
 0&-\mathfrak iE_-&-\mathfrak iE_-&-E_+\\
 \mathfrak iE_-&0&E_+&-\mathfrak iE_-\\
 \mathfrak iE_-&E_+&0&-\mathfrak iE_-\\
 -E_+&\mathfrak iE_-&\mathfrak iE_-&0
 \end{bmatrix},
 \label{eq: f8 slice Gram matrix}
\end{equation}
where
$D=4+8\operatorname{Re}(\alpha+\beta)
  +64(|\alpha|^2+|\beta|^2)$.
The matrix in \eqref{eq: f8 slice Gram matrix} is scalar precisely when
$E_-=E_+=0$, and these two equations are equivalent to
\eqref{eq: f8 second orth circles}.  Pair transitivity proves sufficiency
for every coordinate pair.
\end{proof}


\begin{lemma}
\label{lm: f8 self mating}
Let $h$ satisfy \eqref{eq: f8 kernel family} and
\eqref{eq: f8 second orth circles}.  
Unless $\Holant(\mathcal{F})$ is
\#P-hard, there exists $\theta\in\mathbb R$ such that
\begin{equation}
 h=R_\theta^{\otimes8}f_8,\qquad
 R_\theta=
 \begin{bmatrix}
 \cos\theta&\sin\theta\\
 -\sin\theta&\cos\theta
 \end{bmatrix}.
 \label{eq: f8 real orbit}
\end{equation}
\end{lemma}

\begin{proof}
Take two copies of $h$, join six corresponding variables using $(=_2)$,
and leave the first two variables of each copy dangling.  
(Notice that this is a self-mating, not mating $h$ and $\overline{h}$.)
Let $g_4$ be the
resulting four-ary signature.  With row order $00,01,10,11$ on the two
variables from the first copy and the same column order on those from the
second copy, substitution of \eqref{eq: f8 kernel family} gives
\begin{equation}\label{eq: f8 self mating matrix}
     M_{12,34}(g_4)=
     \begin{bmatrix}
     4+t&0&0&-t\\
     0&4+t&t&0\\
     0&t&4+t&0\\
     -t&0&0&4+t
     \end{bmatrix},\qquad t=8(\alpha+\beta+16\alpha\beta).
\end{equation}

The signature $g_4$ is nonzero.  By \cref{lm: base case 2n=4},  $g_4\in\mathcal U^{\otimes}$ unless $\Holant(\mathcal{F})$ is \#P-hard. 
An arity $4$ signature in $\mathcal U^{\otimes}$ has rank one in the $2|2$ flattening associated with its binary factorization. 
There are three possibilities, $M_{12,34}(g_4)$, or $M_{13,23}(g_4)$, or $M_{14,23}(g_4)$ has rank $1.$  
If $M_{12,34}(g_4)$ has rank 1, by evaluating all 2 by 2 minors in \eqref{eq: f8 self mating matrix}, we have $(4+t)^2=t^2=0$, which is impossible.
In the second case, evaluate all 2 by 2 minors in
    \[M_{13,24}(g_4)
    =
    \begin{bmatrix}
    4+t&0&0&4+t\\
    0&-t&t&0\\
    0&t&-t&0\\
    4+t&0&0&4+t
    \end{bmatrix},\]
we have $t(4+t)=0$.  
In the third case, 
    \[M_{14,23}(g_4)=
    \begin{bmatrix}
    4+t&0&0&t\\
    0&-t&4+t&0\\
    0&4+t&-t&0\\
    t&0&0&4+t
    \end{bmatrix}.\]
gives $(4+t)t=0$ and $t^2=(4+t)^2$, again impossible.

Thus, only second case is possible.
And $t=0$ or $t=-4.$
Now define
\(
 z=1+16\alpha,\, w=1+16\beta.
\)
The equations in \eqref{eq: f8 second orth circles} gives 
\(
|z|^2
=|1+16\alpha|^2
=1+32\operatorname{Re}(\alpha)+256|\alpha|^2
=1.
\)
Similarly, $|w|=1.$
Also, \eqref{eq: f8 self mating matrix} gives
\(
zw=(1+16\alpha)(1+16\beta)
=1+16(\alpha+\beta)+256\alpha\beta=1+2t.
\)
The alternative $t=-4$ would imply $|zw|=7$, a contradiction.
Thus $t=0$, $zw=1$, and $w=\overline z$. 
Hence
$\beta=\overline\alpha$.

Choose $\theta\in\mathbb R$ with $z=e^{8\mathfrak i\theta}$.
Then $\alpha=\frac{e^{8 \mathfrak i\theta}-1}{16}$ and $\beta=\frac{e^{-8 \mathfrak i\theta}-1}{16}$.
And
\begin{equation}
 h= f_8+\frac{e^{8 \mathfrak i\theta}-1}{16}u_+^{\otimes8}+\frac{e^{-8 \mathfrak i\theta}-1}{16} u_-^{\otimes8}.
 \label{eq: h expansion f8 u+ u-}
\end{equation}

Because the real orthogonal holographic transformation by $R_\theta$ preserves $=_2$, we have $\partial_{ij}(R_\theta f_8)=R_\theta\partial_{ij}f_8=R_\theta(=_2)^{\otimes 3}=(=_2)^{\otimes 3}=\partial_{ij}f_8$, for every pair of distinct indices $\{i,j\}.$ 
By \cref{lm: f8 normalized equations}, 
we know
\(R_\theta^{\otimes 8}f_8-f_8
=
A\,u_+^{\otimes 8}+B\,u_-^{\otimes 8}\) for some $A,B\in \mathbb C.$
Next we determine the coefficients $A$ and $B.$
In the following, view $u_+,u_-$ as vectors in $\mathbb C^2$, and $f_8$ as a vector in $\mathbb C^{256}.$
Consider \((u_-^{\otimes8})^{\tt T}\left(Au_+^{\otimes8}+Bu_-^{\otimes8}\right)=A\cdot 2^8=256A.\)
On the other hand, \((u_-^{\otimes8})^{\tt T}R_\theta^{\otimes8}f_8=e^{8i\theta}(u_-^{\otimes8})^{\tt T}f_8,\)
since $R_\theta u_\pm=e^{\pm\mathfrak i\theta}u_\pm$.
Direct calculation shows \((u_-^{\otimes8})^{\tt T}f_8=16.\)
So $(u_-^{\otimes 8})^{\tt T}(R_\theta^{\otimes 8}f_8-f_8)=16\cdot (e^{8 \mathfrak i\theta}-1)$.
Therefore, $A=\frac{e^{8\mathfrak i \theta}-1}{16}$.
Similarly, $B=\frac{e^{-8\mathfrak i \theta}-1}{16}$. 
Thus,
\begin{equation}
 R_\theta^{\otimes8}f_8
 =f_8+\frac{e^{8\mathfrak i\theta}-1}{16}u_+^{\otimes8}
      +\frac{e^{-8\mathfrak i\theta}-1}{16}u_-^{\otimes8}.
 \label{eq: rotated f8 expansion}
\end{equation}
Equations \eqref{eq: h expansion f8 u+ u-} and
\eqref{eq: rotated f8 expansion} prove \eqref{eq: f8 real orbit}.
\end{proof}

\begin{theorem}[Realization of $f_8$]
\label{thm: realize f8}
Let $\mathcal F=\overline{\mathcal F}$ fail condition
\rm{(\ref{cond: tractable classes})}, and suppose that
$f\in\mathcal F$ has arity eight and
$f\notin\mathcal U^{\otimes}$.  Then either
$\Holant(\mathcal F)$ is \#P-hard, or there is a real orthogonal matrix $R$ such that
\begin{equation}
 \Holant(f_8,R^{-1}\mathcal F)\leq_T\Holant(\mathcal F).
 \label{eq: realize f8 reduction}
\end{equation}
\end{theorem}

\begin{proof}
If $f$ is reducible, \cref{lm: decomposition} makes all its factors
available.  An odd-arity factor gives hardness by
\cref{lm: odd holant with conjugation}.  Otherwise the factors have even
arity, and at least one is outside $\mathcal U^{\otimes}$.  Since the
proper even arities below eight are $2,4,6$, hardness follows from
\cref{lm: base case 2n=2,lm: base case 2n=4,lm: induction 2n=6}.
Thus we may assume that $f$ is irreducible.

If $f$ fails {\sc 2nd-Orth}, hardness follows from
\cref{lm: not 2nd-orth is hard}.  If some
$\partial_{ij}f\notin\mathcal U^{\otimes}$, it is a six-ary obstruction
and \cref{lm: induction 2n=6} applies. 
In the remaining case apply
\cref{lm: unitary wire normalization}.  Its lower-arity outcome is again
handled by \cref{lm: induction 2n=6}; otherwise it gives an arity-$8$
signature $h$ satisfying {\sc 2nd-Orth} and the strong equality property.
Now \cref{lm: f8 normalized equations,%
lm: f8 self mating} give
$h\sim R^{\otimes8}f_8$ for some real orthogonal $R$.

View the problem as $\holant{=_2}{\mathcal F}$ and apply
\cref{thm: holographic transformation} with $T=R^{-1}$.  Orthogonality
gives $(=_2)R^{\otimes2}=(=_2)$, while the transformed signature $h$ is
$f_8$ up to a nonzero scalar.  
So
\[
 \Holant(f_8,R^{-1}\mathcal F)
 \equiv_T\Holant(h,\mathcal F)
 \leq_T\Holant(\mathcal F).
\]
This proves
\eqref{eq: realize f8 reduction}.

For later use, the transformed set $R^{-1}\mathcal F$ is still conjugate
closed because $R$ is real.  Failure of condition
\rm{(\ref{cond: tractable classes})} is also preserved.  Tensor
factorization is invariant under a local invertible transformation; and
if $R^{-1}\mathcal F$ were $\mathscr C$-transformable for one of
$\mathscr P,\mathscr A,\mathscr L$, composing its witnessing matrix with
$R^{-1}$ would show that $\mathcal F$ is $\mathscr C$-transformable.
\end{proof}

\subsection{Hardness When \texorpdfstring{$f_8$}{f8} is Available}
\label{subsec: f8 available hardness}
Recall $\mathcal B=\{b_0, b_1, b_2, b_3\}$, where
\begin{equation}
 M(b_0)=I_2,\quad M(b_1)=X,\quad M(b_2)=Z,\quad M(b_3)=J.
 \label{eq: Bell signature set}
\end{equation}
These are the four Bell signatures.  

\begin{definition}[Strong Bell property]\label{def-strong-bell}
    A signature $f$ satisfies the strong Bell property if for all pairs of indices $\{i, j\}$, and every $b\in \mathcal{B}$, the signature $\partial_{ij}^b f$  realized by merging $x_i$ and $x_j$ of $f$ using $b$ is in $\{b\}^{\otimes}.$
\end{definition}

It is proved in \cite{realholant} that $f_8$ satisfies the strong Bell property.

Let $\mathcal G$ denote the closure of $\mathcal{F}$ under constant-size
gadgets, nonzero scalar multiples, argument permutations, conjugation, and
the factor extraction supplied by \cref{lm: decomposition}.  
If $\mathcal{F}$ is conjugate closed, then so is $\mathcal G$.

\begin{lemma}[Non-$\mathcal{B}$ hard]
\label{lm: f8 complex non Bell}
Let $\mathcal H=\overline{\mathcal H}$ fail condition
\rm{(\ref{cond: tractable classes})} and contain $f_8$.  If the closure of
$\mathcal H$ contains a nonzero binary signature $b$ that is not
projectively in $\mathcal B$, then $\Holant(\mathcal H)$ is \#P-hard.
\end{lemma}

\begin{proof}
If $b\notin\mathcal U$, hardness follows from
\cref{lm: base case 2n=2}.  Hence normalize $B=M(b)$ to be unitary.
Contract $b$ into the two coordinates $000$ and $100$ of $f_8$, and call
the resulting six-ary signature $g$.

First suppose that all four entries of $B$ are nonzero.  Puncturing the
two chosen coordinates is injective on the distance-$4$ code $C$, so
$|\mathscr S(g)|=16$.  The projection of $C$ onto any two of the six
remaining coordinates is surjective: an affine Boolean function can take
arbitrarily prescribed values at two distinct points.  Therefore
$\mathscr S(g)$ satisfies no fixed equality or disequality relation on
any pair of variables.

If $g\in\mathcal U^{\otimes}$, its support size is the product of the
support sizes of three projectively unitary binary factors.  Such a factor
has support size two when its matrix contains a zero and support size four
otherwise.  The equation $16=2\cdot2\cdot4$ forces two sparse factors,
and each sparse unitary factor imposes an equality or disequality relation
on its two variables.  This contradicts the preceding surjectivity.
Thus $g\notin\mathcal U^{\otimes}$.

If a $2$-by-$2$ unitary matrix contains a zero, it is diagonal or
anti-diagonal.  Suppose first that $B=\operatorname{diag}(a,d)$, where
$|a|=|d|\neq0$.  For the remaining triples
$y=(x_2,x_3,x_4)$ and $z=(x_6,x_7,x_8)$, the code definition gives
\begin{equation}
 g(y,z)=
 \begin{cases}
 a,&z=y\text{ and }\operatorname{wt}(y)\text{ is even},\\
 d,&z=y\text{ and }\operatorname{wt}(y)\text{ is odd},\\
 0,&z\neq y.
 \end{cases}
 \label{eq: diagonal Bell test}
\end{equation}
Here $|\mathscr S(g)|=8$, so a factorization in
$\mathcal U^{\otimes}$ must have three support-$2$ factors.  The support
in \eqref{eq: diagonal Bell test} forces their unique pairing to be
$(x_2,x_6),(x_3,x_7),(x_4,x_8)$, and all three factors are diagonal.
Let $r_k$ be the $11/00$ ratio of the $k$th factor.  The three weight-two
words give $r_1r_2=r_1r_3=r_2r_3=1$.  Hence
$r_1=r_2=r_3=\varepsilon$ with $\varepsilon^2=1$, and a weight-one word
gives $d/a=\varepsilon$.  Thus $b$ is projectively $b_0$ or $b_2$.

Now let
$B=\left[\begin{smallmatrix}0&a\\d&0\end{smallmatrix}\right]$.
Then
\[
 g(y,z)=
 \begin{cases}
 a,&z=y\oplus111\text{ and }\operatorname{wt}(y)\text{ is even},\\
 d,&z=y\oplus111\text{ and }\operatorname{wt}(y)\text{ is odd},\\
 0,&z\neq y\oplus111.
 \end{cases}
\]
The support again has size eight, so a factorization in
$\mathcal U^{\otimes}$ has three anti-diagonal factors with the unique
pairing above.  If $r_k$ is the ratio of the factor value for $y_k=1$ to
that for $y_k=0$, the three weight-two inputs give
$r_1r_2=r_1r_3=r_2r_3=1$.  Thus every $r_k$ equals one sign
$\varepsilon$, and a weight-one input gives $d/a=\varepsilon\in\{1,-1\}$.
These are exactly the projective classes of $b_1$ and $b_3$.
Consequently
\begin{equation}
 \partial_{000,100}^{b}f_8\in\mathcal U^{\otimes}
 \quad\Longleftrightarrow\quad
 b\text{ is projectively in }\mathcal B.
 \label{eq: non Bell iff six obstruction}
\end{equation}
The assumed $b$ therefore yields a six-ary signature outside
$\mathcal U^{\otimes}$, and \cref{lm: f6 is hard} proves hardness.
\end{proof}

For a signature $b$, let $\Holant(b^{\leq k},\mathcal H)$ denote the restriction in which at most $k$ vertices are labeled by $b$.
The following lemma and its proof are exactly the same as
\cite[Lemma~8.9]{realholant}.

\begin{lemma}
\label{lm: f8 three Bell copies}
    For every $b\in\mathcal B$,
    $
     \Holant(b,f_8,\mathcal F)
     \leq_T\Holant(b^{\leq2},f_8,\mathcal F).
    $
\end{lemma}

The following lemma and its two-step oracle reconstruction are essentially
the same as \cite[Lemma~8.10]{realholant}.  We give the proof because the
present version uses \cref{lm: f8 complex non Bell} for complex binary
signatures and must track the known orientation sign of $b_3$.

\begin{lemma}
\label{lm: f8 Bell simulation}
If $\mathcal F=\overline{\mathcal F}$ fails condition
\rm{(\ref{cond: tractable classes})}, then
\(
 \Holant(\mathcal B,f_8,\mathcal F)
 \leq_T\Holant(f_8,\mathcal F).
 \label{eq: simulate Bell f8}
\)
\end{lemma}

\begin{proof}
Assume the problem on the right is not already \#P-hard.  By
\cref{lm: f8 complex non Bell}, every nonzero realizable binary signature
is projectively in $\mathcal B$.  By \cref{lm: base case 2n=4}, every
nonzero realizable four-ary signature lies in
$\mathcal U^{\otimes}$; its two binary factors are available by
\cref{lm: decomposition} and hence are Bell signatures by
\cref{lm: f8 complex non Bell}.

\emph{Step 1: simulate one nonidentity Bell signature $c_1$.}
If a nonidentity Bell signature is already realizable, use it.  Otherwise
every binary gadget over $\{f_8\}\cup\mathcal F$ is zero or a multiple of
$b_0$, and every four-ary gadget is zero or a multiple of
$b_0\otimes b_0$, with an unknown pairing of its external edges.  Fix any
$c_1\in\mathcal B\setminus\{b_0\}$ and consider an instance with at most
two occurrences of $c_1$.  A self-loop on one copy, or an edge directly
joining the two copies, reduces to a known scalar or to the binary case.
We may therefore assume that all their external half-edges meet the rest
of the instance.

With one occurrence, the rest of the instance is a binary gadget
$\lambda b_0$.  The bilinear contraction of two distinct Bell signatures
is zero, so the answer is zero.  With two occurrences, the rest is a
four-ary gadget
\begin{equation}
 h=\lambda b_0(x_1,x_j)b_0(x_k,x_\ell),
 \qquad \{j,k,\ell\}=\{2,3,4\}.
 \label{eq: unknown equality pairing}
\end{equation}
Close the four external edges in each of the three possible pairings using
$b_0=(=_2)$.  The three oracle values are $4\lambda$ for the actual
pairing and $2\lambda,2\lambda$ for the other two.  If all values vanish,
$h=0$; otherwise they determine both $\lambda$ and the pairing.  Since the
two $c_1$ vertices are known, their contribution can then be computed
exactly.  Thus
$\Holant(c_1^{\leq2},f_8,\mathcal F)\leq_T
  \Holant(f_8,\mathcal F)$, and
\cref{lm: f8 three Bell copies} removes the occurrence bound.

\emph{Step 2: simulate a second independent Bell signature $c_2$.}
Step 1 proves
$\Holant(c_1,f_8,\mathcal F)\leq_T\Holant(f_8,\mathcal F)$, so under the
standing assumption the $c_1$-augmented problem is not \#P-hard.  Work
over $\{c_1,f_8\}\cup\mathcal F$.  If another Bell signature is
realizable, use it.  Otherwise \cref{lm: f8 complex non Bell} shows that
every nonzero binary gadget is projectively in $\{b_0,c_1\}$.  Likewise,
\cref{lm: base case 2n=4} puts every nonzero four-ary gadget in
$\mathcal U^{\otimes}$, and \cref{lm: decomposition} extracts its binary
factors.  The binary conclusion just proved then shows that every such
four-ary gadget has the form
\begin{equation}
 h=\lambda c(x_1,x_j)d(x_k,x_\ell),
 \qquad c,d\in\{b_0,c_1\}.
 \label{eq: unknown two Bell pairing}
\end{equation}
Fix $c_2\in\mathcal B\setminus\{b_0,c_1\}$.  It again suffices, by
\cref{lm: f8 three Bell copies}, to handle at most two occurrences.
Self-loops and direct edges reduce to the already known arity-zero or
arity-two cases.  With one occurrence, the rest is $b_0$ or $c_1$, both
bilinearly orthogonal to $c_2$, so the answer is zero.

With two occurrences, identify \eqref{eq: unknown two Bell pairing} using
at most six oracle queries.  First close the three possible pairings using
$b_0$ on both closing edges.  Up to the fixed orientation signs, the
responses are
\[
\begin{array}{c|c}
(c,d)&\text{three responses}\\ \hline
(b_0,b_0)&(4\lambda,2\lambda,2\lambda)\\
(c_1,c_1)&(0,2\varepsilon\lambda,2\lambda)\\
(b_0,c_1)\text{ or }(c_1,b_0)&(0,0,0),
\end{array}
\]
where $\varepsilon=1$ for the symmetric matrices $X,Z$; when
$M(c_1)=J$, the sign and its position are fixed by the known argument
orientations.  Only the mixed case requires three more queries: for each
possible pairing, close the pair incident with $x_1$ using $c_1$ and the
other pair using $b_0$.  If the factor incident with $x_1$ is $b_0$, the
response pattern is $(0,2\varepsilon\lambda,2\lambda)$; if it is $c_1$,
the pattern is $(4\varepsilon\lambda,2\lambda,2\lambda)$, again with known
orientation signs.  These responses determine the pairing, the two
factors, and $\lambda$, unless they all vanish, in which case $h=0$.
The value after attaching the two known $c_2$ vertices is therefore
computable.  This simulates $c_2$.

The fourth Bell signature is obtained, up to a nonzero scalar and argument
orientation, by connecting $c_1$ and $c_2$ as a path.  Combining the two
steps proves \eqref{eq: simulate Bell f8}.
\end{proof}

We use the real Bell-hardness theorem
\cite[Theorem~7.19]{realholant}: if a real-valued signature set
$\mathcal H$ fails condition \rm{(\ref{cond: tractable classes})} and is
non-Bell hard, meaning that adjoining any real binary signature outside
$\mathcal B$ gives hardness, then
$\Holant(\mathcal B,\mathcal H)$ is \#P-hard.

The proof of the next lemma is essentially the phase-synchronization proof
of \cref{lm: realification when parity}.  It is repeated because here the
Bell signatures are supplied by the Turing reduction in
\cref{lm: f8 Bell simulation}, rather than by a direct $f_6$ gadget.

\begin{lemma}[Parity realification with $f_8$]
\label{lm: f8 parity realification}
Let $\mathcal F=\overline{\mathcal F}$ fail condition
\rm{(\ref{cond: tractable classes})}, and assume
$\Holant(f_8,\mathcal F)$ is not \#P-hard.  Every nonzero parity
signature in the gadget-and-factor closure of
$\mathcal B\cup\{f_8\}\cup\mathcal F$ is projectively real.
\end{lemma}

\begin{proof}
By \cref{lm: f8 Bell simulation}, the Bell-augmented problem is also not
\#P-hard.  Work in its conjugate-closed gadget-and-factor closure
$\mathcal G$ and induct on arity.

A nonzero nullary signature is a scalar and hence is projectively real.
A nonzero odd-arity member of $\mathcal G$ would give hardness by
\cref{lm: odd holant with conjugation}; since every member of $\mathcal G$
is available from $\{f_8\}\cup\mathcal F$, this would contradict the
standing assumption.  It therefore remains to treat positive even arity.
A nonzero parity binary signature must be projectively in $\mathcal B$ by
\cref{lm: f8 complex non Bell}, and hence is real.  At arity four,
\cref{lm: base case 2n=4} gives membership in
$\mathcal U^{\otimes}$; \cref{lm: decomposition} extracts its binary
factors, which have parity and hence are real Bell signatures.  The same
argument at arity six uses \cref{lm: f6 is hard}.  If a larger signature
is reducible, \cref{lm: decomposition} extracts its factors.  Each factor
of a nonzero parity signature has parity; an odd-arity factor would give
hardness by \cref{lm: odd holant with conjugation}.  The even-arity factors
are projectively real by induction, and so is their tensor product.

Let $g$ now be irreducible of larger even arity.  By
\cref{lm: not 2nd-orth is hard}, it satisfies {\sc 2nd-Orth}.  Fix a pair
$\{i,j\}$.  By \cref{lm: 2nd-orth contraction nonzero} and the
induction hypothesis, we have
\begin{equation}
 \partial_{ij}^{b_s}g=\lambda_s r_s,
 \qquad 0\leq s\leq3,
 \label{eq: projectively real Bell contractions}
\end{equation}
where every $\lambda_s\neq0$ and every $r_s$ is a nonzero real parity
signature.  The contractions for $b_0,b_2$ have the same parity as $g$,
whereas those for $b_1,b_3$ have the opposite parity.

Choose $a\in\{0,2\}$ and $c\in\{1,3\}$.  By
\cref{lm: non-zero Bell contraction}, there are a disjoint pair
$\{u,v\}$ and $b_s\in\mathcal B$ such that
$\partial_{uv}^{b_s}r_a$ and $\partial_{uv}^{b_s}r_c$ are both nonzero.
The signature $\partial_{uv}^{b_s}g$ is a lower-arity parity signature, so
write it as $\mu r$ with $\mu\neq0$ and $r$ real.  Commutativity gives
\[
 \lambda_a\partial_{uv}^{b_s}r_a
 =\mu\partial_{ij}^{b_a}r,\qquad
 \lambda_c\partial_{uv}^{b_s}r_c
 =\mu\partial_{ij}^{b_c}r.
\]
All four signatures remaining after the scalar coefficients are nonzero
and real.  Hence $\lambda_a/\mu$ and $\lambda_c/\mu$ are nonzero real
numbers, so $\lambda_a/\lambda_c\in\mathbb R^\times$.  Applying this to
$(a,c)=(0,1),(2,1),(0,3)$ shows that the four $\lambda_s$ have one common
complex phase.  Absorb their remaining nonzero real multipliers into the
corresponding $r_s$.  After factoring out the common phase, recover the
four slices of $g$ by
\[
 \mathbf g_{ij}^{00}=\tfrac12(r_0+r_2),\quad
 \mathbf g_{ij}^{11}=\tfrac12(r_0-r_2),\quad
 \mathbf g_{ij}^{01}=\tfrac12(r_1+r_3),\quad
 \mathbf g_{ij}^{10}=\tfrac12(r_1-r_3).
\]
All four slices are real, so $g$ is projectively real.  This completes the
induction.
\end{proof}

\begin{theorem}
\label{thm: f8 available is hard}
If $\mathcal F=\overline{\mathcal F}$ fails condition
\rm{(\ref{cond: tractable classes})}, then
\(
 \Holant(f_8,\mathcal F)\text{ is \#P-hard}.
\)
\end{theorem}

\begin{proof}
By \cref{lm: f8 Bell simulation}, it suffices to prove hardness after
adjoining $\mathcal B$.  Assume temporarily that the resulting problem is
not already hard and work in the Bell-augmented gadget-and-factor closure.

If the closure contains a signature without parity, it cannot have odd
arity, since \cref{lm: odd holant with conjugation} would already give
hardness.  Thus it has positive even arity.  If it is not already binary,
repeatedly apply \cref{lm: non-pairty reduction} to obtain a binary
signature without parity.  Every member of $\mathcal B$ has parity, so
this binary signature is non-Bell.  Then
\cref{lm: f8 complex non Bell} gives hardness, a contradiction.

Otherwise every signature has parity.  By
\cref{lm: f8 parity realification}, multiply each signature of
$\mathcal F$ by an irrelevant nonzero phase and obtain a real-valued set
$\mathcal F_{\mathbb R}$.  These rescalings preserve condition
\rm{(\ref{cond: tractable classes})}, so the real set still fails it.
Moreover, \cref{lm: f8 complex non Bell}, restricted to real binary
signatures, says that $\{f_8\}\cup\mathcal F_{\mathbb R}$ is non-Bell
hard.  Therefore \cite[Theorem~7.19]{realholant} gives
\[
 \Holant(\mathcal B,f_8,\mathcal F_{\mathbb R})
 \text{ is \#P-hard}.
\]
Undoing the fixed projective scalars identifies this problem, under a
direct polynomial-time reduction, with
$\Holant(\mathcal B,f_8,\mathcal F)$.  Finally,
\cref{lm: f8 Bell simulation} reduces the latter problem to
$\Holant(f_8,\mathcal F)$.
\end{proof}

\begin{theorem}
\label{thm: base case 2n=8}
Let $\mathcal F=\overline{\mathcal F}$ fail condition
\rm{(\ref{cond: tractable classes})}.  If $\mathcal F$ contains an
arity-$8$ signature $f\notin\mathcal U^{\otimes}$, then
$\Holant(\mathcal F)$ is \#P-hard.
\end{theorem}

\begin{proof}
Apply \cref{thm: realize f8}.  Its first outcome is hardness.  In the
second outcome, the real orthogonal transform $R^{-1}\mathcal F$ remains
conjugate closed and outside condition
\rm{(\ref{cond: tractable classes})}, as proved there.  Hence
\cref{thm: f8 available is hard} makes the left-hand side of
\eqref{eq: realize f8 reduction} \#P-hard.  The reduction then proves
hardness of $\Holant(\mathcal F)$.
\end{proof}

\section{The Induction Proof of the Main Theorem: \texorpdfstring{$2n\ge 10$}{2n >= 10}}
\label{sec: 2n >= 10}
\label{sec: high arity induction}

We now complete the induction framework of
\cref{subsec: Induction framework}.  The normalization from
\cref{lm: unitary wire normalization} is valid in every arity at least
eight.  At arity at least ten, however, its strong equality outcome is
combinatorially impossible.

\begin{lemma}
\label{lm: high arity equality rigidity}
No nonzero signature of arity $2n\geq10$ can simultaneously satisfy
{\sc 2nd-Orth} and the strong equality property.
\end{lemma}

\begin{proof}
Similar to \cref{def: f8 matching}, for any pair of distinct indices $e=\{i,j\}$, we can define $M_e, F_e$ and the relation $\sim.$
Also similar to \cref{lm: f8 matching equivalence}, $\sim$ is an equivalence relation.
The proof there makes no special use of $2n=8$, and it holds for any $2n\ge10.$
Consequently the perfect matchings $F_e$ partition $E(K_{2n})$ into a
$1$-factorization.

Choose distinct perfect matchings $A,B$ and disjoint edges $e\in A$, $d\in B$.
Such a choice exists because $n\geq5$.  
Because $A$ and $B$ are distinct, they are disjoint as two equivalent classes of $\sim.$
Compute $\partial_d\partial_eh$ in the two orders.  Starting with $e$, the two
endpoints of $d$ lie in two distinct edges of $A\setminus\{e\}$.
Contracting $d$ switches those two edges into one edge and leaves the
other $m-3$ edges of $A$ unchanged.  Hence the residual matching consists of one switched edge and $m-3$ unchanged $A$-edges.  Reversing the order
gives one switched edge and $m-3$ unchanged $B$-edges.

The two double contractions are the same nonzero equality product, so
their residual matchings must coincide.  
Since $m\geq5$, the first residual
matching contains at least two unchanged $A$-edges.  
Neither can be an unchanged $B$-edge, since $A$ and $B$ are disjoint.
However, at most one can equal the single switched edge
in the reverse-order matching.  This is impossible.  Therefore the two
hypotheses cannot hold simultaneously.
\end{proof}

The following descent theorem plays the same role as
\cite[Lemma~9.1]{realholant}. 

\begin{theorem}
\label{thm: strict arity descent}
Let $\mathcal F=\overline{\mathcal F}$ fail condition
\rm{(\ref{cond: tractable classes})}, and let $f\in \mathcal{F}$ is a
signature of arity $2n\geq10$ with
$f\notin\mathcal U^{\otimes}$.  Then either
$\Holant(\mathcal F)$ is \#P-hard, or we can realize an even arity signature
$g\notin\mathcal U^{\otimes}$ of arity at most $2n-2$.
\end{theorem}

\begin{proof}
If $f$ is reducible, \cref{lm: decomposition} makes its proper factors
available.  A nonzero odd-arity factor gives hardness by
\cref{lm: odd holant with conjugation}.  If all factors have even arity,
at least one proper factor lies outside $\mathcal U^{\otimes}$, and it is
the desired $g$.

Assume therefore that $f$ is irreducible.  If it fails {\sc 2nd-Orth},
\cref{lm: not 2nd-orth is hard} gives hardness.  If some equality
contraction $\partial_{ij}f$ lies outside
$\mathcal U^{\otimes}$, take $g=\partial_{ij}f$, whose arity is
$2n-2$.

The only remaining case is that $f$ is irreducible, satisfies
{\sc 2nd-Orth}, and every $\partial_{ij}f$ belongs to
$\mathcal U^{\otimes}$.  Apply
\cref{lm: unitary wire normalization}.  Its first outcome is precisely the
required lower-arity $g$.  Its second outcome is a nonzero arity-$2n$
signature satisfying both {\sc 2nd-Orth} and the strong equality property,
contradicting \cref{lm: high arity equality rigidity}.  Thus one of the
first two outcomes must occur.
\end{proof}

The induction below is essentially the same as the final induction in
\cite[Theorem~9.2]{realholant}, with the complex arities $6$ and $8$
supplied by \cref{lm: f6 is hard,thm: base case 2n=8}.  We include it to
make the proof of the main theorem complete.

\begin{theorem}[Main theorem]
\label{thm: conjugate closed Holant dichotomy}
Let $\mathcal F$ be a conjugate-closed set of algebraic complex-valued Boolean
signatures.  
If $\mathcal F$ satisfies condition
\rm{(\ref{cond: tractable classes})}, then $\Holant(\mathcal F)$ is
polynomial-time computable.  Otherwise $\Holant(\mathcal F)$ is
\#P-hard.
\end{theorem}

\begin{proof}
The tractable direction is \cref{thm: tractable classes}.  Suppose that
$\mathcal F$ does not satisfy condition
\rm{(\ref{cond: tractable classes})}.  If it contains a nonzero odd-arity
signature, hardness follows from
\cref{lm: odd holant with conjugation}.  We may therefore restrict to
nonzero even-arity signatures.

Since $\mathcal U^{\otimes}\subseteq\mathscr T$, the assumption that
$\mathcal F$ fails condition \rm{(\ref{cond: tractable classes})} implies
that some $f\in\mathcal F$ has even arity $2n$ and
$f\notin\mathcal U^{\otimes}$.  We induct on $2n$.  The cases
$2n=2,4,6,8$ are respectively
\cref{lm: base case 2n=2,lm: base case 2n=4,lm: induction 2n=6,%
thm: base case 2n=8}.

For $2n\geq10$, apply \cref{thm: strict arity descent}.  Either hardness
already follows, or it supplies an available
$g\notin\mathcal U^{\otimes}$ of smaller even arity.  
Since $\mathcal{F}$ is conjugate closed, $\overline{g}$ is also realizable.
Let $\mathcal{G}=\mathcal{F}\cup\{g,
\overline{g}\}$.
Clearly $\mathcal{G}$ fails
condition \rm{(\ref{cond: tractable classes})} because it contains $\mathcal{F}.$
The induction hypothesis makes the augmented problem $\Holant(\mathcal{G})$
\#P-hard, while gadget realization reduces it to
$\Holant(\mathcal F)$.  
Since each descent lowers the fixed signature
arity by at least two, the process reaches one of the four base arities
after finitely many constant-size reductions.  This completes the proof.
\end{proof}

\bibliographystyle{alpha}
\bibliography{ref}

\newpage
\appendix
\section{An Exposition of Xia's Framework Based on Classification of Projective Binary Groups}\label{appendix: mingji}

In this section, we will give an exposition of Xia's framework (\cite{MingjiProgram}, v1) of classifying all $\Holant(\mathcal{F})$ problems with unknown complexity according to the projective binary group generated by $\mathcal{F}$.
This exposition is short and self-contained.
Specifically, we will give a proof for \cref{thm: binary set is finite group}:

\MingjiTheorem*

Recall in \cref{subsec: Induction framework}, we identify $B(f)$ by a subset of $\mathbf{PU}(2)$, and $B(f)$ is closed under multiplication.
Suppose $b\in B(f)$ (we equivalently say the matrix $M(b)\in B(f)$). 
By \cref{def: generating process}, $B(f)$ is closed under conjugation and transpose. 
So $\overline{M(b)}^{\tt T}\in B(f)$.
Let $\overline{b}^{\tt T}$ denote the binary signature in $B(f)$ with matrix $\overline{M(b)}^{\tt T}$.
Notice that $\overline{M(b)}^{\tt T}M(b)=\lambda I_2$ for some $\lambda\in \mathbb
C^{\times}$ since $M(b)\in \mathbf{PU}(2)$.
So $b^{-1}=\overline{b}^{\tt T}$.
Thus $B(f)$ is a group.
In the following we show it is a \emph{finite} group, thus $B(f)$ belongs to a finite list of subgroups of $\mathbf{SO}(3)$, unless $\Holant(\mathcal{F})$ is \#P-hard.
\InterpolationBinary*

By \cref{lm: 2 by 2 interpolation}, we may assume that every binary matrix in $B(f)$ has finite projective order, otherwise we can realize $M(g)=\left[\begin{smallmatrix}
    1 & 0\\
    0 & 0
\end{smallmatrix}\right]$ and thus $\Holant(\mathcal{F})$ is \#P-hard by \cref{lm: odd holant with conjugation}.
The following lemma shows $B(f)$ has finite exponent, i.e., there exists a uniform $N\in \N$ such that $b$ has order less than $N$ for every $b\in B(f)$.

\begin{lemma}\label{lm: finite exponent}
    There exists $N\in \N$, such that $b$ has order less than $N$ for every $b\in B(f)$. 
\end{lemma}
\begin{proof}
    Recall that we assume $f$ takes algebraic complex value. 
    Adjoint all possible values of $f$ into $\Q$, and obtain a finite field extension $E/\mathbb{Q}$.
    All possible values of $b$ also lie in $E$ for every $b\in B(f)$, since gadget operations and signature factorizations do not leave $E.$
    Fix an arbitrary $b\in B(f)$.
    By \cref{lm: 2 by 2 interpolation}, $M(b)$ has finite projective order $d$ for some $d\in \N$.
    So the eigenvalues $\lambda_1(b)$ and $\lambda_2(b)$ are non-zero, and $\mu(b):=\frac{\lambda_2(b)}{\lambda_1(b)}$ is a primitive $d$-th root of unity.
    Thus, its algebraic degree is $\varphi(d)$, where $\varphi(\cdot)$ is Euler's totient function.

    Now suppose for a contradiction that for every $n\in\N$, there exists $b_n\in B(f)$, such that $b_n$ has order $d_n>n.$ 
    Then $\mu(b_n):=\frac{\lambda_2(b_n)}{\lambda_1(b_n)}$ is a primitive $d_n$-th root of unity.
    However, $\lambda_1(b_n),\lambda_2(b_n)$ are two roots of the characteristic polynomial  of $M(b)$, which is a quadratic polynomial in $E[X]$.
    So $\lambda_2(b_n),\lambda_2(b_n)$ lie in a quadratic extension of $E$, as well as $\mu(b_n)$.
    So $\mu(b_n)$ has a finite algebraic degree, independent of $n$.
    This contradicts $\varphi(d_n)\to +\infty$ as $n\to +\infty.$
\end{proof}
\begin{theorem}[Burnside. See for example, \cite{rowen2008graduate} Theorem 19A.9]\label{thm: Burnside}
   Suppose $F$ is a field of characteristic 0.
   Then every subgroup $G \subseteq \mathrm{GL}(n,F)$ of finite exponent is a finite group.
\end{theorem}

\begin{proof}[Proof of \cref{thm: binary set is finite group}]
    Lifting $B(f)$ from $\mathbf{PU}(2)$ to $\mathbf{SU}(2)$, it is a subgroup of  $\mathrm{GL(2,\mathbb{C})}$.
    By \cref{lm: finite exponent}, $B(f)$ has finite exponent.
    By \cref{thm: Burnside}, $B(f)$ is a finite group.
    \cref{thm: binary set is finite group} is proved.
\end{proof}
Since $\mathbf{SU}(2)/\{\pm I_2\}\cong \mathbf{SO}(3)$, $B(f)$ belongs to the following list: cyclic groups, dihedral groups, or polyhedral groups.

\section{All AME(6,2) States Have One Local-Unitary Normal Form}\label{sec: AME62 has a common local unitary normal}

This Appendix~\ref{sec: AME62 has a common local unitary normal} has proved that every $\mathrm{AME}(6,2)$ state is local unitary equivalent to a fixed 6-qubit stabilizer state, i.e., for any $\mathrm{AME}(6,2)$ state $\ket{\psi}$, we have $\ket{\psi}=\lambda (U_1\otimes \cdots \otimes U_6)\ket{H}$ where $\lambda\in\mathbb{C}^\times$, $U_1,\dots,U_6$ are unitary matrix and $\ket{H}$ is determined explicitly. The proof of this section is used in the~\cref{sec:recover-f6} to recover a special 6-ary signature from any given an $\mathrm{AME}(6,2)$ state.

\subsection{Definitions and Notations}

\begin{definition}
    For $n\geq 0$, let $\mathcal{H}_n:=(\mathbb{C}^2)^{\otimes n}$, with the standard basis $\{\ket{x}:x\in\{0,1\}^n\}$. A linear operator $M:\mathcal{H}_n\rightarrow \mathcal{H}_n$ is represented in this basis by a $2^n\times 2^n$ complex matrix whose entry in row x and column y is $(M)_{x,y}:=\bra{x}M\ket{y},\quad x,y\in\{0,1\}^n$. The trace of $M$ is 
    \begin{equation}
        \operatorname{Tr}(M)=\sum_{x\in \{0,1\}^n}\bra{x}M\ket{x}.
    \end{equation}

    Let $A\subseteq [n], \overline{A}=[n]\setminus A$. The partial trace of $M$ over the coordinates in A, denoted by $\operatorname{Tr}_A(M)$, is the
\end{definition}

Let $g:\{0,1\}^n\rightarrow \mathbb{C}$ be a nonzero Boolean signature. 
Its associated normalized state is
\begin{equation}
    \ket{g}=\frac{1}{\lVert g\rVert_2}\sum_{x\in\{0,1\}^n} g(x)\ket{x},\quad\quad \lVert g\rVert_2^2=\sum_{x} \lvert g(x)\rvert^2.
\end{equation}
Let $\rho_g=\ket{g}\bra{g}$ be an $2^n \times 2^n$ matrix where $x,y\in\{0,1\}^n$ are row index and column index respectively, and each entry is $(\rho_g)_{x,y}=\bra{x}\rho_g\ket{y}=\frac{g(x)\overline{g(y)}}{\lVert g\rVert_2^2}$. For $S\subseteq [n]$ and $\overline{S}=[n]\setminus S$, the reduced density matrix $\rho_S$ is a $2^{\lvert S\rvert}\times 2^{\lvert S\rvert}$ matrix, where $x_S,y_S\in \{0,1\}^{\lvert S\rvert}$ are row index and column index respectively, and each entry 
\begin{equation}
    (\rho_S)_{x_S,y_S}=\bra{x_S}\rho_S\ket{y_S}=\sum_{z\in \{0,1\}^{\lvert \overline{S}\rvert}} \bra{x_Sz}\rho_g\ket{y_Sz}.
\end{equation}
Equivalently, $\rho_S=\operatorname{Tr}_{\overline{S}}(\rho_g)$. %

\begin{definition}
    A normalized $n$-qubit state $\ket{\psi}$ is $\operatorname{AME}(n,2)$ if 
    \begin{equation}\label{eq: AME}
        \rho_S=2^{-\lvert S\rvert} I_{2^{\lvert S\rvert}}
    \end{equation}
    for every $S\in [n]$ with $\lvert S\rvert\leq \lfloor n/2 \rfloor$.
\end{definition}

\begin{remark}
    For $n=6$, it is enough to impose~\eqref{eq: AME} for $\lvert S\rvert=3$, because tracing out further qubits gives the cases $\lvert S\rvert=1,2$.
\end{remark}

Write 
\begin{equation}
    I=\left(\begin{matrix}
        1 & 0\\
        0 & 1
    \end{matrix}\right), \quad
    X=\left(\begin{matrix}
        0 & 1\\
        1 & 0
    \end{matrix}\right), \quad 
    Y=\left(\begin{matrix}
        0 & -i\\
        i & 0
    \end{matrix}\right), \quad
    Z_0=\left(\begin{matrix}
        1 & 0\\
        0 & -1
    \end{matrix}\right).
\end{equation}

An $n$-qubit Pauli string is an operator $E=E_1\otimes \cdots \otimes E_n$ with $E_i\in \{I,X,Y,Z_0\}$ for $i\in [n]$. Its support and weight are 
\begin{equation}
    \operatorname{supp}(E)=\{i:E_i\neq I\}, \quad\quad \operatorname{wt}(E)=\lvert \operatorname{supp}(E)\rvert. 
\end{equation}

\begin{definition}
    Let $\mathcal{C}$ be a $K$-dimensional subspace of $(C^2)^{\otimes n}$. If $\ket{\psi_1},\dots,\ket{\psi_K}$ is an orthonormal basis of $\mathcal{C}$, define the orthogonal projector onto $\mathcal{C}$ by $P_{\mathcal{C}}:=\sum_{i=1}^K\ket{\psi_i}\bra{\psi_i}$. We call $\mathcal{C}$ a pure $((n,K,d))_2$ code if
    \begin{equation}
        P_{\mathcal{C}}EP_{\mathcal{C}}=0,
    \end{equation}
    for every nonidentity $n$-qubit Pauli string $E$ with $\operatorname{wt}(E)<d$. Equivalently, $\bra{\psi_i}E\ket{\psi_j}=0$ for all $i,j\in [K]$.

    In particular, for $K=1$, $\mathcal{C}=\operatorname{span}\{\ket{\psi}\},P_{\mathcal{C}}=\ket{\psi}\bra{\psi}$, the condition is $\bra{\psi}E\ket{\psi}=0$.
\end{definition}

\begin{definition}
    Two nonzero $n$-qubit tensors $u$ and $v$ are locally unitarily equivalent, up to input permutation, if there are exist $2\times 2$ unitary matrices $U_1,\dots,U_n\in U(2)$, a permutation $\pi\in S_n$, and a nonzero scalar $\lambda\in \mathbb{C}\setminus \{0\}$ such that 
    \begin{equation}
        \ket{u}=\lambda P_{\pi}(U_1\otimes \dots \otimes U_n)\ket{v}.
    \end{equation}
\end{definition}

\subsection{AME States and Pure Quantum Codes}

\begin{lemma}\label{lem:614 code}
    A normalized $6$-qubit state $\ket{\psi}$ is $\operatorname{AME}(6,2)$ if and only if its one-dimensional span is a pure $((6,1,4))_2$ code.
\end{lemma}

\begin{proof}
    For a subset $S\subseteq [6]$, the reduced density matrix on $S$ is $\rho_S$. Let $\mathcal{E}_{S}=\{I,X,Y,Z_0\}^{\otimes \lvert S\rvert}$ be the set of all $\lvert S\rvert$-qubit Pauli strings and for $E,F\in \mathcal{E}_S$, $\operatorname{Tr}(EF)=2^{\lvert S\rvert}\delta_{E,F}$. Hence the matrices in $\mathcal{E}_S$ form an orthogonal basis of this $2^{\lvert S\rvert}\times 2^{\lvert S\rvert}$ matrix space.
    Then there are unique coefficients $c_E$ such that $\rho_S=\sum_{E\in \mathcal{E}_S} c_E E$. For a fixed $F\in \mathcal{E}_S$, $\operatorname{Tr}(F\rho_S)=\operatorname{Tr}\left(F\sum_{E\in \mathcal{E}_S }c_E E\right)=\sum_{E\in \mathcal{E}_S}c_E \operatorname{Tr}(FE)=2^{\lvert S\rvert}c_F$. Therefore 
    $$\rho_S=2^{-\lvert S\rvert}\sum_{E\in \mathcal{E}_S}\operatorname{Tr}(E\rho_S)E$$.

    By the definition of trace,
    \begin{align}
        \operatorname{Tr}(E\rho_S)&=\sum_{x_S}\bra{x_S} E\rho_S\ket{x_S}\\
        &= \sum_{x_S,y_S}\bra{x_S}E \ket{y_S}\bra{y_S}\rho_S\ket{x_S}\\
        &= \sum_{x_S,y_S} \bra{x_S}E \ket{y_S} \sum_{z\in \{0,1\}^{\lvert \overline{S}\rvert}}\bra{y_Sz}\rho_g\ket{x_Sz}\\
        &= \sum_{x_S,y_S,z}\bra{x_S}E \ket{y_S}\braket{y_Sz |\psi}\braket{\psi|x_Sz}\\
        &=\sum_{x_S,y_S,z}\braket{\psi|x_Sz} \bra{x_S}E \ket{y_S}\braket{y_Sz |\psi}\\
        &=\sum_{x_S,y_S,z}\braket{\psi|x_Sz} \bra{x_Sz}(E\otimes I_{\overline{S}}) \ket{y_Sz}\braket{y_Sz |\psi}\\
        &=\bra{\psi}(E\otimes I_{\overline{S}})\ket{\psi}.
    \end{align}
    Hence,
    \begin{equation}
        \rho_S=2^{-\lvert S\rvert}\sum_{E\in \mathcal{E}_S}\bra{\psi}(E\otimes I_{\overline{S}})\ket{\psi}E.
    \end{equation}
    The coefficient corresponding to $E=I_S$ is $\bra{\psi}I_S\otimes I_{\overline{S}}\ket{\psi}=\braket{\psi|\psi}=1$. Consequently, by uniqueness of the Pauli expansion, $\rho_S=2^{-\lvert S\rvert}I_S$ if and only if 
    \begin{equation}\label{eq: pure code}
        \bra{\psi} (E\otimes I_{\overline{S}})\ket{\psi}=0
    \end{equation}
    for every nonidentity Pauli string $E\in \mathcal{E}_S$.
    Since $\psi$ is $\operatorname{AME}(6,2)$ when $\rho_S=2^{-\lvert S\rvert}I_S$ for every $S\subseteq [6],|S|\leq 3$, then~\eqref{eq: pure code} is hold. Equivalently,
    \begin{equation}
        \bra{\psi} E\ket{\psi}=0,
    \end{equation}
    for every $E\neq I$ with $1\leq \operatorname{wt}(E)\leq 3$.

    Therefore, $\ket{\psi}$ is $\operatorname{AME}(6,2)$ if and only if $\operatorname{span}(\ket{\psi})$ is pure $((6,1,4))_2$ code.
\end{proof}

\subsection{The Standard Four-qubit Code}

\begin{definition}
    Let $S_X:=X^{\otimes 4}$ and $S_Z:=Z_0^{\otimes 4}$. Since $XZ_0=-Z_0X$, we have $S_XS_Z=(-1)^4S_ZS_X=S_ZS_X$, so $S_X$ and $S_Z$ commute. Define the standard four-qubit code $\mathcal{C}_{\mathrm{std}}\subseteq (\mathbb{C}^2)^{\otimes 4}$ to be their simultaneous $+1$ eigenspace:
    \begin{equation}
        \mathcal{C}_{\mathrm{std}}:=\{\ket{\phi}\in (\mathbb{C}^2)^{\otimes 4}:S_X\ket{\phi}=\ket{\phi},S_Z\ket{\phi}=\ket{\phi}\}.
    \end{equation}
    Equivalently, $$\mathcal{C}_{\mathrm{std}}=\operatorname{span}\{\frac{\ket{0000}+\ket{1111}}{\sqrt 2},\frac{\ket{0011}+\ket{1100}}{\sqrt 2},\frac{\ket{0101}+\ket{1010}}{\sqrt 2},\frac{\ket{0110}+\ket{1001}}{\sqrt 2}\}.$$
    These four vectors are orthonormal, and hence $\mathrm{dim} \mathcal{C}_{\mathrm{std}}=4$. The code $\mathcal{C}_{\mathrm{std}}$ is conventionally called the standard $[[4,2,2]]_2$ code. Here the notation $[[n,k,d]]_2$ means that $n$ physical qubits encode $k$ logical qubits with distance $d$.

    And let $\ket{c_0}=\frac{\ket{0000}+\ket{1111}}{\sqrt 2},\ket{c_1}=\frac{\ket{0011}+\ket{1100}}{\sqrt 2},\ket{c_2}=\frac{\ket{0101}+\ket{1010}}{\sqrt 2},\ket{c_3}=\frac{\ket{0110}+\ket{1001}}{\sqrt 2}$, define the orthogonal projector onto $\mathcal{C}_{\mathrm{std}}$ by $P_{\mathcal{C}_{\mathrm{std}}}:=\sum_{i=0}^3\ket{c_i}\bra{c_i}=\frac{1}{4}(I^{\otimes 4}+X^{\otimes 4}+Y^{\otimes 4}+Z_0^{\otimes 4}).$
\end{definition}

\begin{theorem}[Rains~\cite{rains1999quantum}]\label{thm: rains standard code}
    Let $\mathcal{C}\subseteq (\mathbb{C}^2)^{\otimes 4}$ be a 4-dimensional subspace. If $\ket{\psi_1},\ket{\psi_2},\ket{\psi_3},\ket{\psi_4}$ is an orthonormal basis of $\mathcal{C}$, define the orthogonal projector onto $\mathcal{C}$ by $P_{\mathcal{C}}:=\sum_{i=1}^4\ket{\psi_i}\bra{\psi_i}$. Suppose that $P_{\mathcal{C}}EP_{\mathcal{C}}=0$ for every nonidentity 4-qubit Pauli string $E$ of weight 1; equivalently, $\mathcal{C}$ is a a pure $((4,4,2))_2$ code. Then there exist $2\times 2$ unitary matrices $U_1,U_2,U_3,U_4\in U(2)$ such that, with $U=U_1\otimes U_2\otimes U_3\otimes U_4$, we have
    \begin{equation}
        U\mathcal{C}=\mathcal{C}_{\mathrm{std}},\quad UP_{\mathcal{C}}U^\dagger=P_{\mathcal{C}_{\mathrm{std}}}
    \end{equation}
    Thus every pure $((4,4,2))_2$ code is locally unitary equivalent to the standard $[[4,2,2]]_2$ code.
\end{theorem}

\subsection{Four-qubit Reductions of AME States}

\begin{lemma}\label{lem: QA is orth proj of pure 442}
    Let $\ket{\psi}$ be $\operatorname{AME}(6,2)$ state and $P=\ket{\psi}\bra{\psi}$. For every pair $A\subset [6]$, define the $2^4\times 2^4$ matrix $Q_A=4 \operatorname{Tr}_A(P)$ is the rank-$4$ projector of a pure $((4,4,2))_2$ code.
\end{lemma}

\begin{proof}
    Without loss of generality, we can suppose that $A=\{1,2\}$. Dividing the 6 qubits into 2 groups, one is $\{1,2\}$ and the other is $\{3,4,5,6\}$. Thus, we can represent the $\ket{\psi}$, $\ket{\psi}=\sum_{x\in\{0,1\}^2}\sum_{y\in\{0,1\}^4} c_{x,y} \ket{x}_{12}\otimes \ket{y}_{3456}$ where $x=(x_1,x_2)$ are the first two qubits and $y=(x_3,x_4,x_5,x_6)$ are the last four qubits. For any fixed $x$, we define $\ket{w_x}=\sum_{y\in\{0,1\}^4}c_{x,y}\ket{y}_{3456}$. Hence, 
    \begin{equation}\label{eq: psi}
        \ket{\psi}=\sum_{x\in\{0,1\}^2}\ket{x}_{12}\otimes\ket{w_x}_{3456}.
    \end{equation}

    Then $P=\ket{\psi}\bra{\psi}=\sum_{x,x'\in \{0,1\}^2}\ket{x}\bra{x'}_{12}\otimes \ket{w_x}\bra{w_{x'}}_{3456}$. For $A=(1,2)$, we can compute the reduced density matrix $\rho_{12}$,
    \begin{align}
        \rho_{12}&=\operatorname{Tr}_{3456}(P)\\
        &=\sum_{x,x'} \operatorname{Tr}(\ket{w_x}\bra{w_{x'}})\ket{x}\bra{x'}_{12}\\
        &=\sum_{x,x'}\braket{w_{x'}|w_x}\ket{x}\bra{x'}_{12},\label{eq: rho_12}
    \end{align}
    since $\operatorname{Tr}(\ket{u}\bra{v})=\braket{v|u}$. Recall that $\ket{\psi}$ is $\mathrm{AME}(6,2)$, $\rho_{12}=\frac{1}{4}I_4=\frac{1}{4}\sum_{x\in\{0,1\}^2}\ket{x}\bra{x}$. Since $\{\ket{x}\bra{x'}:x,x'\in\{0,1\}^2\}$ is a basis of $4\times 4$ matrix, thus we can get $\braket{w_{x'}|w_x}=\frac{1}{4}\delta_{x,x'}$.
    This indicates that $\lVert w_x\rVert^2=\frac{1}{4}$ if $x=x'$ and $\ket{w_x}$ is orthogonal to $\ket{w_{x'}}$ if $x\neq x'$. For simplicity, we define $\ket{v_x}=2\ket{w_x}$ and $\braket{v_{x'}|v_x}=\delta_{x,x'}$ and $\{\ket{v_x}:x\in\{0,1\}^2\}$ are pairwise orthogonal. Thus, ~\eqref{eq: psi} is equal to
    \begin{equation}\label{eq: psi v}
        \ket{\psi}=\frac{1}{2}\sum_{x\in\{0,1\}^2}\ket{x}_{12}\otimes \ket{v_x}_{3456}.
    \end{equation}

    We can get $Q_A=4\operatorname{Tr}_{12}(P)= 4\sum_{x,x'\in\{0,1\}^2} \operatorname{Tr}(\ket{x}\bra{x'})\ket{w_x}\bra{w_{x'}}=4 \sum_{x,x'\in\{0,1\}^2}\braket{x|x'}\ket{w_x}\bra{w_{x'}}=\sum_{x\in\{0,1\}^2}\ket{v_x}\bra{v_x}$.

    Let $\mathcal{C}=\operatorname{span}\{\ket{v_x}:x\in\{0,1\}^2\}$. Since $\ket{v_x},x\in\{0,1\}^2$ are pairwise orthogonal, $\operatorname{dim}\mathcal{C}=4$. We can get $Q_A=\sum_{x\in\{0,1\}^2}\ket{v_x}\bra{v_x}=Q_A^\dagger$, and $Q_A^2=\sum_{x,x'\in\{0,1\}^2}\ket{v_x}\braket{v_x|v_{x'}}\bra{x'}=\sum_{x,x'}\delta_{x,x'}\ket{v_x}\bra{v_{x'}}=\sum_{x\in\{0,1\}^2}\ket{v_x}\bra{v_{x'}}=Q_A$. Thus, $Q_A$ is the orthogonal projector onto $\mathcal{C}$.

    Now we prove that for any nonidentity 4-qubit Pauli string E with $\operatorname{wt}(E)=1$, we have $Q_AEQ_A=0$.

    Fix a position $k\in \{3,4,5,6\}$, without loss of generality, suppose that $k=3$. Since $\ket{\psi}$ is $\mathrm{AME}(6,2)$, $\rho_{123}=\operatorname{Tr}_{456}(P)=\frac{1}{8}I_4\otimes I_2$. By~\eqref{eq: psi v}, $\rho_{123}=\frac{1}{4}\sum_{x,x'\in\{0,1\}^2}\ket{x}\bra{x'}_{12}\otimes \operatorname{Tr}_{456}(\ket{v_x}\bra{v_{x'}})$. Note that $\frac{1}{8}I_4\otimes I_2=\frac{1}{8}\sum_{x\in\{0,1\}^2}\ket{x}\bra{x}_{12}\otimes I_2$, then we can get
    \begin{align}
        \rho_{123}&=\frac{1}{8}\sum_{x\in\{0,1\}^2}\ket{x}\bra{x}_{12}\otimes I_k.\\
        \rho_{123}&=\frac{1}{4}\sum_{x,x'\in\{0,1\}^2}\ket{x}\bra{x'}_{12}\otimes \operatorname{Tr}_{456}(\ket{v_x}\bra{v_{x'}}).
    \end{align}
    Since $\{\ket{x}\bra{x'}:x,x'\in\{0,1\}^2\}$ is a orthogonal basis of $4\times 4$ matrix, comparing the coefficient, we can get that 
    \begin{equation}
        \operatorname{Tr}_{456}(\ket{v_x}\bra{v_{x'}})=\delta_{x,x'}\frac{I_k}{2}.
    \end{equation}

    Thus, suppose that $E=E_3\otimes I_{456}$ be a 4-qubit Pauli string and one of $\{X,Y,Z_0\}$ in the $k$-th qubit, we have
    \begin{align}
        Q_AEQ_A&=\sum_{x,x'\in\{0,1\}^2}\ket{v_x}\bra{v_{x}}E \ket{v_{x'}}\bra{v_{x'}}\\
        &=\sum_{x,x'\in\{0,1\}^2}\ket{v_x}\operatorname{Tr}_{3456}(E\ket{v_{x'}}\bra{v_x})\bra{v_{x'}}\\
        &=\sum_{x,x'\in\{0,1\}^2}\ket{v_x}\operatorname{Tr}_3[E_3\operatorname{Tr}(\ket{v_x}\bra{v_{x'}})]\bra{v_{x'}}\\
        &=\sum_{x,x'\in\{0,1\}^2}\frac{\delta_{x,x'}}{2} \operatorname{Tr}(E_k) \ket{v_x}\bra{v_{x'}}\\
        &=0,
    \end{align}
    since $\operatorname{Tr}(E_k)=0$. Finally, $Q_A$ is a the rank-4 projector of a pure $((4,4,2))_2$ code.
\end{proof}

\begin{corollary}\label{cor: fixed by SX and SZ}
    Let $\ket{\psi}$ be $\mathrm{AME}(6,2)$ state. For every pair $A\subset [6]$, we have $\mathrm{AME}(6,2)$ state $\ket{\psi^{(1)}}:=(I_A\otimes U_{\overline{A}})\ket{\psi}$ such that $(I_A\otimes S_X)\ket{\psi^{(1)}}=\ket{\psi^{(1)}}$ and $(I_A\otimes S_Z)\ket{\psi^{(1)}}=\ket{\psi^{(1)}}$.
\end{corollary}

\begin{proof}
Fix a pair $A\subseteq[6]$. By
\cref{lem: QA is orth proj of pure 442,thm: rains standard code},
we can write
\[
\ket{\psi}=\frac12\sum_{x\in\{0,1\}^2}
\ket{x}_A\otimes\ket{v_x}_{\overline A},
\qquad
U_{\overline A}\ket{v_x}\in\mathcal C_{\mathrm{std}},
\]
where $U_{\overline A}=\bigotimes_{k\in\overline A}U_k$
and each $U_k$ is unitary.
Set $\ket{\psi^{(1)}}=(I_A\otimes U_{\overline A})\ket{\psi}$.
For $S\in\{S_X,S_Z\}$, the definition of
$\mathcal C_{\mathrm{std}}$ gives
\[
(I_A\otimes S)\ket{\psi^{(1)}}
=\frac12\sum_x\ket{x}_A\otimes S U_{\overline A}\ket{v_x}
=\frac12\sum_x\ket{x}_A\otimes U_{\overline A}\ket{v_x}
=\ket{\psi^{(1)}}.
\]
Set $U_k=I$ for $k\in A$. For every $T\subseteq[6]$
with $|T|\leq3$,
\[
\rho_T^{(1)}
=\left(\bigotimes_{k\in T}U_k\right)
\rho_T
\left(\bigotimes_{k\in T}U_k\right)^\dagger
=2^{-|T|}I_{2^{|T|}}.
\]
Thus $\ket{\psi^{(1)}}$ is $\operatorname{AME}(6,2)$.
\end{proof}

\subsection{The Six-qubit Local-Unitary Normal Form}

\begin{lemma}\label{lem: ame62 is unitary equivalent}
    Every $\mathrm{AME}(6,2)$ state $\ket{\psi^{(0)}}$ is locally unitary equivalent, up to variables permutation, to one fixed 6-qubit stabilizer state, called the quantum hexacode state.
\end{lemma}

\begin{proof}
    Let $P^{(0)}=\ket{\psi^{(0)}}\bra{\psi^{(0)}}$. Apply~\cref{lem: QA is orth proj of pure 442} to the pair $(1,2)$ and the matrix $Q_{12}^{(0)}=4\operatorname{Tr}(P^{(0)})$. By~\cref{thm: rains standard code}, there is a local unitary $U_{3456}=U_3\otimes U_4\otimes U_5\otimes U_6$ such that 
    \begin{equation}
        U_{3456}Q_{12}^{(0)}U_{3456}^\dagger=Q_{\mathrm{std}}.
    \end{equation}
    Define $\ket{\psi^{(1)}}=(I_{12}\otimes U_{3456})\ket{\psi^{(0)}}$ and $P^{(1)}=\ket{\psi^{(1)}}\bra{\psi^{(1)}}$. Local unitaries preserve the $\mathrm{AME}$ equations, thus $\ket{\psi^{(1)}}$ is $\mathrm{AME}(6,2)$ state. And we have 
    \begin{equation}
        Q_{12}^{(1)}:=4\operatorname{Tr}_{12}(P^{(1)})=U_{3456}Q_{12}^{(0)}U_{3456}^\dagger=Q_{\mathrm{std}}.
    \end{equation}
    By~\cref{cor: fixed by SX and SZ}, $\ket{\psi^{(1)}}$ is fixed by (omitting $\otimes$) 
    \begin{equation}
        S_1=IIXXXX,\quad\quad S_2=IIZ_0Z_0Z_0Z_0.
    \end{equation}

    Let $Q:=Q_{56}^{(1)}=4\operatorname{Tr}_{56}(P^{(1)})$, ~\cref{lem: QA is orth proj of pure 442} says that $Q$ is the projector of a pure $((4,4,2))_2$ code. Since $S_1\ket{\psi^{(1)}}=\ket{\psi^{(1)}}$ we have $S_1P^{(1)}S_1=P^{(1)}$. With respect to the cut $1234|56$, the $S_1$ is $(IIXX)_{1234}\otimes (XX)_{56}$, since unitary invariance of the partial trace on qubits 5,6 gives
    \begin{equation}
        (IIXX)Q(IIXX)=4\operatorname{Tr}_{56}(S_1P^{(1)}S_1)=4\operatorname{Tr}_{56}(P^{(1)})=Q.
    \end{equation}
    The same argument with $S_2$ gives $(IIXX)Q=Q(IIXX),(IIZ_0Z_0)Q=Q(IIZ_0Z_0)$, namely $[Q,IIXX]=[Q,IIZ_0Z_0]=0$.

    Let $\ket{b_1}=\ket{\Phi^+}=\frac{\ket{00}+\ket{11}}{\sqrt 2},\ket{b_2}=\ket{\Phi^-}=\frac{\ket{00}-\ket{11}}{\sqrt 2},\ket{b_3}=\ket{\Psi^+}=\frac{\ket{01}+\ket{10}}{\sqrt 2},\ket{b_4}=\frac{\ket{01}-\ket{10}}{\sqrt 2}$ be an orthonormal basis of two qubit on 3,4, thus $I_{34}=\sum_{i=1}^4\ket{b_i}\bra{b_i}$.

    We can rewrite $Q=(I_{12}\otimes I_{34})Q(I_{12}\otimes I_{34})=\sum_{i,j=1}^4(I_{12}\otimes \ket{b_i}\bra{b_i})Q(I_{12}\otimes \ket{b_j}\bra{b_j})$. Define $Q_{ij}=(I_{12}\otimes \bra{b_i})Q(I_{12}\otimes \ket{b_j})$, we can get $Q=\sum_{i,j=1}^4Q_{ij}\otimes\ket{b_i}\bra{b_j}$. Consider the four Bell state's eigenvalue with respect to $XX$ and $ZZ$,
    \[
\begin{array}{c|cc}
 & X\otimes X & Z\otimes Z \\ \hline
|b_1\rangle & +1 & +1 \\
|b_2\rangle & -1 & +1 \\
|b_3\rangle & +1 & -1 \\
|b_4\rangle & -1 & -1
\end{array}
\]
    Then we can get $(I_{12}\otimes X_3X_4)(Q_{ij}\otimes \ket{b_i}\bra{b_j})=Q_{ij}\otimes X_3X_4\ket{b_i}\bra{b_j}=x_iQ_{ij}\otimes\ket{b_i}\bra{b_j}$ where $x_i\in{\pm 1}$ and $x_i$ is the eigenvalue of $X_3X_4$ to $\ket{b_i}$. We also can get $(Q_{ij}\otimes \ket{b_i}\bra{b_j})(I_{12}\otimes X_3X_4)=x_jQ_{ij}\otimes \ket{b_i}\bra{b_j}$ where $x_j\in\{\pm 1\}$ and $x_j$ is the eigenvalue of $X_3X_4$ to $\ket{b_j}$. Since $[Q,IIXX]=0$, $(x_i-x_j)Q_{ij}=0$. Similarly, $(z_i-z_j)Q_{ij}=0$ where $z_i,z_j\in \{\pm 1\}$ are the eigenvalue of $Z_3Z_4$ with respect to $\ket{b_i},\ket{b_j}$ respectively. Since $i\neq j$ we have at least one of $x_i\neq x_j$ and $z_i\neq z_j$ hold, thus $Q_{ij}=0$ for $i\neq j$. Hence, $Q=\sum_{i=1}^4Q_{ii}\otimes \ket{b_i}\bra{b_i}$.

    Since $Q^2=Q$, expanding $Q^2$ we can get $Q^2=\sum_{i=1}^4 Q_{ii}^2\otimes \ket{b_i}\otimes\bra{b_i}$. Comparing coefficient, we can get $Q_{ii}^2=Q_{ii}$. And since $Q^\dagger=Q$, we can get $Q^\dagger=\sum_{i=1}^4 Q_{ii}^\dagger \ket{b_i}\bra{b_i}$. Comparing coefficient, we can also get $Q_{ii}^\dagger=Q_{ii}$. Hence each $Q_{ii}$ is an orthogonal projector.

    Since $\ket{\psi^{(1)}}$ is $\mathrm{AME}(6,2)$ state, $\rho_{34}^{(1)}=\frac{1}{4}I_{34}$. Then $\operatorname{Tr}_{12}Q=4\operatorname{Tr}_{12}(\operatorname{Tr}_{56}P^{(1)})=4\operatorname{Tr}_{1256}P^{(1)}=4\rho_{34}^{(1)}=I_{34}=\sum_{i=1}^4\ket{b_i}\bra{b_i}$. On the other hand, $\operatorname{Tr}_{12}(Q)=\sum_{i=1}^4\operatorname{Tr}(Q_{ii})\ket{b_i}\bra{b_i}$. Since $\ket{b_i}\bra{b_i}$ are orthogonal, comparing the coefficient, we can get $\operatorname{Tr}(Q_{ii})=1$ for each $i\in [4]$. Since the trace of orthogonal projector is equal to the rank, $\operatorname{rank}(Q_{ii})=1$. Thus, exist uniform vector $\ket{u_i}$ such that $Q_{ii}=\ket{u_i}\bra{u_i}$ for each $i\in [4]$.

    Consider $\operatorname{Tr}_{34}Q$, by $\mathrm{AME}(6,2)$, we have $\operatorname{Tr}_{34}Q=I_{12}$. On the other hand, we have $\operatorname{Tr}_{34}Q=\sum_{i=1}^4Q_{ii} \operatorname{Tr}(\ket{b_i}\bra{b_i})=\sum_{i=1}^4 Q_{ii}$. Fixed $j\in [4]$, $\sum_{i=1}^4\ket{u_j}Q_{ii}\bra{u_j}=\ket{u_j}I_{12}\bra{u_j}=1$. We also know that $\sum_{i=1}^4\ket{u_j}Q_{ii}\bra{u_j}=\sum_{i=1}^4\braket{u_j|u_i}\braket{u_j|u_i}=\sum_{i=1}^4\lvert \braket{u_i|u_j}\rvert^2$. Thus $1=\sum_{i=1}^4\lvert \braket{u_i|u_j}\rvert^2$, and since $\lvert \braket{u_j|u_j}\rvert^2=1$. Hence, $\braket{u_i|u_j}=0$ for $i\neq j$. So the four vectors $\ket{u_i},i\in[4]$ form an orthonormal basis of the two-qubit space.

    Since $Q$ is the orthogonal projector of pure $((4,4,2))_2$ code, for any nonidentity Pauli string $E$ with weight one, we have $Q(E\otimes I_{34})Q=0$. Expanding, 
    \begin{align}
        Q(E\otimes I_{34})Q&=\sum_{i,j\in\{0,1\}^2}Q_{ii}EQ_{jj}\otimes \ket{b_i}\braket{b_i|b_j}\bra{b_j}\\
        &=\sum_{i\in\{0,1\}^2}Q_{ii}EQ_{ii}\otimes\ket{b_i}\bra{b_i}\\
        &=\sum_{i\in\{0,1\}^2}\ket{u_i}\bra{u_i}E\ket{u_i}\bra{u_i}\otimes \ket{b_i}\bra{b_i}\\
        &=\sum_{i\in\{0,1\}^2}\bra{u_i}E\ket{u_i}Q_{ii}\otimes \ket{b_i}\bra{b_i}\\
        &=0.
    \end{align}
    Since distinct $i$ on qubit $3,4$ has orthogonal space $\ket{b_i}\bra{b_i}$, thus $\bra{u_i}E\ket{u_i}=0$ for each $i\in [3]$.

    Let $\rho_1=\operatorname{Tr}_2(\ket{u_i}\bra{u_i})$ be the density matrix on qubit 1, note that every single qubit density matrix can be write as $\rho_1=\frac{1}{2}(I+\operatorname{Tr}(\rho_1X)X+\operatorname{Tr}(\rho_1Y)Y+\operatorname{Tr}(\rho_1Z)Z)$. Let $E=X\otimes I$, we can get $\operatorname{Tr}(\rho_1X)=\bra{u_i}X\otimes I\ket{u_i}=0$. Thus, $\operatorname{Tr}(\rho_1X)=0$. Similarly, let $E=Y\otimes I,Z_0\otimes I$, we can get $\operatorname{Tr}(\rho_1Y)=\operatorname{Tr}(\rho_1Z_0)=0$. Hence, $\rho_1=\frac{1}{2}I$. Let $\rho_2=\operatorname{Tr}_1(\ket{u_i}\bra{u_i})$, we can also get $\rho_2=\frac{1}{2}I$ similarly. This indicates that each $\ket{u_i}$ is maximally entangled.

    Now we prove that for each $i\in [4]$ there exist a unitary $M_i$ such that $\ket{u_i}=(M_i\otimes I)\ket{\Phi^+}$ and $\operatorname{Tr}(M_i^\dagger M_j)=2\delta_{ij}$. Since we can write $\ket{u_i}=\sum_{a,b=0}^1c_{ab}^{(i)}\ket{a}\ket{b}$, let $C_i=\left(\begin{matrix}
        c_{00}^{(i)} & c_{01}^{(i)}\\
        c_{10}^{(i)} & c_{11}^{(i)}
    \end{matrix}\right)$. Since $\ket{u_i}$ is maximally entangled, $C_iC_i^\dagger=\frac{I}{2}$. Define $M_i:=\sqrt{2} C_i$ and we have $M_iM_i^\dagger=I$. Thus, $\ket{u_i}=(M_i\otimes I)\ket{\Phi^+}$.

    For any $2\times 2$ matrix $A$, we have $\bra{\Phi^+}(A\otimes I)\ket{\Phi^+}=\frac{1}{2}\operatorname{Tr}(A)$. Thus $\braket{u_i|u_j}=\bra{\Phi^+}(M_i^\dagger M_j\otimes I)\ket{\Phi^+}=\frac{1}{2}\operatorname{Tr}(M_i^\dagger M_j)$. And since $\braket{u_i|u_j}=\delta_{ij}$ by orthogonal, we have $\operatorname{Tr}(M_i^\dagger M_j)=2\delta_{ij}$. This indicates that the four $2\times 2$ unitary matrices $M_1,M_2,M_3,M_4$ is the orthonormal basis of $2\times 2$ matrix under the Hilbert-Schmidt inner product. Define $N_i:=M_1^\dagger M_i$, in particular $N_1=I$. Thus $N_1=I,N_2,N_3,N_4$ is a orthonormal basis and $\operatorname{Tr}(N_1)=2$ and $\operatorname{Tr}(N_j)=\operatorname{Tr}(N_1^\dagger N_j)=2\delta_{1j}=0$ for $j=2,3,4$.

    Now, we have $N_1=I,N_2,N_3,N_4$ are unitary matrix and $N_1,N_2,N_3$ are traceless unitary matrix. For any $N_j,j\neq 1$, its two eigenvalue have norm 1 and the sum of two eigenvalues is equal to 1. Thus, we can suppose that $\lambda_j,-\lambda_j$ are $N_j$'s two eigenvalues and $\lvert \lambda_j\rvert=1$. Hence, define $H_j:=\frac{1}{\lambda_j}N_j$. According $N_j$ is unitary matrix, then $H_j$ is normal matrix, then we can get $H_j=U\operatorname{diag}(1,-1)U^\dagger=H^\dagger$ and $\operatorname{Tr}(H_j)=0$. So $H_j$ is Hermitian matrix. Since every traceless $2\times 2$ Hermitian matrix has unique form: $H_j=n_{j,x}X+n_{j,y}Y+n_{j,z}Z_0$ where $n_{j,x},n_{j,Y},n_{j,z}\in \mathbb{R}$. Let $\mathbf{n_j}=(n_{j,x},n_{j,Y},n_{j,z})^T\in \mathbb{R}^3$. Then we can define $A(\mathbf{v})=v_xX+v_yY+v_zZ_0$ for $\mathbf{v}=(v_x,v_y,v_z)^T$. Hence $H_j=A(\mathbf{n_j})$.

    Computing $H_j^2$, we can get 
    \begin{align*}
        H_j^2&=(n_{j,x}X+n_{j,y}Y+n_{j,Z}Z_0)^2\\
        &=n_{j,x}^2X^2+n_{j,y}^2Y^2+n_{j,z}^2Z_0^2\\
        &+n_{j,x}n_{j,y}(XY+YX)\\
        &+n_{j,x}n_{j,z}(XZ_0+Z_0X)\\
        &+n_{j,y}n_{j,z}(YZ_0+Z_0Y)\\
        &=(n_{j,x}^2+n_{j,y}^2+n_{j,z}^2)I\\
        &=\lVert\mathbf{n_j}\rVert^2I.
    \end{align*}
    Since $H_j^2=H_jH_j^\dagger=I$, $\lVert\mathbf{n_j}\rVert=1$. For distinct $i,j$, $\frac{1}{2}\operatorname{Tr}(H_iH_j)=\frac{1}{2}\operatorname{Tr}[(n_{i,x}X+n_{i,y}Y+n_{i,z}Z_0)(n_{j,x}X+n_{j,y}Y+n_{j,z}Z_0)]=\mathbf{n_i}\cdot\mathbf{n_j}$. And since $\frac{1}{2}\operatorname{Tr}(H_iH_j)=\frac{1}{2}\operatorname{Tr}(H_iH_j^\dagger)=\frac{1}{2}\operatorname{Tr}(N_iN_j^\dagger)=\delta_{ij}=0$. So $\mathbf{n_i}\cdot \mathbf{n_j}=0$, and finally $\mathbf{n_1},\mathbf{n_2},\mathbf{n_3}$ is orthonormal basis in $\mathbb{R}^3$. 

    Let $\mathbf{e_x}=(1,0,0)^T,\mathbf{e_y}=(0,1,0)^T,\mathbf{e_z}=(0,0,1)^T$. Let $N=(\mathbf{n_1,n_2,n_3})$. Since $\mathbf{n_1,n_2,n_3}$ are orthonormal, $N^TN=I$. Thus define $R:=N^T$, we have $R\mathbf{n_1}=\mathbf{e_x}, R\mathbf{n_2}=\mathbf{e_y}, R\mathbf{n_3}=\mathbf{e_z}$. Since $(\det N)^2=\det N^TN=\det I=1$, so $\det N\in\{\pm1\}$. If $\det N=-1$, then we can just define $\mathbf{n_1'}=(-n_{1,x},-n_{1,y},-n_{1,z}), \mathbf{n_2'}=\mathbf{n_2},\mathbf{n_3'}=\mathbf{n_3}$, thus $N'=(\mathbf{n_1',n_2',n_3'})$ and $\det N'=1$ and $-n_{1,x}X-n_{1,y}Y-n_{1,z}Z=-H_1$ and others are unchanged. We can suppose that $\det N=1$. In this case $R\in \mathrm{SO}(3)=\{R\in \mathbb{R}^{3\times 3}: R^TR=I, \det R=1\}$. Since $\mathrm{SO}(3)\cong \mathrm{SU}(2)/\{I,-I\}$, for every $R\in \operatorname{SO}(3)$, there exist a $W\in \operatorname{SU}(2)$ such that $WA(\mathbf{v})W^\dagger=A(R\mathbf{v})$ for every $\mathbf{v}\in \mathbb{R}^3$. Thus for this $R=N^T$, there exist $W\in \mathrm{SU}(2)$ such that $WH_1W^\dagger=A(R\mathbf{n_1})=A(\mathbf{e_1})=X$, $WH_2W^\dagger=Y$ and $WH_3W^\dagger=Z_0$.

    Recall that $\ket{u_i}=(M_i\otimes I)\ket{\Phi^+}$, let $V_1:=WM_0^\dagger,V_2:=\overline{W}$. Since $(A\otimes B)\ket{\Phi^+}=(AB^T\otimes I)\ket{\Phi^+}$ for every $2\times 2$ matrix $A,B$.
    Then we can get $(V_1\otimes V_2)\ket{u_i}=(WM_0^\dagger\otimes \overline{W})(M_i\otimes I)\ket{\Phi^+}=(WM_0^\dagger M_i\otimes \overline{W})\ket{\Phi^+}=(WN_i\otimes \overline{W})\ket{\Phi^+}=(WN_iW^\dagger \otimes I)\ket{\Phi^+}$. Since $N_0=I,N_i=\lambda_i H_i$ for $i=1,2,3$. Then we have $(V_1\otimes V_2)\ket{u_0}=\ket{\Phi^+}$, $(V_1\otimes V_2)\ket{u_1}=\lambda_1(X\otimes I)\ket{\Phi^+}=\lambda_1\ket{\Psi^+}$, $(V_1\otimes V_2)\ket{u_2}=\lambda_1(Y\otimes I)\ket{\Phi^+}=-i\lambda_2\ket{\Psi^-}$ and $(V_1\otimes V_2)\ket{u_3}=\lambda_3(Z_0\otimes I)\ket{\Phi^+}=\lambda_3\ket{\Phi^-}$. Hence, for every $i\in\{0,1,2,3\}$, there exist a $\omega_i\in \mathbb{C},\lvert w_i\rvert=1$ such that $(V_1\otimes V_2)\ket{u_i}=\omega_i\ket{b_i}$. Hence, $(V_1\otimes V_2)\ket{u_i}\bra{u_i}(V_1\otimes V_2)^\dagger=\ket{b_i}\bra{b_i}.$

    Recall that $Q=\sum_{i=0}^3\ket{u_i}\bra{u_i}_{12}\otimes \ket{b_i}\bra{b_i}_{3456}$. After add $V_1\otimes V_2$ on qubit 1,2, we can get $Q'=(V_1\otimes V_2\otimes I_{34})Q(V_1\otimes V_2\otimes I_{34})^\dagger=\sum_{i=0}^3\ket{b_i}\bra{b_i}_{12}\otimes \ket{b_i}\bra{b_i}_{34}=Q_{\mathrm{std}}$.

    Define the twice standardized state and its projector by
    \begin{equation}
        \ket{\psi^{(2)}}:=(V_1\otimes V_2\otimes I_{3456})\ket{\psi^{(1)}},\quad \quad P^{(2)}=\ket{\psi^{(2)}}\bra{\psi^{(2)}}.
    \end{equation}
    The $\ket{\psi^{(2)}}$ is also $\mathrm{AME}(6,2)$ state. Then $Q_{56}^{(2)}=4\operatorname{Tr}_{56}(P^{(2)})=Q_{\mathrm{std}}$.
    One can verify that $S_1\ket{\psi^{(2)}}=\ket{\psi^{(2)}},S_2\ket{\psi^{(2)}}=\ket{\psi^{(2)}}$. Since $XXXX,IIII$ fixed $Q_{56}^{(2)}$, then $XXXXII,ZZZZII$ fixed $\ket{\psi^{(2)}}$. Hence $(IIXXXX)S_1=(IIXXXX)(XXXXII)=(XXIIXX),(IIZZZZ)S_2=(ZZIIZZ)$ also fixed $\ket{\psi^{(2)}}$. Now we can get that $\ket{\psi^{(2)}}$ is fixed by 
    \begin{align*}
        &S_1=IIXXXX, \quad S_2=IIZ_0Z_0Z_0Z_0\\
        &S_3=XXIIXX, \quad S_4=Z_0Z_0IIZ_0Z_0.
    \end{align*}

    Now consider $Q_{13}^{(2)}=4\operatorname{Tr}_{13}(P^{(2)})$ and order its four qubits as $(2,4,5,6)$. Thus by~\cref{lem: QA is orth proj of pure 442}, $Q_{13}^{(2)}$ is a rank-4 projector of pure $((4,4,2))_2$ code. Restrict the four stabilizers $S_1,S_2,S_3,S_4$ to qubits $(2,4,5,6)$, we can $IXXX,IZZZ,XIXX,ZIZZ$. Since $S_1P^{(2)}S_1=P^{(2)}$, the partial trace gives $(IXXX)Q_{13}^{(2)}(IXXX)=Q_{13}^{(2)}$, $IZZZ,XIXX,ZIZZ$ is similar. Thus $[Q_{13}^{(2)},IXXX]=[Q_{13}^{(2)},IZZZ]=[Q_{13}^{(2)},XIXX]=[Q_{13}^{(2)},ZIZZ]=0$.

    Since $Q_{13}^{(2)}$ is $16\times 16$ matrix, then $Q_{13}^{(2)}=\sum_{E\in\{I,X,Y,Z\}^{\otimes 4}}c_E E$. Note that the $256$ different $E$ is a orthonormal basis of $16\times 16$ matrix. Encode the 4 qubit Pauli string as $(x_1,x_2,x_3,x_4|z_1,z_2,z_3,z_4)$, $(x_i,z_i)=(0,0)$ means that the $i$'th Pauli Matrix is $I$. $(1,0)$ is $X$, $(0,1)$ is $Z$ and $(1,1)$ is $Y$. Let $\mathbf{x}=(x_1,x_2,x_3,x_4)^T,\mathbf{z}=(z_1,z_2,z_3,z_4)^T$. A key property is that two Pauli strings  $(\mathbf{x|z})$ and $(\mathbf{x'|z'})$ are commutate if and only if $\mathbf{x^T z'+z^T x'}=0 \mod{2}$.

    Since $IXXX,IZZZ,XIXX,ZIZZ$'s corresponding Pauli string encodes are $(0111|0000),(0000|0111),(1011|0000)$ and $(0000|1011)$ respectively. Then the possible Pauli string $E=(\mathbf{x|z})$ with nonzero coefficient must satisfy 
    \begin{align*}
        z_2+z_3+z_4=0, \quad x_2+x_3+x_4=0\\
        z_1+z_3+z_4=0, \quad x_1+x_3+x_4=0\\
    \end{align*}
    Thus, $x_1=x_2,z_1=z_2,x_4=x_1+x_3,z_4=z_1+z_3$. There are four free bits and hence exactly sixteen solutions. Grouped by weight, they are

    \begin{equation}
\begin{array}{c|l}
\text{weight} & \text{Pauli strings} \\ \hline
0 & IIII \\[2pt]
2 & IIXX,\ IIYY,\ IIZ_0Z_0 \\[2pt]
3 & XXIX,\ XXXI,\ YYIY,\ YYYI,\ Z_0Z_0IZ_0,\ Z_0Z_0Z_0I \\[2pt]
4 & XXYZ_0,\ YYZ_0X,\ Z_0Z_0XY,\ XXZ_0Y,\ YYXZ_0,\ Z_0Z_0YX
\end{array}
\end{equation}

Note that $e_E=\operatorname{Tr}(Q_{13}^{(2)}E)=4\operatorname{Tr}(P^{(2)}(I_{13}\otimes E))=4\bra{\psi^{(2)}}I_{13}\otimes E\ket{\psi^{(2)}}=0$ if $\operatorname{wt}(E)\leq 3$ and $E$ is nonidentity by~\cref{lem:614 code}. Thus the Pauli expansion of $Q_{13}^{(2)}$ contains no terms of weights 1,2 or 3. Since $Q_{13}^{(2)}$ is rank 4 orthonormal projector, then $\operatorname{Tr}(Q_{13}^{(2)})=\operatorname{rank}(Q_{13}^{(2)})=4$. Hence
\begin{equation}
    Q_{13}^{(2)}=\frac{1}{4}I+bXXYZ_0+cYYZ_0X+dZ_0Z_0XY+eXXZ_0Y+fYYXZ_0+gZ_0Z_0YX,
\end{equation}
where $b,c,d,e,f,g\in \mathbb{R}$. Let $A_1:=XXYZ_0,A_2:=YYZ_0X,A_3:=Z_0Z_0XY, B_1:XXZ_0Y,B_2:=YYXZ_0,B_3:=Z_0Z_0YX$, we have $(A_i,A_j)$ are commute and $A_i^2=B_i^2=I,A_1A_2=A_3,A_2A_3=A_1,A_3A_1=A_2,B_1B_2=B_3,B_2B_3=B_1,B_3B_1=B_2$ and $A_1B_1=IIXX,A_2B_2=IIYY,A_3B_3=IIZ_0Z_0$ and $(A_i,B_j)$ are anticommute when $i\neq j$.

Since $Q_{13}^{(2)}$ is orthogonal projector, thus $(Q_{13}^{(2)})^2=Q_{13}^{(2)}$. Let $L=bA_1+cA_2+dA_3+eB_1+fB_2+gB_3$. Expanding $(Q_{13}^{(2)})^2$, we can get
\begin{align}
    (Q_{13}^{(2)})^2&=\frac{1}{16}I+\frac{1}{2}L+L^2\\
    &=(\frac{1}{16}+b^2+c^2+d^2+e^2+f^2+g^2)I\\
    &+(\frac{b}{2}+2cd)A_1+(\frac{c}{2}+2bd)A_2+(\frac{d}{2}+2bc)A_3\\
    &+(\frac{e}{2}+2fg)B_1+(\frac{f}{2}+2eg)B_2+(\frac{g}{2}+2ef)B_3\\
    &+2beIIXX+2cfIIYY+2dgIIZ_0Z_0.
\end{align}
Comparing $Q_{13}^{(2)}=\frac{1}{4}I+bA_1+cA_2+dA_3+eB_1+fB_2+gB_3$, we can get
\begin{align}
    & b^2+c^2+d^2+e^2+f^2+g^2=\frac{3}{16},\\
    & be=cf=dg=0,\\
    & b=4cd,\quad c=4bd,\quad d=4bc,\\
    &e=4fg,\quad f=4eg,\quad g=4ef.\\
\end{align}

Since $be=0$, gives $b=0$ or $e=0$. Swapping qubits 5,6 will exchanges the triples $(b,c,d)$ and $(e,f,g)$. Thus, we can suppose that $e=0$. Then we can get $f=g=0$. If one of $b,c,d$ were zero, this gives that all three zero, contary to the first equation. Hence,
\begin{equation}
    \lvert b\rvert=\lvert c\rvert=\lvert d\rvert=\frac{1}{4}, \mathrm{sgn}(d)=\mathrm{sgn}(b)\mathrm{sgn}(c).
\end{equation}
Thus, let $s_b=4b\in\{\pm1\},s_c=4c\in\{\pm1\}$, we have 
\begin{equation}
    Q_{13}^{(2)}=\frac{1}{4}(I+s_bA_1+s_cA_2+s_bs_cA_3)=\frac{1}{4}(I+s_bA_1)(I+s_cA_2),
\end{equation}
since $A_1A_2=A_3$. This indicates that $\mathrm{im} Q_{13}^{(2)}$ is the simultaneous $+1$ eigenspace of $s_bA_1$ and $s_cA_2$. Thus, lifting to $\ket{\psi^{(2)}}$, we can get $T_5=s_bIXIXYZ_0,T_6=s_cIYIYZ_0X$ and
\begin{equation}
    s_bIXIXYZ_0\ket{\psi^{(2)}}=\ket{\psi^{(2)}}, \quad s_cIYIYZ_0X\ket{\psi^{(2)}}=\ket{\psi^{(2)}}.
\end{equation}
Let $S_5=T_5=s_bIXIXYZ_0,S_6=T_5T_6=s_bs_cIZ_0IZ_0XY$.

Until now, we have prove that $S_i\ket{\psi^{(2)}}=\ket{\psi^{(2)}}, 1\leq i\leq 6$ and 
\begin{align*}
    S_1&=IIXXXX\\
    S_2&=IIZ_0Z_0Z_0Z_0\\
    S_3&=XXIIXX\\
    S_4&=Z_0Z_0IIZ_0Z_0\\
    S_5&=s_bIXIXYZ_0\\
    S_6&=s_bs_cIZ_0IZ_0XY,
\end{align*}
where $s_b,s_c\in\{\pm 1\}$. Define $S_i=G_i$ for $1\leq i\leq 4$ and $S_5=s_bG_5,S_6=s_cG_6$. We can encode each 6 qubits Pauli string as $r=(\mathbf{x|z})\in \mathbb{F}_2^{12}$ where $\mathbf{x}=(x_1,x_2,x_3,x_4,x_5,x_6)$ and $\mathbf{z}=(z_1,z_2,z_3,z_4,z_5,z_6)$. $(x_i,z_i)=(0,0),(0,1),(1,0),(1,1)$ represent the $i$'th Pauli matrix is $I,Z_0,X,Y$ respectively. Thus we can get the six $G_i$'s encode are as following:
\begin{equation}
\begin{array}{c|c|c}
 & x & z \\ \hline
G_1 & 001111 & 000000 \\
G_2 & 000000 & 001111 \\
G_3 & 110011 & 000000 \\
G_4 & 000000 & 110011 \\
G_5 & 010110 & 000011 \\
G_6 & 000011 & 010101
\end{array}
\end{equation}

One can check that $\mathbf{xz'^T+x'z^T}=0 \mod{2}$ for any $(\mathbf{x,z}),(\mathbf{x',z'})$. Thus for any $S_i,S_j$, $[S_i,S_j]=0$. Multiplication of Pauli string is equal to addition of encode mod two. So the independence of the six Pauli strings is equal to the six encoders are linear independent on $\mathbb{F}_2$. Suppose that $\mathbf{r_i}=(\mathbf{x_i|z_i})$ for $i\in [6]$. If $\sum \alpha_i \mathbf{r_i}=0$, we can only get $\alpha_1=\alpha_2=\alpha_3=\alpha_4=\alpha_5=\alpha_6=0$. Thus, the six $G_i$ are independent.

Let $\mathcal{V^+}(S)=\{\ket{\phi}:S\ket{\phi}=\ket{\phi}\}$ be the $+1$ eigenspace of $S$. Now the eigenspace of $S_1,\dots,S_6$ is $\mathcal{V}=\{\ket{\phi}:S_j\ket{\phi}=\ket{\phi},j=1,\dots,6\}=\mathcal{V^+}(S_1)\cap\cdots \cap \mathcal{V^+}(S_6)$. We know that $\ket{\psi^{(2)}}\in \mathcal{V}$, next we want to prove that $\mathrm{dim}(\mathcal{V})=1$.

Since each $S_j$ is Hermitian Pauli string with $S_j^\dagger=S_j,S_j^2=I$, define $\Pi_j=\frac{S_j+I}{2}$. Since $\mathrm{im} \Pi_j=\{\Pi_j\ket{\eta}:\ket{\eta}\in (\mathbb{C}^2)^{\otimes 6}\}=\mathcal{V}^+(S_j)$ and $\Pi_j^\dagger=\Pi_j,\Pi_j^2=\Pi_j$, then $\Pi_j$ is the orthogonal projector of $\mathcal{V}^+(S_j)$. Since $S_j$ are pairwise commute, then $\Pi_i$ are pairwise commute. Thus, $\Pi=\Pi_1\cdots\Pi_6$ is the orthogonal projector of $\mathcal{V}$, namely $\mathrm{im} \Pi=\mathcal{V}$.

We can compute the $\Pi=\Pi_1\cdots \Pi_6=2^{-6}\Pi_{j=1}^6(I+S_j)=2^{-6}\sum_{\alpha\in \mathbb{F}_2^{6}}S_1^{\alpha_1}\cdots S_6^{\alpha_6}$. We know the trace of nonidentity Pauli strings are zero, thus $\operatorname{Tr}(\Pi)=2^{-6}\operatorname{Tr}(I^{\otimes 6})=1$. Since $\Pi$ is orthogonal projector, it has only eigenvalue $0$ or $1$, and its trace is equal to its rank and its dimension of image phase space. Hence $\mathrm{dim}\mathcal{V}=\operatorname{rank} \Pi=\operatorname{Tr}\Pi=1$. Finally, we can get $\mathcal{V}=\operatorname{span}\{\ket{\psi^{(2)}}\}$.

Recall that $S_j=G_j$ for $1\leq j\leq 4$ and $G_5\ket{\psi^{(2)}}=s_b\ket{\psi^{(2)}},G_6\ket{\psi^{(2)}}=s_c\ket{\psi^{(2)}}$.
Since $\ket{\psi^{(2)}}$ is the eigenvalue $+1$ of $G_j$, but maybe eigenvalue $-1$ of $G_5,G_6$. So we need Pauli transformation to change $-1$ to $+1$. Define $C_5=Z_0Z_0IIII,C_6=XXIIII$, thus $C_5G_5=-G_5C_5,C_5G_j=G_jC_5$ for $j\neq 5$ and $C_6G_6=-G_6C_6,C_6G_j=G_jC_6$ for $j\neq 6$. Let $\epsilon_5=\frac{1-s_b}{2},\epsilon_6=\frac{1-s_c}{2}$, thus $\epsilon_5,\epsilon_6\in\{0,1\}$. Define $U_{\mathrm{sgn}}=C_5^{\epsilon_5}C_6^{\epsilon_6}$ and $\ket{\phi}=U_{\mathrm{sgn}}\ket{\psi^{(2)}}$. Since $[U_{\mathrm
sgn},G_j]=0$ for $j=1,2,3,4$, $G_j\ket{\phi}=\ket{\phi}$. For $G_5$, $G_5\ket{\phi}=G_5C_5^{\epsilon_5}C_6^{\epsilon^6}\ket{\psi^{(2)}}=(-1)^{\epsilon_5}s_b\ket{\phi}=\ket{\phi}$ when $s_b=1,-1$. Similarly, $G_6\ket{\phi}=\ket{\phi}$. Hence, define $\mathcal{V}_0=\{\ket{\eta}:G_j\ket{\eta}=\ket{\eta},j=1,2,\dots,6\}$. In fact, $U_{\mathrm{sgn}}S_jU_{\mathrm{sgn}}^\dagger=G_j$ for $1\leq j\leq 6$, thus $\mathcal{V}_0=U_{\mathrm{sgn}}\mathcal{V}$. So $\mathrm{dim} \mathcal{V}_0=1$.

Choose and fix a unit vector spanning $\mathcal{V}_0$, and denote it by
$\ket{H_6}$. We call $\ket{H_6}$ the quantum hexacode state. Since
$\dim \mathcal{V}_0=1$ and
$U_{\mathrm{sign}}\ket{\psi^{(2)}}\in\mathcal{V}_0$, there exists
$\theta\in\mathbb{R}$ such that
\begin{equation}
    U_{\mathrm{sign}}\ket{\psi^{(2)}}
    =e^{\mathrm{i}\theta}\ket{H_6}.
\end{equation}

Combining the two standardizing local unitaries, the possible
transposition of qubits $5$ and $6$, and the final local Pauli operator
$U_{\mathrm{sign}}$, we conclude that there exist one-qubit unitaries
$L_1,\ldots,L_6$, a permutation $\pi\in S_6$, and
$\theta\in\mathbb{R}$ such that
\begin{equation}
    (L_1\otimes\cdots\otimes L_6)
    P_\pi\ket{\psi^{(0)}}
    =
    e^{\mathrm{i}\theta}\ket{H_6}.
\end{equation}
Equivalently,
\begin{equation}
    \ket{\psi^{(0)}}
    =
    e^{\mathrm{i}\theta'}
    P_{\pi^{-1}}
    (L_1^\dagger\otimes\cdots\otimes L_6^\dagger)
    \ket{H_6}
\end{equation}
for some $\theta'\in\mathbb{R}$. Hence every $\mathrm{AME}(6,2)$ state
is locally unitarily equivalent, up to a permutation of the qubits and
a global phase, to the quantum hexacode state $\ket{H_6}$. 
This completes the proof.
    
\end{proof}

\section{Proofs of Some Lemmas in Preliminaries}\label{app: disequality reduction}
This appendix provides the deferred proofs of the preliminary lemmas.

\subsection{Proof of Some Lemmas in~\cref{subsec: tractable signatures}}\label{subapp: proof of subsec tractable signature}
This subappendix provides the deferred proof of the preliminary lemmas in~\cref{subsec: tractable signatures}.

\InvarianceUnderRealOrth*
\begin{proof}[Proof of~\cref{lem: invariance under real orthogonal transformation}]
    We prove the three statements separately.

    For the first statement, let $f\in \F$ has arity $n$. Since $O$ is real orthogonal matrix,
    we have 
    \begin{equation}
        \overline{Of}
        =\overline{O}^{\otimes n}\overline f 
        =O^{\otimes n}\overline f 
        =O\overline f.
    \end{equation}
    Therefore, if $\F$ is conjugate closed, then $\overline{f}\in \F$. Thus, $\overline{Of}=O\overline f\in O\F$. The converse follows by applying the same argument to the real orthgonal matrix $O^{-1}$. Hence, the first statement is hold.

    For the second statement, since holographic transformation preserves tensor decompositions into unary and binary signatures, let $f=f_1\otimes \cdots \otimes f_k$ where $f_i$ has arity at most two, then $Of=(Of_1)\otimes \cdots \otimes (Of_k)$.
    Applying the same argument to $O^{-1}$ gives
\[
    \mathcal F\subseteq\mathscr T
    \quad\Longleftrightarrow\quad
    O\mathcal F\subseteq\mathscr T.
\]
Now let
$
    \mathscr C\in\{\mathscr P,\mathscr A,\mathscr L\},
$
and suppose that $\mathcal F$ is $\mathscr C$-transformable, witnessed
by some $T\in\mathrm{GL}_2(\mathbb C)$. We claim that $TO^{-1}$
witnesses that $O\mathcal F$ is $\mathscr C$-transformable. First,
\[
    (TO^{-1})(O\mathcal F)
    =T\mathcal F
    \subseteq\mathscr C.
\]
Moreover, since $O$ is orthogonal, the binary equality signature is
invariant under $O$:
\[
    (=_2)O^{\otimes 2}=(=_2).
\]
It follows that
\[
\begin{aligned}
    (=_2)\bigl((TO^{-1})^{-1}\bigr)^{\otimes 2}
    &=(=_2)(OT^{-1})^{\otimes 2}\\
    &=\bigl((=_2)O^{\otimes 2}\bigr)
      (T^{-1})^{\otimes 2}\\
    &=(=_2)(T^{-1})^{\otimes 2}
      \in\mathscr C.
\end{aligned}
\]
Hence $O\mathcal F$ is $\mathscr C$-transformable. Applying the same
argument to $O^{-1}$ proves the converse. Therefore, all four
alternatives in condition \eqref{cond: tractable classes} are
invariant under $O$.

For the third statement, let $f$ be a nonzero signature. Suppose that $f=g\otimes h$, since the holographic transformation preserves tensor decompositions, we have $Of=Og\otimes Oh$. Thus, reducibility of $f$ implies reducibility of $Of$. Conversely, if $Of=g'\otimes h'$, then $f=(O^{-1}g')\otimes (O^{-1} h')$. Hence $f$ is reducible if and only if $Of$ is reducible, and $f$ is irreducible if and only if $Of$ is irreducible.
\end{proof}

\subsection{Proof of Some Lemmas in~\cref{subsec: hardness results}}
This subappendix provides the deferred proofs of the preliminary lemmas in~\cref{subsec: hardness results}
on tensor factorization and hardness from odd-arity signatures or
$\neq_4$, which are used in the orthogonality arguments and the
low-arity base cases of the main proof.

\CSPwithEQ*
\begin{proof}
    Let $\mathcal{G}=T(\F\cup \{=_2\})$.
    Suppose that $\#\operatorname{CSP}_2(T(\F\cup \{=_2\}))$ is not $\#\operatorname{P}$-hard, then by~\cite{real_holantc}, $\mathcal{G}\subseteq \mathscr{A},\mathscr{P},\mathscr{L}$ or $\mathcal{G}\subseteq \alpha\mathscr{A}$.

    We consider the first three cases. 
    
    If $\mathcal{G}\subseteq \mathscr{A}$, then $T\F\subseteq \mathscr{A}$ and $T(=_2)=TT^{\mathsf T}\in \mathscr{A}$. Let $E=TT^{\mathsf T}\in \mathscr{A}$. Then we can get $(=_2)T^{-1}=(T^{-1})^{\mathsf T}T^{-1}=(TT^{\mathsf T})^{-1}=E^{-1}$, thus $E$ is invertible. By~\cref{lm: affine singular value}, $E$'s two row has the same norm and either the two rows are proportional or orthogonal. The proportional property gives $\operatorname{rank}(E)=1$, which contradiction to the invertibility of $E$. Hence, $EE^\dagger=\rho I_2$ and $E^{-1}=\rho^{-1}E^\dagger$. Since $\mathcal{A}$ is closed under conjugation and transposition, $E^{-1}\in \mathscr{A}$. Therefore, $(=_2)T^{-1}\in \mathscr{A}$ and $\mathcal{F}\subseteq \mathscr{A}$-transformable. This contradicts the assumption that $\F$ does not satisfy condition~\eqref{cond: tractable classes}.

    If $\mathcal{G}\subseteq \mathscr{P}$, then $T\F\subseteq \mathscr{P}$ and $T(=_2)=TT^{\mathsf T}\in \mathscr{P}$. Let $E=TT^{\mathsf T}\in \mathscr{P}$. Then we can get $(=_2)T^{-1}=(T^{-1})^{\mathsf T}T^{-1}=(TT^{\mathsf T})^{-1}=E^{-1}$, thus $E$ is invertible. If the two variables of $E$ are independent, then $\operatorname{rank}(E)=1$ which contradicts to the invertibility of $E$. Thus, either $E$ is diagonal matrix or antidiagonal matrix, both give that $E^{-1}\in \mathscr{P}$. Therefore, $(=_2)T^{-1}\in \mathscr{P}$ and $\mathcal{F}\subseteq \mathscr{P}$-transformable. This contradicts the assumption that $\F$ does not satisfy condition~\eqref{cond: tractable classes}.

    If $\mathcal{G}\subseteq \mathscr{L}$, then $T\F\subseteq \mathscr{L}$ and $T(=_2)=TT^{\mathsf T}\in \mathscr{L}$. Let $E=TT^{\mathsf T}\in \mathscr{L}$. Then we can get $(=_2)T^{-1}=(T^{-1})^{\mathsf T}T^{-1}=(TT^{\mathsf T})^{-1}=E^{-1}$, thus $E$ is invertible. Let $\alpha=e^{\frac{\pi \mathfrak i}{4}}$ and $D:=\left( \begin{matrix}
        1 & 0\\
        0 & \alpha 
    \end{matrix}\right).$ The definition of $\mathscr{L}$ means that for any $(s_1,s_2)\in \mathscr{S}(E)$, we have $\left(\left( \begin{matrix}
        1 & 0\\
        0 & \alpha^{s_1}
    \end{matrix}\right)\otimes \left( \begin{matrix}
        1 & 0\\
        0 & \alpha^{s_2}
    \end{matrix}\right)\right)E\in \mathscr{A}$.

    If $00\in \mathscr{S}(E)$, then $E\in \mathscr{A}$. If $01\in \mathscr{S}(E)$, then $(I_2\otimes D)E\in \mathscr{A}$. If $10\in \mathscr{S}(E)$, then $(D\otimes I_2)E\in \mathscr{A}$. If $11\in \mathscr{S}(E)$, then $(D\otimes D)E\in \mathscr{A}$. This diagonal matrix does not change the support, thus the support of $E$ is also affine. If $|\mathscr{S}(E)|=1$, then contradicts to the invertibility of $E$. If $|\mathscr{S}(E)|=2$ and the nonzero position in the same row or column, then contradicts to invertibility. If $\mathscr{S}(E)=\{00,01,10,11\}$, then $E(x,y)=\lambda\mathfrak i^{Q(x,y)}$. And we have $E(x,y)/E(x',y')\in \{1,\mathfrak i,-1,\mathfrak{-i}\}$. In particular, $E(0,1)/E(0,0)=\mathfrak i^k$ for some $k$. Since $01\in \mathscr{S}(E)$, $(I_2\otimes D)E\in \mathscr{A}$ gives that $E'=\left( \begin{matrix}
        E(0,0) & \alpha E(0,1)\\
        E(1,0) & \alpha E(1,1)
    \end{matrix}\right)\in \mathscr{A}$. Thus, $\alpha E(0,1)/E(0,0)=\alpha \mathfrak i^k\notin \{1,\mathfrak i,-1,\mathfrak{-i}\}$. This contradicts to $E'\in \mathscr{A}$. If $\mathscr{S}(E)=\{01,10\}$, then we can suppose that $E=\left(\begin{matrix}
        0 & b\\
        b & 0
    \end{matrix}\right)$ since $E$ is symmetric. Since $E'=(I_2\otimes D)E=\left( \begin{matrix}
        0 & \alpha b\\
        b & 0
    \end{matrix}\right)\in \mathscr{A}$, but $\alpha b/b=\alpha\notin \{1,\mathfrak i,-1,\mathfrak{-i}\}$. This contradicts to $E'\in\mathscr{A}$.
    Lastly, $E=\left( \begin{matrix}
        a & 0\\
        0 & b
    \end{matrix}\right),$ and $b/a=\mathfrak i^k$. This gives $E=\lambda \operatorname{diag}(1,\mathfrak i^k)\in \mathscr{A}$. Also we can get $E^{-1}=\lambda^{-1}\operatorname{diag}(1,\mathfrak i^{-k})\in \mathscr{A}$.
    Finally, we can get that $\F\in \mathscr{L}$-transformable, which contradicts to $\F$ does not satisfies the condtion~\eqref{cond: tractable classes}.

    Next, we consider the case that $\mathcal{G}\subseteq \alpha \mathscr{A}$. Let $D:=\operatorname{diag}(1,\alpha)$, $E=TT^{\mathsf T}$ and $T'=D^{-1}T$. Since $T\F\subseteq \alpha \mathscr{A},T(=_2)\in \alpha \mathscr{A}$, we have $T'\F=D^{-1}T\F\in \mathscr{A}$ and $D^{-1}ED^{-1}\in \mathscr{A}$. Thus, after a holographic transformation $D^{-1}$, we get that $\#\operatorname{CSP}_2(D^{-1}T(\F\cup \{=_2\}))$ where $D^{-1}T(\F\cup \{=_2\})\in \mathscr{A}$, this case is equivalent to the first case $\mathcal{G}\in \mathscr{A}$.

    Eventually, $\#\operatorname{CSP}_2(T(\F\cup \{=_2\}))$ is $\#\operatorname{P}$-hard.
\end{proof}

\DisFourgivesHard*
\begin{proof}[Proof of~\cref{lm: disequality 4 gives dichotomy}]

    Let $g_4:=K^{\otimes 4}\neq_4$. We may discard zero signatures and handle nullary signatures as scalar factors. We divide the proof into three cases.

    \textbf{Case 1: $\F$ contains a nonzero signature of odd arity.}

    By the holographic transformation $K$, $\holant{\neq_2}{\neq_4,\hF}\equiv_T \holant{=_2}{g_4,\F}$, and by~\Cref{thm: holant odd} together with~\cite{GSS26}, $\holant{=_2}{g_4,\F}$ is either polynomial-time solvable or is $\#\operatorname{P}$-hard.

    \textbf{Case 2: every signature in $\hF$ is an $\eo$ signature.}
    Since $\neq_4$ is also an $\eo$ signature, then the problem $\holant{\neq_2}{\neq_4,\hF}$ is a $\ceo$ problem, namely $\holant{\neq_2}{\neq_4,\hF}\equiv_T \ceo(\neq_4,\hF)$. Then by~\Cref{thm: dichotomy for EO} together with~\cite{GSS26}, $\ceo(\neq_4,\hF)$ is either polynomial-time solvable or $\#\operatorname{P}$-hard.

    \textbf{Case 3: a non-$\eo$ signature occurs.}

    Let $\widehat{f}\in \hF$ have arity $2n$ and suppose that $\widehat{f}$ is not an $\eo$ signature. Then there exists $\alpha\in\mathrm{supp}(\widehat{f})$ such that $\mathrm{wt}(\alpha)=k\neq n$. We first show that, if $k<n$, we can realize a signature  $\widehat{g}$ of arity $r=2n-2k>0$ such that $\widehat{g}(\vec{0})\neq 0$, if $k>n$, we can realize a signature  $\widehat{g}$ of arity $r=2k-2n>0$ such that $\widehat{g}(\vec{1})\neq 0$. This follows from~\cite{Holant_Odd}.

    We treat the case $k<n$ firstly. Recall that we realized a signature $\widehat{g}$ of arity $r=2n-2k$ such that $\widehat{g}(\vec{0})\neq 0$. Let $a:=\widehat{g}(\vec{0})\neq 0$ and $b:=\widehat{g}(\vec{1})$. By the gadget construction in~\cite{cai-fu-shao-eo}, $\neq_{2r}$ can be realized on the RHS of $\holant{\neq_2}{\neq_4,\hF}$. We permute its variables so that $\mathrm{supp}(\neq_{2r})=\{0^r1^r,1^r0^r\}$. Then we connect the $r$ variables of $\widehat{g}$ with the first $r$ variables of $\neq_{2r}$ using $\neq_2$, leaving the last $r$ variables of $\neq_{2r}$ dangling. The resulting signature $\widehat{h}$ of arity $r$ is $\widehat{h}=[a,0,\dots,0,b]_r$. Since $\widehat{h}$ is realized on the RHS of $\holant{\neq_2}{\neq_4,\hF}$, we have $\holant{\neq_2}{\neq_4,\hF}\equiv_T \holant{\neq_2}{\neq_4,\widehat{h},\hF}$.

    We now distinguish whether $b$ is zero.
    
    \textbf{Subcase 3.1: $b=0$.}

    Then $\widehat{h}=[a,0,\dots,0]_r=a\Delta_0^{\otimes r}=(\sqrt[r]a \Delta_0)^{\otimes r}$. By~\Cref{lm: decomposition of tensor power}, we can get $\holant{\neq_2}{\neq_4,\hF}\equiv_T \holant{\neq_2}{\neq_4,\widehat{h},\hF}\equiv_T \holant{\neq_2}{\neq_4,\Delta_0,\hF}$. By a holographic transformation $K$, we can get $\holant{\neq_2}{\neq_4,\Delta_0,\hF}\equiv_T \holant{=_2}{g_4,u_0,\F}$ where $u_0=K\Delta_0=\frac{1}{\sqrt{2}}\begin{bmatrix}1\\ \mathfrak i\end{bmatrix}.$ Thus a nonzero odd arity signature is available, and by~\Cref{thm: holant odd} together with~\cite{GSS26}, $\holant{=_2}{g_4,u_0,\F}$ is either polynomial time solvable or $\#\operatorname{P}$-hard.

    \textbf{Subcase 3.2: $b\neq 0$.}

    Now $ab\neq 0$. We first suppose that $r\geq 4$. Let $q\in \mathbb{C}$ satisfying $q^{2r}=\frac{b}{a}$. Define 
    \[
        \widehat{Q}^{-1}
        =
        \begin{pmatrix}
        q^{-1}&0\\
        0&q
        \end{pmatrix},
        \qquad
        \widehat Q
        =
        \begin{pmatrix}
        q&0\\
        0&q^{-1}
        \end{pmatrix}.
    \]
    We have 
    \[
        (\widehat{Q}^{-1})^{\mathsf T}N_2\widehat Q^{-1}
        =
        \begin{pmatrix}
        q^{-1}&0\\
        0&q
        \end{pmatrix}
        \begin{pmatrix}
        0&1\\
        1&0
        \end{pmatrix}
        \begin{pmatrix}
        q^{-1}&0\\
        0&q
        \end{pmatrix}
        =N_2.
    \]
    By a holographic transformation $\widehat{Q}$,
    \[
    \begin{aligned}
        &\holant{\neq_2}{\neq_4,\widehat{h},\hF}\\
        &\equiv_T \holant{(\widehat{Q}^{-1})^{\mathsf T}N_2\widehat{Q}^{-1}}{\widehat{Q}\neq_4,\widehat{Q}\widehat{h},\widehat{Q}\hF}\\
        &\equiv_T \holant{\neq_2}{\widehat{Q}\neq_4,\widehat{Q}\widehat{h},\widehat{Q}\hF}
        \equiv_T \holant{\neq_2}{\neq_4,\widehat{Q}\widehat{h},\widehat{Q}\hF}.
    \end{aligned}
    \]
    Furthermore,
    \[
    \begin{aligned}
        \widehat{Q}\widehat{h}
        &=\widehat{Q}^{\otimes r}\widehat{h}
        =[aq^r,0,\dots,0,bq^{-r}]_r\\
        &=aq^r[1,0,\dots,0,1]_r=aq^r(=_r)
    \end{aligned}
    \]
    since $q^{2r}=b/a$. Thus $\holant{\neq_2}{\neq_4,\widehat{h},\hF}\equiv_T \holant{\neq_2}{\neq_4,=_r,\widehat{Q}\hF}$. By\cite{odd_holant_real}, we have
    \[
        \#\operatorname{CSP}_r(\neq_2,\neq_4,\widehat{Q}\hF)
        \equiv_T \holant{\neq_2}{\neq_4,=_r,\widehat{Q}\hF}.
    \]
    Then by~\cite{odd_holant_real}, $\#\operatorname{CSP}_r(\neq_2,\neq_4,\widehat{Q}\hF)$ is either polynomial-time solvable or $\#\operatorname{P}$-hard.

    It remains to consider $r=2$. In this case $\widehat{h}=[a,0,0,b]$. Recall that $\neq_8$ can be realized on the RHS of $\holant{\neq_2}{\neq_4,\widehat{h},\hF}$. Permuting the variables of $\neq_8$ so that $\mathrm{supp}(\neq_8)=\{0^41^4,1^40^4\}$.
    Connecting the first four variables of $\neq_8$ with two copies of $\widehat{h}$ using $\neq_2$, leaving the last four variables of $\neq_8$ dangling. The resulting quaternary signature is $\widehat{h}_4=[a^2,0,\dots,0,b^2]_4$. Since $a^2b^2\neq 0$, the preceding argument applies with $r=4$. This completes the case $k<n$.

    The case $k>n$ is similar.
\end{proof}

\ConjuDisFourHard*
\begin{proof}[Proof of~\cref{lem: In conjugate closed setting disequality 4 gives hardness}]
     By~\Cref{lm: disequality 4 gives dichotomy}, we just need to verify that $\holant{\neq_2}{\neq_4,\hF}$ will not fall into the polynomial-time solvable side of case 1, case 2, and case 3.

     \textbf{Case 1: $\F$ contains a nonzero signature of odd arity.}

     By~\Cref{lm: odd holant with conjugation}, this case is $\#\operatorname{P}$-hard.

     \textbf{Case 2: every signature in $\hF$ is an $\eo$ signature.}

     Suppose that $\hF\subseteq\POLUP$ and either $\hF\subseteq \eo^{\mathscr{A}}$ or $\hF\subseteq \eo^{\mathscr{P}}$. Since $\neq_4\in \POLUP\cap\eo^{\mathscr{A}}\cap\eo^{\mathscr{P}}$, then by~\Cref{lm: conjugation_polymorphism_affine support} we can get $\neq_4$ and every $\widehat{f}\in\hF$ has affine support. By~\cite{cai-fu-shao-eo}, these signatures are pairwise opposite. Combined with the definition of $\eo^{\mathscr{A}}$ and $\hF\subseteq \eo^{\mathscr{P}}$, we can get $\{\neq_4\}\cup \hF\subseteq \mathscr{A}$ or $\{\neq_4\}\cup\hF\subseteq\mathscr{P}$. Since $K^{\mathsf T}K=N_2\in\mathscr{A}\cap\mathscr{P}$, the transformation $K^{-1}$ witnesses $\mathscr{A}$- or $\mathscr{P}$-transformability of $\F$. This contradicts that $\F$ does not satisfy the condition \rm{(\ref{cond: tractable classes})}. The case $\hF\subseteq\POLDOWN$ and either $\hF\subseteq \eo^{\mathscr{A}}$ or $\hF\subseteq \eo^{\mathscr{P}}$ is symmetric. Thus this case is $\#\operatorname{P}$-hard.

     \textbf{Case 3: a non-$\eo$ signature occurs.}
     If exactly one of $a,b$ is nonzero in the preceding construction, we can get realize a nonzero unary signature, this case is $\#\operatorname{P}$-hard by~\Cref{lm: odd holant with conjugation}. Otherwise, by~\cite{odd_holant_real}, if $\F$ does not satisfy the condition \rm{(\ref{cond: tractable classes})}, $\#\operatorname{CSP}_r(\neq_2,\neq_4,\widehat{Q}\hF)$ is $\#\operatorname{P}$-hard for $r\geq 3$; when the original $r=2$, we first use the preceding quaternary construction and take $r=4$. Thus this case is also $\#\operatorname{P}$-hard.
\end{proof}

\subsection{Proof of Some Lemmas in~\cref{subsec: signatures and quantum states}}\label{subapp: proof of subsec signatures and quantum states}

\begin{proof}[Proof of~\cref{prop: invariant kth Orth}]
    For every subset $S\subseteq [n]$ of size $k$, if $f$ satisfies {\sc $k$th-Orth}, then there exists some $\mu\neq 0$ such that 
    \begin{equation}
        M_S(f)M_S(f)^\dagger=\mu I_{2^k}.
    \end{equation}

    We first prove invariance under conjugation. Since $M_S(\overline f)=\overline{M_S(f)}$, we have 
    \begin{equation}
        \begin{aligned}
            M_S(\overline f)M_S(\overline f)^\dagger
            &=\overline{M_S(f)}\,M_S(f)^{\mathsf T}\\
            &=\overline{M_S(f)M_S(f)^\dagger}\\
            &=\overline{\mu I_{2^k}}\\
            &=\mu I_{2^k}.
        \end{aligned}
    \end{equation}
   Therefore, $\overline{f}$ satisfies {\sc $k$th-Orth}. 

   For avoiding conflict between transpose $\mathsf{T}$ and holographic transformation $T$, we use $U$ instead of the holographic transformation $T$.
   We next prove invariance under uniatry holographic transformation $U\in \mathbf{U}_2(\mathbb{C})$. We know that $M_S(Uf)=U^{\otimes k}M_S(f)(U^{\mathsf{T}})^{\otimes(n-k)}$. Since $T$ is unitary, we have $U^{\mathsf{T}}\overline U=I_2$. Then we can get that 
   \begin{equation}
    \begin{aligned}
        M_S(Uf)M_S(Uf)^\dagger
        &= U^{\otimes k}M_S(f) (U^{\mathsf{T}})^{\otimes (n-k)}\overline{U^{\otimes (n-k)}}M_S(f)^\dagger (U^\dagger)^{\otimes k}\\
        &=U^{\otimes}M_S(f)M_S(f)^\dagger (U^\dagger)^{\otimes k}\\
        &=\mu U^{\otimes}I_{2^k} (U^\dagger)^{\otimes k}\\
        &=\mu I_{2^k}.
    \end{aligned}
   \end{equation}
   Thus, $Uf$ satisfies {\sc $k$th-Orth}.
\end{proof}

\end{document}